\documentclass[11pt, a4paper]{article}
\usepackage{url}
\usepackage{a4wide}
\usepackage{graphicx}
\usepackage{natbib}
\usepackage{enumitem}
\usepackage{hyperref}
\usepackage{setspace}
\usepackage{multirow, booktabs}
\usepackage{amsmath, amsthm, amssymb, amsfonts}
\usepackage{bm, bbm}
\usepackage[T1]{fontenc}
\usepackage[ruled, vlined]{algorithm2e}
\usepackage{color}
\usepackage{subcaption}

\allowdisplaybreaks

\DeclareMathAlphabet\mathbfcal{OMS}{cmsy}{b}{n}

\newcommand{\mbf}{\mathbf}
\newcommand{\mc}{\mathcal}
\newcommand{\bmx}{\begin{bmatrix}}
\newcommand{\emx}{\end{bmatrix}}

\newcommand{\vep}{\varepsilon}

\renewcommand{\l}{\left}
\renewcommand{\r}{\right}

\def\wh{\widehat}
\def\wt{\widetilde}

\newcommand{\E}[0]{\mathsf{E}}
\newcommand{\Var}[0]{\mathsf{Var}}
\newcommand{\Cov}[0]{\mathsf{Cov}}
\newcommand{\tr}[0]{\mathsf{tr}}
\newcommand{\p}{\mathsf{P}}
\newcommand{\R}{\mathbb{R}}
\newcommand{\Z}{\mathbb{Z}}
\newcommand{\N}{\mathbb{N}}
\newcommand{\C}{\mathbb{C}}
\newcommand{\iid}{\text{\upshape iid}}
\newcommand{\nn}{\nonumber}

\newcommand{\static}{\text{\upshape {res}}}
\newcommand{\va}{\text{\upshape{unr}}}

\newcommand{\las}{\text{\upshape{las}}}
\newcommand{\ds}{\text{\upshape{DS}}}
\newcommand{\dir}{\text{\upshape{G}}}
\newcommand{\undir}{\text{\upshape{C}}}
\newcommand{\lr}{\text{\upshape{L}}}

\newcommand{\bbG}{\mathbbm{G}}
\newcommand{\bbg}{\mathbbm{g}}

\theoremstyle{definition}
\newtheorem{thm}{Theorem}[section]
\theoremstyle{definition}
\newtheorem{cor}[thm]{Corollary}
\theoremstyle{definition}
\newtheorem{lem}[thm]{Lemma}
\theoremstyle{definition}
\newtheorem{prop}[thm]{Proposition}
\theoremstyle{definition}
\newtheorem{assum}{Assumption}[section]
\theoremstyle{remark}
\newtheorem{rem}{Remark}[section]
\theoremstyle{definition}

\theoremstyle{definition}

\title{FNETS: Factor-adjusted network estimation and forecasting for high-dimensional time series}
\author{Matteo Barigozzi$^1$ and Haeran Cho$^2$ and Dom Owens$^3$}

\begin{document}

\maketitle

\footnotetext[1]{Department of Economics, Universit\`a di Bologna.
Email: \url{matteo.barigozzi@unibo.it}.
Supported by MIUR (PRIN 2017, Grant 2017TA7TYC).}

\footnotetext[2]{School of Mathematics, University of Bristol.
Email: \url{haeran.cho@bristol.ac.uk}.
Supported by the Leverhulme Trust (RPG-2019-390).}

\footnotetext[3]{School of Mathematics, University of Bristol.
Email: \url{dom.owens@bristol.ac.uk}.}


\begin{abstract}
We propose FNETS, a methodology for network estimation and forecasting of high-dimensional time series exhibiting strong serial- and cross-sectional correlations. We operate under a factor-adjusted vector autoregressive (VAR) model which, after accounting for pervasive co-movements of the variables by {\it common} factors, models the remaining {\it idiosyncratic} dynamic dependence between the variables as a sparse VAR process. Network estimation of FNETS consists of three steps: (i) factor-adjustment via dynamic principal component analysis, (ii) estimation of the latent VAR process via $\ell_1$-regularised Yule-Walker estimator, and (iii) estimation of partial correlation and long-run partial correlation matrices. In doing so, we learn three networks underpinning the VAR process, namely a directed network representing the Granger causal linkages between the variables, an undirected one embedding their contemporaneous relationships and finally, an undirected network that summarises both lead-lag and contemporaneous linkages. In addition, FNETS provides a suite of methods for forecasting the factor-driven and the idiosyncratic VAR processes. Under general conditions permitting tails heavier than the Gaussian one, we derive uniform consistency rates for the estimators in both network estimation and forecasting, which hold as the dimension of the panel and the sample size diverge. Simulation studies and real data application confirm the good performance of FNETS.
\end{abstract}

\noindent%

\section{Introduction}

Vector autoregressive (VAR) models are popularly adopted for time series analysis in economics and finance.
Fitting a VAR model to the data enables inferring dynamic interdependence between the variables as well as forecasting the future. 
VAR models are particularly appealing for network analysis 
since estimating the non-zero elements of the VAR parameter matrices, a.k.a.\ transition matrices,
recovers directed edges between the components of vector time series
in a Granger causality network. 
In addition, by estimating a precision matrix
(inverse of the covariance matrix) of the VAR innovations, 
we can also define a network capturing contemporaneous linear dependencies. 
For the network interpretation of VAR modelling, see e.g.\ 
\citet{dahlhaus2000graphical, eichler2007granger, billio2012econometric, ahelegbey2016bayesian, barigozzi2019, gudhmundsson2021detecting, uematsu2023discovering}.

Estimation of VAR models quickly becomes a high-dimensional problem
as the number of parameters grows quadratically with the dimensionality.
There is a mature literature on estimation of high-dimensional VAR models under the sparsity 
\citep{hsu2008subset, song2011large, basu2015, han2015direct, kock2015oracle, barigozzi2019, nicholson2020high} and low-rank plus sparsity \citep{basu2019} assumptions,
see also \citet{dgr08} and \cite{banbura2010large} for Bayesian approaches.
In all above, either explicitly or implicitly, the spectral density of the time series is required to have eigenvalues which are uniformly bounded over frequencies. 
Indeed, this condition is crucial for controlling the deviation bounds involved in theoretical investigation of regularised estimators.

\begin{figure}[htbp]
\centering
\begin{tabular}{ccc}
\includegraphics[width = .3\textwidth]{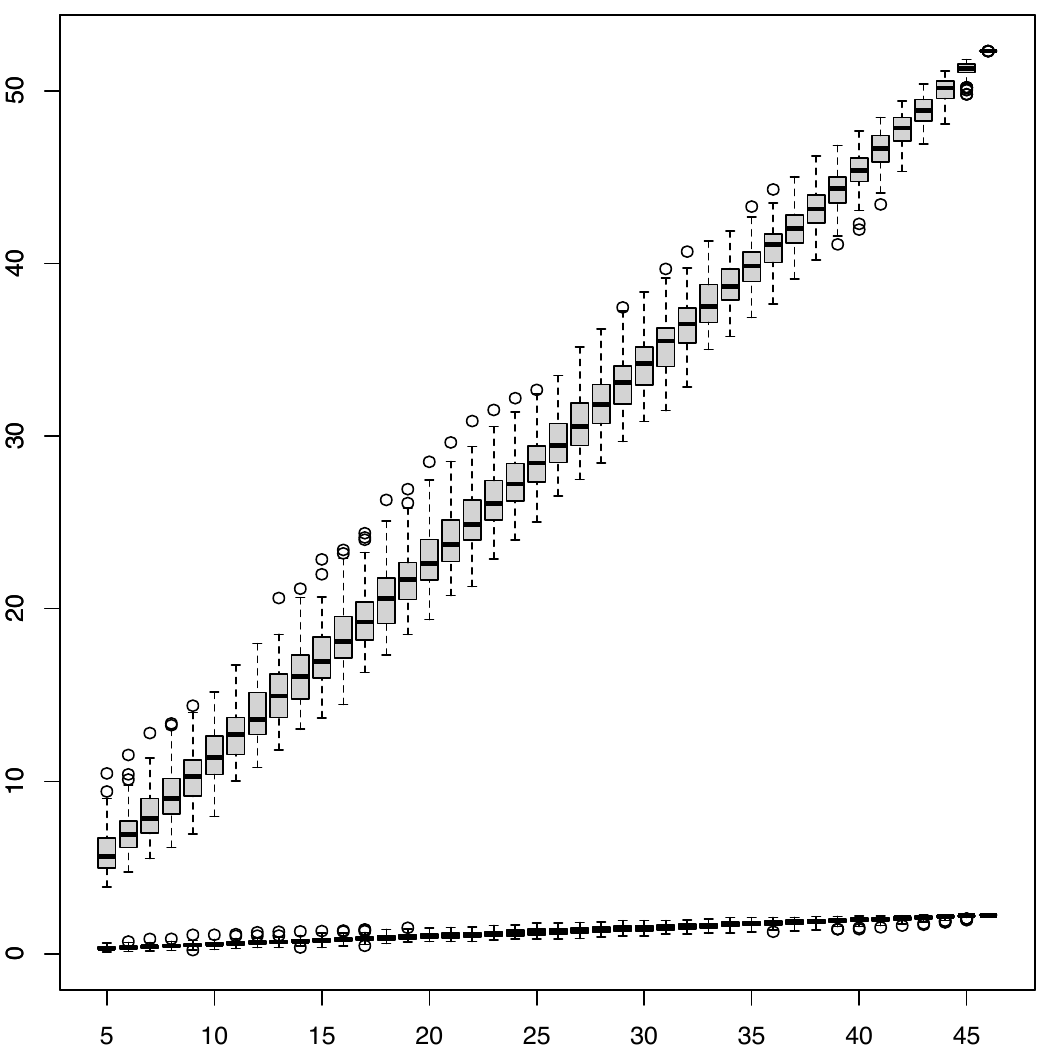} &
\includegraphics[width = .35\textwidth]{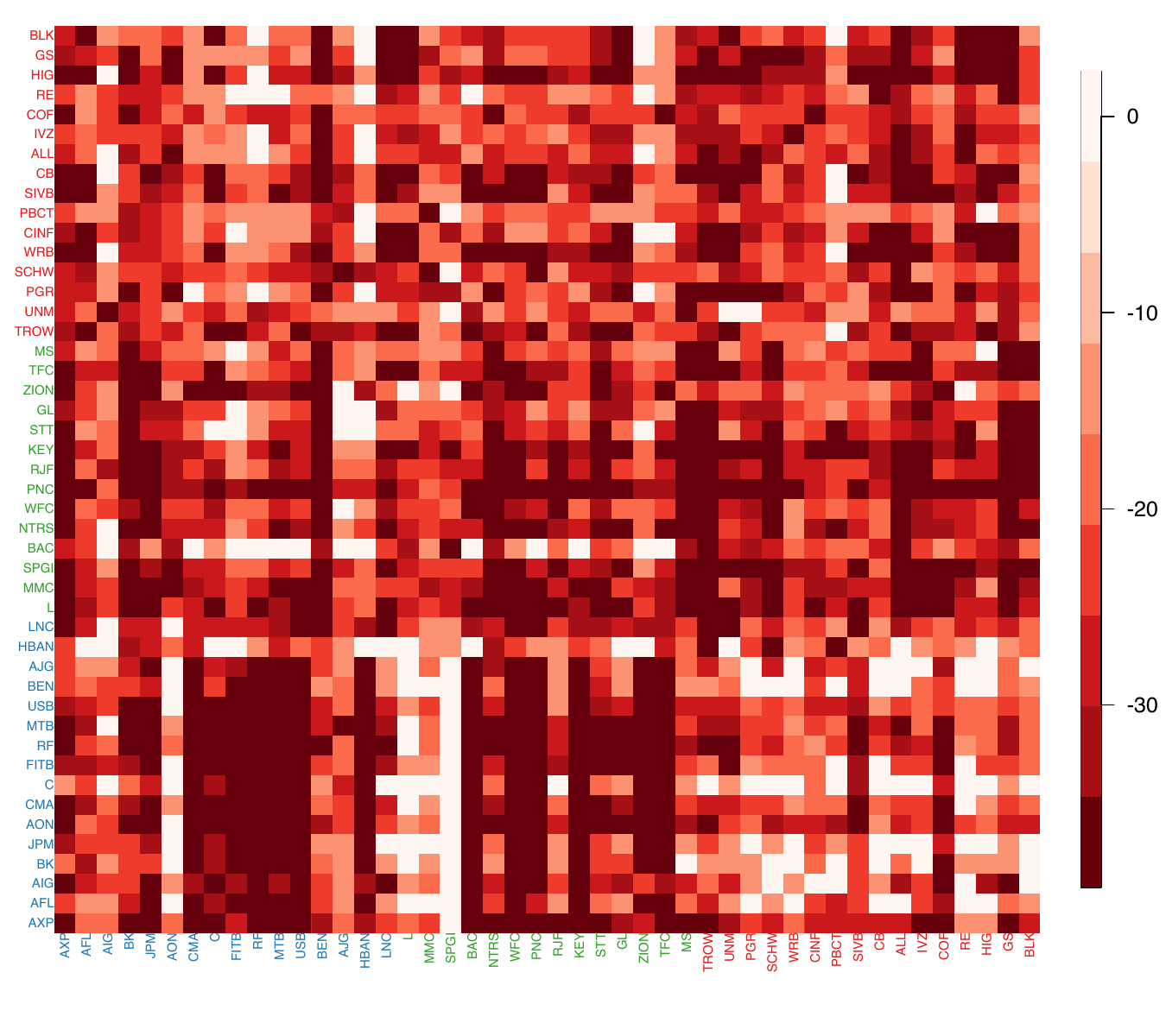} 
\end{tabular}
\caption{\small Left: The two largest eigenvalues ($y$-axis) of the long-run covariance matrix
estimated from the volatility panel analysed in Section~\ref{sec:real}
(March $2008$ to March $2009$, $n = 252$)
with subsets of cross-sections randomly sampled $100$ times
for each given dimension $p \in \{5, \ldots, 46\}$ ($x$-axis).
Right: logged and truncated $p$-values 
(truncation level chosen by Bonferroni correction with the significance level $0.1$)
from fitting a VAR($5$) model to the same dataset using ridge regression and generating $p$-values corresponding to each coefficient as described in \cite{cule2011significance}.
For each pair of variables (corresponding tickers given in $x$- and $y$-axes), the minimum $p$-value over the five lags is reported.}
\label{fig:ex}
\end{figure}

\cite{lin2020} observe that for VAR processes, this assumption restricts the parameters to be either dense but small in their magnitude (which makes their estimation using the shrinkage-based methods challenging) or highly sparse,
while \cite{giannone2021economic} note the difficulty of identifying sparse predictive representations in many economic applications.
Moreover, some datasets typically exhibit strong serial and cross-sectional correlations and violate the bounded spectrum assumption.
The left panel of Figure~\ref{fig:ex} provides an illustration of this phenomenon;
with the increase of dimensionality, a volatility panel dataset (see Section~\ref{sec:real} for its description) exhibits a linear increase in the leading eigenvalue of the estimate of its spectral density matrix at frequency~$0$ (i.e.\ long-run covariance).
The right panel visualises the outcome from fitting a VAR($5$) model to the same dataset
without making any adjustment of the strong correlations (see the caption for further details),
from which we cannot infer meaningful, sparse pairwise relationship.


In this paper, we propose to model high-dimensional time series by means of a factor-adjusted VAR approach, which simultaneously accounts for strong serial and cross-sectional correlations attributed to factors, as well as sparse, idiosyncratic correlations among the variables that remain after factor adjustment. 
We take the most general approach to factor modelling based on the generalised dynamic factor model, where factors are dynamic in the sense that they are allowed to have not only contemporaneous but also lagged effects on the variables \citep{FHLR00}.
We propose FNETS, a suite of tools accompanying the model for estimation and forecasting with a particular focus on network analysis, which addresses the challenges arising from the latency of the VAR process as well as high dimensionality. 

We make the following methodological and theoretical contributions.
\begin{enumerate}[noitemsep, wide, labelwidth=!, labelindent=0pt, label = (\alph*)]
\item We propose an $\ell_1$-regularised Yule-Walker estimation method for estimating the factor-adjusted, idiosyncratic VAR, while permitting the number of non-zero parameters to slowly grow with the dimensionality.
Estimating the VAR parameters and the inverse of the innovation covariance, and then combining them allow us to define three networks underlying the latent VAR process, namely a direct network representing Granger causal linkages, an undirected one underpinning their contemporaneous relationships, as well as an undirected network summarising both. 
Under general conditions permitting weak factors and heavier tails than the sub-Gaussian one, we show the consistency of FNETS in estimating the edge sets of these networks, which holds uniformly over all $p^2$ entries of the networks (Propositions~\ref{prop:idio:est:lasso} and~\ref{prop:idio:delta}).

\item We provide new consistency rates for the estimation and forecasting approaches considered by \citet{forni2005generalized, forni2017dynamic}, which hold uniformly for the entire cross-sections of $p$-dimensional time series (Propositions~\ref{thm:2005} and~\ref{thm:common:var}).
In doing so, we establish uniform consistency of the estimators of high-dimensional spectral density matrices of the factor-driven and the idiosyncratic components, extending the results of \citet{zhang2021} to the presence of latent factors.
\end{enumerate}

Our approach differs from the existing ones for factor-adjusted regression problems \citep{fan2020factor, fan2021bridging, fan2023latent, krampe2021}, 
as: (i)~it allows for the presence of dynamic factors, thus including all possible dynamic linear co-dependencies, and (ii)~it relies only on the estimators of the autocovariances of the latent idiosyncratic process, and avoids estimating the entire latent process and controlling the errors arising from such a step, which increase with the sample size. 
The price to pay for the generality of the factor modelling in~(i), is an extra term appearing in the rate of consistency which represents the bandwidth for spectral density estimation required for factor-adjustment in the frequency domain. 
We make explicit the role played by this bandwidth in the theoretical results, and also present the results under a more restricted static factor model for ease of comparison.
We mention two more differences between this paper and \cite{fan2021bridging, fan2023latent}. 
First, they additionally consider the problem of testing hypotheses on the idiosyncratic covariance and the adequacy of factor/sparse regression, while we focus on network estimation. 
Secondly, their methods accommodate models for the idiosyncratic component other than VAR.

FNETS is another take at the popular low-rank plus sparsity modelling framework 
in the high-dimensional learning literature, 
Also, it is in line with a frequently adopted practice in financial time series analysis
where factor-driven common components representing the systematic sources of risk, 
are removed prior to inferring a network structure via (sparse) regression modelling
and identifying the most central nodes representing the systemic sources of risk \citep{diebold2014network, barigozzi2019}.
We provide a rigorous theoretical treatment of this empirical approach
by accounting for the effect of the factor-adjustment step on the second step regression.



The rest of the paper is organised as follows.
Section~\ref{sec:model} introduces the factor-adjusted VAR model.
Sections~\ref{sec:fnets} and~\ref{sec:forecast} describe 
the network estimation and forecasting methodologies comprising FNETS, respectively,
and provide their theoretical consistency.
In Section~\ref{sec:num}, we demonstrate the good estimation and forecasting performance of 
FNETS on a panel of volatility measures.
Section~\ref{sec:conc} concludes the paper,
and all the proofs and complete simulation results are presented in Supplementary Appendix.
The R software {\tt fnets} implementing FNETS is available from CRAN \citep{fnets}.


\paragraph{Notations.}
By $\mbf I$, $\mbf O$ and $\mbf 0$, we denote an identity matrix, a matrix of zeros and a vector of zeros whose dimensions depend on the context. 
For a matrix $\mbf A = [a_{ii'}, \, 1 \le i \le m, \, 1 \le i' \le n]$, 
we denote by $\mbf A^\top$ its transpose.
The element-wise $\ell_\infty$, $\ell_0$, $\ell_1$ and $\ell_2$-norms are denoted
by $\vert \mbf A \vert_\infty = \max_{1 \le i \le m} \max_{1 \le i' \le n} \vert a_{ii'} \vert$,
$\vert \mbf A \vert_0 = \sum_{i = 1}^m \sum_{i' = 1}^n \mathbb{I}_{\{a_{ii'} \ne 0\}}$,
$\vert \mbf A \vert_1 = \sum_{i = 1}^m \sum_{i' = 1}^n \vert a_{ii'} \vert$
and
$\vert \mbf A \vert_2 = \sqrt{\sum_{i = 1}^m \sum_{i' = 1}^n \vert a_{ii'} \vert^2}$.
The Frobenius, spectral, induced $L_1$ and $L_\infty$-norms are denoted by
$\Vert \mbf A \Vert_F = \vert \mbf A \vert_2$,  
$\Vert \mbf A \Vert = \sqrt{\Lambda_{\max}(\mbf A^\top \mbf A)}$
(with $\Lambda_{\max}(\mbf A)$ and $\Lambda_{\min}(\mbf A)$ denoting its largest and smallest eigenvalues in modulus),
$\Vert \mbf A \Vert_1 = \max_{1\le i' \le n} \sum_{i = 1}^m \vert a_{ii'}\vert$
and $\Vert \mbf A \Vert_\infty = \max_{1 \le i \le n} \sum_{i' = 1}^m \vert a_{ii'}\vert$.
Let $\mbf A_{i \cdot}$ and $\mbf A_{\cdot k}$ denote
the $i$-th row and the $k$-th column of $\mbf A$.
For two real numbers, set $a \vee b = \max(a, b)$ and $a \wedge b = \min(a, b)$.
Given two sequences $\{a_n\}$ and $\{b_n\}$, we write $a_n = O(b_n)$
if, for some finite constant $C > 0$
there exists $N \in \N_0 = \N \cup \{0\}$ such that
$|a_n| |b_n|^{-1} \le C$ for all $n \ge N$; we denote by $O_P$ the stochastic boundedness.
We write $a_n\asymp b_n$ when $a_n=O(b_n)$ and $b_n=O(a_n)$.
Throughout, $L$ denotes the lag operator and $\iota = \sqrt{-1}$. 
Finally, $\mathbb{I}_{\mc A} = 1$ if the event $\mc A$ takes place and $0$ otherwise.

\section{Factor-adjusted vector autoregressive model}
\label{sec:model}

Consider a zero-mean, second-order stationary $p$-variate process $\mbf X_t = (X_{1t}, \ldots, X_{pt})^\top$, $1 \le t \le n$, which is decomposed into the sum of two latent components: 
a factor-driven, {\it common} component $\bm\chi_t = (\chi_{1t}, \ldots, \chi_{pt})^\top$,
and an {\it idiosyncratic} component $\bm\xi_t = (\xi_{1t}, \ldots, \xi_{pt})^\top$ modelled as a VAR process. 
That is, $\mbf X_t = \bm\chi_t + \bm\xi_t$ where
\begin{align}
\bm\chi_t &= \mc B(L) \mbf u_t = \sum_{\ell = 0}^\infty \mbf B_\ell \mbf u_{t - \ell}\;\;
\text{ with } \mbf u_t = (u_{1t}, \ldots, u_{qt})^\top,
\quad \text{and} \label{eq:gdfm}
\\
\mc A(L) \bm\xi_t &= \bm\xi_t - \sum_{\ell = 1}^d \mbf A_{\ell} \bm\xi_{t - \ell} = \bm\Gamma^{1/2} \bm\vep_t\;\;
\text{ with } \bm\vep_t = (\vep_{1t}, \ldots, \vep_{pt})^\top. \label{eq:idio:var}
\end{align}
In~\eqref{eq:gdfm}, the latent random vector $\mbf u_t$, referred to as the vector of {\it common factors} or {\it common shocks}, is assumed to satisfy $\E(\mbf u_t) = \mbf 0$ and $\Cov(\mbf u_t) = \mbf I_q$, and
are loaded on each $\chi_{it}$ via square summable, one-sided filters 
$\mc B_{ij}(L) = \sum_{\ell = 0}^\infty B_{\ell, ij} L^\ell$, where $\mbf B_\ell = [B_{\ell, ij},\, 1\le i\le p,\, 1\le j\le q] \in \R^{p \times q}$.
This defines the generalised dynamic factor model (GDFM) proposed by \cite{FHLR00} and \citet{fornilippi01}, which provides the most general approach to high-dimensional time series factor modelling. 

In~\eqref{eq:idio:var}, the idiosyncratic component $\bm\xi_t$ is modelled as a VAR($d$) process for some finite positive integer $d$, with innovations $\bm\Gamma^{1/2} \bm\vep_t$ 
where $\bm\Gamma \in \R^{p \times p}$ is some positive definite matrix and $\bm\Gamma^{1/2}$ its symmetric square root matrix, and $\E(\bm\vep_t) = \mbf 0$ and $\Cov(\bm\vep_t) = \mbf I_p$. 
We assume that $\bm\xi_t$ is causal (see Assumption~\ref{assum:idio}~\ref{cond:idio:transition} below), i.e.\ it admits the Wold representation:
\begin{align}
\label{eq:idio:wold}
\bm\xi_t = \mc D(L)\bm\Gamma^{1/2} \bm\vep_t = \sum_{\ell = 0}^\infty \mbf D_\ell \bm\Gamma^{1/2} \bm\vep_{t - \ell} \text{ \ with \ } \mc D(L) = \mc A^{-1}(L),
\end{align}
such that $\bm\Gamma^{1/2} \bm\vep_t$ is seen as a vector of {\it idiosyncratic shocks} loaded on each $\xi_{it}$ via square summable, one-sided filters $\mc D_{ik}(L) = \sum_{\ell = 0}^\infty D_{\ell, ik} L^\ell$ where $\mbf D_\ell = [D_{\ell, ik}, \, 1 \le i, k \le p]$.
After accounting for the dominant cross-sectional dependence in the data (both contemporaneous and lagged) by factors,
it is reasonable to assume that the dependences left in $\bm\xi_t$ are weak and, therefore, that the VAR structure is sufficiently sparse.
Discussion on the precise requirement on the sparsity of $\mbf A_\ell, \, 1 \le \ell \le d$, and $\bm\Gamma^{-1}$ is deferred to Section~\ref{sec:fnets}.

\begin{rem}
\label{rem:staticFM}
A special case of the GDFM is the popularly adopted static factor model where the factors are loaded only contemporaneously (see e.g.\ \citealp{stock2002forecasting, bai2003, fan2013large}). 
This is formalised in Assumption~\ref{assum:static} below, where we consider forecasting under a static representation. 
A sufficient condition to obtain a static representation from the GDFM in~\eqref{eq:gdfm}, is to assume $\mc B(L) = \sum_{\ell = 0}^s \mbf B_\ell L^\ell$ for some finite integer $s \ge 0$. 
For example, if $s = 0$, the model reduces to $\bm\chi_t = \mbf B_{0} \mbf u_t$ while if $s > 0$, it can be written as $\bm\chi_t = \bm\Lambda \mbf F_t$ with $\bm\Lambda = [\mbf B_{\ell}, \, 0 \le \ell \le s]$ and $\mbf F_t = (\mbf u_t^\top, \ldots, \mbf u_{t - s}^\top)^\top$. 
Under the static factor model, $\mbf X_t$ admits a factor-augmented VAR representation (see Remark~\ref{rem:static:FVAR} below).
\end{rem}

In the remainder of this section, we list the assumptions required for identification and estimation of~\eqref{eq:gdfm}--\eqref{eq:idio:var}.
Since $\bm\chi_t$ and $\bm\xi_t$ are latent, some assumptions are required to ensure their (asymptotic) identifiability which are made in the frequency domain.
Denote by $\bm\Sigma_x(\omega)$ the spectral density matrix of $\mbf X_t$ at frequency $\omega \in [-\pi, \pi]$, and $\mu_{x, j}(\omega)$ its dynamic eigenvalues which are real-valued and ordered in the decreasing order.
We similarly define $\bm\Sigma_\chi(\omega)$, $\mu_{\chi, j}(\omega)$, $\bm\Sigma_\xi(\omega)$ and $\mu_{\xi, j}(\omega)$.

\begin{assum}
\label{assum:factor}
{\it 
There exist a positive integer $p_0 \ge 1$,
constants $\rho_j \in (3/4, 1]$ with $\rho_1 \ge \ldots \ge \rho_q$,
and pairs of continuous functions
$\omega \mapsto \alpha_{\chi, j}(\omega)$
and $\omega \mapsto \beta_{\chi, j}(\omega)$
for $\omega \in [-\pi, \pi]$ and $1 \le j \le q$, such that
for all $p \ge p_0$,
\begin{align*}
& \beta_{\chi, 1}(\omega) \ge \frac{\mu_{\chi, 1}(\omega)}{p^{\rho_1}} \ge \alpha_{\chi, 1}(\omega) >
\ldots > \beta_{\chi, q}(\omega) \ge 
\frac{\mu_{\chi, q}(\omega)}{p^{\rho_q}} \ge \alpha_{\chi, q}(\omega) > 0.
\end{align*}}
\end{assum}
Under the assumption, if $\rho_j = 1$ for all $1 \le j \le q$, 
then we are in presence of $q$ factors that are equally pervasive for the whole cross-section. The left panel of Figure~\ref{fig:ex} depicts the case when $\rho_1 = 1$.
If $\rho_j < 1$ for some $j$, we permit the presence of `weak' factors
and our theoretical analysis explicitly reflects this;
see \cite{dgr08}, \cite{onatski12}, \cite{freyaldenhoven2021factor} and \cite{uematsu2023estimation} 
for estimation under static factor models permitting weak factors. 
When weak factors are present, the ordering of the variables becomes important as $p \to \infty$,
whereas the case of linearly diverging factor strengths is compatible with completely arbitrary cross-sectional ordering.
The requirement that $\rho_j > 3/4$ is a minimal one, and generally larger values of $\rho_j$ are required as the dimensionality increases and heavier tails are permitted as discussed later.

Assumptions~\ref{assum:common} and~\ref{assum:idio} are made to control the serial dependence in $\mbf X_t$.

\begin{assum}
\label{assum:common}
{\it There exist some constants $\Xi > 0$ and $\varsigma > 2$ such that for all $\ell \ge 0$,
\begin{align*}
\max_{1 \le i \le p} \vert \mbf B_{\ell, i \cdot} \vert_2 \le \Xi (1 + \ell)^{-\varsigma} 
\quad \text{and} \quad
\l( \sum_{j = 1}^q \vert \mbf B_{\ell, \cdot j} \vert_\infty^2 \r)^{1/2} \le \Xi (1 + \ell)^{-\varsigma}.
\end{align*}
}
\end{assum}

\begin{assum}
\label{assum:idio}
{\it
\begin{enumerate}[noitemsep, wide, labelindent = 0pt, label = (\roman*)]
\item \label{cond:idio:transition} 
$d$ is a finite positive integer and $\det(\mc A(z)) \ne 0$ for all $\vert z \vert \le 1$.
\item \label{cond:idio:innov}
There exist some constants $0 < m_\vep \le M_\vep$ such that
$\Vert \bm\Gamma \Vert \le M_\vep$ and $\Lambda_{\min}(\bm\Gamma) \ge m_\vep$. 
\item \label{cond:idio:minspec} 
There exist a constant $m_\xi >0$ such that
$\inf_{\omega \in [-\pi, \pi]} \mu_{\xi, p}(\omega) \ge m_\xi$.
\item \label{cond:idio:coef} 
There exist some constants $\Xi > 0$ and $\varsigma > 2$ such that for all $\ell \ge 0$,
\begin{align*}
\vert D_{\ell, ik} \vert \le C_{ik} (1 + \ell)^{-\varsigma} \text{ with }
\max\l\{ \max_{1 \le k \le p} \sum_{i = 1}^p C_{ik}, \, 
\max_{1 \le i \le p} \sum_{k = 1}^p C_{ik},
\max_{1 \le i \le p} \sqrt{\sum_{k = 1}^p C_{ik}^2} \r\} \le \Xi. 
\end{align*}
\end{enumerate}}
\end{assum}

Assumption~\ref{assum:idio}~\ref{cond:idio:transition} and~\ref{cond:idio:innov} are standard in the literature \citep{lutkepohl2005} and imply that $\bm\xi_t$ is causal and has finite and non-zero covariance.
Under Assumptions~\ref{assum:common} and~\ref{assum:idio}~\ref{cond:idio:coef} (imposed on the Wold decomposition of $\bm\xi_t$ in~\eqref{eq:idio:wold}), the serial dependence in $\mbf X_t$ decays at an algebraic rate.
Further, we obtain a uniform bound for $\mu_{\xi, j}(\omega)$ under Assumption~\ref{assum:idio}~\ref{cond:idio:coef}:
\begin{prop}
\label{prop:idio:eval}
\it{
Under Assumption~\ref{assum:idio}, uniformly over all $\omega \in [-\pi, \pi]$,
there exists some constant $B_\xi >0$ depending only on $M_\vep$,  $\Xi$ and $\varsigma$, 
defined in Assumption~\ref{assum:idio}~\ref{cond:idio:minspec} and~\ref{cond:idio:coef}, such that
$\sup_{\omega \in [-\pi, \pi]} \mu_{\xi, 1}(\omega) \le B_\xi$.
}
\end{prop}

\begin{rem}
\label{rem:stability}
Proposition~\ref{prop:idio:eval} and Assumption~\ref{assum:idio}~\ref{cond:idio:minspec} jointly establish the uniform boundedness of $\mu_{\xi, 1}(\omega)$ and $\mu_{\xi, p}(\omega)$, which is commonly assumed in the literature on high-dimensional VAR estimation via $\ell_1$-regularisation.
A sufficient condition for Assumption~\ref{assum:idio}~\ref{cond:idio:minspec} is that $\max\l\{\max_{1 \le i \le p} \sum_{\ell = 1}^d \vert \mbf A_{\ell, i \cdot} \vert_1,
\max_{1 \le j \le p} \sum_{\ell = 1}^d \vert \mbf A_{\ell, \cdot j} \vert_1 \r\}~\le~\Xi$
for some constant $\Xi > 0$ \citep{basu2015}.
Further, when e.g.\ $d = 1$, Assumption~\ref{assum:idio}~\ref{cond:idio:coef} follows if $\vert \mbf A_1 \vert_\infty \le \gamma < 1$ since $\max(\Vert \mbf D_\ell \Vert_1, \Vert \mbf D_\ell \Vert_\infty) \le \Xi \gamma^\ell$ with $\mbf D_\ell = \mbf A_1^\ell$.
\end{rem}

The two latent components $\bm\chi_t$ and $\bm\xi_t$, and the number of factors $q$, are identified thanks to the large gap between the eigenvalues of their spectral density matrices, which follows from Assumption~\ref{assum:factor} and Proposition~\ref{prop:idio:eval}. 
Then by Weyl's inequality, the $q$-th dynamic eigenvalue $\mu_{x, q}(\omega)$ 
diverges almost everywhere in $[-\pi, \pi]$ as $p \to \infty$, whereas $\mu_{x, q + 1}(\omega)$ is uniformly bounded for any $p \in \N$ and $\omega$.
This property is exploited in the FNETS methodology as later described in Section~\ref{sec:dpca}.
It is worth stressing that Assumption~\ref{assum:factor} and Proposition~\ref{prop:idio:eval} jointly constitute both a necessary and sufficient condition for the process $\mbf X_t$ to admit the dynamic factor representation in~\eqref{eq:gdfm}, see \cite{fornilippi01}.

Finally, we characterise the common and idiosyncratic innovations.

\begin{assum} 
\label{assum:innov:new}
\it{
\begin{enumerate}[noitemsep, wide, labelindent = 0pt, label = (\roman*)] 
\item \label{cond:iid:new}  $\{\mbf u_t\}_{t \in \Z}$ is a sequence of zero-mean, $q$-dimensional martingale difference vectors with $\Cov(\mbf u_t) = \mbf I_q$, and $u_{it}$ and $u_{jt}$ are independent for all $1 \le i, j \le q$ with $i \ne j$ and all $t\in \Z$.

\item \label{cond:iid:new_eps} $\{\bm\vep_t\}_{t \in \Z}$ is a sequence of zero-mean, $p$-dimensional martingale difference vectors with $\Cov(\bm\vep_t) = \mbf I_p$, and $\vep_{it}$ and $\vep_{jt}$ are independent for all $1 \le i, j \le p$ with $i \ne j$ and all $t\in \Z$.

\item \label{cond:uncor:new} 
$\E(u_{jt} \vep_{it'})~=~0$ for all $1 \le j \le q$, $1 \le i \le p$ and $t, t' \in \Z$.

\item \label{cond:dist:new}  There exist some constants $\nu > 4$ and $\mu_\nu>0$ such that 
\begin{center}
$\max\l\{ \max_{1 \le j \le q} \E(\vert u_{jt} \vert^\nu), \max_{1 \le i \le p} \E(\vert \vep_{it} \vert^\nu) \r\} \le \mu_\nu$.
\end{center}
\end{enumerate}
}
\end{assum}

Assumption~\ref{assum:innov:new}~\ref{cond:iid:new} and~\ref{cond:iid:new_eps} allow the common and idiosyncratic innovations to be sequences of martingale differences, relaxing the assumption of serial independence found in \cite{forni2017dynamic}.
Condition~\ref{cond:uncor:new} is standard in the factor modelling literature. 
Under~\ref{cond:dist:new}, we require that the innovations have $\nu > 4$ moments, which is considerably weaker than the Gaussianity assumed in the literature on VAR modelling of high-dimensional time series \citep{basu2015, han2015direct}.
In Appendix~\ref{sec:proof}, we separately consider the case when $\mbf u_t$ and $\bm\vep_t$ are Gaussian for for the sake of comparison.

\section{Network estimation via FNETS}
\label{sec:fnets}

\subsection{Networks underpinning factor-adjusted VAR processes}
\label{sec:networks}

Under the latent VAR model in~\eqref{eq:idio:var}, we can define three types of networks
underpinning the interconnectedness of $\mbf X_t$ after factor adjustment  \citep{barigozzi2019}.

Let $\mc V = \{1, \ldots, p\}$ denote the set of vertices representing the $p$ time series.
Firstly, the transition matrices $\mbf A_\ell = [A_{\ell, ii'}, \, 1 \le i, i' \le p]$, encode
the directed network $\mc N^{\dir} = (\mc V, \mc E^{\dir})$ representing Granger causal linkages,~with
\begin{align}\label{eq:net:dir}
\mc E^{\dir} = \l\{(i, i') \in \mc V \times \mc V: \, A_{\ell, ii'} \ne 0 \text{ for some } 1 \le \ell \le d \r\}
\end{align}
as the set of edges.
Here, the presence of an edge $(i, i') \in \mc E^{\dir}$
indicates that $\xi_{i', t - \ell}$ Granger causes $\xi_{it}$ at some lag $1 \le \ell \le d$.

The second network contains undirected edges representing 
contemporaneous dependence between VAR innovations $\bm\Gamma^{1/2} \bm\vep_t$,
denoted by $\mc N^{\undir} = (\mc V, \mc E^{\undir})$;
we have $(i, i') \in \mc E^{\undir}$ iff the partial correlation between 
the $i$-th and $i'$-th elements of $\bm\Gamma^{1/2} \bm\vep_t$ is non-zero. 
Specifically, letting $\bm\Gamma^{-1} = \bm\Delta = [\delta_{ii'}, \, 1 \le i, i' \le p]$, the set of edges is given by
\begin{align}
\label{eq:net:undir}
\mc E^{\undir} = \l\{ (i, i') \in \mc V \times \mc V: \, i \ne i' \text{ and }
- \frac{\delta_{ii'}}{\sqrt{\delta_{ii} \cdot \delta_{i'i'}}} \ne 0 \r\}.
\end{align}
Finally, we summarise the aforementioned 
lead-lag and contemporaneous relations between the variables
in a single, undirected network $\mc N^{\lr} = (\mc V, \mc E^{\lr})$
by means of the long-run partial correlations of $\bm\xi_t$. 
Let $\bm\Omega = [\omega_{ii'}, \, 1 \le i, i' \le p]$ denote 
the long-run partial covariance matrix of $\bm\xi_t$,
i.e.\ $\bm\Omega = (\bm\Sigma_\xi(0))^{-1} = 2\pi \mc A^\top(1) \bm\Delta \mc A(1)$ under~\eqref{eq:idio:var}.
Then, the set of edges of $\mc N^{\lr}$ is 
\begin{align}
\label{eq:net:lr}
\mc E^{\lr} = \l\{ (i, i') \in \mc V \times \mc V: \, i \ne i' \text{ and }
- \frac{\omega_{ii'}}{\sqrt{\omega_{ii} \cdot \omega_{i'i'}}} \ne 0 \r\}.
\end{align}
Generally, $\mc E^{\lr}$ is greater than $\mc E^{\dir} \cup \mc E^{\undir}$, see Appendix~\ref{app:lr} for a sufficient condition for the absence of an edge $(i, i')$ from $\mc N^{\lr}$.
In the remainder of Section~\ref{sec:fnets}, we describe the network estimation methodology of FNETS which, consisting of three steps, estimates the three networks while fully accounting for the challenges arising from not directly observing the VAR process $\bm\xi_t$, and investigate its theoretical properties.

\subsection{Step~1: Factor adjustment via dynamic PCA}
\label{sec:dpca}

As described in Section~\ref{sec:model}, under our model~\eqref{eq:gdfm}--\eqref{eq:idio:var}, there exists a large gap in $\mu_{x, j}(\omega)$, the dynamic eigenvalues of the spectral density matrix of $\mbf X_t$, between those attributed to the factors ($ j \le q$) and those which are not ($j \ge q + 1$). 
With the goal of estimating the autocovariance (ACV) matrix of the latent VAR process $\bm\xi_t$, we exploit this gap in the factor-adjustment step based on dynamic principal component analysis (PCA);
see Chapter~9 of \cite{brillinger1981} for the definition of dynamic PCA and \cite{FHLR00} for its use in the estimation of GDFM. 
Throughout, we treat~$q$ as known and refer to \cite{hallin2007} for its consistent estimation under~\eqref{eq:gdfm}.

Denote the ACV matrices of $\mbf X_t$ by
$\bm\Gamma_x(\ell) = \E(\mbf X_{t - \ell}\mbf X_t^\top)$ for $\ell \ge 0$
and $\bm\Gamma_x(\ell) = \bm\Gamma_x^\top(-\ell)$ for $\ell \le -1$,
and analogously define $\bm\Gamma_\chi(\ell)$ and $\bm\Gamma_\xi(\ell)$
with $\bm\chi_t$ and $\bm\xi_t$ replacing $\mbf X_t$, respectively. 
Then, $\bm\Sigma_x(\omega)$ and $\bm\Gamma_x(\ell)$ satisfy
$\bm\Sigma_x(\omega) = (2\pi)^{-1} \sum_{\ell=-\infty}^\infty \bm\Gamma_x(\ell) \exp(-\iota\ell\omega)$
for all $\omega\in[-\pi,\pi]$. 
Motivated by this, we estimate $\bm\Sigma_x(\omega)$ by
\begin{align}
\wh{\bm\Sigma}_x(\omega) = \frac{1}{2\pi} \sum_{\ell = -m}^m K\l(\frac{\ell}{m}\r)
\wh{\bm\Gamma}_x(\ell) \exp(-\iota \ell \omega), \label{eq:periodogram}
\end{align}
with the sample ACV $\wh{\bm\Gamma}_x(\ell) = n^{-1} \sum_{t = \ell + 1}^n \mbf X_{t - \ell} \mbf X_t^\top$ when $\ell \ge 0$, and $\wh{\bm\Gamma}_x(\ell) = \wh{\bm\Gamma}_x(-\ell)^\top$ for $\ell < 0$, and the kernel bandwidth $m = \lfloor n^\beta \rfloor$ for some $\beta \in (0, 1)$. 
We adopt the Bartlett kernel as $K(\cdot)$
which ensures positive semi-definiteness of $\wh{\bm\Sigma}_x(\omega)$
(see Appendix~\ref{pf:prop:idio:est:lasso}).
Then, we evaluate $\wh{\bm\Sigma}_x(\omega)$ at the $2m + 1$ Fourier frequencies $\omega_k, \, -m \le k \le m$ ($\omega_k = 2\pi k / (2m + 1)$ for $0 \le k \le m$, and $\omega_k = - \omega_{\vert k \vert}$ for $-m \le k \le -1$),
and estimate $\bm\Sigma_\chi(\omega_k)$ by retaining the contribution from the $q$ largest eigenvalues and eigenvectors only. 
That is, we obtain $\wh{\bm\Sigma}_\chi(\omega_k) = \sum_{j = 1}^q \wh{\mu}_{x, j}(\omega_k) \wh{\mbf e}_{x, j}(\omega_k) (\wh{\mbf e}_{x, j}(\omega_k))^*$ (with $*$ denoting the transposed complex conjugate),
where $\wh\mu_{x, 1}(\omega) \ge \ldots \ge \wh\mu_{x, q}(\omega)$, denote the $q$ leading eigenvalues of $\wh{\bm\Sigma}_x(\omega)$ and $\wh{\mbf e}_{x, j}(\omega)$ the associated (normalised) eigenvectors.
From this, an estimator of $\bm\Gamma_\chi(\ell)$ at a given lag $\ell \in \mathbb{N}$, is obtained via inverse Fourier transform as
$\wh{\bm\Gamma}_\chi(\ell) = 2\pi (2m + 1)^{-1} \sum_{k = -m}^m \wh{\bm\Sigma}_\chi(\omega_k) \exp(\iota \ell \omega_k)$
and finally, we estimate the ACV matrices of $\bm\xi_t$ with $\wh{\bm\Gamma}_\xi(\ell) = \wh{\bm\Gamma}_x(\ell) - \wh{\bm\Gamma}_\chi(\ell)$, by virtue of Assumption~\ref{assum:innov:new}~\ref{cond:uncor:new}.

\subsection{Step~2: Estimation of VAR parameters and $\mc N^{\dir}$}
\label{sec:idio:beta}

Recalling the VAR($d$) model in~\eqref{eq:idio:var}, let $\bm\beta = [\mbf A_\ell, \, 1 \le \ell \le d]^\top \in \R^{(pd) \times p}$ denote the matrix collecting all the VAR parameters.
When $\bm\xi_t$ is directly observable, $\ell_1$-regularised least squares or maximum likelihood estimators have been proposed for $\bm\beta$, see the references given in Introduction.
In the context of factor-adjusted regression modelling where the aim is to estimate the regression structure in the latent idiosyncratic process, it has been proposed to apply the $\ell_1$-regularisation methods after estimating the entire latent process by, say,~$\wh{\bm\xi}_t$ \citep{fan2020factor, fan2021bridging, fan2023latent, krampe2021}.
However, such an approach possibly suffers from the lack of statistical efficiency due to having to control the estimation errors in $\wh{\bm\xi}_t$ uniformly for all $1 \le t \le n$.
 Instead, we make use of the Yule-Walker (YW) equation $\bm\beta = \bbG^{-1} \bbg$, where
\begin{align*}
\bbG = \bmx 
\bm\Gamma_\xi(0) & \bm\Gamma_\xi(-1) & \ldots & \bm\Gamma_\xi(-d + 1) 
\\
\vdots& \vdots& \ddots &\vdots 
\\
\bm\Gamma_\xi(d - 1) & \bm\Gamma_\xi(d - 2) & \ldots & \bm\Gamma_\xi(0)
\emx 
\quad \text{and} \quad
\bbg = \bmx
\bm\Gamma_\xi(1)
\\
\vdots
\\
\bm\Gamma_\xi(d)
\emx,
\end{align*}
with $\bbG$ being always invertible since $\Lambda_{\min}(\bbG) \ge 2\pi m_\xi > 0$
by Assumption~\ref{assum:idio}~\ref{cond:idio:minspec}. 
We propose to estimate $\bm\beta$ as a regularised YW estimator based on $\wh{\bbG}$ and $\wh{\bbg}$, which are obtained by replacing $\bm\Gamma_\xi(\ell)$ with $\wh{\bm\Gamma}_\xi(\ell)$ derived in Step~1 of FNETS via dynamic PCA, in the definitions of $\bbG$ and $\bbg$, respectively.

To handle the high dimensionality, we consider an $\ell_1$-regularised estimator for $\bm\beta$ 
which solves the following $\ell_1$-penalised $M$-estimation problem
\begin{align}
\label{eq:lasso}
\wh{\bm\beta} = {\arg\min}_{\mbf M \in \R^{pd \times p}} \
\tr\l( \mbf M^\top \wh{\bbG} \mbf M - 2 \mbf M^\top\wh{\bbg} \r)
+ \lambda \vert \mbf M \vert_1
\end{align}
with a tuning parameter $\lambda > 0$.
Note that the matrix $\wh{\bbG}$ is guaranteed to be positive semi-definite
(see Appendix~\ref{pf:prop:idio:est:lasso}),
thus the problem in~\eqref{eq:lasso} is convex with a global minimiser.
We note the similarity between~\eqref{eq:lasso} and the Lasso estimator, but our estimator is specifically tailored for the problem of estimating the parameters for the latent VAR process $\bm\xi_t$ by means of second-order moments only, and thus differs fundamentally from the Lasso-type estimators proposed for high-dimensional VAR estimation.
In Appendix~\ref{sec:beta:ds}, we propose an alternative estimator based on a constrained $\ell_1$-minimisation approach closely related to the Dantzig selector \citep{candes2007}.

Once the VAR parameters are estimated, we propose to estimate the edge set of $\mc N^{\dir}$ in~\eqref{eq:net:dir} by the set of indices of the non-zero elements of a thresholded version of $\wh{\bm\beta}$, denoted by $\wh{\bm\beta}(\mathfrak{t}) = [\wh{\beta}_{ij} \cdot \mathbb{I}_{\{\vert \wh{\beta}_{ij} \vert > \mathfrak{t} \}} ]$, with some threshold $\mathfrak{t} > 0$.

\subsection{Step~3: Estimation of $\mc N^{\undir}$ and $\mc N^{\lr}$}
\label{sec:idio:lrpc}

Recall that the edge sets of $\mc N^{\undir}$ and $\mc N^{\lr}$ defined in~\eqref{eq:net:undir}--\eqref{eq:net:lr}, are given by the supports of $\bm\Delta$ and $\bm\Omega$.
Given $\wh{\bm\beta}$ in~\eqref{eq:lasso} which estimates $\bm\beta$, a natural estimator of $\bm\Gamma$ arises from the YW equation
$\bm\Gamma = \bm\Gamma_\xi(0) - \sum_{\ell = 1}^d \mbf A_{\ell} \bm\Gamma_\xi(\ell) = \bm\Gamma_\xi(0) - \bm\beta^\top \bbg$,
as $\wh{\bm\Gamma} = \wh{\bm\Gamma}_\xi(0) - \wh{\bm\beta}^\top \wh{\bbg}$.
Then, we propose to estimate $\bm\Delta = \bm\Gamma^{-1}$ via constrained $\ell_1$-minimisation~as
\begin{align}
\label{eq:inv:delta}
\check{\bm\Delta} & = {\arg\min}_{\mbf M \in \R^{p \times p}} \vert \mbf M \vert_1 
\quad \text{subject to} \quad
\l\vert \wh{\bm\Gamma} \mbf M - \mbf I \r\vert_\infty \le \eta,
\end{align}
where $\eta > 0$ is a tuning parameter.
This approach has originally been proposed for estimating the precision matrix of independent data \citep{cai2011}, which we extend to time series settings.
Since $\check{\bm\Delta} = [\check\delta_{ii'}, \, 1 \le i, j \le p]$ is not guaranteed to be symmetric, a symmetrisation step is performed to obtain
$\wh{\bm\Delta} = [\wh\delta_{ii'}, \, 1 \le i, i' \le p]$
with
$\wh\delta_{ii'} = \check\delta_{ii'} \cdot \mathbb{I}_{\{\vert \check\delta_{ii'} \vert
\le \vert \check\delta_{i'i} \vert \}}
+ \check\delta_{i'i} \cdot \mathbb{I}_{\{\vert \check\delta_{i'i} \vert
< \vert \check\delta_{ii'} \vert \}}$.
Then, the edge set of $\mc N^{\undir}$ in~\eqref{eq:net:undir} is estimated by the support of the thresholded estimator $\wh{\bm\Delta}(\mathfrak{t}_\delta) = [\wh\delta_{ii'} \cdot \mathbb{I}_{\{ \vert \wh\delta_{ii'} \vert > \mathfrak{t}_\delta\}}, \, 1 \le i, i' \le p]$ with some threshold $\mathfrak{t}_\delta > 0$.

Finally, we estimate $\bm\Omega = 2\pi (\mc A(1))^\top \bm\Delta \mc A(1)$ by replacing $\mc A(1)$ and $\bm\Delta$ with their estimators.
We adopt the thresholded estimator $\wh{\bm\beta}(\mathfrak{t}) = [\wh{\mbf A}_1(\mathfrak{t}), \ldots, \wh{\mbf A}_d(\mathfrak{t})]^\top$, to obtain $\wh{\mc A}(1) = \mbf I - \sum_{\ell = 1}^d \wh{\mbf A}_\ell(\mathfrak{t})$ and set $\wh{\bm\Omega} = 2\pi (\wh{\mc A}(1))^\top \wh{\bm\Delta} \wh{\mc A}(1)$.
Analogously, the edge set of $\mc N^{\lr}$ in~\eqref{eq:net:lr} is obtained by thresholding $\wh{\bm\Omega} = [\wh{\omega}_{ii'}, \, 1 \le i, i' \le p]$ with some threshold $\mathfrak{t}_\omega > 0$, as the support of $\wh{\bm\Omega}(\mathfrak{t}_\omega) = [\wh{\omega}_{ii'} \cdot \mathbb{I}_{\{ \vert \wh\omega_{ii'} \vert > \mathfrak{t}_\omega\}}, \, 1 \le i, i' \le p]$.

\subsection{Theoretical properties}
\label{sec:fnets:network:theory}

We prove the consistency of FNETS in network estimation by establishing the theoretical properties of each of its three steps in Sections~\ref{sec:theor:fnets:one}--\ref{sec:theor:fnets:three}. Then in Section~\ref{sec:theory:static}, we present the results for a special case where $\bm\chi_t$ admits a static representation, $n \asymp p$ and $\E(\vert X_{it} \vert^\nu) < \infty$ for $\nu > 8$, for ease of comparing our results to the existing ones.

Hereafter, we define
\begin{align}
\label{eq:rates:one}
\psi_n = \l( \frac{m}{n^{1 - 2/\nu}} \vee \sqrt{\frac{m\log(m)}{n}} \r) 
\text{ and }
\vartheta_{n, p} = \l(\frac{m (np)^{2/\nu} \log^{7/2}(p)}{n} \vee \sqrt{\frac{m \log(mp)}{n}} \r),
\end{align}
where the dependence of these quantities on $\nu$ is omitted for simplicity.

\subsubsection{Factor adjustment via dynamic PCA}
\label{sec:theor:fnets:one}

We first establish the consistency of the dynamic PCA-based estimator of $\bm\Gamma_\chi(\ell)$.
\begin{thm}
\label{thm:common:spec}
\it{Suppose that Assumptions~\ref{assum:factor},~\ref{assum:common},~\ref{assum:idio} and~\ref{assum:innov:new} are met. Then, for any finite positive integer $s\le d$, as $n, p \to \infty$,
\begin{align*}
\max_{\ell: \, \vert \ell \vert \le s} \frac{1}{p} 
\l\Vert \wh{\bm\Gamma}_\chi(\ell) - \bm\Gamma_\chi(\ell) \r\Vert_F
=&\, O_P\l(q  p^{2(1 - \rho_q)} \l(\psi_n \vee \frac{1}{m} \vee \frac{1}{\sqrt p}\r)\r), \\
\max_{\ell: \, \vert \ell \vert \le s} 
\l\vert \wh{\bm\Gamma}_\chi(\ell) - \bm\Gamma_\chi(\ell) \r\vert_\infty
=&\, O_P\l(q p^{2(1 - \rho_q)} \l(\vartheta_{n, p} \vee \frac{1}{m} \vee \frac{1}{\sqrt p}\r)\r).
\end{align*}
}
\end{thm}

\begin{rem}
\label{rem:thm:dpca}
\begin{enumerate}[noitemsep, wide, labelwidth=!, labelindent=5pt, label = (\alph*)]

\item Theorem~\ref{thm:common:spec} is complemented by Proposition~\ref{prop:spec:common} in Appendix which establishes the consistency of the spectral density matrix estimator $\wh{\bm\Sigma}_\chi(\omega)$ uniformly over $\omega \in [-\pi, \pi]$ in both Frobenius and $\ell_\infty$-norms.

\item \label{rem:bandwidth} Both $\psi_n$ and $\vartheta_{n, p}$ in~\eqref{eq:rates:one} increase with the bandwidth $m$.  
It is possible to find $m$ that minimises e.g.\ $(\vartheta_{n, p} \vee m^{-1})$ which, roughly speaking, represents the bias-variance trade-off in the estimation of the spectral density matrix $\bm\Sigma_x(\omega)$.
For example, in light-tailed settings with large enough $\nu$, 
the choice $m \asymp (n\log^{-1}(np))^{1/3}$ 
leads to the minimal rate in $\ell_\infty$-norm $(\vartheta_{n, p} \vee m^{-1}) \asymp (\log(np)/n)^{1/3}$ 
which nearly matches the optimal non-parametric rate when using the Bartlett kernel as in~\eqref{eq:periodogram} \citep[p.\ 463]{priestley1982spectral}.

\item \label{rem:tails} 
Consistency in Frobenius norm depends on $\psi_n$ which tends to zero as $n \to \infty$ without placing any constraint on the relative rate of divergence between $n$ and~$p$.
Consistency in $\ell_\infty$-norm
is determined by $\vartheta_{n, p}$ which depends on the interplay between the dimensionality and the tail behaviour.
Generally, the estimation error worsens as weaker factors are permitted ($\rho_q < 1$ in Assumption~\ref{assum:factor}) and as $p$ grows, and also when $\nu$ is small such that heavier tails are permitted.
Consider the case when all factors are strong (i.e.\ $\rho_j = 1$).
If $p \asymp n$, then $\ell_\infty$-consistency holds with an appropriately chosen $m = n^\beta, \, \beta \in (0, 1)$, that leads to $\vartheta_{n, p} = o(1)$, provided that $\nu > 4$.
When all moments of $u_{jt}$ and $\vep_{it}$ exist, we achieve $\ell_\infty$-consistency even in the ultra high-dimensional case where $\log(p) = o(n)$.
\end{enumerate}
\end{rem}

From Theorem~\ref{thm:common:spec}, the following proposition immediately follows.
\begin{prop}
\label{prop:acv:idio}
\textit{Suppose that the conditions in Theorem~\ref{thm:common:spec} are met and let Assumption~\ref{assum:factor} hold with $\rho_j = 1, \, 1 \le j \le q$. 
Then, $\p(\mc E_{n, p}) \to 1$ as $n, p \to \infty$, where
\begin{align}
\mc E_{n, p} &= \l\{
\max_{-d \le \ell \le d} \l\vert \wh{\bm\Gamma}_\xi(\ell) - \bm\Gamma_\xi(\ell) \r\vert_\infty
\le C_\xi \l(\vartheta_{n, p} \vee \frac{1}{m} \vee \frac{1}{\sqrt p}\r) \r\}.
\label{eq:idio:set}
\end{align}
for some constant $C_\xi > 0$. 
}
\end{prop}

From Proposition~\ref{prop:acv:idio}, we have $\ell_\infty$-consistency of $\wh{\bm\Gamma}_\xi(\ell)$ in the presence of strong factors. Although it is possible to trace the effect of weak factors on the estimation of $\bm\Gamma_\xi(\ell)$ (see Corollary~\ref{prop:acv:idio:new}), we make this simplifying assumption to streamline the presentation of the theoretical results of the subsequent Steps~2--3 of FNETS. 

\begin{rem}
\label{rem:static} 
In Appendix~\ref{sec:rem:static}, we show that if $\bm\chi_t$ admits the static representation discussed in Remark~\ref{rem:staticFM}, the rate in Proposition~\ref{prop:acv:idio} is further improved as
\begin{align}
& \max_{\ell: \, \vert \ell \vert \le d} 
\l\vert \wh{\bm\Gamma}_\xi(\ell) - \bm\Gamma_\xi(\ell) \r\vert_\infty
= O_P\l(\wt{\vartheta}_{n, p} \vee \frac{1}{\sqrt p}\r)
\text{ with } \wt\vartheta_{n, p} = 
\l( \frac{p^{2/\nu} \log^3(p)}{n^{1 - 2/\nu}} \vee \sqrt{\frac{\log(p)}{n}} \r).
\label{eq:tilde:vartheta}
\end{align}
The term $\wt{\vartheta}_{n, p}$ comes from bounding $\max_{\ell: \, \vert \ell \vert \le d} \vert \wh{\bm\Gamma}_x(\ell) - \bm\Gamma_x(\ell) \vert_\infty$.
Hence, the improved rate in~\eqref{eq:tilde:vartheta} is comparable to the rate attained when we directly observe $\bm\xi_t$ apart from the presence of $p^{-1/2}$, which is due to the presence of latent factors; similar observations are made in Theorem~3.1 of \cite{fan2013large}. 
\end{rem}

\subsubsection{Estimation of VAR parameters and $\mc N^{\dir}$}
 \label{sec:theor:fnets:two}

We measure the sparsity of $\bm\beta$ by $s_{0, j} = \vert \bm\beta_{\cdot j} \vert_0$, $s_0 = \sum_{j = 1}^p s_{0, j}$ and $s_{\text{\upshape in}} = \max_{1 \le j \le p} s_{0, j}$. When $d = 1$, the quantity $s_{\text{\upshape in}}$ coincides with the maximum in-degree per node of $\mc N^{\dir}$.

\begin{prop}
\label{prop:idio:est:lasso}
{\it 
Suppose that $C_\xi  s_{\text{\upshape in}} (\vartheta_{n, p} \vee m^{-1} \vee p^{-1/2} ) \le \pi m_\xi/16$, where $m_\xi$ is defined in Assumption~\ref{assum:idio}~\ref{cond:idio:minspec}. 
Also, set $\lambda \ge 4C_\xi (\Vert \bm\beta \Vert_1 + 1) (\vartheta_{n, p} \vee m^{-1} \vee p^{-1/2})$ in~\eqref{eq:lasso}. Then, conditional on $\mc E_{n, p}$ defined in~\eqref{eq:idio:set}, we have
\begin{align*}
& \max_{1 \le j \le p} \l\vert \wh{\bm\beta}_{\cdot j} - \bm\beta_{\cdot j} \r\vert_2 \le 
\frac{6 \sqrt{s_{\text{\upshape in}}} \lambda}{\pi m_\xi}, \quad
\max_{1 \le j \le p} \l\vert \wh{\bm\beta}_{\cdot j} - \bm\beta_{\cdot j} \r\vert_1 \le \frac{24 s_{\text{\upshape in}} \lambda}{\pi m_\xi} \quad \text{ and}
\\
& \max_{1 \le j \le p} \l\vert \wh{\bm\beta}_{\cdot j} - \bm\beta_{\cdot j} \r\vert_\infty \le \min\l(4\Vert \bbG^{-1} \Vert_1\lambda, \frac{6 \sqrt{s_{\text{\upshape in}}} \lambda}{\pi m_\xi} \r).
\end{align*}
}
\end{prop}
Following \cite{loh2012}, the proof of Proposition~\ref{prop:idio:est:lasso} proceeds by showing that, conditional on $\mc E_{n, p}$, the matrix $\wh{\bbG}$ meets a restricted eigenvalue condition \citep{BRT09} and the deviation bound is controlled as $\vert \wh{\bbG} \bm\beta - \wh{\bbg} \vert_\infty \le \lambda/4$.
Then, thanks to Proposition~\ref{prop:acv:idio}, as $n, p \to \infty$, the estimation errors of $\wh{\bm\beta}$ in $\ell_2$, $\ell_1$- and $\ell_\infty$-norms, are bounded as in Proposition \ref{prop:idio:est:lasso} with probability tending to one.

\begin{rem}
\label{rem:idio:est:lasso}
As noted in Remark~\ref{rem:stability}, the boundedness of $\mu_{\xi, j}(\omega)$ follows from that of $\Vert \bm\beta \Vert_1 = \max_{1 \le j \le p}  \sum_{\ell = 1}^d \vert \mbf A_{\ell, j \cdot} \vert_1$, in which case $\Vert \bm\beta \Vert_1$ appearing in the assumed lower bound on $\lambda$, does not inflate the rate of the estimation errors.
In the light-tailed situation, with the optimal bandwidth $m \asymp (n \log^{-1}(np))^{1/3}$ as specified in Remark~\ref{rem:thm:dpca}~\ref{rem:bandwidth}, it is required that $s_{\text{\upshape in}} = O((n \log^{-1}(np))^{1/3} \wedge \sqrt{p})$, which still allows the number of non-zero entries in {\it each} row of $\mbf A_\ell$ to grow with~$p$.
Here, the exponent $1/3$ in place of $1/2$ often found in the literature, comes from adopting the most general approach to time series factor modelling which necessitates selecting a bandwidth for frequency domain-based factor adjustment. 
\end{rem}

For sign consistency of the Lasso estimator, the (almost) necessary and sufficient condition is the so-called irrepresentable condition \citep{zhao2006model}, which is known to be highly stringent \citep{tardivel2022}. 
Alternatively, \cite{medeiros2016} propose an adaptive Lasso estimator with data-driven weights for high-dimensional VAR estimation when $\bm\xi_t$ is directly observed.
Instead, we propose to additionally threshold $\wh{\bm\beta}$ and obtain $\wh{\bm\beta}(\mathfrak{t})$, whose support consistently estimates the edge set of $\mc N^{\dir}$. 


\begin{cor}
\label{cor:idio:est:dir:lasso}
{\it Suppose that the conditions of Proposition~\ref{prop:idio:est:lasso} are met.
If 
\begin{align}
\label{eq:beta:min}
\min_{(i, j): \, \vert \beta_{ij} \vert > 0} \vert \beta_{ij} \vert > 2\mathfrak{t}
\end{align}
with $\mathfrak{t} = \min(4 \Vert \bbG^{-1} \Vert_1\lambda, 6 \sqrt{s_{\text{\upshape in}}} \lambda/(\pi m_\xi))$, then $\text{\upshape{sign}}(\wh{\bm\beta}(\mathfrak{t})) = \text{\upshape{sign}}(\bm\beta)$ conditional on $\mc E_{n, p}$.}
\end{cor}

\subsubsection{Estimation of  $\mc N^{\undir}$ and $\mc N^{\lr}$} 
\label{sec:theor:fnets:three}

Let $s_\delta(\varrho) = \max_{1 \le i \le p} \sum_{i' = 1}^p \vert \delta_{ii'} \vert^\varrho, \, \varrho \in [0, 1)$, denote the (weak) sparsity of $\bm\Delta = [\delta_{ii'}, \, 1 \le i, i' \le p]$.
Also, define $s_{\text{\upshape out}} = \max_{1 \le j \le p} \sum_{\ell = 1}^d \vert \mbf A_{\ell, \cdot j} \vert_0$ which, complementing $s_{\text{\upshape in}}$, represents the sparsity of the out-going edges of $\mc N^{\dir}$. 
Analogously as in Proposition~\ref{prop:idio:est:lasso}, we establish deterministic guarantees for $\wh{\bm\Delta}$ and $\wh{\bm\Omega}$ conditional on $\mc E_{n, p}$.

\begin{prop}
\label{prop:idio:delta}
{\it Suppose that the conditions in Propositions~\ref{prop:idio:est:lasso} are met, and set $\eta = C s_{\text{\upshape in}} \Vert \bm\Delta \Vert_1 (\Vert \bm\beta \Vert_1 + 1) (\vartheta_{n, p} \vee m^{-1} \vee p^{-1/2})$ in~\eqref{eq:inv:delta}, with $C$ depending only on $C_\xi$ and $m_\xi$.
Then, conditional on $\mc E_{n, p}$ defined in~\eqref{eq:idio:set}, we have:
\begin{enumerate}[noitemsep, wide, labelwidth=!, labelindent=5pt, label = (\roman*)]
\item \label{prop:idio:delta:one} 
$\vert \wh{\bm\Delta} - \bm\Delta \vert_\infty \le 4 \Vert \bm\Delta \Vert_1 \eta$ and
$\Vert \wh{\bm\Delta} - \bm\Delta \Vert \le 12 s_\delta(\varrho)
(4 \Vert \bm\Delta \Vert_1 \eta)^{1 - \varrho}$.

\item \label{prop:idio:delta:two} If also $s_{\text{\upshape out}} \mathfrak{t} \le \Vert \mc A(1) \Vert_1$ with $\mathfrak{t}$ chosen as in Corollary~\ref{cor:idio:est:dir:lasso},
then, 
\begin{align*}
\l\vert \wh{\bm\Omega} - \bm\Omega \r\vert_\infty \le
4\pi \Vert \mc A(1) \Vert_1 \l(3 \Vert \bm\Delta \Vert s_{\text{\upshape out}}\mathfrak{t} + 
16 \Vert \mc A(1) \Vert_1 \Vert \bm\Delta \Vert_1 \eta \r).
\end{align*}
\end{enumerate}}
\end{prop}

Together with Assumption~\ref{assum:idio}~\ref{cond:idio:innov},
Proposition~\ref{prop:idio:delta}~\ref{prop:idio:delta:one}
indicates asymptotic positive definiteness of $\wh{\bm\Delta}$
provided that $\bm\Delta$ is sufficiently sparse, as measured by $\Vert \bm\Delta \Vert_1$ and $s_\delta(\varrho)$.
By definition, $\mc N^{\lr}$ combines $\mc N^{\dir}$ and $\mc N^{\undir}$ and consequently, its sparsity structure is determined by the sparsity of the other two networks, which is reflected in Proposition~\ref{prop:idio:delta}~\ref{prop:idio:delta:two}.
Specifically, the term $\Vert \mc A(1) \Vert_1$ is related to the out-going property of $\mc N^{\dir}$, and satisfies $\Vert \mc A(1) \Vert_1 \le \max_{1 \le j \le p} \sum_{\ell = 1}^d \vert \mbf A_{\ell, \cdot j} \vert_1$, where the boundedness of the right-hand side is sufficient for the boundedness of $\mu_{\xi, j}(\omega)$ (Remark~\ref{rem:stability}).
Also, $\Vert \bm\Delta \Vert_1$ reflects the sparsity of the edge set of $\mc N^{\undir}$, and the tuning parameter $\eta$ depends on the sparsity of the in-coming edges of $\mc N^{\dir}$ through $\Vert \bm\beta \Vert_1 = \max_{1 \le j \le p} \sum_{\ell = 1}^d \vert \mbf A_{\ell, j \cdot} \vert_1$ and $s_{\text{\upshape in}}$.

Similarly as in Corollary~\ref{cor:idio:est:dir:lasso}, we can show the consistency of the thresholded estimators $\wh{\bm\Delta}(\mathfrak{t}_\delta)$ and $\wh{\bm\Omega}(\mathfrak{t}_\omega)$ in estimating the edge sets of $\mc N^{\undir}$ and $\mc N^{\lr}$, respectively.
\begin{cor}
\label{cor:idio:delta}
{\it Suppose that conditions of Proposition~\ref{prop:idio:delta} are met.
Conditional on $\mc E_{n, p}$:
\begin{enumerate}[noitemsep, wide, labelwidth=!, labelindent=5pt, label = (\roman*)]
\item \label{cor:idio:delta:one} If $\min_{(i, i'): \, \vert \delta_{ii'} \vert > 0} \vert \delta_{ii'} \vert > 2\mathfrak{t}_\delta$ with $\mathfrak{t}_\delta = 4\Vert \bm\Delta \Vert_1 \eta$, we have $\text{\upshape sign}(\wh{\bm\Delta}(\mathfrak{t}_\delta)) = \text{\upshape sign}(\bm\Delta)$.

\item \label{cor:idio:delta:two} If $\min_{(i, i'): \, \vert \omega_{ii'} \vert > 0} \vert \omega_{ii'} \vert > 2\mathfrak{t}_\omega$ with $\mathfrak{t}_\omega = 4\pi \Vert \mc A(1) \Vert_1 (3 \Vert \bm\Delta \Vert s_{\text{\upshape out}}\mathfrak{t} + 16 \Vert \mc A(1) \Vert_1 \Vert \bm\Delta \Vert_1 \eta )$, \\ we have $\text{\upshape sign}(\wh{\bm\Omega}(\mathfrak{t}_\omega)) = \text{\upshape sign}(\bm\Omega)$.
\end{enumerate}
}
\end{cor}

\subsubsection{The case of the static factor model} 
\label{sec:theory:static}

For ease of comparing the performance of FNETS with the existing results, we focus on the static factor model setting discussed in Remark~\ref{rem:staticFM}, and assume that $n \asymp p$ and $\max( \Vert \bm \beta \Vert_1, \Vert \mc A(1) \Vert_1) = O(1)$.
Then, from Remark~\ref{rem:static} and the proof of Proposition~\ref{prop:idio:est:lasso}, we obtain $\max_{1\le j\le p} \vert \wh{\bm\beta}_{\cdot j} - \bm\beta_{\cdot j} \vert_2 
= O_P(\sqrt{s_{\text{\upshape in}} \log(n)/n})$ provided that $\nu > 8$, such that the condition in~\eqref{eq:beta:min} is written with $\mathfrak{t} \asymp \sqrt{s_{\text{\upshape in}} \log(n)/n}$. That is, $\wh{\bm\beta}$ and its thresholded counterpart proposed for the estimation of the latent VAR process, perform as well as the benchmark derived under independence and Gaussianity in the Lasso literature \citep{van2011adaptive}.
In this same setting, the factor-adjusted regression estimation method of \cite{fan2021bridging}, when applied to the problem of VAR parameter estimation, yields an estimator $\wh{\bm\beta}^{\text{FARM}}$ which attains $\max_{1\le j\le p} \vert \wh{\bm\beta}^{\text{FARM}}_{\cdot j} - \bm\beta_{\cdot j} \vert_2 = O_P(\sqrt{s_{\text{\upshape in}}} n^{-1/2 + 5/\nu})$ 
under strong mixingness, see their Theorem~3.
Here, the larger $O_P$-bound compared to ours stems from that $\wh{\bm\beta}^{\text{FARM}}$ requires the estimation of $\xi_{it}$ for all $i$ and $t$, the error from which increases with $n$ as well as $p$. This demonstrates the efficacy of adopting our regularised YW estimator.

Continuing with the same setting, Propositions~\ref{prop:idio:delta} implies that
\begin{align*}
\l\vert \wh{\bm\Delta} - \bm\Delta \r\vert_\infty = O_P\l( \Vert \bm\Delta \Vert_1^2 s_{\text{\upshape in}} \sqrt{\frac{\log(n)}{n}} \r) \text{ and }
\l\vert \wh{\bm\Omega} - \bm\Omega \r\vert_\infty = O_P\l( \l( s_{\text{\upshape out}} \vee
\Vert \bm\Delta \Vert_1^2 \sqrt{s_{\text{\upshape in}}} \r) \sqrt{\frac{s_{\text{\upshape in}} \log(n)}{n}} \r).
\end{align*}
The former is comparable (up to $s_{\text{\upshape in}}$) to the results in Theorem~4 of \cite{cai2011} derived for estimating a sparse precision matrix of independent random vectors.

\section{Forecasting via FNETS}
\label{sec:forecast}


\subsection{Forecasting under the static factor model representation}
\label{sec:2005}


For given time horizon $a \ge 0$, the best linear predictor of $\bm\chi_{n + a}$ 
based on $\bm\chi_{n - \ell}, \, \ell \ge 0$, is
\begin{align}
\label{eq:gdfm:best:lin}
\bm\chi_{n + a \vert n} = \sum_{\ell = 0}^\infty \mbf B_{\ell + a} \mbf u_{n  - \ell}.
\end{align}
under~\eqref{eq:gdfm}.
Following \cite{forni2005generalized}, we consider a forecasting method for the factor-driven component which estimates $\bm\chi_{n + a \vert n}$ under a {\it restricted} GDFM that admits a static representation of finite dimension.  
We formalise the static factor model discussed in Remark~\ref{rem:staticFM} in the following assumption.
\begin{assum}
\label{assum:static}
\it{
\begin{enumerate}[noitemsep, wide, labelwidth=!, labelindent=5pt, label = (\roman*)]
\item \label{cond:static:one} There exist two finite positive integers $m_1$ and $m_2$ such that $m_1 + 1 \ge m_2$, $\bm\chi_t = \mc M^{(1)}(L) \mbf f_t$ and $\mbf f_t=\mc M^{(2)}(L)\mbf u_t$ 
where $\mc M^{(1)}(L) = \sum_{\ell = 0}^{m_1} \mbf M^{(1)}_\ell L^\ell$ with $\mbf M^{(1)} \in \R^{p \times q}$, $\mc M^{(2)}(L) = \sum_{\ell = 0}^{m_2} \mbf M^{(2)}_\ell L^\ell$ with  $\mbf M^{(2)} \in \R^{q \times q}$ and $\det(\mc M^{(2)}(z)) \ne 0$ for all $\vert z \vert \le 1$.

\item \label{cond:static:linear} Let $\mu_{\chi, j}, \, 1 \le j \le r$, denote the $j$-th largest eigenvalue of $\bm\Gamma_\chi(0)$.
Then, there exist a positive integer $p_0 \ge 1$,
constants $\varrho_j \in (7/8, 1]$ with $\varrho_1 \ge \ldots \ge \varrho_r$,
and pairs of positive constants $(\alpha_{\chi, j}, \beta_{\chi, j}), \, 1 \le j \le r$, 
such that for all $p \ge p_0$,
\begin{align*}
& \beta_{\chi, 1} \ge \frac{\mu_{\chi, 1}}{p^{\varrho_1}} \ge \alpha_{\chi, 1} >
\beta_{\chi, 2} \ge \frac{\mu_{\chi, 2}}{p^{\varrho_2}} 
\ge \ldots \ge 
\alpha_{\chi, r - 1} > \beta_{\chi, r} \ge 
\frac{\mu_{\chi, r}}{p^{\varrho_r}} \ge \alpha_{\chi, r} > 0.
\end{align*}
\end{enumerate}
}
\end{assum}

In part~\ref{cond:static:one}, $\bm\chi_t$ admits a static representation with $r = q(m_1 + 1)$ factors: $\bm\chi_t = \bm\Lambda \mbf F_t$, where
$\bm\Lambda = [ \mbf M^{(1)}_\ell, \, 0 \le \ell \le m_1 ]$, $\mbf F_t = (\mbf f_t^\top, \ldots, \mbf f_{t - m_1}^\top )^\top$ and $\mbf f_t = \mc M^{(2)}(L) \mbf u_t$. 
The condition that $m_1 + 1 \ge m_2$ is made for convenience, and the proposed estimator of $\bm\chi_{n + a \vert n}$ can be modified accordingly when it is relaxed. 

\begin{rem}\label{rem:static:FVAR}
Under Assumption~\ref{assum:static}~\ref{cond:static:one}, the $r$-vector of static factors, $\mbf F_t$, is driven by the $q$-dimensional common shocks $\mbf u_t$. 
If $q < r$, \citet{anderson2008generalized} show that $\mbf F_t$ always admits a VAR($h$) representation:
$\mbf F_t = \sum_{\ell = 1}^h \mbf G_\ell \mbf F_{t - \ell} + \mbf H \mbf u_t$ for some finite positive integer $h$ and $\mbf H \in \R^{r \times q}$. 
Then, $\mbf X_t$ has a factor-augmented VAR representation:
\begin{align}
\mbf X_t & = \bm\Lambda \sum_{\ell=1}^h \mbf G_\ell \mbf F_{t-\ell} +  \bm\Lambda \mbf H \mbf u_t + \sum_{\ell = 1}^d \mbf A_\ell \bm\xi_{t-\ell}+\bm\Gamma^{1/2}\bm\vep_t=  \sum_{\ell=1}^{d \vee h} \mbf C_\ell \mbf F_{t-\ell} + \sum_{\ell=1}^d \mbf A_\ell \mbf X_{t-\ell} +\bm \nu_t, \nonumber
\end{align}
with $\mbf C_\ell = \bm \Lambda \mbf G_\ell \mathbb{I}_{\{\ell \le h\}} - \mbf A_\ell \bm\Lambda \mathbb{I}_{\{\ell \le d\}}$ and $\bm \nu_t =  \bm\Lambda \mbf H \mbf u_t + \bm\Gamma^{1/2} \bm\vep_t$.
This model is a generalisation of the factor augmented forecasting model considered by \citet{stock2002forecasting} where only the factor-driven component is present, and it is also considered by \cite{fan2021bridging}.
\end{rem}

It immediately follows from Proposition~\ref{prop:idio:eval} that $\Vert \bm\Gamma_\xi(0) \Vert \le 2\pi B_\xi$.
This, combined with Assumption~\ref{assum:static}~\ref{cond:static:linear},
indicates the presence of a large gap in the eigenvalues of $\bm\Gamma_x(0)$,
which allows the asymptotic identification of $\bm\chi_t$ and $\bm\xi_t$ in the time domain, as well as that of the number of static factors $r$.
Throughout, we treat $r$ as known, and refer to e.g.\ \cite{bai2002, onatski10, ahn2013, trapani2018randomized}, for its estimation.


Let $(\mu_{\chi, j}, \mbf e_{\chi, j}), \, 1 \le j \le r$, denote the pairs of eigenvalues and eigenvectors of $\bm\Gamma_\chi(0)$ ordered such that $\mu_{\chi, 1} \ge \ldots \ge \mu_{\chi, r}$.
Then, $\bm\Gamma_\chi(0) = \mbf E_\chi \bm{\mc M}_\chi \mbf E_\chi^\top$ with $\bm{\mc M}_\chi = \text{diag}(\mu_{\chi, j}, \, 1 \le j \le r)$ and $\mbf E_\chi = [\mbf e_{\chi, j}, \, 1 \le j \le r]$.
Under Assumption~\ref{assum:static}~\ref{cond:static:one}, we have
$\bm\chi_{n + a \vert n}$ in~\eqref{eq:gdfm:best:lin} satisfy
$\bm\chi_{n + a \vert n} = \text{Proj}\l( \bm\chi_{n + a} \vert \mbf F_{n - \ell}, \ell \ge 0\r)
=  \text{Proj}\l( \bm\chi_{n + a} \vert \mbf F_n\r)
= \bm\Gamma_\chi(-a) \mbf E_\chi \bm{\mc M}^{-1}_\chi \mbf E_\chi^\top \bm\chi_n$,
where $\text{Proj}(\cdot \vert \mbf z)$ denotes the linear projection operator
onto the linear space spanned by $\mbf z$.
When $a = 0$, we trivially have $\bm\chi_{t \vert n} = 
\bm\chi_t$ for $t \le n$.
Then, a natural estimator of $\bm\chi_{n + a \vert n}$ is
\begin{align}
\label{eq:sta:chi}
\wh{\bm\chi}^{\static}_{n + a \vert n} = \wh{\bm\Gamma}_\chi(-a) \wh{\mbf E}_\chi \wh{\bm{\mc M}}_\chi^{-1} \wh{\mbf E}_\chi^\top \mbf X_n,
\end{align}
where $(\wh\mu_{\chi, j}, \wh{\mbf e}_{\chi, j}), \, 1 \le j \le r$, denote the pairs of eigenvalues and eigenvectors of $\wh{\bm\Gamma}_\chi(0)$, and $\wh{\bm\Gamma}_\chi(\ell), \, \ell \in \{0, a\}$, are estimated as described in Section~\ref{sec:dpca}.
As a by-product, we obtain the in-sample estimator by setting $a = 0$, as
$\wh{\bm\chi}^{\static}_t = \wh{\mbf E}_\chi \wh{\mbf E}_\chi^\top \mbf X_t$ for $1 \le t \le n$.

\begin{rem}
Our proposed estimator $\wh{\bm\chi}^{\static}_{n + a \vert n}$ differs from that of \cite{forni2005generalized}, as they estimate the factor space via generalised PCA on $\wh{\bm\Gamma}_\chi(0)$.
This in effect replaces $\wh{\mbf E}_\chi$ in~\eqref{eq:sta:chi} with the eigenvectors of $\mbf W^{-1} \wh{\bm\Gamma}_\chi(0)$ where $\mbf W$ is a diagonal matrix containing the estimators of the sample variance of $\bm\xi_t$. 
Such an approach may gain in efficiency compared to ours in the same way a weighted least squares estimator is more efficient than the ordinary one in the presence of heteroscedasticity.
However, since we investigate the consistency of $\wh{\bm\chi}^{\static}_{n + a \vert n}$
without deriving its asymptotic distribution, we do not explore such approach in this paper.
\end{rem}

In Appendix~\ref{sec:2017}, we present an alternative forecasting method that operates under an {\it unrestricted} GDFM, i.e.\ it does not require Assumption~\ref{assum:static}. 
Referred to as $\wh{\bm\chi}^{\va}_{n + a \vert n}$, we compare its performance with that of
$\wh{\bm\chi}^{\static}_{n + a \vert n}$ in numerical studies. 


Once VAR parameters are estimated by $\wh{\bm\beta} = [\wh{\mbf A}_1, \ldots, \wh{\mbf A}_d]^\top$ as in~\eqref{eq:lasso}, we produce a forecast of $\bm\xi_{n + a}$ given $\mbf X_t, \, t \le n$, by estimating the best linear predictor $\bm\xi_{n + a \vert n} = \sum_{\ell = 1}^d \mbf A_{\ell} \bm\xi_{n + 1 - \ell \vert n}$ (with $\bm\xi_{t \vert n} = \bm\xi_t$ for $t \le n$), as
\begin{align}
\label{eq:idio:best:lin:pred}
\wh{\bm\xi}_{n + a \vert n} = \sum_{\ell = 1}^{\max(1, a) - 1} \wh{\mbf A}_{\ell} \wh{\bm\xi}_{n + a - \ell \vert n} + \sum_{\ell = \max(1, a)}^{d} \wh{\mbf A}_{\ell} \wh{\bm\xi}_{n + a - \ell}.
\end{align}
When $a \le d$, the in-sample estimators appearing in~\eqref{eq:idio:best:lin:pred} are obtained as $\wh{\bm\xi}_t = \mbf X_t - \wh{\bm\chi}_t, \,  n + a - d \le t \le n$, with either $\wh{\bm\chi}_t^{\static}$ or $\wh{\bm\chi}_t^{\va}$ as $\wh{\bm\chi}_t$.

\subsection{Theoretical properties}
\label{sec:fnets:forecast:theory}

Proposition~\ref{thm:2005} establishes the consistency of $\wh{\bm\chi}^{\static}_{n + a \vert n}$ in estimating the best linear predictor of $\bm\chi_{n + a}$, where we make it explicit the effects of the presence of weak factors, both dynamic (as measured by $\mu_{\chi, j}(\omega)$ in Assumption~\ref{assum:factor}) and static (as measured by $\mu_{\chi, j}$ in Assumption~\ref{assum:static}~\ref{cond:static:linear}),
and the tail behaviour (through $\psi_n$ and $\vartheta_{n, p}$ defined in~\eqref{eq:rates:one}).

\begin{prop}
\label{thm:2005}
\it{Suppose that the conditions in Theorem~\ref{thm:common:spec} are met
and, in addition, we assume that Assumption~\ref{assum:static} holds.
Then, for any  finite $a \ge 0$, we have
\begin{align*}
\l\vert \wh{\bm\chi}^{\static}_{n + a \vert n}  - \bm\chi_{n + a \vert n} \r\vert_\infty
= O_P\l(  p^{4 - 2 \rho_q - 2\varrho_r}
\l(\psi_n \vee  p^{\varrho_r - 1} \vartheta_{n, p} 
\vee \frac{1}{m} \vee \frac{1}{\sqrt p}\r) \r).
\end{align*}}
\end{prop}

As noted in Remark~\ref{rem:thm:dpca}~\ref{rem:tails}, weaker factors and heavier tails impose a stronger requirement on the dimensionality $p$.
If all factors are strong ($\varrho_r = 1$), the rate becomes $(\vartheta_{n, p} \vee m^{-1} \vee p^{-1/2})$.
When $a = 0$, Proposition~\ref{thm:2005} provides in-sample estimation consistency for any given $t \le n$. 
The next proposition accounts for the irreducible error in $\bm\chi_{n + a \vert n}$, with which we conclude the analysis of the forecasting error $\vert \wh{\bm\chi}_{n + a \vert n}^{\static} - \bm\chi_{n + a} \vert_\infty$ when $a \ge 1$.
\begin{prop}
\label{prop:common:irreduce}
{\it Suppose that Assumptions~\ref{assum:common} and~\ref{assum:innov:new} hold.
Then for any finite $a \ge 1$, $\vert \bm\chi_{n + a \vert n} - \bm\chi_{n + a} \vert_\infty = 
O_P(q^{1/\nu} \mu_\nu^{1/\nu} \log^{1/2}(p))$.}
\end{prop}

Recall the definition of $s_{\text{\upshape in}}$ given in Section~\ref{sec:fnets:network:theory}.
The next proposition investigates the performance of $\wh{\bm\xi}_{n + a \vert n}$ when $a = 1$, which can easily be extended to any finite $a \ge 2$.

\begin{prop}
\label{prop:idio:pred}
\textit{
Suppose that the in-sample estimator of $\bm\xi_t$ and $\wh{\bm\beta}$ satisfy
\begin{align}
\l\vert \wh{\bm\xi}_{n + 1 - \ell} - \bm\xi_{n + 1 - \ell} \r\vert_\infty = O_P\l(\bar{\zeta}_{n, p}\r) 
\text{ \ for $1 \le \ell \le d$}
\text{ \ and \ }
\l\Vert \wh{\bm\beta} - \bm\beta \r\Vert_1 = O_P(s_{\text{\upshape in}} \zeta_{n, p}).
\label{cond:idio:pred:two}
\end{align}
Also, 
let Assumptions~\ref{assum:idio} and~\ref{assum:innov:new} hold.
Then, 
\begin{align*}
\l\vert \wh{\bm\xi}_{n + 1 \vert n} - \bm\xi_{n + 1} \r\vert_\infty
&= O_P\l( s_{\text{\upshape in}} \zeta_{n, p} \l( \log^{1/2}(p) p^{1/\nu} \mu_\nu^{1/\nu} + \bar{\zeta}_{n, p} \r) + \Vert \bm\beta \Vert_1 \bar{\zeta}_{n, p}
+ p^{1/\nu} \mu_\nu^{1/\nu}\r).
\end{align*}
}
\end{prop}

Either of the in-sample estimators $\wh{\bm\chi}^{\static}_t$ (described in Section~\ref{sec:2005}) or $\wh{\bm\chi}^{\va}_t$ (Appendix~\ref{sec:2017}), can be used in place of $\wh{\bm\chi}_t$.
Accordingly, the rate $\bar{\zeta}_{n, p}$ in~\eqref{cond:idio:pred:two} is inherited by that of $\wh{\bm\chi}^{\static}_t$ (given in Proposition~\ref{thm:2005}) or $\wh{\bm\chi}^{\va}_t$ (Proposition~\ref{thm:common:var}~\ref{thm:common:var:three}).
From the proof of Proposition~\ref{prop:idio:est:lasso}, we have $\zeta_{n, p} \asymp (\Vert \bm\beta \Vert_1 + 1)(\vartheta_{n, p} \vee m^{-1} \vee p^{-1/2})$ in~\eqref{cond:idio:pred:two}.

\section{Numerical studies}
\label{sec:num}

\subsection{Tuning parameter selection}
\label{sec:tuning}


We briefly discuss the choice of the tuning parameters for FNETS.
For full details, see \cite{owens2023fnets}. 

\paragraph{Related to $\bm\chi_t$.}
We set the kernel bandwidth at $m = \lfloor 4 (n/\log(n))^{1/3} \rfloor$ based on the case when sufficiently large number of moments exist and $n \asymp p$ (Remark~\ref{rem:thm:dpca}~\ref{rem:bandwidth}).
In simulation studies reported in Appendix~\ref{sec:complete:sim}, we treat the number of factors $q$ (required for Step~1 of FNETS) known, and also treat the number of static factors $r$ (for generating the forecast) as known if it is finite;
when $\bm\chi_t$ does not admit a static factor model (i.e.\ $r = \infty$), we use the value returned by the ratio-based estimator of \cite{ahn2013}.
In real data analysis reported in Section~\ref{sec:real}, we estimate both $q$ and $r$, the former with the estimator proposed in \cite{hallin2007}, the latter as in \cite{ahn2013}.

\paragraph{Related to $\bm\xi_t$.}
We select the tuning parameter $\lambda$ in~\eqref{eq:lasso} 
jointly with the VAR order $d$, by adopting cross validation (CV); in time series settings, a similar approach is explored in \cite{wang2021robust}.
For this, the data is partitioned into $M$ consecutive folds with indices $\mc I_l = \{n_l + 1, \ldots, n_{l + 1}\}$ where $n_l = \min(l\lceil n/M \rceil, n), \, 0 \le l \le M$, and each fold is split into $\mc I^{\text{train}}_l = \{n_l + 1, \ldots, \lceil (n_l + n_{l + 1})/2 \rceil\}$ and $\mc I^{\text{test}}_l = \mc I_l \setminus \mc I^{\text{train}}_l$.
Then with $\wh{\bm\beta}^{\text{train}}_l(\mu, b)$ obtained from $\{\mbf X_t, \, t \in \mc I^{\text{train}}_l\}$ with the tuning parameter $\mu$ and the VAR order $b$, we evaluate 
\begin{align*}
\text{CV}(\mu, b) = \sum_{l = 1}^M & \,
\tr\l( \wh{\bm\Gamma}^{\text{test}}_{\xi, l}(0) - (\wh{\bm\beta}_l^{\text{train}}(\mu, b))^\top \wh{\bbg}_l^{\text{test}}(b) - 
(\wh{\bbg}_l^{\text{test}}(b))^\top \wh{\bm\beta}_l^{\text{train}}(\mu, b)
\r.
\\
& \l. +  (\wh{\bm\beta}_l^{\text{train}}(\mu, b))^\top \wh{\bbG}_l^{\text{test}}(b)
\wh{\bm\beta}_l^{\text{train}}(\mu, b) \r),
\end{align*}
where $\wh{\bm\Gamma}^{\text{test}}_{\xi, l}(\ell)$, $\wh{\bbG}_l^{\text{test}}(b)$ and $\wh{\bbg}_l^{\text{test}}(b)$ are generated analogously as $\wh{\bm\Gamma}_\xi(\ell)$, $\wh{\bbG}$ and $\wh{\bbg}$, respectively, using the test set $\{\mbf X_t, \, t \in \mc I^{\text{test}}_l\}$.
The measure $\text{CV}(\mu, b)$ approximates the prediction error while accounting for that we do not directly observe $\bm\xi_t$.
Minimising it over varying $\mu$ and $b$, we select $\lambda$ and $d$.
In simulation studies, we treat $d$ as known
while in real data analysis, we select it from the set $\{1, \ldots, 5\}$ via CV.
For selecting $\eta$ in~\eqref{eq:inv:delta}, we adopt the Burg matrix divergence-based CV measure:
\begin{align*}
\text{CV}(\mu) = \sum_{l = 1}^M \tr\l( \wh{\bm\Delta}_l^{\text{train}}(\mu) \wh{\bm\Gamma}_l^{\text{test}} \r) - \log \l\vert \wh{\bm\Delta}_l^{\text{train}}(\mu) \wh{\bm\Gamma}_l^{\text{test}} \r\vert - p.
\end{align*}
For both CV procedures, we set $M = 1$ in the numerical results reported below.
In simulation studies, we compare the estimators with their thresholded counterparts in estimating the network edge sets with the thresholds $\mathfrak{t}$, $\mathfrak{t}_\delta$ and $\mathfrak{t}_\omega$ selected according to a data-driven approach motivated by \cite{liu2021simultaneous}.
Details are in Appendix~\ref{sec:thresh}.

\subsection{Simulations}
\label{sec:sim}

In Appendix~\ref{sec:complete:sim}, we investigate the estimation and forecasting performance of FNETS ondatasets simulated under a variety of settings, from Gaussian innovations $\mbf u_t$ and $\bm\vep_t$ with~\ref{e:one}~$\bm\Delta = \mbf I$ and~\ref{e:two}~$\bm\Delta \ne \mbf I$, to~\ref{e:three}~heavy-tailed ($t_5$) innovations with $\bm\Delta = \mbf I$,
and when $\bm\chi_t$ is generated from~\ref{m:ar}~fully dynamic or~\ref{m:ma}~static factor models.
In addition, we consider the `oracle' setting~\ref{m:oracle} $\bm\chi_t = \mbf 0$ where, in the absence of the factor-driven component, the results obtained can serve as a benchmark. 
For comparison, we consider the factor-adjusted regression method of \cite{fan2021bridging} and present the performance of their estimator of VAR parameters and forecasts.

\subsection{Application to a panel of volatility measures}
\label{sec:real}

We investigate the interconnectedness in a panel of volatility measures 
and evaluate its out-of-sample forecasting performance using FNETS.
For this purpose, we consider a panel of $p = 46$ stock prices
retrieved from the Wharton Research Data Service,
of US companies which are all classified as `financials' according 
to the Global Industry Classification Standard;
a list of company names and industry groups are found in Appendix~\ref{app:data}.
The dataset spans the period between January 3, 2000 and December 31, 2012 ($3267$ trading days).
Following \cite{diebold2014network}, we measure the volatility using the high-low range as
$\sigma_{it}^2 = 0.361 (p^{\text{high}}_{it} - p^{\text{low}}_{it})^2$
where $p^{\text{high}}_{it}$ and $p^{\text{low}}_{it}$ denote, respectively,
the maximum and the minimum log-price of stock~$i$ on day~$t$,
and set $X_{it} = \log(\sigma_{it}^2)$; \citet{brownlees2010comparison} support this choice of volatility measure over more sophisticated alternatives.

\subsubsection{Network analysis}
\label{sec:real:network}

We focus on the period 03/2006--02/2010 corresponding to the Great Financial Crisis. 
We partition the data into four segments of length $n = 252$ each
(corresponding to the number of trading days in a single year)
and on each segment, we apply FNETS to estimate the three networks
$\mc N^{\dir}$, $\mc N^{\undir}$ and $\mc N^{\lr}$ described in Section~\ref{sec:networks}. 

Each row of Figure~\ref{fig:real:lasso} plots the heat maps of the matrices
underlying the three networks of interest.
From all four segments, the CV-based approach described in Section~\ref{sec:tuning} returns $d = 1$ from the candidate VAR order set $\{1, \ldots, 5\}$.
Hence in each row, the left panel represents the estimator $\wh{\mbf A}_1 = \wh{\bm\beta}^\top$,
and the middle and the right show the (long-run) partial correlations from the corresponding $\wh{\bm\Delta}$ and $\wh{\bm\Omega}$ (with their diagonals set to be zero).
Locations of the non-zero elements estimate the edge sets of the corresponding networks, and the hues represent the (signed) edge weights.

Prior to March 2007, all networks exhibit a low degree of interconnectedness
but the number of edges increases considerably in 03/2007--02/2008 
due mainly to an overall increase in dynamic co-dependencies
and a prominent role of banks (blue group)
not only in $\mc N^{\dir}$ but also in $\mc N^{\undir}$. 
In 03/2008--02/2009, the companies belonging to the insurance sector (red group) 
play a central role and in 03/2009--02/2010,
the companies become highly interconnected 
with two particular firms having many outgoing edges in $\mc N^{\dir}$.
Also, while most edges in $\mc N^{\lr}$, which captures the overall long-run dependence,
have positive weights across time and companies,
their weights become negative in this last segment.
We highlight that FNETS is able to capture the aforementioned group-specific activities although this information is not supplied to the estimation method.

\begin{figure}[htb!]
\centering
\begin{tabular}{ccc}
\footnotesize $\mc N^{\dir}$ & \footnotesize $\mc N^{\undir}$ & \footnotesize $\mc N^{\lr}$
\\
\includegraphics[width = .21\textwidth]{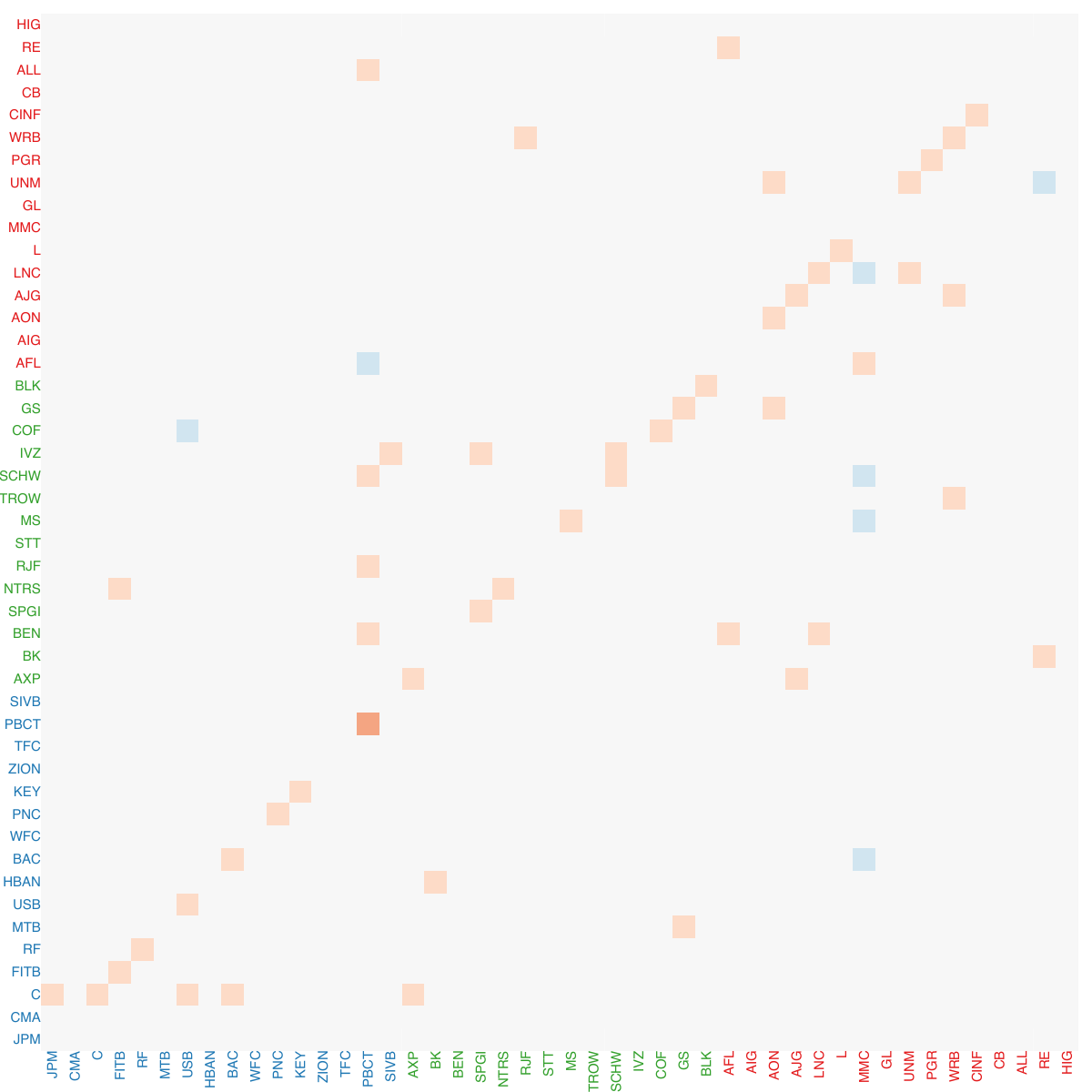} 
&\includegraphics[width = .21\textwidth]{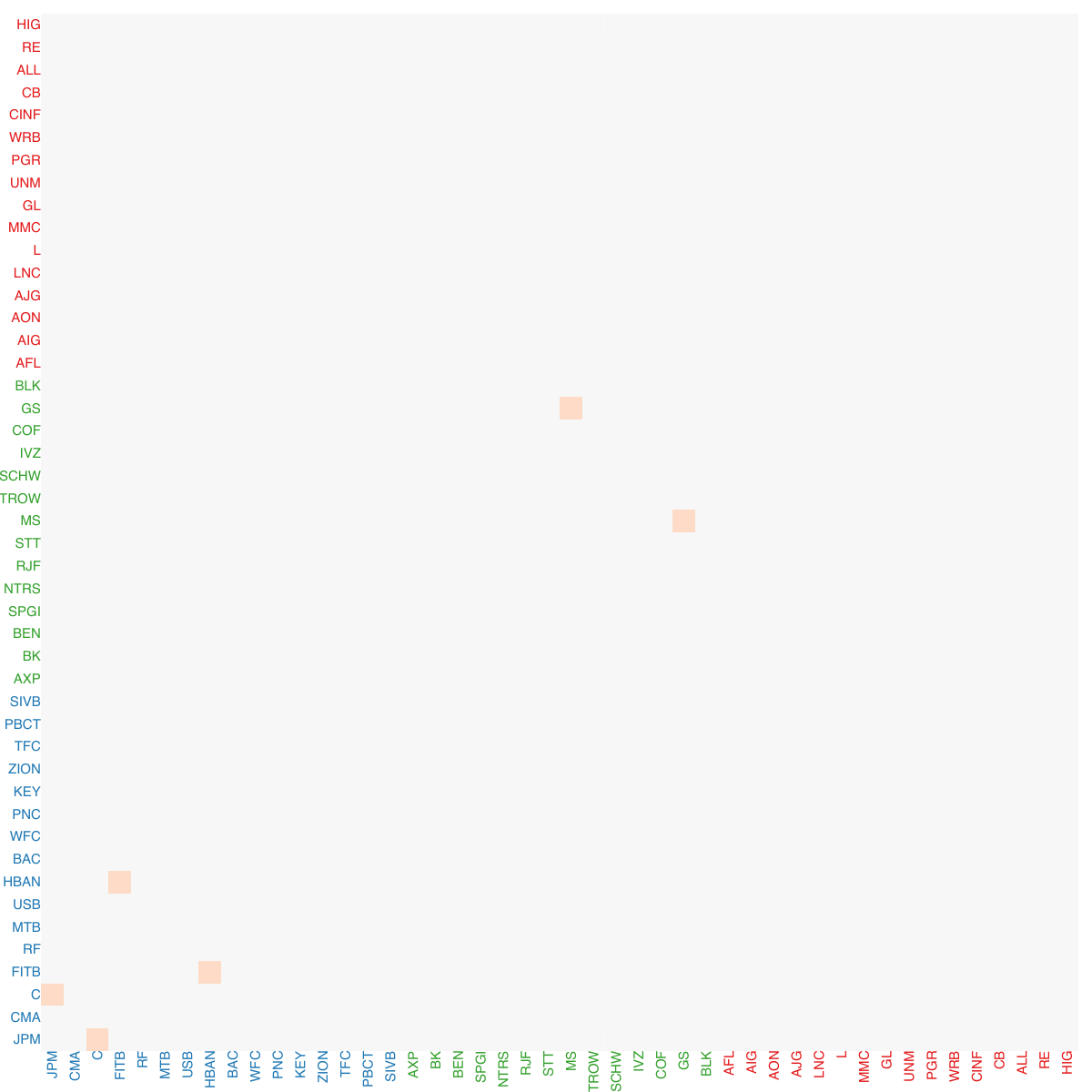} 
&\includegraphics[width = .21\textwidth]{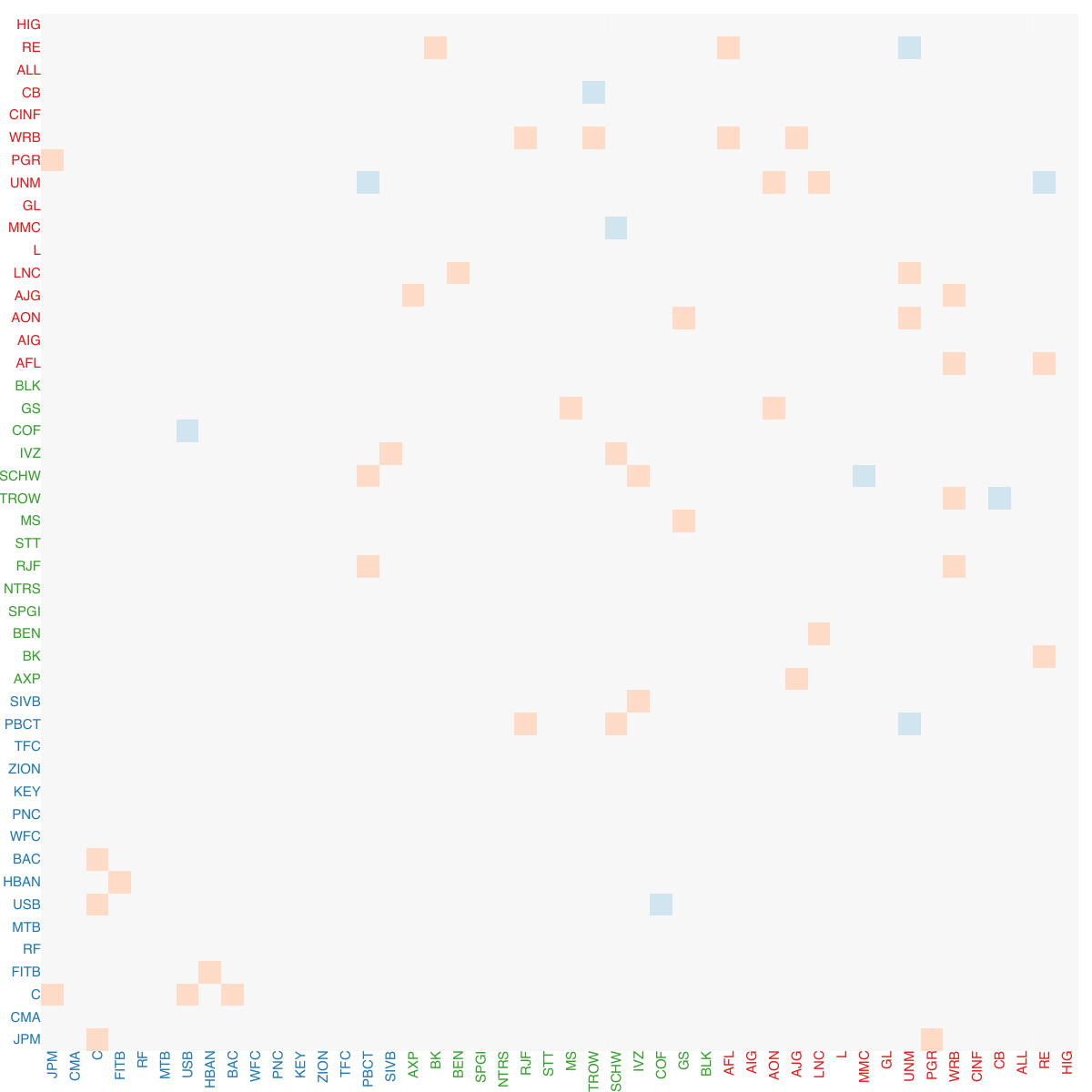} \\[-.1pt]
\multicolumn{3}{c}{\footnotesize 03/2006--02/2007}
\\
\includegraphics[width = .21\textwidth]{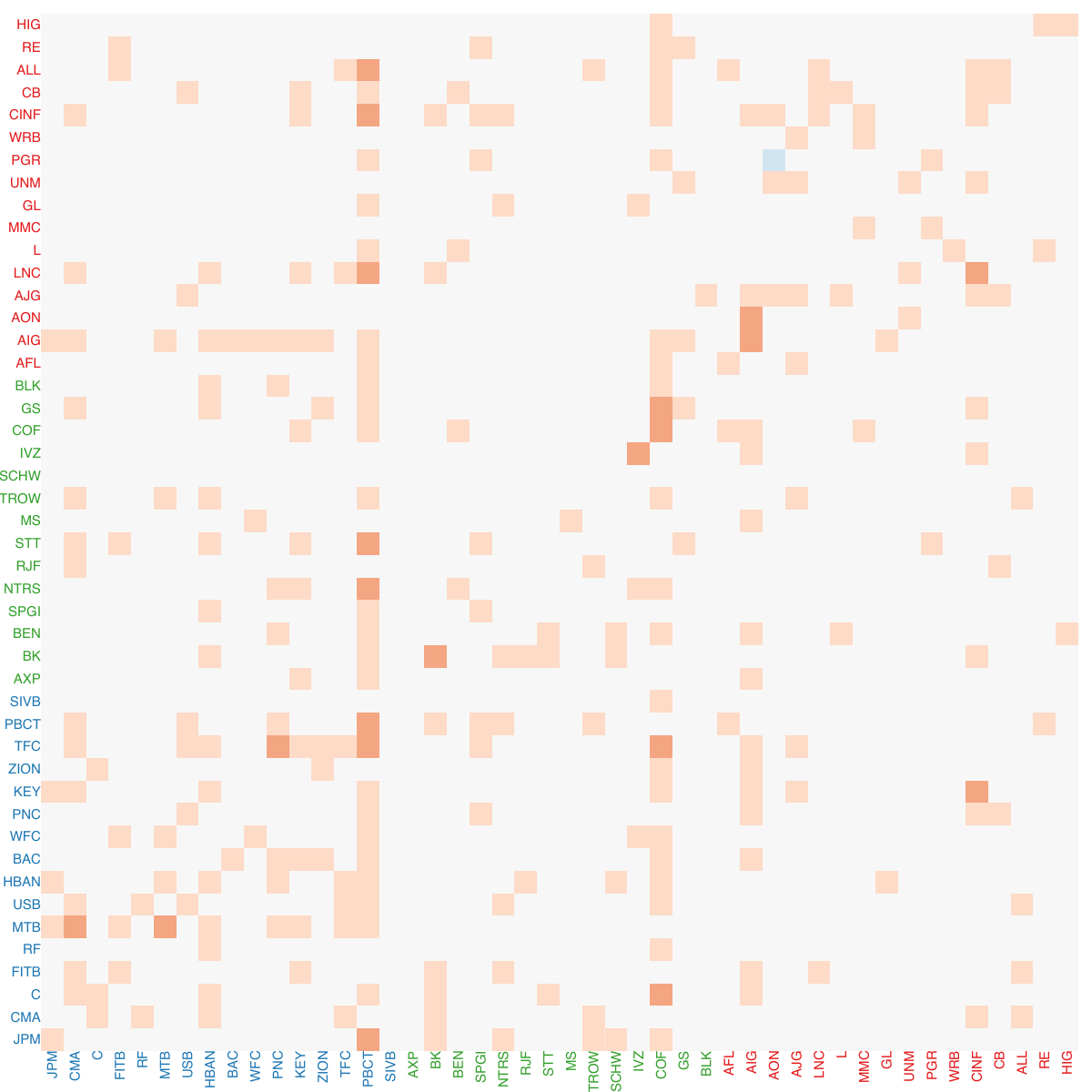} 
&\includegraphics[width = .21\textwidth]{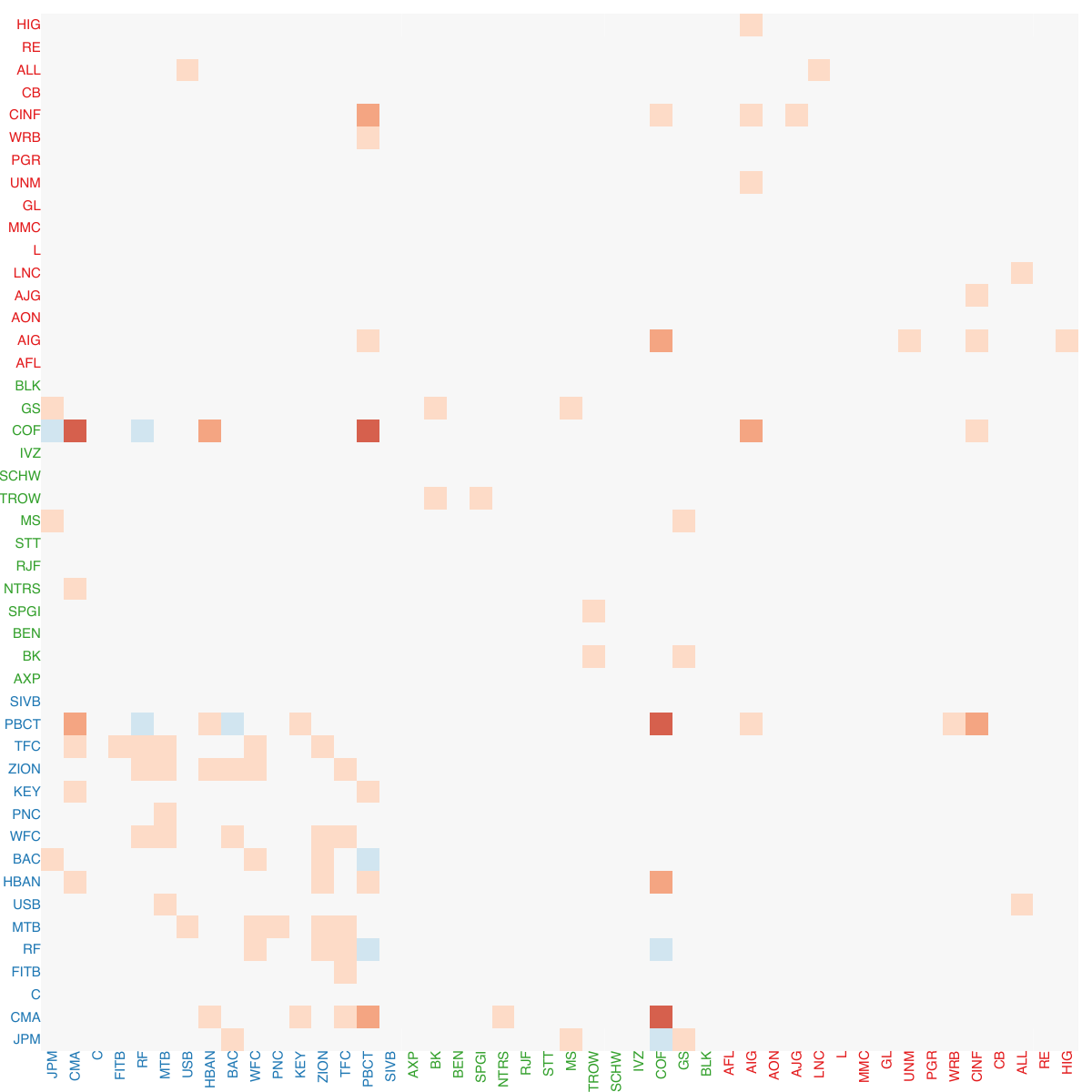} 
&\includegraphics[width = .21\textwidth]{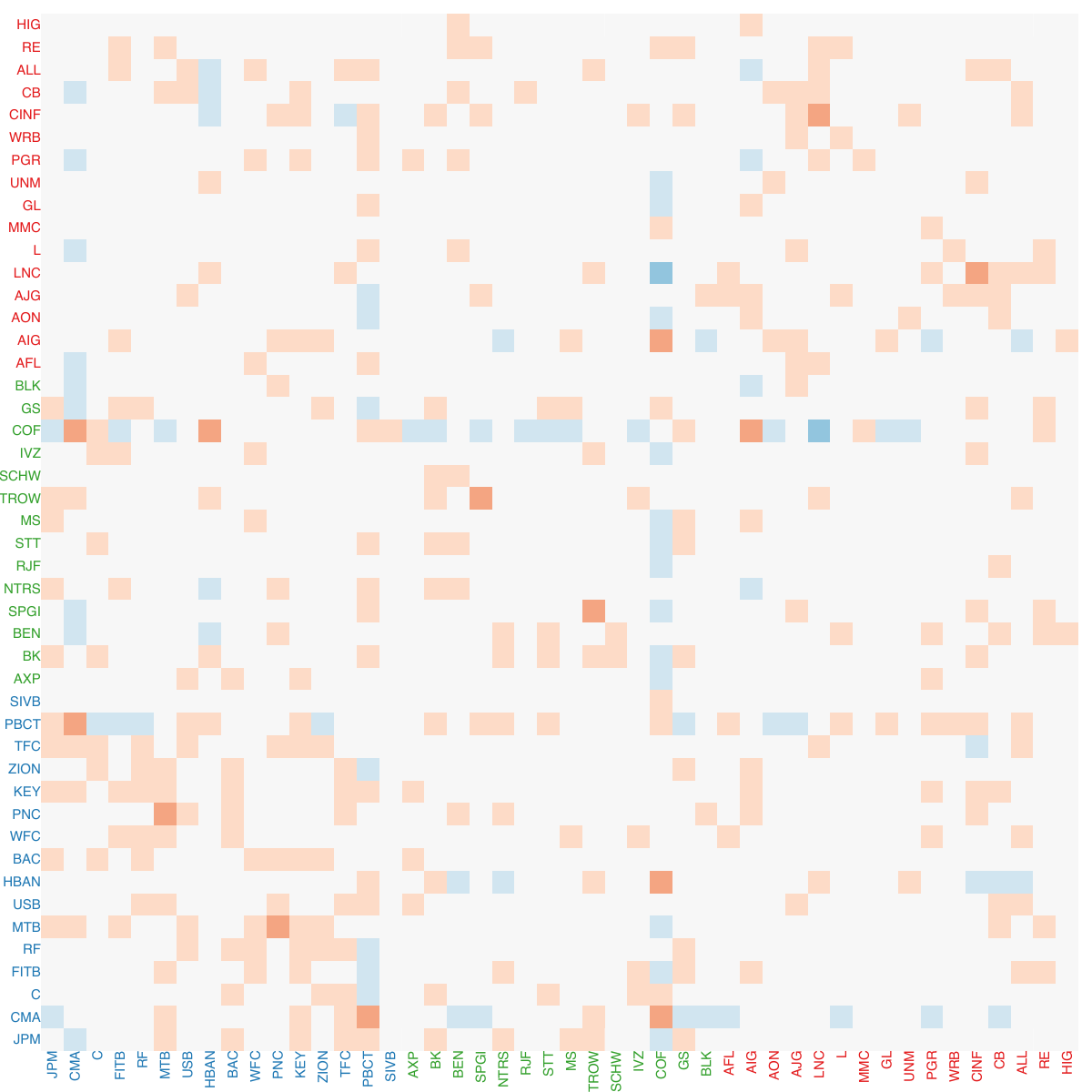} \\[-.1pt]
\multicolumn{3}{c}{\footnotesize 03/2007--02/2008}
\\
\includegraphics[width = .21\textwidth]{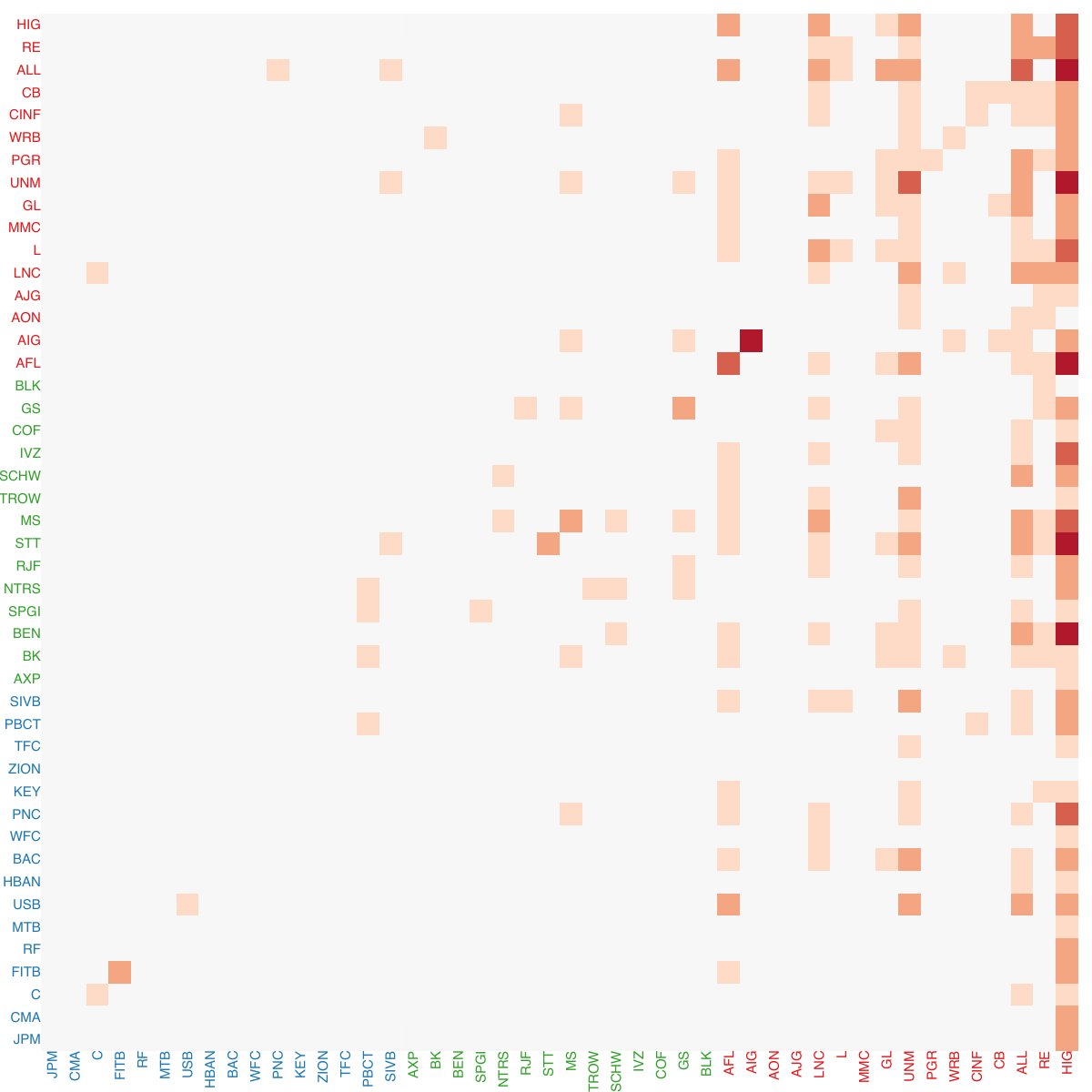} 
&\includegraphics[width = .21\textwidth]{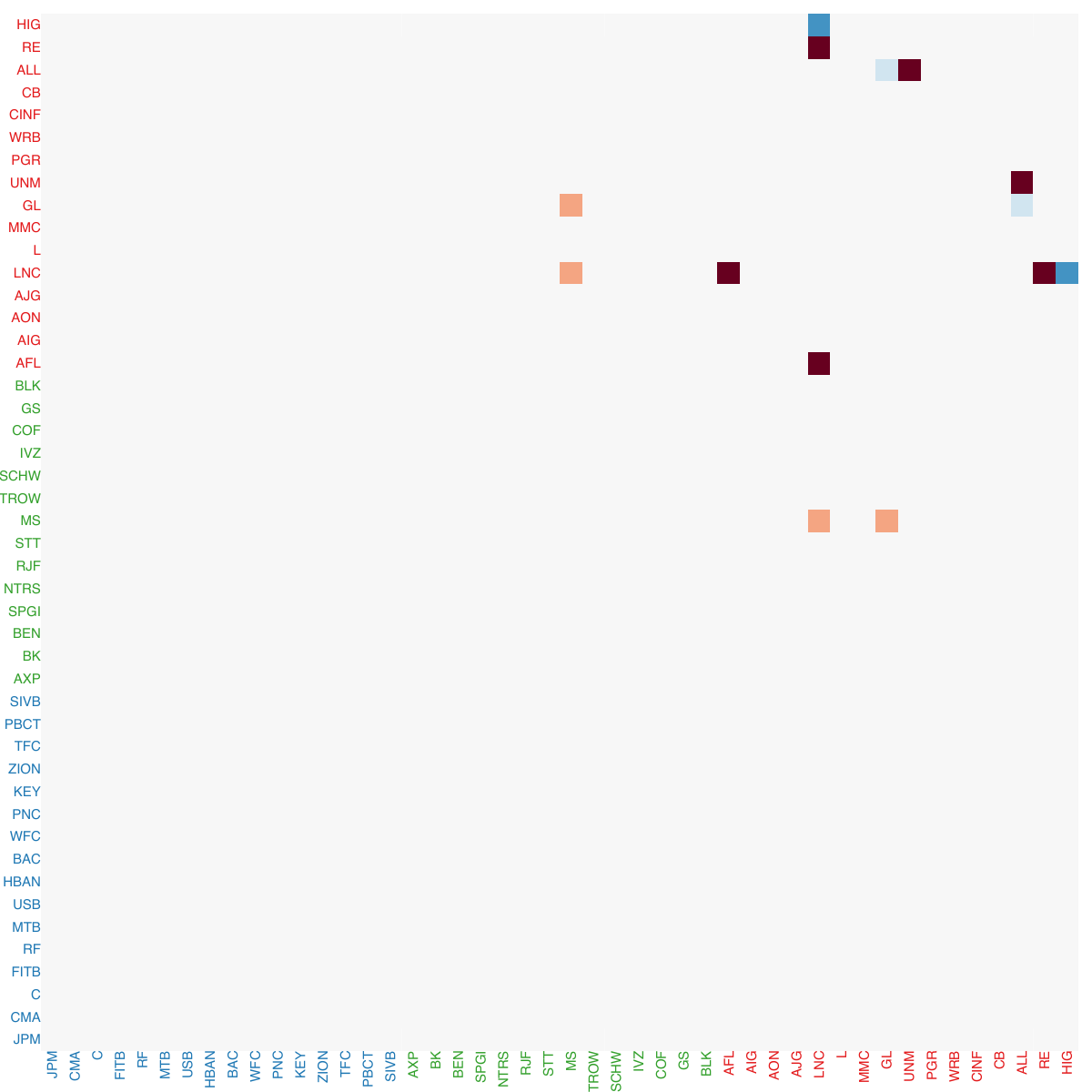} 
&\includegraphics[width = .21\textwidth]{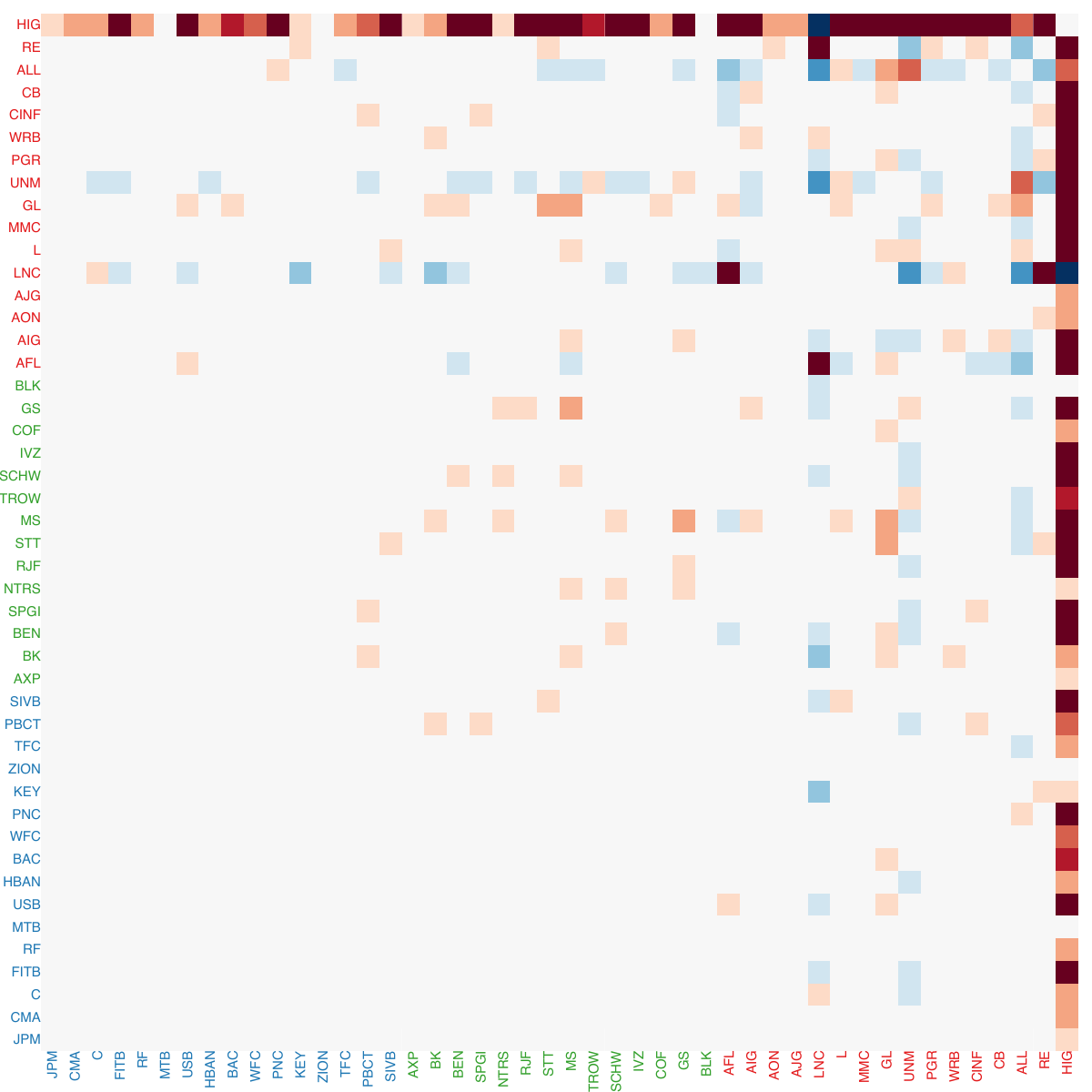} \\[-.1pt]
\multicolumn{3}{c}{\footnotesize 03/2008--02/2009}\\
\\
\includegraphics[width = .21\textwidth]{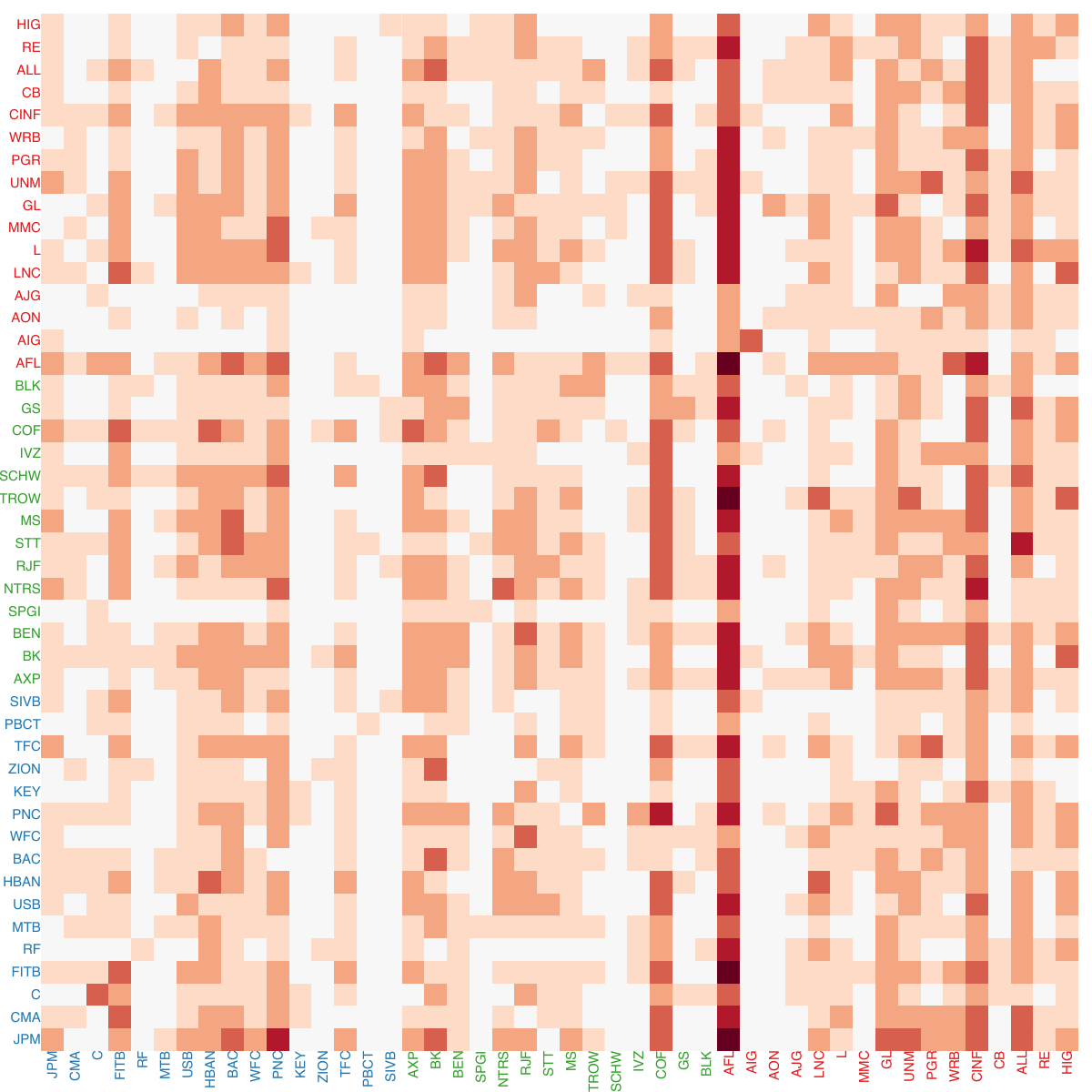} 
&\includegraphics[width = .21\textwidth]{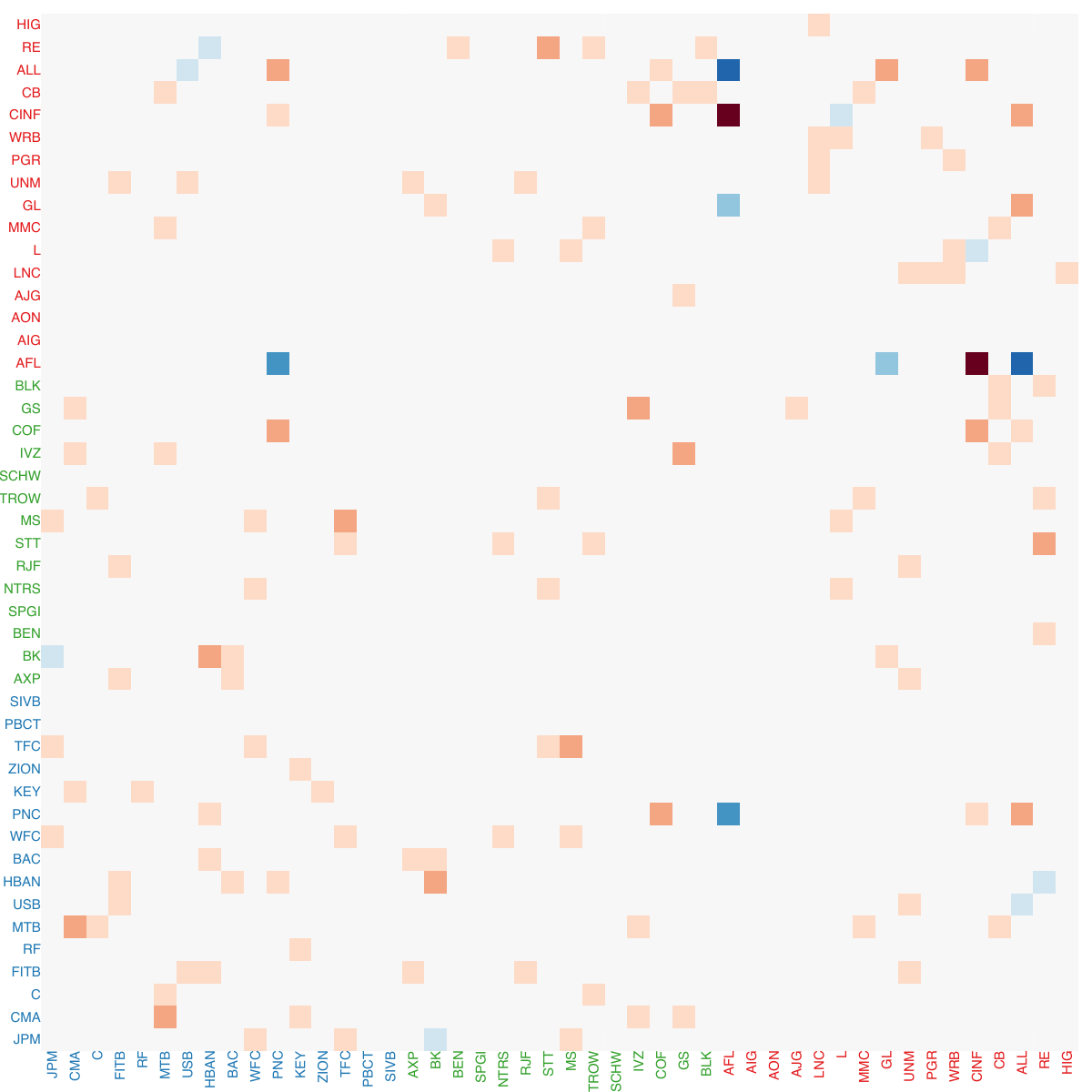} 
&\includegraphics[width = .21\textwidth]{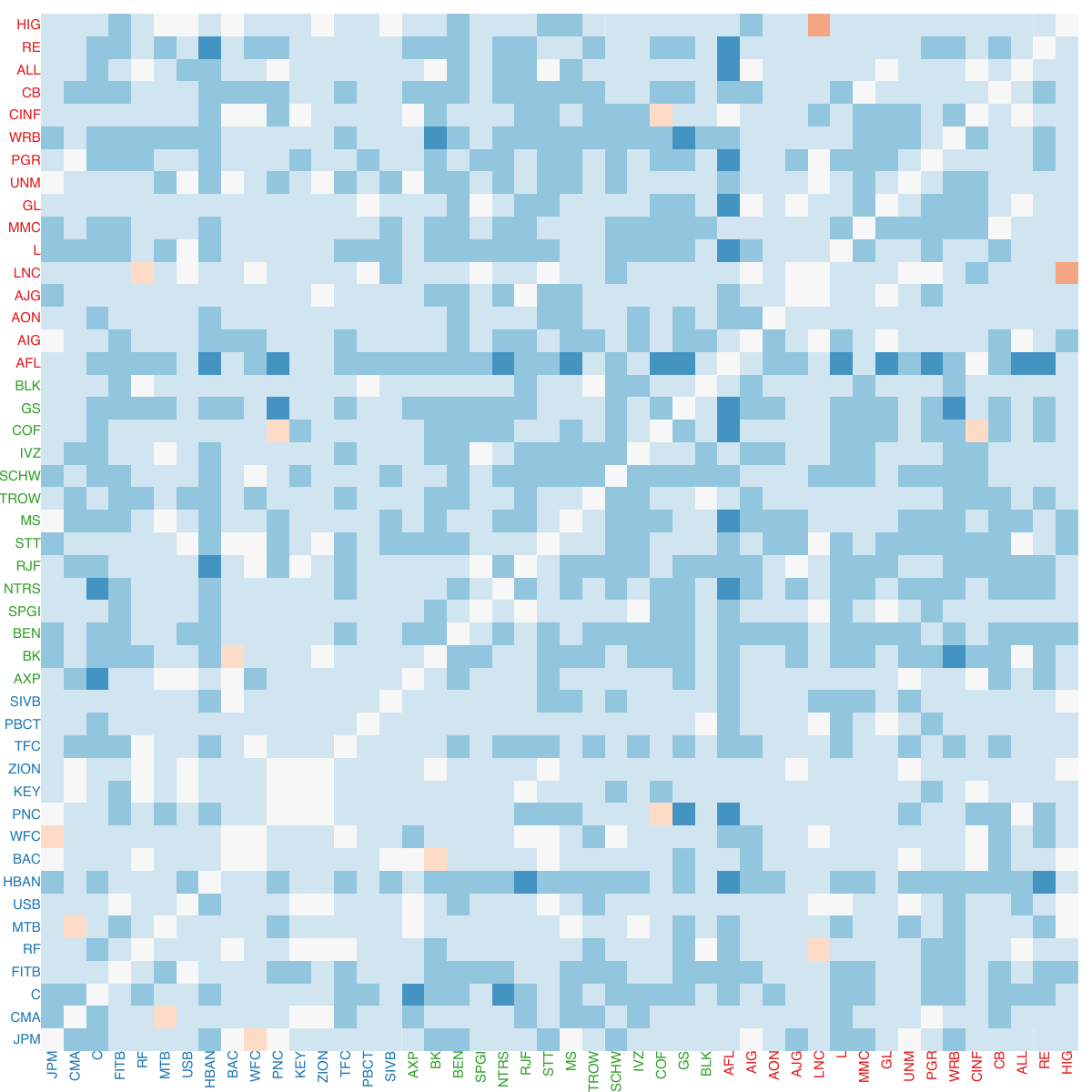}\\ [-.1pt]
\multicolumn{3}{c}{\footnotesize 03/2009--02/2010}\\ [1pt]
\end{tabular}

\caption{\small Heat maps of the estimators of the VAR transition matrices $\wh{\mbf A}_1$, partial correlations from $\wh{\bm\Delta}$ and long-run partial correlations from $\wh{\bm\Omega}$ (left to right),
which in turn estimate the networks $\mc N^{\dir}$, $\mc N^{\undir}$ and $\mc N^{\lr}$,
respectively,
over three selected periods.
The grouping of the companies according to their industry classifications 
are indicated by the axis label colours.
The heat maps in the left column are in the scale of $[-0.81, 0.81]$
while the others are in the scale of $[-1, 1]$, 
with red hues denoting large positive values and blue hues large negative values.}
\label{fig:real:lasso}
\end{figure}


\subsubsection{Forecasting}
\label{sec:real:forecast}

\begin{table}[htb!]
\caption{\small Mean, median and standard errors of $\text{FE}^{\text{avg}}_{T + 1}$ and 
$\text{FE}^{\text{max}}_{T + 1}$ on the trading days in $2012$
for $\wh{\mbf X}_{T + 1 \vert T}(n)$
in comparison with AR and FarmPredict \citep{fan2021bridging} forecasts.
The smallest entry in each row is highlighted in bold.}
\label{table:real:forecast}
\centering
{\footnotesize
\begin{tabular}{cc cc cc}
\toprule										
&	&	\multicolumn{2}{c}{FNETS} &				&			\\	
&	&	$\wh{\bm\chi}^{\static}_{T + 1 \vert T}(n)$ &	 $\wh{\bm\chi}^{\va}_{T + 1 \vert T}(n)$ & AR &	 	FarmPredict 	\\	\cmidrule(lr){1-2} \cmidrule(lr){3-4} \cmidrule(lr){5-5} \cmidrule(lr){6-6}
$\text{FE}^{\text{avg}}$ &	Mean &	 {\bf 0.7258} &	0.7466 &	0.7572 &	 	0.7616	\\	
&	Median &	{\bf 0.6029} &	0.6412 &	0.6511 &	 	0.6243	\\	
&	SE &	0.4929 &	{\bf 0.3748} &	0.4162 &	 	0.4946	\\	
\cmidrule(lr){1-2} \cmidrule(lr){3-4} \cmidrule(lr){5-5} \cmidrule(lr){6-6}
$\text{FE}^{\text{max}}$ &	Mean & {\bf 	0.8433} &	0.8729 &	0.879 &	 	0.8745	\\	
&	Median &	{\bf 0.7925} &	0.8088 &	0.8437 &	 	0.8259	\\	
&	SE &	0.2331 &	 0.2246 &		{\bf 0.2169} &	 	0.2337	\\	\bottomrule
\end{tabular}
}
\end{table}

We perform a rolling window-based forecasting exercise
on the trading days in $2012$.
Starting from $T = 3016$ (the first trading day in $2012$), we forecast $\mbf X_{T + 1}$ as
$\wh{\mbf X}_{T + 1 \vert T}(n) = \wh{\bm\chi}_{T + 1 \vert T}(n) + \wh{\bm\xi}_{T + 1 \vert T}(n)$,
where $\wh{\bm\chi}_{T + 1 \vert T}(n)$ (resp. $\wh{\bm\xi}_{T + 1 \vert T}(n)$)
denotes the forecast of $\bm\chi_{T + 1}$ (resp. $\bm\xi_{T + 1}$) using 
the preceding $n$ data points $\{\mbf X_t, \, T - n + 1 \le t \le T\}$. We set $n = 252$.
After the forecast $\wh{\mbf X}_{T + 1 \vert T}(n)$ is generated, 
we update $T \leftarrow T + 1$ and repeat the above procedure until $T = 3267$ 
(the last trading day in $2012$) is reached.

For $\wh{\bm\chi}_{T + 1 \vert T}(n)$, we consider the forecasting methods
derived under the static factor model (Section~\ref{sec:2005}, denoted by $\wh{\bm\chi}^{\static}_{T + 1 \vert T}(n)$) 
and unrestricted GDFM (Appendix~\ref{sec:2017}, $\wh{\bm\chi}^{\va}_{T + 1 \vert T}(n)$).
Following the analysis in Section~\ref{sec:real:network}, we set $d = 1$ when producing $\wh{\bm\xi}_{T + 1 \vert T}(n)$.
Additionally, we report the forecasting performance of FarmPredict \citep{fan2021bridging}, which first fits an AR model to each of the $p$ series (`AR'), projects the residuals on their principal components, and then fits VAR models to what remains via Lasso. 
Combining the three steps gives the final forecast $\wh{\mbf X}^{\text{FARM}}_{T + 1 \vert T}(n)$.
The forecast produced by the first step univariate AR modelling,
denoted by $\wh{\mbf X}^{\text{AR}}_{T + 1 \vert T}(n)$, is also included for comparison.

We evaluate the performance of $\wh{\mbf X}_{T + 1 \vert T}$ 
using two measures of errors
$\text{FE}^{\text{avg}}_{T + 1} = 
\vert \mbf X_{T + 1} \vert_2^{-2} \cdot \vert \mbf X_{T + 1} - \wh{\mbf X}_{T + 1 \vert T} \vert_2^2$ and
$\text{FE}^{\text{max}}_{T + 1} = 
\vert \mbf X_{T + 1} \vert_\infty^{-1} \cdot \vert \mbf X_{T + 1} - \wh{\mbf X}_{T + 1 \vert T} \vert_\infty$,
see Table~\ref{table:real:forecast} for the summary of the forecasting results.
Among the forecasts generated by FNETS, the one based on $\wh{\bm\chi}^{\static}_{T + 1 \vert T}(n)$ performs the best in this exercise,
which outperforms $\wh{\mbf X}^{\text{AR}}_{T + 1 \vert T}(n)$
and $\wh{\mbf X}^{\text{FARM}}_{T + 1 \vert T}(n)$ according 
to both $\text{FE}^{\text{avg}}$ and $\text{FE}^{\text{max}}$ on average. 
As noted in Appendix~\ref{app:sim:forecast}, the forecast based on $\wh{\bm\chi}^{\va}_{T + 1 \vert T}$ shows instabilities and generally is outperformed by the one based on $\wh{\bm\chi}^{\static}_{T + 1 \vert T}$, but nonetheless performs reasonably well.
Given the high level of co-movements and persistence in the data, the good performance of FNETS is mainly attributed to the way we forecast the factor-driven component, which is based on the estimators derived under GFDM that fully exploit all the dynamic co-dependencies (see also the results obtained by \citealp{barigozzi2017generalized} on a similar dataset).


\section{Conclusions}
\label{sec:conc}

We propose and study the asymptotic properties of FNETS, a network estimation and forecasting methodology
for high-dimensional time series under a dynamic factor-adjusted VAR model.
Our estimation strategy fully takes into account the latency of the VAR process of interest 
via regularised YW estimation which, distinguished from the existing approaches, brings in methodological simplicity as well as theoretical benefits.
We investigate the theoretical properties of FNETS under general conditions permitting weak factors and heavier tails than sub-Gaussianity commonly imposed in the high-dimensional VAR literature, and provide new insights into the interplay between various quantities determining the sparsity of the networks underpinning VAR processes, factor strength and tail behaviour, on the estimation of those networks.
Simulation studies and an application to a panel of financial time series show that FNETS is particularly useful for network analysis as it is able to discover  group structures as well as producing accurate forecasts for highly co-moving and persistent time series such as log-volatilities. 
The R software {\tt fnets} implementing FNETS is available from CRAN \citep{fnets}.

\bibliographystyle{apalike}
\bibliography{fbib}

\clearpage

\appendix

\numberwithin{equation}{section}
\numberwithin{figure}{section}
\numberwithin{table}{section}

\section{Estimation of VAR parameters and $\mc N^{\dir}$ via Dantzig selector estimator}
\label{sec:beta:ds}

Recalling the notations in Section~\ref{sec:idio:beta}, we consider the following constrained $\ell_1$-minimisation approach closely related to the Dantzig selector proposed for high-dimensional linear regression \citep{candes2007}, for the estimation of $\bm\beta = [\mbf A_1, \ldots, \mbf A_d]^\top$:
\begin{align}
\label{eq:ds}
\wh{\bm\beta}^{\ds} = {\arg\min}_{\mbf M \in \R^{pd \times p}} \ \vert \mbf M \vert_1
\quad \text{subject to} \quad
\l\vert \wh{\bbG} \mbf M - \wh{\bbg} \r\vert_\infty \le \lambda^{\ds},
\end{align}
where $\lambda^{\ds} > 0$ is a tuning parameter.

To investigate the theoretical properties of $\wh{\bm\beta}^{\ds}$, we measure the (weak) sparsity of $\bm\beta$ by $s_0(\varrho) = \sum_{j = 1}^p s_{0, j}$ with $s_{0, j}(\varrho) = \sum_{\ell = 1}^d \sum_{k = 1}^p \vert A_{\ell, jk} \vert^\varrho$ for some $\varrho \in [0, 1)$. In particular, when $\varrho = 0$, they coincide with the sparsity measures defined in the main text, as $s_{0, j} = s_{0, j}(0) = \vert \bm\beta_{\cdot j} \vert_0$, $s_0 = \sum_{j = 1}^p s_{0, j}$ and $s_{\text{\upshape in}} = \max_{1 \le j \le p} s_{0, j}$.

\begin{prop}
\label{prop:idio:est:dantzig}
{\it Set $\lambda^{\ds} \ge C_\xi (\Vert \bm\beta \Vert_1 + 1) (\vartheta_{n, p} \vee m^{-1} \vee p^{-1/2})$ in~\eqref{eq:ds}. Then on $\mc E_{n, p}$ defined in~\eqref{eq:idio:set},
for any $\varrho \in [0, 1)$,
conditional on $\mc E_{n, p}$ defined in~\eqref{eq:idio:set}, we have
\begin{align*}
& \max_{1 \le j \le p} \l\vert \wh{\bm\beta}^{\ds}_{\cdot j} - \bm\beta_{\cdot j} \r\vert_2 \lesssim
\frac{\sqrt{s_{\text{\upshape in}}} \lambda^{\ds}}{\pi m_\xi}, \quad
\max_{1 \le j \le p} \l\vert \wh{\bm\beta}^{\ds}_{\cdot j} - \bm\beta_{\cdot j} \r\vert_1 \lesssim \frac{s_{\text{\upshape in}} \lambda^{\ds}}{\pi m_\xi},
\\
& \max_{1 \le j \le p} \l\vert \wh{\bm\beta}^{\ds}_{\cdot j} - \bm\beta_{\cdot j} \r\vert_\infty \lesssim \min\l(\Vert \bbG^{-1} \Vert_1\lambda^{\ds}, \frac{\sqrt{s_{\text{\upshape in}}} \lambda^{\ds}}{\pi m_\xi} \r) \quad \text{ and}
\\
& \l\vert \wh{\bm\beta}^{\ds} - \bm\beta \r\vert_1 
\le 6 s_0(\varrho) \l(\Vert \bbG^{-1} \Vert_1 \lambda^{\ds}\r)^{1 - \varrho}.
\end{align*}}
\end{prop}

Proposition~\ref{prop:idio:est:dantzig} shows that $\wh{\bm\beta}^{\ds}$ achieves consistency when $\bm\beta$ is only weakly sparse with $\varrho > 0$, but the estimation error involves a multiplicative factor $\Vert \bbG^{-1} \Vert_1$.
This is linked to the sparsity of the ACV matrices of $\bm\xi_t$
which is related to, but is not fully captured by, the sparsity of $\bm\beta$.
For example, \cite{han2015direct} show that when $\bbG$ is strictly diagonally dominant \citep[Definition~6.1.9]{horn1985} with $\min_{1 \le i \le p} \min_{1 \le k \le d} ( 2 \vert \gamma_{\xi, ii}(0) \vert 
- \sum_{\ell = -d + k}^{k - 1} \sum_{i' = 1}^p \vert \gamma_{\xi, ii'}(\ell) \vert ) \ge \gamma_\circ > 0$ (where $\bm\Gamma_\xi(\ell) = [\gamma_{\xi, ii'}(\ell)]_{i, i' = 1}^p$),
we have $\Vert \bbG^{-1} \Vert_1 \le \gamma_\circ^{-1}$.

\begin{cor}
\label{cor:idio:est:dir:ds}
{\it Suppose that the conditions of Proposition~\ref{prop:idio:est:dantzig} are met.
If 
\begin{align*}
\min_{(i, j) \in \text{\upshape supp}(\bm\beta)} \vert \beta_{ij} \vert > 2\mathfrak{t}
\end{align*}
with $\mathfrak{t} = 2 \Vert \bbG^{-1} \Vert_1 \lambda^{\ds}$,
we have $\text{\upshape{sign}}(\wh{\bm\beta}^{\ds}(\mathfrak{t})) = \text{\upshape{sign}}(\bm\beta)$ on $\mc E_{n, p}$.}
\end{cor}

\section{Forecasting the factor-driven component under unrestricted GDFM}
\label{sec:2017}

In this section, we present an alternative method for estimating the best linear predictor~\eqref{eq:gdfm:best:lin} of the common component without the restrictive assumption made in Section~\ref{sec:2005}.
The estimator, denoted by $\bm\chi^{\va}_{t + a \vert t}$, has been proposed by \cite{forni2017dynamic}, and we provide a new theoretical result that establishes the $\ell_\infty$-norm consistency of this estimator.

We first make a mild assumption that the filter $\mc B(L)$ in~\eqref{eq:gdfm} is rational:
\begin{assum}
\label{cond:ration}
{\it  Each filter $\mc B_{ij}(L) = \sum_{\ell = 0}^\infty B_{\ell, ij} L^\ell$ is a ratio of finite-order polynomials in $L$, i.e.\
$\mc B_{ij}(L) = (\mc B^{(2)}_{ij}(L))^{-1}\mc B^{(1)}_{ij}(L)$ with
$\mc B^{(k)}_{ij}(L) = \sum_{\ell = 0}^{s^{(k)}} B^{(k)}_{\ell, ij} L^\ell$, $k = 1, 2$,
for all $1 \le i \le p$ and $1 \le j \le q$
with $s^{(1)}, s^{(2)} \in \N$ not dependent on $i$ and $j$. 
Moreover,
\begin{enumerate}[wide, label = (\roman*)]
\item $\max_{1 \le i \le p} \max_{1 \le j \le q} \max_{0 \le \ell \le s^{(1)}} \vert B^{(1)}_{\ell, ij} \vert \le B^{(1)}_\chi$ for some constant $B^{(1)}_\chi >0$, and

\item for all $1 \le i \le p$ and $1 \le j \le q$, we have
$\mc B^{(2)}_{ij}(z) \ne 0$ for all $\vert z \vert \le 1$.
\end{enumerate}}    
\end{assum}

\cite{forni2015} establish that for generic values of the parameters $B^{(1)}_{\ell, ij}$ and $B^{(2)}_{\ell, ij}$ defined in Assumption~\ref{cond:ration} (i.e.\ outside a countable union of nowhere dense subsets),
$\bm\chi_t$ admits a blockwise VAR representation.
Supposing that $p = N(q + 1)$ with some integer $N$ for convenience,
each $(q+1)$-dimensional block  
$\bm\chi_t^{(h)} = (\chi_{(q + 1)(h - 1) + i, t}, \, 1 \le i \le q + 1)^\top$, $1 \le h \le N$,
admits a singular, {\it finite-order} VAR representation 
with $\mbf u_t$ as the $q$-dimensional innovations. 
We formally impose this genericity result as an assumption.

\begin{assum} [Blockwise VAR representation]
\label{assum:common:var} \hfill
\it{
\begin{enumerate}[wide, label = (\roman*)]
\item \label{cond:block:var} 
$\bm\chi_t$ admits a blockwise VAR representation
\begin{align}
\label{eq:gdfm:var}
\mc A_\chi(L) \bm\chi_t = \bmx
\mc A_\chi^{(1)}(L) & \mbf O & \ldots & \mbf O \\
\mbf O & \mc A_\chi^{(2)}(L) & \ldots & \mbf O \\
\vdots&\vdots & \ddots &\vdots \\
\mbf O & \mbf O & \ldots & \mc A_\chi^{(N)}(L) 
\emx \; 
\bmx \bm\chi^{(1)}_t \\ \bm\chi^{(2)}_t \\ \vdots \\ \bm\chi^{(N)}_t \emx
= \bmx \mbf R^{(1)} \\ \mbf R^{(2)} \\ \vdots \\ \mbf R^{(N)} \emx \mbf u_t =: \mbf R \mbf u_t
\end{align}
where, for each $h = 1, \ldots, N$:
\begin{enumerate}[wide, label = (\alph*)]
\item $\mc A_\chi^{(h)}(L) = \mbf I - \sum_{\ell = 1}^{s_h} \mbf A_{\chi, \ell}^{(h)} L^\ell$
with $\mbf A_{\chi, \ell}^{(h)} \in \R^{(q + 1) \times (q + 1)}$
is of degree $s_h \le q s^{(1)} + q^2 s^{(2)}$
and $\det(\mc A_\chi^{(h)}(z)) \ne 0$ for all $\vert z \vert \le 1$.
Then, we define $\mc A_\chi(L) = \mbf I - \sum_{\ell = 1}^s \mbf A_{\chi, \ell} L^\ell$
with $s = \max_{1 \le h \le N} s_h$ where $\mbf A^{(h)}_{\chi, \ell} = \mbf O$ for $\ell \ge s_h + 1$.

\item $\mbf R^{(h)} \in \R^{(q + 1) \times q}$ is of rank $q$
with $[\mbf R^{(h)}]_{ij} = B^{(1)}_{0, ij}$. 

\item The VAR representation in~\eqref{eq:gdfm:var} is unique, i.e.\ if 
$\wt{\mc A}_\chi^{(h)}(L) \bm\chi_t = \wt{\mbf R}^{(h)} \wt{\mbf u}_t$,
the degree of $\wt{\mc A}_\chi^{(h)}(L)$ does not exceed $s$ and
$\wt{\mbf u}_t$ is $q$-dimensional white noise, such that
$\wt{\mc A}_\chi^{(h)}(L) = \mc A_\chi^{(h)}(L)$,
$\wt{\mbf R}^{(h)} = \mbf R^{(h)}\bm{\mc O}^\top$ and $\wt{\mbf u}_t = \bm{\mc O} \mbf u_t$
for some orthogonal matrix $\bm{\mc O}\in \R^{q \times q}$.
\end{enumerate}

\item \label{cond:det} For $h = 1, \ldots, N$, define
$\bm\Gamma^{(h)}_\chi(\ell) = \E(\bm\chi^{(h)}_{t - \ell} (\bm\chi^{(h)}_t)^\top)$,
$\mbf B^{(h)}_\chi = [(\bm\Gamma^{(h)}_\chi(\ell))^\top, \, 1 \le \ell \le s_h]$ and
\begin{align*}
\mbf C^{(h)}_\chi = \bmx \bm\Gamma^{(h)}_\chi(0) & \bm\Gamma^{(h)}_\chi(- 1) & \ldots & \bm\Gamma^{(h)}_\chi(-s_h + 1) \\
\vdots&\vdots & \ddots & \vdots\\
\bm\Gamma^{(h)}_\chi(s_h - 1) & \bm\Gamma^{(h)}_\chi(s_h - 2) & \ldots & \bm\Gamma^{(h)}_\chi(0)
\emx,
\end{align*}
such that $\mbf A^{(h)}_\chi = [\mbf A^{(h)}_{\chi, \ell}, \, 1 \le \ell \le s_h]
= \mbf B^{(h)}_\chi (\mbf C^{(h)}_\chi)^{-1}$.
Then, there exists a constant $c_0>0$ such that 
$\min_{1 \le h \le N} \det(\mbf C^{(h)}_\chi) > c_0$.
\end{enumerate}
}
\end{assum}

Non-singularity of $\mbf C^{(h)}_\chi$ is implied by 
Assumption~\ref{assum:common:var}~\ref{cond:block:var}
but the condition~\ref{cond:det} is imposed
to ensure the boundedness of $\Vert \mbf A^{(h)}_\chi \Vert$. 
Let $\mbf Z_t = \mc A_\chi(L) \mbf X_t$ such that under~\eqref{eq:gdfm:var},
\begin{align}
\label{eq:static:two}
\mbf Z_t = \mbf R \mbf u_t + \mc A_\chi(L) \bm\xi_t =: \mbf V_t + \mbf W_t.
\end{align}
Then, $\mbf Z_t$ admits a static factor representation with $\mbf u_t$ as (static) factors.
Let $\bm\Gamma_z$ denote the covariance matrix of $\mbf Z_t$
and $\mu_{z, j}$ its $j$-th largest eigenvalue,
and similarly define $\bm\Gamma_v$, $\mu_{v, j}$, $\bm\Gamma_w$ and $\mu_{w, j}$.
As a consequence of Assumption~\ref{assum:idio},
the eigenvalues of $\bm\Gamma_w$ are bounded as below.
\begin{prop}
\label{prop:idio:eval:two}
{\it Under Assumptions~\ref{assum:common}, \ref{assum:idio}, \ref{assum:innov} and~\ref{assum:common:var},
there exists some constant $B_w~>~0$ such that $\mu_{w, 1} \le B_w$.}
\end{prop}

The following assumption,
imposed on the strength of the factors in~\eqref{eq:static:two}
similarly as in Assumptions~\ref{assum:factor} and~\ref{assum:static}~\ref{cond:static:linear},
enables asymptotic identification of $\mbf V_t$ and $\mbf W_t$.
Due to the complexity associated with theoretical analysis under the model~\eqref{eq:gdfm:var},
we consider the case of strong factors only.
\begin{assum}
\label{assum:static:two}
{\it There exist a positive integer $p_0 \ge 1$
and pairs of positive constants $(\alpha_{v, j}, \beta_{v, j}), \, 1 \le j \le q$, such that
for all $p \ge p_0$,
\begin{align*}
\beta_{v, 1} \ge \frac{\mu_{v, 1}}{p} \ge \alpha_{v, 1} >
\beta_{v, 2} \ge \frac{\mu_{v, 2}}{p} \ge \ldots \ge 
\alpha_{v, q - 1} > \beta_{v, q} \ge \frac{\mu_{v, q}}{p} \ge \alpha_{v, q} > 0.
\end{align*}
}
\end{assum}

We propose to perform in-sample estimation and forecasting of the common component 
by directly utilising the expression of $\bm\chi_{t + a \vert t}$ in~\eqref{eq:gdfm:best:lin},
following the method proposed in \cite{forni2017dynamic}
that makes use of the outcome of the dynamic PCA step outlined in Section~\ref{sec:dpca}.

\begin{enumerate}[wide, label = {\bf Step~\arabic*}:]
\item Estimate $\mc A_\chi^{(h)}(L)$ 
with $\wh{\mc A}_\chi^{(h)}(L) = \mbf I - \sum_{\ell = 1}^{s_h} \wh{\mbf A}^{(h)}_{\chi, \ell} L^\ell$
via Yule-Walker estimator 
$\wh{\mbf A}^{(h)}_\chi = [\wh{\mbf A}^{(h)}_{\chi, \ell}, \, 1 \le \ell \le s_h]
= \wh{\mbf B}^{(h)}_\chi (\wh{\mbf C}^{(h)}_\chi)^{-1}$,
where $\wh{\mbf B}^{(h)}_\chi$ and $\wh{\mbf C}^{(h)}_\chi$ are defined analogously as
$\mbf B^{(h)}_\chi$ and $\mbf C^{(h)}_\chi$,
with $\wh{\bm\Gamma}_\chi(\ell)$ defined in Section~\ref{sec:dpca} replacing $\bm\Gamma_\chi(\ell)$.
Let $\wh{\mc A}_\chi(L) = \text{diag}(\wh{\mc A}_\chi^{(h)}(L), \, 1 \le h \le N)$.

\item Obtain the filtered process
$\wh{\mbf Z}_t = \wh{\mc A}_\chi(L) \mbf X_t$ for $s + 1 \le t \le n$,
and let $\wh{\bm\Gamma}_z = (n - s)^{-1} \sum_{t = s + 1}^n \wh{\mbf Z}_t\wh{\mbf Z}_t^\top$,
and denote by $\wh{\bm\Gamma}_z = \sum_{j = 1}^{p \wedge (n - s)} \wh{\mu}_{z, j} \wh{\mbf e}_{z, j} \wh{\mbf e}_{z, j}^\top$ its eigendecomposition.
Then, set
$\wh{\mbf R} = [\wh{\mbf e}_{z, j} \sqrt{\wh{\mu}_{z, j}}, \, 1 \le j \le q]
= \wh{\mbf E}_{z, 1:q} \wh{\bm{\mc M}}_{z, 1:q}^{1/2}$
and $\wh{\mbf u}_t = \wh{\bm{\mc M}}_{z, 1:q}^{-1/2} \wh{\mbf E}_{z, 1:q}^\top \wh{\mbf Z}_t$
where $\wh{\bm{\mc M}}_{z, 1:q} = \text{diag}(\wh\mu_{z, j}, \, 1 \le j \le q)$
and $\wh{\mbf E}_{z, 1:q} = [\wh{\mbf e}_{z, j}, \, 1 \le j \le q]$.

\item Denoting by $\wh{\mbf B}_\ell$ the coefficient matrix multiplied to $L^\ell$
in expanding $\wh{\mc A}_\chi^{-1}(L) \wh{\mbf R}$, 
with some truncation lag $K$, we estimate $\bm\chi_{t + a \vert t}$ by
\begin{align}
\label{eq:common:forecast:infinite}
\wh{\bm\chi}^{\va}_{t + a \vert t} = \sum_{\ell = 0}^K \wh{\mbf B}_{\ell + a} \wh{\mbf u}_{t  - \ell}.
\end{align}
\end{enumerate}

\begin{rem}
\label{rem:2017}
\begin{enumerate}[wide, label = (\alph*)]
\item In our numerical studies reported in Section~\ref{sec:sim} and Appendix~\ref{app:sim}, we apply the Schwarz criterion \citep[Chapter~4]{lutkepohl2005} to each block for selecting the order $s_h$ of the blockwise VAR model (see Assumption~\ref{assum:common:var}).
Also, we set the truncation lag at $K = 20$ in~\eqref{eq:common:forecast:infinite} and the number of cross-sectional permutations to be $30$ by default.

\item \label{rem:2017:perm} The cross-sectional ordering of the panel has an impact on the selection of the diagonal blocks when estimating $\mathcal A_\chi(L)$. Each cross-sectional permutation of the panel leads to distinct estimators, all sharing the same asymptotic properties. Using a Rao-Blackwell-type argument, \citet{forni2017dynamic} advocates the aggregation of these estimators into a unique one by simple averaging (after obvious reordering of the cross-sections). 
Although averaging over all $p!$ permutations is infeasible, as argued by \citet{forni2017dynamic} and empirically verified by \citet{forni2018dynamic}, a few of them are enough in practice to deliver stable averages.

\item When $p$ is not an integer multiple of $(q + 1)$,
we can consider $\lfloor p/(q + 1)\rfloor - 1$ blocks of size $(q + 1)$ 
along with a block of the remaining $(q+1)+p-\lfloor p/(q+1)\rfloor(q+1)$ variables. 
All the theoretical arguments used in \citet{forni2017dynamic} and in this paper apply to any partition of the cross-section into blocks of size $(q + 1)$ or larger (but finite).

\item \label{rem:2017:hormann} It is known that as the VAR order $s_h$ increases, 
the estimation of a singular VAR via Yule-Walker methods 
might become unstable since it requires inverting $\wh{\mbf C}^{(h)}_\chi$,
a Toeplitz matrix of dimension $s_h(q+1) \times s_h(q+1)$. 
To address potential issues arising from this, 
\citet{hormann2021prediction} propose a regularised approach 
aimed at stabilising the estimate of $\mc A_\chi^{(h)}(L)$. 
Empirically, such an approach leads to better performance
and can be taken when $s_h$ is large.
\end{enumerate}
\end{rem}

Let us denote the matrix collecting all the transition matrices involved
in the blockwise VAR model in~\eqref{eq:gdfm:var} by
$\mbf A_\chi = [ \mbf A_{\chi, 1}, \ldots, \mbf A_{\chi, s} ] \in \R^{p \times ps}$
and its estimated counterpart by $\wh{\mbf A}_\chi$. 
Then, the estimators from the above procedure satisfy the following.

\begin{prop}
\label{thm:common:var}
{\it Suppose that the conditions in Theorem~\ref{thm:common:spec} are met,
including Assumption~\ref{assum:factor} which is satisfied with
$\rho_j = 1, \, 1 \le j \le q$.
Further let Assumptions~\ref{assum:common:var}--\ref{assum:static:two} hold.

\begin{enumerate}[wide, label = (\roman*)]
\item \label{thm:common:var:one}
Denoting by $\bm\varphi_i$ a vector of zeros
except for its $i$-th element being set to one, we have
\begin{align*}
\frac{1}{\sqrt p} \l\Vert \wh{\mbf A}_\chi - \mbf A_\chi \r\Vert_F &= O_P
\l(\vartheta_{n, p} \vee \frac{1}{m} \vee \frac{1}{\sqrt p}\r) \quad \text{and} 
\\
\max_{1 \le i \le p}
\l\vert \bm\varphi_i^\top \l(\wh{\mbf A}_\chi - \mbf A_\chi \r) \r\vert_2 &= O_P
\l( \vartheta_{n, p} \vee \frac{1}{m} \vee \frac{1}{\sqrt p}\r),
\end{align*}
where $\vartheta_{n, p}$ is defined in~\eqref{eq:rates:two}.

\item \label{thm:common:var:two} For any given $t$, we have
\begin{align*}
\l\vert \wh{\mbf R}\wh{\mbf u}_t - \mbf R \mbf u_t \r\vert_\infty &= 
O_P\l( \vartheta_{n, p} \vee \frac{1}{m} \vee \frac{1}{\sqrt p} \r).
\end{align*}

\item \label{thm:common:var:three} 
For any fixed $a \ge 0$ and $t$, we have
\begin{align*}
\l\vert \wh{\bm\chi}^{\va}_{t + a \vert t} - {\bm\chi}_{t + a \vert t} \r\vert_\infty
= O_P\l( K M_\chi^K \l(\vartheta_{n, p} \vee \frac{1}{m} \vee \frac{1}{\sqrt p}\r) + 
\log^{1/2}(p) K^{-\varsigma + 1} \r),
\end{align*}
where $M_\chi > 0$ is a constant. 
Further, when $p = O(n^\kappa)$ for some $\kappa > 0$,
we can select $K = K_n = \lfloor c \log(n) \rfloor$ with a small enough $c > 0$
such that 
$\l\vert \wh{\bm\chi}^{\va}_{t + a \vert t} - {\bm\chi}_{t + a \vert t} \r\vert_\infty = o_P(1)$.
\end{enumerate}}
\end{prop}

Proposition~\ref{thm:common:var}~\ref{thm:common:var:three},
combined with Proposition~\ref{prop:common:irreduce},
concludes the analysis of the forecasting error $\vert \wh{\bm\chi}^{\va}_{t + a \vert t} - \bm\chi_{t + a} \vert_\infty$.
\begin{rem}
\label{rem:var}
\begin{enumerate}[wide, label = (\alph*)]
\item \label{rem:var:one} In Proposition~\ref{thm:common:var}, 
we assume a stronger factor structure 
given by Assumption~\ref{assum:factor} with $\rho_j = 1$ for all $1 \le j \le q$,
and Assumption~\ref{assum:static:two} for ease of presentation;
as noted in Remark~\ref{rem:thm:dpca}~\ref{rem:tails}, 
weaker factor strength would imply worse estimation and forecasting performance.

\item Proposition~\ref{thm:common:var}~\ref{thm:common:var:one} extends
Proposition~9 of \cite{forni2017dynamic} that considers the consistency of $\wh{\mbf A}_\chi$
for a given single block.
The result in~\ref{thm:common:var:two} indicates that
we can consistently recover $\mbf V_t = \mbf R\mbf u_t$
and hence the static factor space. 
Moreover, parts~\ref{thm:common:var:one} and~\ref{thm:common:var:two} imply that we can consistently recover the space spanned by $\mbf u_t$ and their associated impulse response functions ${\mc B}(L) = {\mc A}^{-1}_\chi(L)\mbf R$
for finitely many lags up to a linear transformation.
This is particularly useful in empirical macroeconomic literature where typically,
$\mbf u_t$ carries a specific economic implication 
(e.g.\ shocks related to monetary policy, fiscal policy, 
demand, supply, technology and oil, to name a few). 
Indeed, it is a common practice to use the above estimation method to 
recover the space spanned by these shocks and then to identify their dynamic effect $\mc B(L)$ 
by imposing {\it ad hoc} economic based restrictions;
see, e.g.\ \cite{stock2016dynamic}, for a review of this approach.

\item In Proposition~\ref{thm:common:var}~\ref{thm:common:var:three},
the constant $M_\chi$ is associated with the operator norm
of the transition matrix involved in the VAR($1$) representation of~\eqref{eq:gdfm:var}
and, provided that the maximum VAR order $s = \max_{1 \le h \le N} s_h \ge 2$, 
we have $M_\chi \ge 1$ (see Appendix~E of \citealp{basu2015}).
The $O_P$ bound in~\ref{thm:common:var:three} 
can be made to converge to zero as $n, p \to \infty$ by selecting an appropriate truncation lag $K$, 
but it is at a slower rate compared to the $O_P$ bound derived for $\wh{\bm\chi}^{\static}_t$ 
studied in Section~\ref{sec:2005}
when comparable strong factors are assumed ($\rho_j = 1, \, 1 \le j \le q$, in Assumption~\ref{assum:factor} 
and $\varrho_j = 1, \, 1 \le j \le r$, in Assumption~\ref{assum:static}~\ref{cond:static:linear}).
This shows the advantage of working under a more restrictive factor model
in Section~\ref{sec:2005} when it comes to forecasting.
\end{enumerate}
\end{rem}

\clearpage

\section{Sparsity structure of $\mc N^{\lr}$}
\label{app:lr}

Recall that $\bm\Omega = [\omega_{ii'}]$ and $\bm\Delta = [\delta_{ii'}]$,
and let $\mc A(1) = [a_{ii'}]$. Also, define $g_{i, i'} = 1$ if $i = i'$ and $g_{i, i'} = 0$ otherwise. 
Recall that $a_{ii'} = g_{i, i'} - \sum_{\ell = 1}^d A_{\ell, ii'}$.
Then,
\begin{align}
\frac{1}{2\pi} \omega_{ii'} &= \sum_{h = 1}^p \sum_{k = 1}^p a_{hi} \delta_{hk} a_{ki'}
= a_{ii} \sum_k \delta_{ik} a_{ki'}
+ \sum_h \sum_k a_{hi} \delta_{hk} a_{ki'} (1 - g_{h, i})
\nn \\
&= a_{ii} \sum_\ell \delta_{i \ell} a_{\ell i'}
+ a_{i'i'} \sum_\ell a_{\ell i} \delta_{\ell i'} (1 - g_{\ell, i})
+ \sum_h \sum_k a_{hi} \delta_{hk} a_{ki'} (1 - g_{h, i}) (1 - g_{k, i'})
\nn \\
&= a_{ii} \delta_{ii'} a_{i' i'} + a_{i'i'} \delta_{i'i'} a_{i' i} + a_{ii} \delta_{ii} a_{i i'} 
+ a_{ii} \sum_\ell \delta_{i \ell} a_{\ell i'} (1 - g_{\ell, i}) (1 - g_{\ell, i'}) +
\nn \\
& \qquad a_{i'i'} \sum_\ell a_{\ell i} \delta_{\ell i'} (1 - g_{\ell, i}) (1 - g_{\ell, i'})
+ \sum_{\ell} a_{\ell i} \delta_{\ell \ell} a_{\ell i'} (1 - g_{\ell, i}) (1 - g_{\ell, i'}) +
\nn \\
& \qquad \sum_h \sum_k a_{hi} \delta_{hk} a_{ki'} (1 - g_{h, i}) (1 - g_{k, i'})(1 - g_{h, k}).
\label{eq:nlr}
\end{align}
We conclude that $\omega_{ii'} = 0$ if none of the following holds:
(i) variables $i$ and $i'$ are partially correlated, i.e.\ $\delta_{ii'} \ne 0$ (from the first term in~\eqref{eq:nlr});
(ii) $i$ Granger causes $i'$ in the long run, i.e.\ $\sum_{\ell = 1}^d A_{\ell, i'i} \ne 0$ (from the second term);
(iii) $i'$ Granger causes $i$ in the long run, i.e.\ $\sum_{\ell = 1}^d A_{\ell, ii'} \ne 0$ (from the third term);
(iv) there exists a variable $j \in \mc V \setminus \{i, i'\}$ such that 
$\sum_{\ell = 1}^d A_{\ell, ji} \ne 0$ and $\delta_{i'j} \ne 0$, or
$\sum_{\ell = 1}^d A_{\ell, ji'} \ne 0$ and $\delta_{ij} \ne 0$, or
$\sum_{\ell = 1}^d A_{\ell, ji} \ne 0$ and $\sum_{\ell = 1}^d A_{\ell, ji'} \ne 0$ (from the fourth, the fifth and the sixth terms); or
(v) there exist a pair of variables $j, j' \in \mc V \setminus \{i, i'\}$ such that
$\sum_{\ell = 1}^d A_{\ell, ji} \ne 0$,  
$\sum_{\ell = 1}^d A_{\ell, j'i'} \ne 0$ and $\delta_{jj'} \ne 0$ (from the last term).
As such, the edge set $\mc E^{\lr}$ is typically larger than $\mc E^{\dir} \cup \mc E^{\undir}$.

\clearpage

\section{Data-driven choice of the thresholds}
\label{sec:thresh}

Motivated by \cite{liu2021simultaneous}, we propose a method for data-driven selection of the threshold~$\mathfrak{t}$,
which is applied to the estimators of $\bm\beta = [\mbf A_\ell, \, 1 \le \ell \le d]^\top$, $\bm\Delta$ or $\bm\Omega$ for estimating the edge sets of $\mc N^{\dir}$, $\mc N^{\undir}$ or $\mc N^{\lr}$, respectively.

Let $\mbf B = [b_{ij}] \in \R^{m \times n}$ denote a matrix for which a threshold is to be selected, i.e.\ $\mbf B$ may be either $\wh{\bm\beta}$, $\wh{\bm\Delta}_0$ ($\wh{\bm\Delta}$ with diagonals set to zero) or $\wh{\bm\Omega}_0$ ($\wh{\bm\Omega}$ with diagonals set to zero) obtained from Steps~2 and~3 of FNETS.
We work with $\wh{\bm\Delta}_0$ and $\wh{\bm\Omega}_0$ since we do not threshold the diagonal entries of $\wh{\bm\Delta}$ and $\wh{\bm\Omega}$.
As such estimators have been shown to achieve consistency in $\ell_\infty$-norm (Propositions~\ref{prop:idio:est:lasso} and~\ref{prop:idio:delta}), we expect there exists a large gap between the entries of $\mbf B$ corresponding to true positives and false positives.
Further, it is expected that the number of edges reduces at a faster rate when increasing the threshold from $0$ towards this (unknown) gap, compared to when increasing the threshold from the gap to $\vert \mbf B \vert_\infty$.
Therefore, we propose to identify this gap by casting the problem as that of locating a single change point in the trend of the ratio of edges to non-edges,
\begin{align*}
\text{Ratio}_k = \frac{\vert \mbf B(\mathfrak{t}_k) \vert_0}{ \max( N - \vert \mbf B(\mathfrak{t}_k) \vert_0, 1) }, 
\quad k = 1, \dots, M.
\end{align*}
Here, $\mbf B(\mathfrak{t}) = [b_{ij} \cdot \mathbb{I}_{\{\vert b_{ij} \vert > \mathfrak{t} \}}]$, 
$\vert \mbf B(\mathfrak{t}) \vert_0 = \sum_{i = 1}^{m_1} \sum_{j = 1}^{m_2} \mathbb{I}_{\{\vert b_{ij} \vert > \mathfrak{t} \}}$
and $\{\mathfrak{t}_k, \, 1 \le k \le M: \, 0 = \mathfrak{t}_1 < \mathfrak{t}_2 < \dots < \mathfrak{t}_M = \vert \mbf B \vert_\infty\}$ denotes a sequence of candidate threshold values.  
We recommend using an exponentially growing sequence for $\{\mathfrak{t}_k\}_{k = 1}^M$ since the size of the false positive entries tends to be very small.
The quantity $N$ in the denominator of Ratio$_k$ is set as $N = p^2d$ when $\mbf B = \wh{\bm\beta}$, and $N = p(p - 1)$ when $\mbf B = \wh{\bm\Delta}_0$ or $\mbf B = \wh{\bm\Omega}_0$.
Then, from the difference quotient
\begin{align*} 
\text{Diff}_k = \frac{\text{Ratio}_k - \text{Ratio}_{k - 1}}{\mathfrak{t}_k - \mathfrak{t}_{k - 1}},
\quad k = 2, \ldots, M,
\end{align*}
we compute the cumulative sum (CUSUM) statistic
\begin{align*}
\text{CUSUM}_k = \sqrt{\frac{k (M - k)}{M}} \l\vert \frac{1}{k} \sum_{l = 2}^k \text{Diff}_l - \frac{1}{M - k} \sum_{l = k + 1}^M \text{Diff}_l \r\vert, \quad k = 2, \ldots, M-1,
\end{align*}
and select $\mathfrak{t} = \mathfrak{t}_{k^*}$ with $k^* = {\arg\max}_{2 \le k \le M - 1} \text{CUSUM}_k$.

We investigate the performance of the thus-chosen thresholds on simulated datasets in Appendix~\ref{sec:complete:sim}, see also \cite{owens2023fnets}.

\clearpage

\section{Simulation studies}
\label{sec:complete:sim}

\subsection{Set-up}
\label{sec:sim:model}

We apply FNETS to datasets simulated under a variety of settings, from Gaussian innovations $\mbf u_t$ and $\bm\vep_t$ with~\ref{e:one}~$\bm\Delta = \mbf I$ and~\ref{e:two}~$\bm\Delta \ne \mbf I$, to~\ref{e:three}~heavy-tailed ($t_5$) innovations with $\bm\Delta = \mbf I$,
and when $\bm\chi_t$ is generated from~\ref{m:ar}~fully dynamic or~\ref{m:ma}~static factor models.
In addition, we consider the `oracle' setting~\ref{m:oracle} $\bm\chi_t = \mbf 0$ where, in the absence of the factor-driven component, the results obtained can serve as a benchmark. 
We also include the factor-adjusted regression method of \cite{fan2021bridging} which is referred to as FARM, and present the performance of their estimator $\wh{\bm\beta}^{\text{FARM}}$ of VAR parameters and forecasts (see Appendix~\ref{sec:sim:model} for full descriptions).
For each setting, $100$ realisations are generated.

The idiosyncratic component is generated as a VAR($1$) process.
Let $\mc N^{\dir}$ denote a directed Erd\H{o}s-R\'{e}nyi random graph 
on $\mc V = \{1, \ldots, p\}$ with the link probability $1/p$.
Then, the entries of $\mbf A_1$ are $A_{1, ii'} = 0.275$ when $(i, i') \in \mc E^{\dir}$
and $A_{1, ii'} = 0$ otherwise.
The innovations are generated according to the following three scenarios:
\begin{enumerate}[wide, label = (E\arabic*)]
\item \label{e:one} Gaussian with $\bm\Gamma = \mbf I$.
\item \label{e:two} Gaussian with $\bm\Gamma = \bm\Delta^{-1}$, where
$\delta_{ii} = 1.5$ for $1 \le i \le p$, $\delta_{ii'} = -1/\sqrt{d_i d_{i'}}$ if $(i, i') \in \mc E^{\undir}$
and $\delta_{ii'} = 0$ otherwise.
Here, $\mc N^{\undir}$ is an undirected Erd\H{o}s-R\'{e}nyi random graph on $\mc V$
with the link probability $1/p$, and $d_i$ denotes the degree of the $i$-th node in $\mc E^{\undir}$.
This model is taken from \cite{barigozzi2019}.
\item \label{e:three} Heavy-tailed with 
$\sqrt{5/3} \cdot \vep_{it} \sim_{\iid} t_5$ (such that $\Var(\vep_{it}) = 1$)
and $\bm\Gamma = \mbf I$.
\end{enumerate}

We consider two models for the generation of factor-driven common component:
\begin{enumerate}[wide, label = (C\arabic*)]
\item \label{m:ar} Taken from \cite{forni2017dynamic},
$\chi_{it}$ is generated as sum of $q$ AR processes
$\chi_{it} = \sum_{j = 1}^q a_{ij} (1 - \alpha_{ij} L)^{-1} u_{jt}$,
where $a_{ij} \sim_{\iid} \mc U[-1, 1]$ and $\alpha_{ij} \sim_{\iid} \mc U[-0.8, 0.8]$
with $\mc U[a, b]$ denoting a uniform distribution.
This model does not admit a static factor model representation,
and we consider $q = 2$.

\item \label{m:ma} $\chi_{it}$ admits a static factor model representation as
$\chi_{it} = a_i \sum_{\ell = 1}^2 \bm\lambda_{i \ell}^\top \mbf f_{t - \ell + 1}$
with $\mbf f_t = \mbf D \mbf f_{t - 1} + \mbf u_t$;
here, $\mbf F_t = (\mbf f_t^\top, \mbf f_{t - 1}^\top)^\top \in \R^r$
denotes the static factor with $r = 2q$,
$\mbf f_t \in \R^q$ the dynamic factor and
and $\mbf u_t = (u_{1t}, \ldots, u_{qt})^\top$ the common shocks.
The entries of the loadings $\bm\lambda_{i \ell} \in \R^q$ 
are generated i.i.d.\ from $\mc N(0, 1)$,
and $\mbf D = 0.7 \cdot \mbf D_0 / \Lambda_{\max}(\mbf D_0)$
where the off-diagonal entries of $\mbf D_0 \in \R^{q \times q}$ 
are generated i.i.d.\ from $\mc U[0, 0.3]$
and its diagonal entries from $\mc U[0.5, 0.8]$.
The multiplicative factor $a_i$ is chosen for each realisation to keep
sample estimate of $\Var(\chi_{it})/\Var(\xi_{it})$ at one.
We fix $q = 2$ (such that $r = 4$).
\end{enumerate}
Additionally, we consider the following `oracle' setting:
\begin{enumerate}[wide, label = (C\arabic*)]
\setcounter{enumi}{-1}
\item \label{m:oracle} $\bm\chi_t = 0$,
i.e.\ the idiosyncratic VAR process is directly observed as $\mbf X = \bm\xi_t$.
\end{enumerate}
We vary $(n, p) \in \{(100, 50), (100, 100), (200, 50), (200, 100), (500, 100), (500, 200)\}$.
According to the distribution of $\bm\vep_t$, we also vary the distribution of $\mbf u_t$;
under~\ref{e:one} or~\ref{e:two}, $u_{jt} \sim_{\iid} \mc N(0, 1)$ while
under~\ref{e:three}, $\sqrt{5/3} \cdot u_{jt} \sim_{\iid} t_5$.

\subsection{Results}
\label{app:sim}

\subsubsection{Network estimation}
\label{app:sim:est}

Throughout, we refer to the estimator in~\eqref{eq:lasso} by $\wh{\bm\beta}^{\las}$, and report the results from both $\wh{\bm\beta}^{\las}$ and $\wh{\bm\beta}^{\ds}$ obtained as in~\eqref{eq:ds}.
Also, $\wh{\bm\Omega}^{\las}$ and $\wh{\bm\Omega}^{\ds}$ denote the estimators of $\bm\Omega$ obtained with $\wh{\bm\beta}^{\las}$ and $\wh{\bm\beta}^{\ds}$, respectively see Section~\ref{sec:idio:lrpc}.
For comparison, we consider the FARM methodology of \cite{fan2021bridging}:
we implement their factor-adjustment step under a static factor model with the information criterion-based factor number estimator of \cite{ABC10} and, to the residuals from removing factors, we apply the Lasso to estimate the VAR parameters using the R package {\tt glmnet} \citep{glmnet}.
The resultant estimator is referred to as $\wh{\bm\beta}^{\text{FARM}}$.

In Tables~\ref{table:est:one:beta}--\ref{table:est:two:beta}, we report the estimation errors
of $\wh{\bm\beta}^{\las}$, $\wh{\bm\beta}^{\ds}$ and $\wh{\bm\beta}^{\text{FARM}}$ in estimating $\bm\beta$,
and in Tables~\ref{table:est:one:omega}--\ref{table:est:two:omega},
those of $\wh{\bm\Omega}^{\las}$ and $\wh{\bm\Omega}^{\ds}$ 
in estimating $\bm\Omega$ averaged over $100$ realisations, and the corresponding standard errors.
We also report the results from estimating $\bm\Delta$ by $\wh{\bm\Delta}^{\las}$ and $\wh{\bm\Delta}^{\ds}$ in Table~\ref{table:est:delta}.

With a matrix $\bm\gamma$ as an estimand
we measure the estimation error of its estimator $\wh{\bm\gamma}$
using the (scaled) matrix norms:
\begin{align*}
L_F = \frac{\Vert \wh{\bm\gamma} - \bm\gamma \Vert_F}{\Vert \bm\gamma \Vert_F} \quad \text{and} \quad
L_2 = \frac{\Vert \wh{\bm\gamma} - \bm\gamma \Vert}{\Vert \bm\gamma \Vert}.
\end{align*}
To assess the performance of $\wh{\bm\gamma}$
in recovering of the support of $\bm\gamma = [\gamma_{ii'}]$,  
i.e.\ $\{(i, i'): \, \gamma_{ii'} \ne 0 \}$,
we generate receiver operating characteristic (ROC) curves of
true positive rate (TPR) against false positive rate (FPR), 
averaged over $100$ realisations for each setting:
\begin{align}
\label{eq:tpfp}
\text{TPR} = \frac{ \vert \{ (i, i'): \, \wh\gamma_{ii'} \ne 0 \text{ and } \gamma_{ii'} \ne 0 \} \vert }{\vert \{ (i, i'): \, \gamma_{ii'} \ne 0 \} \vert}
\quad \text{and} \quad
\text{FPR} = \frac{ \vert \{ (i, i'): \, \wh\gamma_{ii'} \ne 0 \text{ and } \gamma_{ii'} = 0 \} \vert }{\vert \{ (i, i'): \, \gamma_{ii'} = 0 \} \vert},
\end{align}
see Figures~\ref{fig:roc:beta:two}--\ref{fig:roc:omega:one}.
We additionally report the results of the TPR value when FPR is set at $0.05$, with and without thresholding the estimators as described in Section~\ref{sec:tuning},
in Tables~\ref{table:est:one:beta}--\ref{table:est:delta}.

Overall, we observe that with increasing $n$, the performance of all estimators improve according to all metrics regardless of the data generating processes while increasing $p$ has an adverse effect.
Generally, whether the factor-driven component admits a static representation as in~\ref{m:ma} or not as in~\ref{m:ar}, FNETS produces estimators of $\bm\beta$ that perform as well as those applied under the oracle setting of~\ref{m:oracle} with $\bm\chi_t = \mbf 0$.
Also $\wh{\bm\beta}$ outperforms $\wh{\bm\beta}^{\text{FARM}}$ in all settings, particularly in estimating the support of $\bm\beta$ (i.e.\ the edge set of $\mc N^{\dir}$) even without any thresholding in all scenarios.
FARM tends to produce highly sparse estimators with low TPR, see Figure~\ref{fig:roc:beta:one} (averaged ROC curves are not necessarily monotonic as it contains pointwise average TPR at given FPR).
This is attributed to the accumulation of errors from estimating $\bm\xi_t, \, 1 \le t \le n$, which possibly leads to low signal-to-noise ratio when estimating the VAR parameters via Lasso. 
This difference vanishes as both $n$ and $p$ increase. 
As noted in Introduction, FNETS and FARM have distinctive objectives, and FNETS is specifically proposed for network estimation under the proposed factor-adjusted VAR model.
When $\bm\Delta = \mbf I$ (as in~\ref{e:one} and~\ref{e:three}), FNETS estimates $\bm\Omega$ with accuracy regardless of the tail behaviour of $\bm\vep_t$ and $\mbf u_t$.
When $\bm\Delta \ne \mbf I$, it tends to incur larger errors in estimating $\bm\Omega$ compared to when $\bm\Delta = \mbf I$, which is more noticeable in terms of support recovery (see Figure~\ref{fig:roc:omega:two}).
This possibly stems from the performance of $\wh{\bm\Delta}$ (see Table~\ref{table:est:delta}) rather than $\wh{\bm\beta}$, which becomes worse when $\bm\Delta \ne \mbf I$.
Since the support of $\bm\Omega$ depends on those of $\bm\beta$ and $\bm\Delta$ in a complex way (see Proposition~\ref{prop:idio:delta}~\ref{prop:idio:delta:two} and Appendix~\ref{app:lr}), its estimation tends to be more challenging.

\begin{figure}[htb!]
\centering
\includegraphics[width = .8\textwidth]{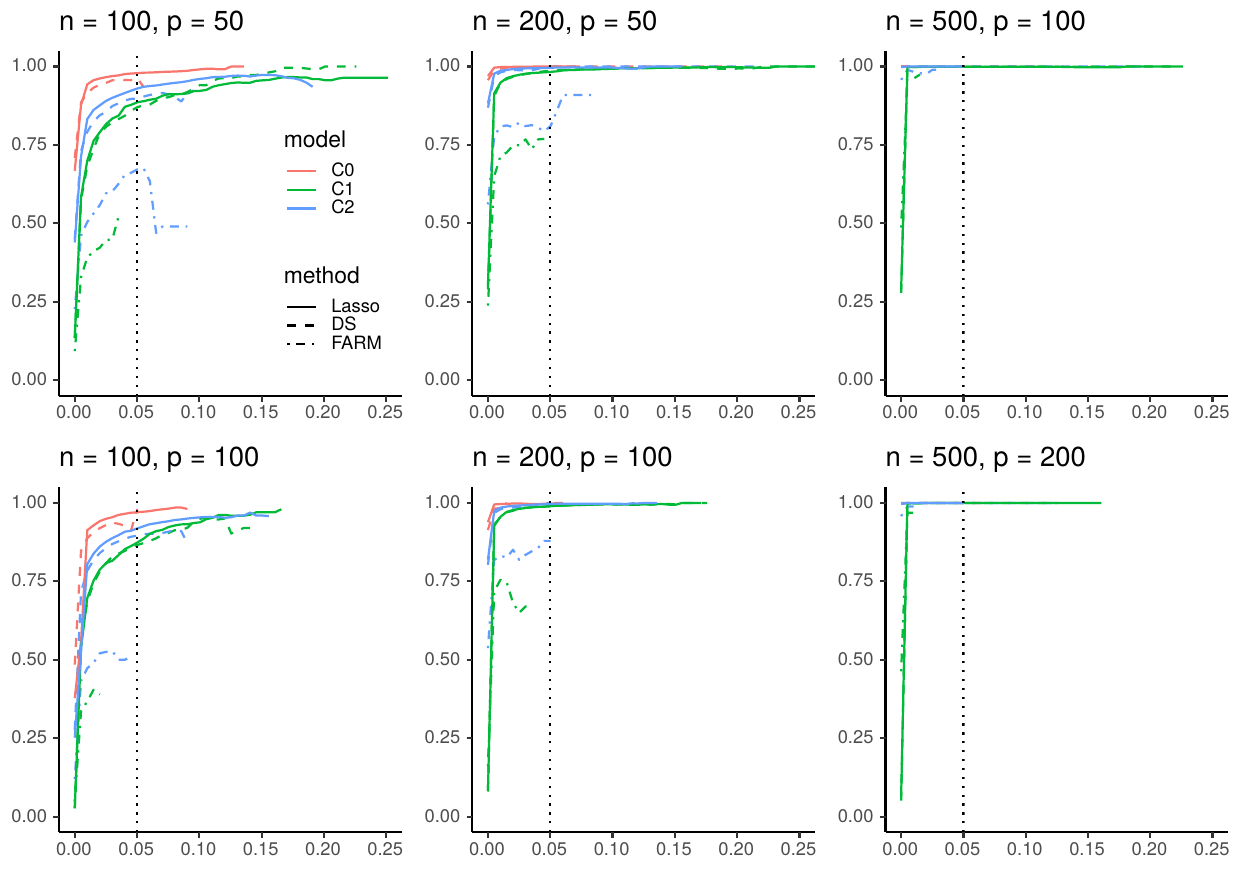}
\caption{\small ROC curves of TPR against FPR for
$\wh{\bm\beta}^{\las}$, $\wh{\bm\beta}^{\ds}$ and $\wh{\bm\beta}^{\text{FARM}}$
in recovering the support of $\bm\beta$
when $\bm\chi_t$ is generated under~\ref{m:ar}--\ref{m:ma}
and $\bm\xi_t$ is generated under~\ref{e:one} with varying $n$ and $p$, 
averaged over $100$ realisations.
Vertical lines indicate FPR $= 0.05$.
For comparison, we also plot the corresponding curves (from $\wh{\bm\beta}^{\las}$ and $\wh{\bm\beta}^{\ds}$) obtained under~\ref{m:oracle} 
i.e.\ when $\bm\chi_t = \mbf 0$.}
\label{fig:roc:beta:one}
\end{figure}

\begin{figure}[htb!]
\centering
\includegraphics[width = .8\textwidth]{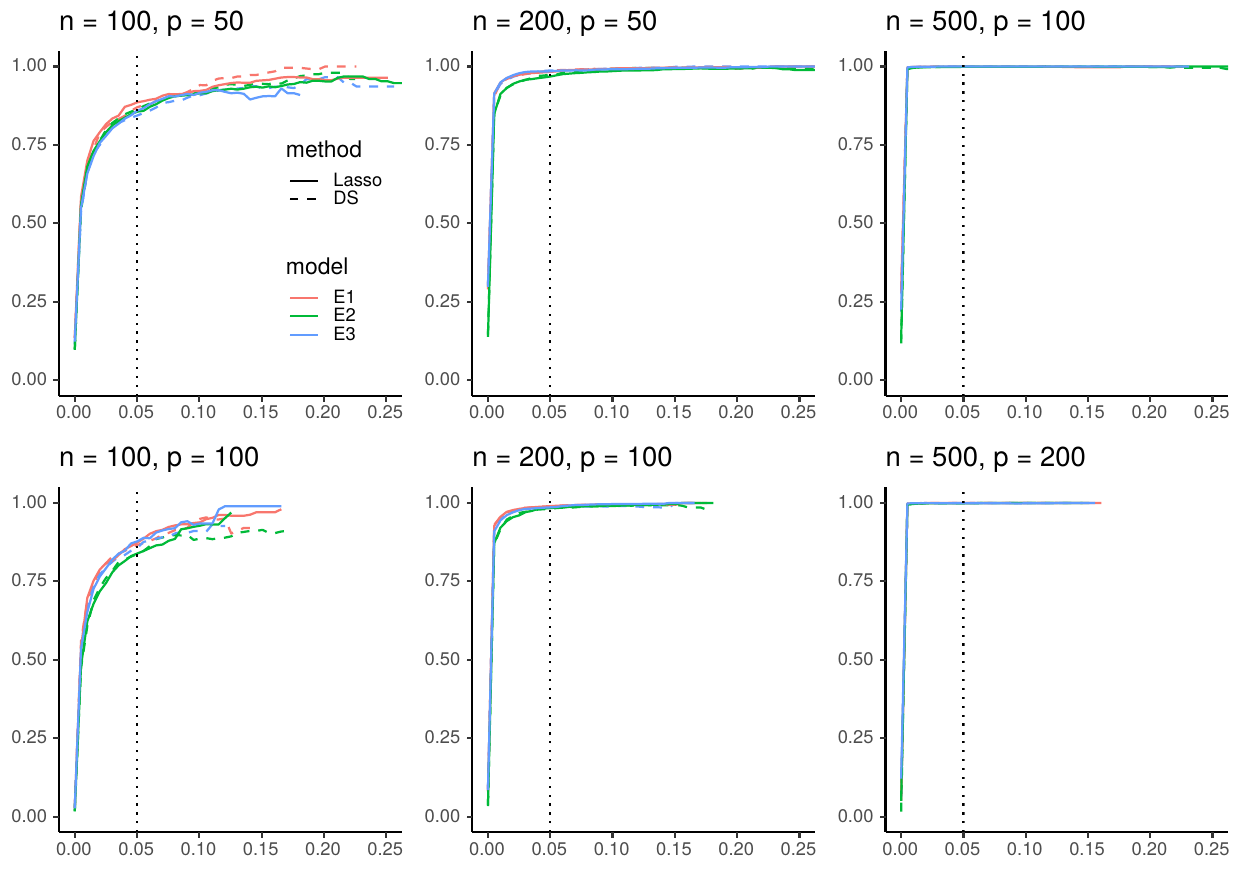}
\caption{\small ROC curves of TPR against FPR for
$\wh{\bm\beta}^{\las}$ and $\wh{\bm\beta}^{\ds}$ in recovering the support of $\bm\beta$
when $\bm\chi_t$ is generated under~\ref{m:ar}
and $\bm\xi_t$ is generated under~\ref{e:one}--\ref{e:three} with varying $n$ and $p$, 
averaged over $100$ realisations.
Vertical lines indicate FPR $= 0.05$.}
\label{fig:roc:beta:two}
\end{figure}

\begin{figure}[htb!]
\centering
\includegraphics[width = .8\textwidth]{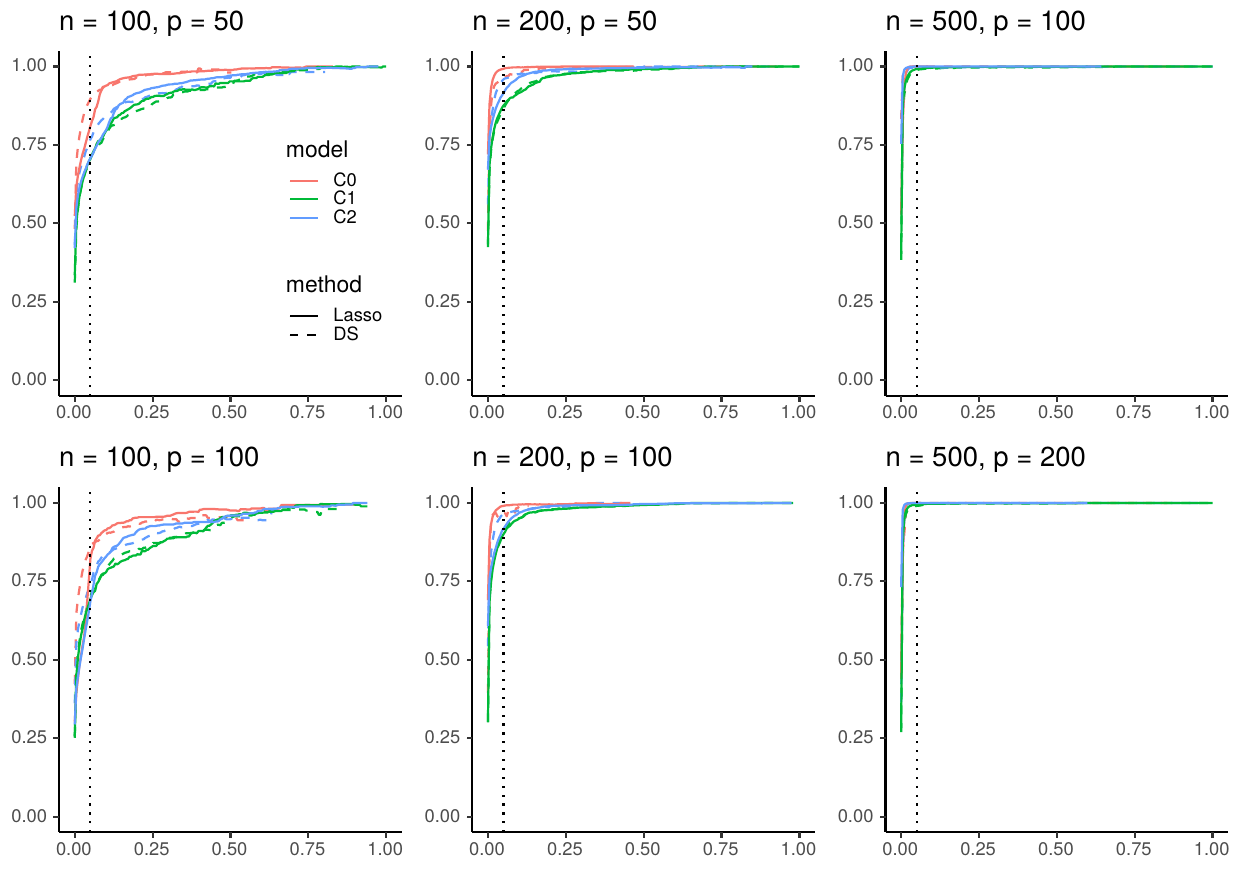}
\caption{\small ROC curves of TPR against FPR for
$\wh{\bm\Omega}^{\las}$ and $\wh{\bm\Omega}^{\ds}$
in recovering the support of $\bm\Omega$
when $\bm\chi_t$ is generated under~\ref{m:ar}--\ref{m:ma}
and $\bm\xi_t$ is generated under~\ref{e:one} with varying $n$ and $p$, 
averaged over $100$ realisations.
Vertical lines indicate FPR $= 0.05$.
For comparison, we also plot the corresponding curves obtained under~\ref{m:oracle} 
i.e.\ when $\bm\chi_t = \mbf 0$.}
\label{fig:roc:omega:one}
\end{figure}

\begin{figure}[htb!]
\centering
\includegraphics[width = .8\textwidth]{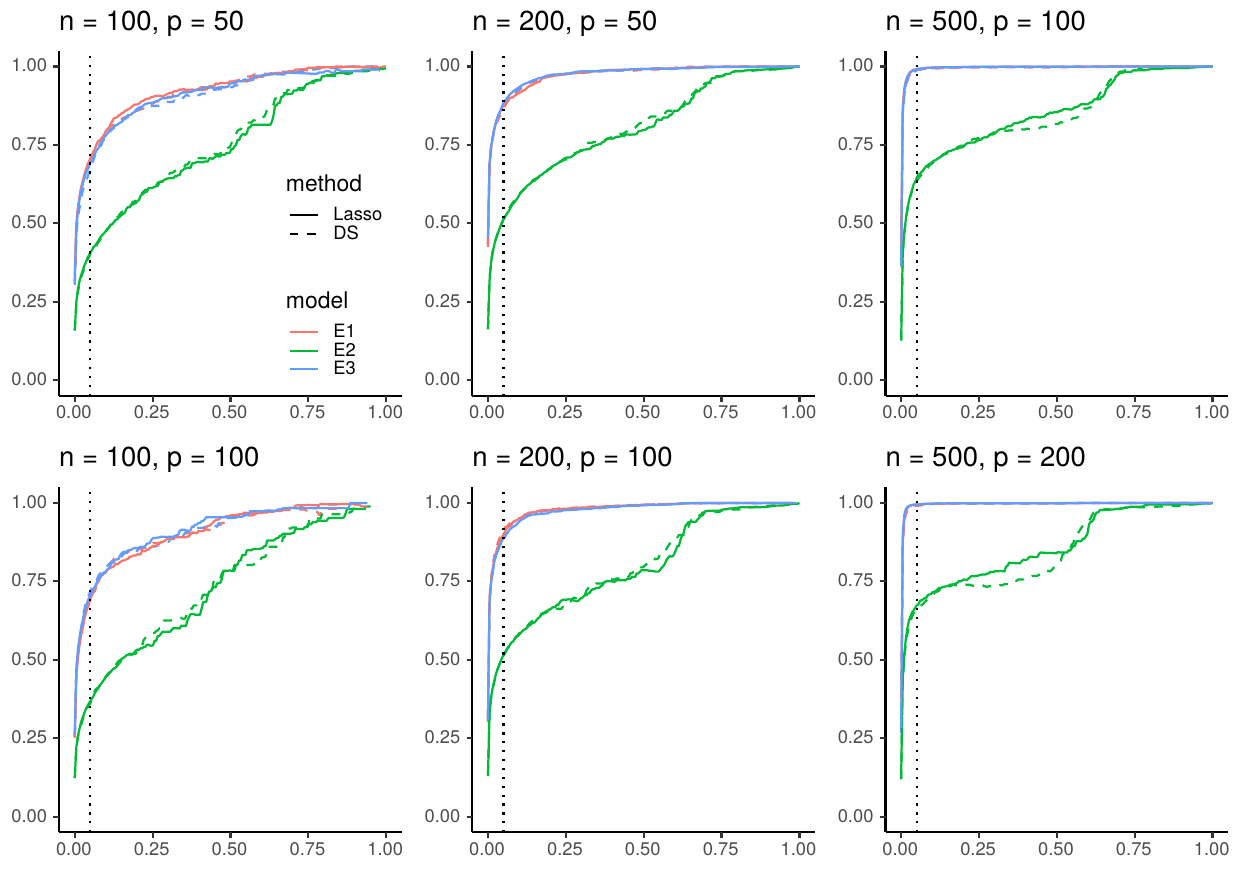}
\caption{\small ROC curves of TPR against FPR for
$\wh{\bm\Omega}^{\las}$ and $\wh{\bm\Omega}^{\ds}$
in recovering the support of $\bm\Omega$
when $\bm\chi_t$ is generated under~\ref{m:ar}
and $\bm\xi_t$ is generated under~\ref{e:one}--\ref{e:three} with varying $n$ and $p$, 
averaged over $100$ realisations.
Vertical lines indicate FPR $= 0.05$.}
\label{fig:roc:omega:two}
\end{figure}

\begin{table}[htb!]
\caption{\small 
Errors of $\wh{\bm\beta}^{\las}$, $\wh{\bm\beta}^{\ds}$ and $\wh{\bm\beta}^{\text{FARM}}$ in estimating $\bm\beta$ measured by $L_F$ and $L_2$ 
averaged over $100$ realisations (also reported are the standard errors)
under the model~\ref{e:one} for the generation of $\bm\xi_t$
and~\ref{m:oracle}--\ref{m:ma} for $\bm\chi_t$ with varying $n$ and $p$.
We also report the TPR when FPR $= 0.05$ without and with thresholding for $\wh{\bm\beta}^{\las}$ and $\wh{\bm\beta}^{\ds}$.}
\label{table:est:one:beta}
\centering
{\scriptsize
\begin{tabular}{cccr cc cc cc cc}
\toprule													
&	&	&	&	&	&	&	&	\multicolumn{4}{c}{TPR ($5$\%)} 				\\	
&	&	&	&	\multicolumn{2}{c}{$L_F$} &		\multicolumn{2}{c}{$L_2$} &		\multicolumn{2}{c}{Without} &		\multicolumn{2}{c}{With} 		\\	
&	$n$ &	$p$ &	Method &	Mean &	SD &	Mean &	SD &	Mean &	SD &	Mean &	SD 	\\	\cmidrule(lr){1-4} \cmidrule(lr){5-6} \cmidrule(lr){7-8} \cmidrule(lr){9-10} \cmidrule(lr){11-12}
\ref{m:oracle} &	100 &	50 &	Lasso &	0.633 &	0.177 &	0.68 &	0.163 &	0.836 &	0.243 &	0.815 &	0.262	\\	
&	&	&	DS &	0.613 &	0.072 &	0.679 &	0.082 &	0.94 &	0.105 &	0.893 &	0.125	\\	\cmidrule(lr){2-4} \cmidrule(lr){5-6} \cmidrule(lr){7-8} \cmidrule(lr){9-10} \cmidrule(lr){11-12}
&	100 &	100 &	Lasso &	0.854 &	0.142 &	0.88 &	0.129 &	0.517 &	0.27 &	0.466 &	0.283	\\	
&	&	&	DS &	0.669 &	0.076 &	0.738 &	0.079 &	0.892 &	0.16 &	0.839 &	0.159	\\	\cmidrule(lr){2-4} \cmidrule(lr){5-6} \cmidrule(lr){7-8} \cmidrule(lr){9-10} \cmidrule(lr){11-12}
&	200 &	50 &	Lasso &	0.421 &	0.048 &	0.473 &	0.059 &	0.999 &	0.004 &	0.997 &	0.009	\\	
&	&	&	DS &	0.532 &	0.071 &	0.589 &	0.084 &	0.989 &	0.019 &	0.982 &	0.024	\\	\cmidrule(lr){2-4} \cmidrule(lr){5-6} \cmidrule(lr){7-8} \cmidrule(lr){9-10} \cmidrule(lr){11-12}
&	200 &	100 &	Lasso &	0.454 &	0.034 &	0.52 &	0.056 &	0.999 &	0.004 &	0.996 &	0.008	\\	
&	&	&	DS &	0.58 &	0.044 &	0.654 &	0.066 &	0.982 &	0.018 &	0.966 &	0.031	\\	\cmidrule(lr){2-4} \cmidrule(lr){5-6} \cmidrule(lr){7-8} \cmidrule(lr){9-10} \cmidrule(lr){11-12}
&	500 &	100 &	Lasso &	0.402 &	0.034 &	0.441 &	0.041 &	1 &	0 &	0.987 &	0.020	\\	
&	&	&	DS &	0.29 &	0.074 &	0.308 &	0.085 &	1 &	0.001 &	0.998 &	0.008	\\	\cmidrule(lr){2-4} \cmidrule(lr){5-6} \cmidrule(lr){7-8} \cmidrule(lr){9-10} \cmidrule(lr){11-12}
&	500 &	200 &	Lasso &	0.425 &	0.034 &	0.47 &	0.048 &	1 &	0.001 &	0.986 &	0.024	\\	
&	&	&	DS &	0.46 &	0.128 &	0.493 &	0.15 &	0.999 &	0.002 &	0.98 &	0.021	\\	\cmidrule(lr){1-4} \cmidrule(lr){5-6} \cmidrule(lr){7-8} \cmidrule(lr){9-10} \cmidrule(lr){11-12}
\ref{m:ar} &	100 &	50 &	Lasso &	0.805 &	0.094 &	0.875 &	0.111 &	0.757 &	0.216 &	0.681 &	0.252	\\	
&	&	&	DS &	0.815 &	0.084 &	0.883 &	0.107 &	0.748 &	0.19 &	0.684 &	0.209	\\	
&	&	&	FARM &	0.914 &	0.047 &	0.954 &	0.088 &	0.404 &	0.127 &	- &	-	\\	\cmidrule(lr){2-4} \cmidrule(lr){5-6} \cmidrule(lr){7-8} \cmidrule(lr){9-10} \cmidrule(lr){11-12}
&	100 &	100 &	Lasso &	0.863 &	0.077 &	0.925 &	0.098 &	0.66 &	0.228 &	0.561 &	0.257	\\	
&	&	&	DS &	0.848 &	0.071 &	0.924 &	0.09 &	0.701 &	0.209 &	0.608 &	0.223	\\	
&	&	&	FARM &	0.927 &	0.026 &	0.96 &	0.086 &	0.361 &	0.086 &	- &	-	\\	\cmidrule(lr){2-4} \cmidrule(lr){5-6} \cmidrule(lr){7-8} \cmidrule(lr){9-10} \cmidrule(lr){11-12}
&	200 &	50 &	Lasso &	0.613 &	0.075 &	0.708 &	0.111 &	0.973 &	0.038 &	0.951 &	0.089	\\	
&	&	&	DS &	0.617 &	0.083 &	0.715 &	0.119 &	0.969 &	0.052 &	0.951 &	0.070	\\	
&	&	&	FARM &	0.804 &	0.057 &	0.871 &	0.135 &	0.726 &	0.106 &	- &	-	\\	\cmidrule(lr){2-4} \cmidrule(lr){5-6} \cmidrule(lr){7-8} \cmidrule(lr){9-10} \cmidrule(lr){11-12}
&	200 &	100 &	Lasso &	0.647 &	0.08 &	0.794 &	0.094 &	0.963 &	0.062 &	0.936 &	0.094	\\	
&	&	&	DS &	0.643 &	0.072 &	0.776 &	0.102 &	0.971 &	0.039 &	0.941 &	0.079	\\	
&	&	&	FARM &	0.794 &	0.045 &	0.841 &	0.095 &	0.733 &	0.098 &	- &	-	\\	\cmidrule(lr){2-4} \cmidrule(lr){5-6} \cmidrule(lr){7-8} \cmidrule(lr){9-10} \cmidrule(lr){11-12}
&	500 &	100 &	Lasso &	0.461 &	0.054 &	0.657 &	0.094 &	0.999 &	0.003 &	0.996 &	0.015	\\	
&	&	&	DS &	0.48 &	0.057 &	0.665 &	0.107 &	0.999 &	0.004 &	0.998 &	0.006	\\	
&	&	&	FARM &	0.625 &	0.037 &	0.725 &	0.124 &	0.961 &	0.03 &	- &	-	\\	\cmidrule(lr){2-4} \cmidrule(lr){5-6} \cmidrule(lr){7-8} \cmidrule(lr){9-10} \cmidrule(lr){11-12}
&	500 &	200 &	Lasso &	0.501 &	0.058 &	0.763 &	0.083 &	0.999 &	0.003 &	0.996 &	0.008	\\	
&	&	&	DS &	0.518 &	0.066 &	0.813 &	0.107 &	0.999 &	0.003 &	0.961 &	0.176	\\	
&	&	&	FARM &	0.611 &	0.035 &	0.704 &	0.122 &	0.969 &	0.021 &	- &	-	\\	\cmidrule(lr){1-4} \cmidrule(lr){5-6} \cmidrule(lr){7-8} \cmidrule(lr){9-10} \cmidrule(lr){11-12}
\ref{m:ma} &	100 &	50 &	Lasso &	0.721 &	0.118 &	0.756 &	0.116 &	0.819 &	0.236 &	0.805 &	0.246	\\	
&	&	&	DS &	0.704 &	0.057 &	0.759 &	0.074 &	0.888 &	0.08 &	0.837 &	0.094	\\	
&	&	&	FARM &	0.857 &	0.046 &	0.888 &	0.071 &	0.534 &	0.137 &	- &	-	\\	\cmidrule(lr){2-4} \cmidrule(lr){5-6} \cmidrule(lr){7-8} \cmidrule(lr){9-10} \cmidrule(lr){11-12}
&	100 &	100 &	Lasso &	0.868 &	0.084 &	0.886 &	0.089 &	0.572 &	0.251 &	0.517 &	0.274	\\	
&	&	&	DS &	0.749 &	0.08 &	0.786 &	0.077 &	0.826 &	0.216 &	0.766 &	0.214	\\	
&	&	&	FARM &	0.882 &	0.031 &	0.894 &	0.071 &	0.483 &	0.085 &	- &	-	\\	\cmidrule(lr){2-4} \cmidrule(lr){5-6} \cmidrule(lr){7-8} \cmidrule(lr){9-10} \cmidrule(lr){11-12}
&	200 &	50 &	Lasso &	0.503 &	0.04 &	0.551 &	0.065 &	0.996 &	0.01 &	0.994 &	0.012	\\	
&	&	&	DS &	0.575 &	0.051 &	0.635 &	0.077 &	0.988 &	0.02 &	0.971 &	0.034	\\	
&	&	&	FARM &	0.737 &	0.057 &	0.774 &	0.093 &	0.821 &	0.089 &	- &	-	\\	\cmidrule(lr){2-4} \cmidrule(lr){5-6} \cmidrule(lr){7-8} \cmidrule(lr){9-10} \cmidrule(lr){11-12}
&	200 &	100 &	Lasso &	0.53 &	0.05 &	0.559 &	0.057 &	0.995 &	0.019 &	0.99 &	0.026	\\	
&	&	&	DS &	0.568 &	0.042 &	0.625 &	0.062 &	0.987 &	0.015 &	0.973 &	0.023	\\	
&	&	&	FARM &	0.726 &	0.046 &	0.722 &	0.064 &	0.83 &	0.078 &	- &	-	\\	\cmidrule(lr){2-4} \cmidrule(lr){5-6} \cmidrule(lr){7-8} \cmidrule(lr){9-10} \cmidrule(lr){11-12}
&	500 &	100 &	Lasso &	0.374 &	0.026 &	0.417 &	0.042 &	1 &	0 &	1 &	0.000	\\	
&	&	&	DS &	0.448 &	0.05 &	0.494 &	0.064 &	1 &	0.001 &	0.994 &	0.016	\\	
&	&	&	FARM &	0.551 &	0.043 &	0.566 &	0.089 &	0.99 &	0.018 &	- &	-	\\	\cmidrule(lr){2-4} \cmidrule(lr){5-6} \cmidrule(lr){7-8} \cmidrule(lr){9-10} \cmidrule(lr){11-12}
&	500 &	200 &	Lasso &	0.383 &	0.023 &	0.425 &	0.035 &	1 &	0 &	1 &	0.000	\\	
&	&	&	DS &	0.478 &	0.033 &	0.528 &	0.045 &	1 &	0.001 &	0.995 &	0.018	\\	
&	&	&	FARM &	0.559 &	0.033 &	0.551 &	0.046 &	0.988 &	0.012 &	- &	-	\\	\bottomrule
\end{tabular}}
\end{table}

\begin{table}[htb!]
\caption{\small 
Errors of $\wh{\bm\Omega}^{\las}$ and $\wh{\bm\Omega}^{\ds}$ in estimating $\bm\Omega$ measured by $L_F$ and $L_2$,
averaged over $100$ realisations (also reported are the standard errors)
under the model~\ref{e:one} for the generation of $\bm\xi_t$
and~\ref{m:oracle}--\ref{m:ma} for $\bm\chi_t$ with varying $n$ and $p$.
We also report the TPR when FPR $= 0.05$ without and with thresholding.}
\label{table:est:one:omega}
\centering
{\scriptsize
\begin{tabular}{cccr cc cc cc cc}
\toprule													
&	&	&	&	&	&	&	&	\multicolumn{4}{c}{TPR ($5$\%)} 				\\	
&	&	&	&	\multicolumn{2}{c}{$L_F$} &		\multicolumn{2}{c}{$L_2$} &		\multicolumn{2}{c}{Without} &		\multicolumn{2}{c}{With} 		\\	
&	$n$ &	$p$ &	Method &	Mean &	SD &	Mean &	SD &	Mean &	SD &	Mean &	SD 	\\	\cmidrule(lr){1-4} \cmidrule(lr){5-6} \cmidrule(lr){7-8} \cmidrule(lr){9-10} \cmidrule(lr){11-12}
\ref{m:oracle} &	100 &	50 &	Lasso &	0.452 &	0.087 &	0.587 &	0.105 &	0.789 &	0.18 &	0.704 &	0.236	\\	
&	&	&	DS &	0.456 &	0.042 &	0.593 &	0.048 &	0.896 &	0.092 &	0.721 &	0.111	\\	\cmidrule(lr){2-4} \cmidrule(lr){5-6} \cmidrule(lr){7-8} \cmidrule(lr){9-10} \cmidrule(lr){11-12}
&	100 &	100 &	Lasso &	0.587 &	0.095 &	0.738 &	0.1 &	0.583 &	0.17 &	0.411 &	0.195	\\	
&	&	&	DS &	0.467 &	0.054 &	0.624 &	0.056 &	0.846 &	0.113 &	0.658 &	0.106	\\	\cmidrule(lr){2-4} \cmidrule(lr){5-6} \cmidrule(lr){7-8} \cmidrule(lr){9-10} \cmidrule(lr){11-12}
&	200 &	50 &	Lasso &	0.373 &	0.026 &	0.488 &	0.038 &	0.993 &	0.016 &	0.854 &	0.090	\\	
&	&	&	DS &	0.423 &	0.043 &	0.553 &	0.058 &	0.979 &	0.042 &	0.755 &	0.134	\\	\cmidrule(lr){2-4} \cmidrule(lr){5-6} \cmidrule(lr){7-8} \cmidrule(lr){9-10} \cmidrule(lr){11-12}
&	200 &	100 &	Lasso &	0.376 &	0.022 &	0.507 &	0.03 &	0.991 &	0.016 &	0.839 &	0.064	\\	
&	&	&	DS &	0.444 &	0.024 &	0.593 &	0.03 &	0.98 &	0.021 &	0.727 &	0.069	\\	\cmidrule(lr){2-4} \cmidrule(lr){5-6} \cmidrule(lr){7-8} \cmidrule(lr){9-10} \cmidrule(lr){11-12}
&	500 &	100 &	Lasso &	0.327 &	0.02 &	0.453 &	0.029 &	1 &	0.001 &	0.784 &	0.040	\\	
&	&	&	DS &	0.236 &	0.045 &	0.328 &	0.064 &	1 &	0.002 &	0.975 &	0.067	\\	\cmidrule(lr){2-4} \cmidrule(lr){5-6} \cmidrule(lr){7-8} \cmidrule(lr){9-10} \cmidrule(lr){11-12}
&	500 &	200 &	Lasso &	0.328 &	0.017 &	0.466 &	0.029 &	1 &	0.001 &	0.77 &	0.029	\\	
&	&	&	DS &	0.336 &	0.059 &	0.473 &	0.09 &	0.999 &	0.003 &	0.845 &	0.123	\\	\cmidrule(lr){1-4} \cmidrule(lr){5-6} \cmidrule(lr){7-8} \cmidrule(lr){9-10} \cmidrule(lr){11-12}
\ref{m:ar} &	100 &	50 &	Lasso &	0.486 &	0.057 &	0.652 &	0.154 &	0.697 &	0.148 &	0.578 &	0.169	\\	
&	&	&	DS &	0.488 &	0.064 &	0.662 &	0.18 &	0.691 &	0.129 &	0.545 &	0.162	\\	\cmidrule(lr){2-4} \cmidrule(lr){5-6} \cmidrule(lr){7-8} \cmidrule(lr){9-10} \cmidrule(lr){11-12}
&	100 &	100 &	Lasso &	0.515 &	0.069 &	0.696 &	0.099 &	0.641 &	0.127 &	0.475 &	0.153	\\	
&	&	&	DS &	0.503 &	0.062 &	0.687 &	0.137 &	0.662 &	0.118 &	0.498 &	0.155	\\	\cmidrule(lr){2-4} \cmidrule(lr){5-6} \cmidrule(lr){7-8} \cmidrule(lr){9-10} \cmidrule(lr){11-12}
&	200 &	50 &	Lasso &	0.474 &	0.723 &	0.812 &	2.66 &	0.876 &	0.106 &	0.769 &	0.145	\\	
&	&	&	DS &	0.403 &	0.052 &	0.563 &	0.123 &	0.872 &	0.089 &	0.769 &	0.131	\\	\cmidrule(lr){2-4} \cmidrule(lr){5-6} \cmidrule(lr){7-8} \cmidrule(lr){9-10} \cmidrule(lr){11-12}
&	200 &	100 &	Lasso &	0.416 &	0.046 &	0.573 &	0.071 &	0.898 &	0.071 &	0.728 &	0.149	\\	
&	&	&	DS &	0.417 &	0.048 &	0.572 &	0.066 &	0.91 &	0.059 &	0.737 &	0.130	\\	\cmidrule(lr){2-4} \cmidrule(lr){5-6} \cmidrule(lr){7-8} \cmidrule(lr){9-10} \cmidrule(lr){11-12}
&	500 &	100 &	Lasso &	0.33 &	0.033 &	0.488 &	0.068 &	0.992 &	0.014 &	0.881 &	0.089	\\	
&	&	&	DS &	0.337 &	0.037 &	0.495 &	0.065 &	0.989 &	0.019 &	0.864 &	0.096	\\	\cmidrule(lr){2-4} \cmidrule(lr){5-6} \cmidrule(lr){7-8} \cmidrule(lr){9-10} \cmidrule(lr){11-12}
&	500 &	200 &	Lasso &	0.348 &	0.04 &	0.523 &	0.055 &	0.995 &	0.008 &	0.841 &	0.088	\\	
&	&	&	DS &	0.35 &	0.046 &	0.535 &	0.062 &	0.992 &	0.018 &	0.828 &	0.103	\\	\cmidrule(lr){1-4} \cmidrule(lr){5-6} \cmidrule(lr){7-8} \cmidrule(lr){9-10} \cmidrule(lr){11-12}
\ref{m:ma} &	100 &	50 &	Lasso &	0.433 &	0.067 &	0.576 &	0.093 &	0.696 &	0.159 &	0.666 &	0.195	\\	
&	&	&	DS &	0.433 &	0.033 &	0.584 &	0.044 &	0.768 &	0.098 &	0.668 &	0.108	\\	\cmidrule(lr){2-4} \cmidrule(lr){5-6} \cmidrule(lr){7-8} \cmidrule(lr){9-10} \cmidrule(lr){11-12}
&	100 &	100 &	Lasso &	0.526 &	0.084 &	0.68 &	0.102 &	0.595 &	0.133 &	0.446 &	0.199	\\	
&	&	&	DS &	0.458 &	0.06 &	0.617 &	0.065 &	0.727 &	0.138 &	0.617 &	0.130	\\	\cmidrule(lr){2-4} \cmidrule(lr){5-6} \cmidrule(lr){7-8} \cmidrule(lr){9-10} \cmidrule(lr){11-12}
&	200 &	50 &	Lasso &	0.349 &	0.03 &	0.48 &	0.046 &	0.915 &	0.068 &	0.843 &	0.077	\\	
&	&	&	DS &	0.399 &	0.034 &	0.541 &	0.043 &	0.96 &	0.047 &	0.769 &	0.097	\\	\cmidrule(lr){2-4} \cmidrule(lr){5-6} \cmidrule(lr){7-8} \cmidrule(lr){9-10} \cmidrule(lr){11-12}
&	200 &	100 &	Lasso &	0.334 &	0.027 &	0.471 &	0.042 &	0.917 &	0.054 &	0.838 &	0.078	\\	
&	&	&	DS &	0.391 &	0.03 &	0.544 &	0.038 &	0.966 &	0.03 &	0.76 &	0.066	\\	\cmidrule(lr){2-4} \cmidrule(lr){5-6} \cmidrule(lr){7-8} \cmidrule(lr){9-10} \cmidrule(lr){11-12}
&	500 &	100 &	Lasso &	0.287 &	0.019 &	0.413 &	0.032 &	1 &	0.001 &	0.884 &	0.080	\\	
&	&	&	DS &	0.321 &	0.043 &	0.456 &	0.058 &	0.998 &	0.008 &	0.819 &	0.086	\\	\cmidrule(lr){2-4} \cmidrule(lr){5-6} \cmidrule(lr){7-8} \cmidrule(lr){9-10} \cmidrule(lr){11-12}
&	500 &	200 &	Lasso &	0.292 &	0.022 &	0.428 &	0.028 &	1 &	0.001 &	0.889 &	0.067	\\	
&	&	&	DS &	0.34 &	0.021 &	0.491 &	0.034 &	1 &	0.002 &	0.775 &	0.050	\\	\bottomrule
\end{tabular}}
\end{table}

\begin{table}[htb!]
\caption{\small 
Errors of $\wh{\bm\beta}^{\las}$ and $\wh{\bm\beta}^{\ds}$ in estimating $\bm\beta$ measured by $L_F$ and $L_2$ 
averaged over $100$ realisations (also reported are the standard errors)
under the models~\ref{e:two}--\ref{e:three} for the generation of $\bm\xi_t$
and~\ref{m:ar} for $\bm\chi_t$ with varying $n$ and $p$.
We also report the TPR when FPR $= 0.05$ without and with thresholding.}
\label{table:est:two:beta}
\centering
{\scriptsize
\begin{tabular}{cccr cc cc cc cc}	
\toprule													
&	&	&	&	&	&	&	&	\multicolumn{4}{c}{TPR ($5$\%)} 				\\	
&	&	&	&	\multicolumn{2}{c}{$L_F$} &		\multicolumn{2}{c}{$L_2$} &		\multicolumn{2}{c}{Without} &		\multicolumn{2}{c}{With} 		\\	
&	$n$ &	$p$ &	Method &	Mean &	SD &	Mean &	SD &	Mean &	SD &	Mean &	SD 	\\	\cmidrule(lr){1-4} \cmidrule(lr){5-6} \cmidrule(lr){7-8} \cmidrule(lr){9-10} \cmidrule(lr){11-12}
\ref{e:two} &	100 &	50 &	Lasso &	0.814 &	0.088 &	0.898 &	0.229 &	0.787 &	0.161 &	0.727 &	0.188	\\	
&	&	&	DS &	0.82 &	0.092 &	0.908 &	0.199 &	0.753 &	0.185 &	0.677 &	0.230	\\	\cmidrule(lr){2-4} \cmidrule(lr){5-6} \cmidrule(lr){7-8} \cmidrule(lr){9-10} \cmidrule(lr){11-12}
&	100 &	100 &	Lasso &	0.883 &	0.063 &	0.963 &	0.099 &	0.636 &	0.216 &	0.536 &	0.234	\\	
&	&	&	DS &	0.889 &	0.074 &	0.983 &	0.143 &	0.662 &	0.228 &	0.552 &	0.244	\\	\cmidrule(lr){2-4} \cmidrule(lr){5-6} \cmidrule(lr){7-8} \cmidrule(lr){9-10} \cmidrule(lr){11-12}
&	200 &	50 &	Lasso &	0.655 &	0.068 &	0.753 &	0.105 &	0.959 &	0.053 &	0.941 &	0.079	\\	
&	&	&	DS &	0.651 &	0.082 &	0.743 &	0.111 &	0.959 &	0.052 &	0.932 &	0.086	\\	\cmidrule(lr){2-4} \cmidrule(lr){5-6} \cmidrule(lr){7-8} \cmidrule(lr){9-10} \cmidrule(lr){11-12}
&	200 &	100 &	Lasso &	0.694 &	0.07 &	0.842 &	0.086 &	0.948 &	0.07 &	0.904 &	0.116	\\	
&	&	&	DS &	0.697 &	0.081 &	0.849 &	0.108 &	0.941 &	0.08 &	0.893 &	0.143	\\	\cmidrule(lr){2-4} \cmidrule(lr){5-6} \cmidrule(lr){7-8} \cmidrule(lr){9-10} \cmidrule(lr){11-12}
&	500 &	100 &	Lasso &	0.519 &	0.062 &	0.73 &	0.109 &	0.998 &	0.004 &	0.996 &	0.009	\\	
&	&	&	DS &	0.524 &	0.06 &	0.755 &	0.111 &	0.999 &	0.004 &	0.997 &	0.007	\\	\cmidrule(lr){2-4} \cmidrule(lr){5-6} \cmidrule(lr){7-8} \cmidrule(lr){9-10} \cmidrule(lr){11-12}
&	500 &	200 &	Lasso &	0.549 &	0.055 &	0.83 &	0.088 &	0.997 &	0.004 &	0.993 &	0.010	\\	
&	&	&	DS &	0.557 &	0.05 &	0.907 &	0.092 &	0.997 &	0.006 &	0.993 &	0.014	\\	\cmidrule(lr){1-4} \cmidrule(lr){5-6} \cmidrule(lr){7-8} \cmidrule(lr){9-10} \cmidrule(lr){11-12}
\ref{e:three} &	100 &	50 &	Lasso &	0.813 &	0.092 &	0.867 &	0.111 &	0.745 &	0.183 &	0.676 &	0.218	\\	
&	&	&	DS &	0.829 &	0.09 &	0.893 &	0.111 &	0.709 &	0.225 &	0.649 &	0.247	\\	\cmidrule(lr){2-4} \cmidrule(lr){5-6} \cmidrule(lr){7-8} \cmidrule(lr){9-10} \cmidrule(lr){11-12}
&	100 &	100 &	Lasso &	0.857 &	0.078 &	0.936 &	0.104 &	0.654 &	0.201 &	0.558 &	0.223	\\	
&	&	&	DS &	0.864 &	0.08 &	0.946 &	0.126 &	0.635 &	0.234 &	0.538 &	0.270	\\	\cmidrule(lr){2-4} \cmidrule(lr){5-6} \cmidrule(lr){7-8} \cmidrule(lr){9-10} \cmidrule(lr){11-12}
&	200 &	50 &	Lasso &	0.617 &	0.07 &	0.701 &	0.095 &	0.972 &	0.048 &	0.95 &	0.086	\\	
&	&	&	DS &	0.617 &	0.075 &	0.699 &	0.094 &	0.97 &	0.037 &	0.949 &	0.060	\\	\cmidrule(lr){2-4} \cmidrule(lr){5-6} \cmidrule(lr){7-8} \cmidrule(lr){9-10} \cmidrule(lr){11-12}
&	200 &	100 &	Lasso &	0.668 &	0.078 &	0.808 &	0.1 &	0.948 &	0.066 &	0.909 &	0.122	\\	
&	&	&	DS &	0.655 &	0.087 &	0.796 &	0.11 &	0.953 &	0.07 &	0.918 &	0.122	\\	\cmidrule(lr){2-4} \cmidrule(lr){5-6} \cmidrule(lr){7-8} \cmidrule(lr){9-10} \cmidrule(lr){11-12}
&	500 &	100 &	Lasso &	0.474 &	0.055 &	0.648 &	0.095 &	0.999 &	0.004 &	0.998 &	0.007	\\	
&	&	&	DS &	0.474 &	0.062 &	0.653 &	0.118 &	0.999 &	0.003 &	0.998 &	0.006	\\	\cmidrule(lr){2-4} \cmidrule(lr){5-6} \cmidrule(lr){7-8} \cmidrule(lr){9-10} \cmidrule(lr){11-12}
&	500 &	200 &	Lasso &	0.489 &	0.055 &	0.766 &	0.085 &	0.999 &	0.002 &	0.998 &	0.005	\\	
&	&	&	DS &	0.516 &	0.053 &	0.811 &	0.11 &	0.999 &	0.003 &	0.988 &	0.082	\\	\bottomrule								
\end{tabular}}
\end{table}

\begin{table}[htb!]
\caption{\small 
Errors of $\wh{\bm\Omega}^{\las}$ and $\wh{\bm\Omega}^{\ds}$ in estimating $\bm\Omega$ measured by $L_F$ and $L_2$,
averaged over $100$ realisations (also reported are the standard errors)
under the models~\ref{e:two}--\ref{e:three} for the generation of $\bm\xi_t$
and~\ref{m:ar} for $\bm\chi_t$ with varying $n$ and $p$.
We also report the TPR when FPR $= 0.05$ without and with thresholding.}
\label{table:est:two:omega}
\centering
{\scriptsize
\begin{tabular}{cccr cc cc cc cc}
\toprule													
&	&	&	&	&	&	&	&	\multicolumn{4}{c}{TPR ($5$\%)} 				\\	
&	&	&	&	\multicolumn{2}{c}{$L_F$} &		\multicolumn{2}{c}{$L_2$} &		\multicolumn{2}{c}{Without} &		\multicolumn{2}{c}{With} 		\\	
&	$n$ &	$p$ &	Method &	Mean &	SD &	Mean &	SD &	Mean &	SD &	Mean &	SD 	\\	\cmidrule(lr){1-4} \cmidrule(lr){5-6} \cmidrule(lr){7-8} \cmidrule(lr){9-10} \cmidrule(lr){11-12}
\ref{e:two} &	100 &	50 &	Lasso &	0.582 &	0.055 &	0.732 &	0.155 &	0.407 &	0.071 &	0.329 &	0.086	\\	
&	&	&	DS &	0.582 &	0.054 &	0.726 &	0.105 &	0.396 &	0.076 &	0.323 &	0.099	\\	\cmidrule(lr){2-4} \cmidrule(lr){5-6} \cmidrule(lr){7-8} \cmidrule(lr){9-10} \cmidrule(lr){11-12}
&	100 &	100 &	Lasso &	0.615 &	0.061 &	0.756 &	0.059 &	0.352 &	0.065 &	0.25 &	0.071	\\	
&	&	&	DS &	0.621 &	0.059 &	0.765 &	0.064 &	0.35 &	0.066 &	0.243 &	0.078	\\	\cmidrule(lr){2-4} \cmidrule(lr){5-6} \cmidrule(lr){7-8} \cmidrule(lr){9-10} \cmidrule(lr){11-12}
&	200 &	50 &	Lasso &	0.504 &	0.066 &	0.649 &	0.157 &	0.513 &	0.072 &	0.454 &	0.080	\\	
&	&	&	DS &	0.502 &	0.038 &	0.638 &	0.06 &	0.514 &	0.065 &	0.452 &	0.087	\\	\cmidrule(lr){2-4} \cmidrule(lr){5-6} \cmidrule(lr){7-8} \cmidrule(lr){9-10} \cmidrule(lr){11-12}
&	200 &	100 &	Lasso &	0.522 &	0.047 &	0.658 &	0.067 &	0.515 &	0.06 &	0.392 &	0.080	\\	
&	&	&	DS &	0.525 &	0.048 &	0.669 &	0.064 &	0.518 &	0.063 &	0.392 &	0.089	\\	\cmidrule(lr){2-4} \cmidrule(lr){5-6} \cmidrule(lr){7-8} \cmidrule(lr){9-10} \cmidrule(lr){11-12}
&	500 &	100 &	Lasso &	0.442 &	0.042 &	0.61 &	0.149 &	0.646 &	0.06 &	0.524 &	0.079	\\	
&	&	&	DS &	0.436 &	0.042 &	0.594 &	0.135 &	0.635 &	0.058 &	0.53 &	0.085	\\	\cmidrule(lr){2-4} \cmidrule(lr){5-6} \cmidrule(lr){7-8} \cmidrule(lr){9-10} \cmidrule(lr){11-12}
&	500 &	200 &	Lasso &	0.457 &	0.039 &	0.608 &	0.066 &	0.674 &	0.043 &	0.484 &	0.059	\\	
&	&	&	DS &	0.447 &	0.038 &	0.598 &	0.063 &	0.659 &	0.041 &	0.493 &	0.064	\\	\cmidrule(lr){1-4} \cmidrule(lr){5-6} \cmidrule(lr){7-8} \cmidrule(lr){9-10} \cmidrule(lr){11-12}
\ref{e:three} &	100 &	50 &	Lasso &	0.495 &	0.057 &	0.652 &	0.113 &	0.684 &	0.125 &	0.546 &	0.171	\\	
&	&	&	DS &	0.502 &	0.056 &	0.655 &	0.097 &	0.661 &	0.147 &	0.525 &	0.168	\\	\cmidrule(lr){2-4} \cmidrule(lr){5-6} \cmidrule(lr){7-8} \cmidrule(lr){9-10} \cmidrule(lr){11-12}
&	100 &	100 &	Lasso &	0.519 &	0.055 &	0.696 &	0.082 &	0.645 &	0.125 &	0.46 &	0.148	\\	
&	&	&	DS &	0.521 &	0.058 &	0.696 &	0.081 &	0.628 &	0.145 &	0.47 &	0.161	\\	\cmidrule(lr){2-4} \cmidrule(lr){5-6} \cmidrule(lr){7-8} \cmidrule(lr){9-10} \cmidrule(lr){11-12}
&	200 &	50 &	Lasso &	0.4 &	0.048 &	0.541 &	0.073 &	0.882 &	0.071 &	0.763 &	0.122	\\	
&	&	&	DS &	0.403 &	0.045 &	0.547 &	0.067 &	0.885 &	0.066 &	0.763 &	0.132	\\	\cmidrule(lr){2-4} \cmidrule(lr){5-6} \cmidrule(lr){7-8} \cmidrule(lr){9-10} \cmidrule(lr){11-12}
&	200 &	100 &	Lasso &	0.429 &	0.05 &	0.586 &	0.074 &	0.893 &	0.073 &	0.69 &	0.160	\\	
&	&	&	DS &	0.423 &	0.049 &	0.571 &	0.071 &	0.895 &	0.074 &	0.72 &	0.154	\\	\cmidrule(lr){2-4} \cmidrule(lr){5-6} \cmidrule(lr){7-8} \cmidrule(lr){9-10} \cmidrule(lr){11-12}
&	500 &	100 &	Lasso &	0.338 &	0.039 &	0.499 &	0.064 &	0.992 &	0.016 &	0.856 &	0.109	\\	
&	&	&	DS &	0.336 &	0.039 &	0.499 &	0.067 &	0.989 &	0.019 &	0.867 &	0.094	\\	\cmidrule(lr){2-4} \cmidrule(lr){5-6} \cmidrule(lr){7-8} \cmidrule(lr){9-10} \cmidrule(lr){11-12}
&	500 &	200 &	Lasso &	0.349 &	0.04 &	0.522 &	0.056 &	0.996 &	0.009 &	0.849 &	0.087	\\	
&	&	&	DS &	0.357 &	0.042 &	0.542 &	0.061 &	0.995 &	0.009 &	0.825 &	0.082	\\	\bottomrule
\end{tabular}}
\end{table}

\begin{table}[htb!]
\caption{\small 
Errors of $\wh{\bm\Delta}^{\las}$ and $\wh{\bm\Delta}^{\ds}$ in estimating $\bm\Delta$ measured by $L_F$ and $L_2$,
averaged over $100$ realisations (also reported are the standard errors)
under the models~\ref{e:one}--\ref{e:two} for the generation of $\bm\xi_t$
and~\ref{m:ar} for $\bm\chi_t$ with varying $n$ and $p$.
We also report the TPR when FPR $= 0.05$ without and with thresholding.}
\label{table:est:delta}
\centering
{\scriptsize
\begin{tabular}{cccr cc cc cc cc}
\toprule													
&	&	&	&	&	&	&	&	\multicolumn{4}{c}{TPR ($5$\%)} 				\\	
&	&	&	&	\multicolumn{2}{c}{$L_F$} &		\multicolumn{2}{c}{$L_2$} &		\multicolumn{2}{c}{Without} &		\multicolumn{2}{c}{With} 		\\	
&	$n$ &	$p$ &	Method &	Mean &	SD &	Mean &	SD &	Mean &	SD &	Mean &	SD 	\\	\cmidrule(lr){1-4} \cmidrule(lr){5-6} \cmidrule(lr){7-8} \cmidrule(lr){9-10} \cmidrule(lr){11-12}
\ref{e:one} &	100 &	50 &	Lasso &	0.244 &	0.052 &	0.585 &	0.217 &	1 &	0 &	0.998 &	0.007	\\	
&	&	&	DS &	0.243 &	0.052 &	0.556 &	0.162 &	1 &	0 &	0.999 &	0.006	\\	\cmidrule(lr){2-4} \cmidrule(lr){5-6} \cmidrule(lr){7-8} \cmidrule(lr){9-10} \cmidrule(lr){11-12}
&	100 &	100 &	Lasso &	0.286 &	0.085 &	0.692 &	0.255 &	1 &	0 &	0.998 &	0.006	\\	
&	&	&	DS &	0.274 &	0.072 &	0.646 &	0.219 &	1 &	0 &	0.997 &	0.007	\\	\cmidrule(lr){2-4} \cmidrule(lr){5-6} \cmidrule(lr){7-8} \cmidrule(lr){9-10} \cmidrule(lr){11-12}
&	200 &	50 &	Lasso &	0.241 &	0.413 &	0.749 &	3.01 &	1 &	0 &	0.988 &	0.092	\\	
&	&	&	DS &	0.201 &	0.04 &	0.471 &	0.192 &	1 &	0 &	0.999 &	0.006	\\	\cmidrule(lr){2-4} \cmidrule(lr){5-6} \cmidrule(lr){7-8} \cmidrule(lr){9-10} \cmidrule(lr){11-12}
&	200 &	100 &	Lasso &	0.222 &	0.037 &	0.499 &	0.095 &	1 &	0 &	0.996 &	0.008	\\	
&	&	&	DS &	0.224 &	0.04 &	0.517 &	0.13 &	1 &	0 &	0.996 &	0.009	\\	\cmidrule(lr){2-4} \cmidrule(lr){5-6} \cmidrule(lr){7-8} \cmidrule(lr){9-10} \cmidrule(lr){11-12}
&	500 &	100 &	Lasso &	0.179 &	0.023 &	0.455 &	0.057 &	1 &	0 &	1 &	0.004	\\	
&	&	&	DS &	0.182 &	0.029 &	0.462 &	0.087 &	1 &	0 &	1 &	0.003	\\	\cmidrule(lr){2-4} \cmidrule(lr){5-6} \cmidrule(lr){7-8} \cmidrule(lr){9-10} \cmidrule(lr){11-12}
&	500 &	200 &	Lasso &	0.194 &	0.03 &	0.512 &	0.081 &	1 &	0 &	1 &	0.004	\\	
&	&	&	DS &	0.192 &	0.035 &	0.497 &	0.07 &	1 &	0 &	1 &	0.003	\\	\cmidrule(lr){1-4} \cmidrule(lr){5-6} \cmidrule(lr){7-8} \cmidrule(lr){9-10} \cmidrule(lr){11-12}
\ref{e:two} &	100 &	50 &	Lasso &	0.432 &	0.053 &	0.665 &	0.129 &	0.583 &	0.061 &	0.539 &	0.047	\\	
&	&	&	DS &	0.435 &	0.054 &	0.686 &	0.148 &	0.584 &	0.06 &	0.54 &	0.045	\\	\cmidrule(lr){2-4} \cmidrule(lr){5-6} \cmidrule(lr){7-8} \cmidrule(lr){9-10} \cmidrule(lr){11-12}
&	100 &	100 &	Lasso &	0.476 &	0.071 &	0.725 &	0.119 &	0.551 &	0.044 &	0.516 &	0.032	\\	
&	&	&	DS &	0.483 &	0.069 &	0.744 &	0.126 &	0.545 &	0.046 &	0.51 &	0.033	\\	\cmidrule(lr){2-4} \cmidrule(lr){5-6} \cmidrule(lr){7-8} \cmidrule(lr){9-10} \cmidrule(lr){11-12}
&	200 &	50 &	Lasso &	0.391 &	0.096 &	0.643 &	0.434 &	0.692 &	0.055 &	0.614 &	0.066	\\	
&	&	&	DS &	0.383 &	0.029 &	0.585 &	0.092 &	0.69 &	0.058 &	0.613 &	0.071	\\	\cmidrule(lr){2-4} \cmidrule(lr){5-6} \cmidrule(lr){7-8} \cmidrule(lr){9-10} \cmidrule(lr){11-12}
&	200 &	100 &	Lasso &	0.402 &	0.038 &	0.641 &	0.134 &	0.66 &	0.046 &	0.584 &	0.048	\\	
&	&	&	DS &	0.406 &	0.04 &	0.645 &	0.119 &	0.656 &	0.044 &	0.581 &	0.048	\\	\cmidrule(lr){2-4} \cmidrule(lr){5-6} \cmidrule(lr){7-8} \cmidrule(lr){9-10} \cmidrule(lr){11-12}
&	500 &	100 &	Lasso &	0.363 &	0.05 &	0.703 &	0.436 &	0.743 &	0.032 &	0.695 &	0.050	\\	
&	&	&	DS &	0.357 &	0.045 &	0.674 &	0.385 &	0.743 &	0.03 &	0.693 &	0.044	\\	\cmidrule(lr){2-4} \cmidrule(lr){5-6} \cmidrule(lr){7-8} \cmidrule(lr){9-10} \cmidrule(lr){11-12}
&	500 &	200 &	Lasso &	0.366 &	0.031 &	0.667 &	0.21 &	0.736 &	0.027 &	0.686 &	0.046	\\	
&	&	&	DS &	0.359 &	0.029 &	0.672 &	0.229 &	0.736 &	0.025 &	0.684 &	0.043	\\	\bottomrule
\end{tabular}}
\end{table}

\clearpage

\subsubsection{Forecasting}
\label{app:sim:forecast}

We assess the performance of the forecasting methodology of FNETS which estimates the best linear predictors $\bm\chi_{n + 1 \vert n}$ by one of the two estimators, $\wh{\bm\chi}_{n + 1 \vert n}^{\static}$ (Section~\ref{sec:2005}) and $\wh{\bm\chi}_{n + 1 \vert n}^{\va}$ (Appendix~\ref{sec:2017}), 
$\bm\xi_{n + 1 \vert n}$ by $\wh{\bm\xi}^{\las}_{n + 1 \vert n}$ and $\wh{\bm\xi}^{\ds}_{n + 1 \vert n}$ (denoting the estimators of $\bm\xi_{n + 1 \vert n}$ with Lasso and DS estimators of $\bm\beta$, respectively) and finally, $\mbf X_{n + 1 \vert n} = \bm\chi_{n + 1 \vert n} + \bm\xi_{n + 1 \vert n}$ by their combinations. 
The estimator $\wh{\bm\xi}_{n + 1 \vert n}$ depends on the choice of the in-sample estimator of $\bm\chi_t$ (which automatically yields the in-sample estimator of $\bm\xi_t$)
but we suppress this dependence in the notations.
\cite{fan2021bridging} propose a forecasting methodology based on VAR modelling of the estimated factors, the results from which we report alongside those from FNETS under the heading FARM. 

In Tables~\ref{table:forecast:one:avg} (under~\ref{e:one}) and~\ref{table:forecast:two:avg} (under~\ref{e:two}--\ref{e:three}), we report the estimation errors of a given forecast, say $\wh{\bm\gamma}_{n + 1 \vert n}$,
in estimating $\bm\gamma_{n + 1 \vert n}$ measured as
\begin{align}
\label{eq:forecast:avg}
\frac{\vert \wh{\bm\gamma}_{n + 1 \vert n} - \bm\gamma_{n + 1 \vert n} \vert_2^2}
{\vert \bm\gamma_{n + 1 \vert n} \vert_2^2}
\end{align}
and additionally, report the in-sample estimation errors of 
$\wh{\bm\chi}_t = \wh{\bm\chi}_t^{\static}$ and $\wh{\bm\chi}_t = \wh{\bm\chi}_t^{\va}$
measured as 
$\sum_t \vert \wh{\bm\chi}_t - \bm\chi_t \vert_2^2 / (\sum_t \vert \bm\chi_t \vert_2^2)$.
Tables~\ref{table:forecast:one:max} (under~\ref{e:one}) and Tables~\ref{table:forecast:two:max} (under~\ref{e:two}--\ref{e:three}) summarise the forecasting errors measured by
\begin{align}
\label{eq:forecast:max}
\frac{\vert \wh{\bm\gamma}_{n + 1 \vert n} - \bm\gamma_{n + 1 \vert n} \vert_\infty}
{\vert \bm\gamma_{n + 1 \vert n} \vert_\infty},
\end{align}
which relates to the norm chosen for theoretical analysis
in Propositions~\ref{thm:2005} and~\ref{thm:common:var} and Proposition~\ref{prop:idio:pred}.
We also report the forecasting errors measured as
\begin{align}
\label{eq:forecast:xavg}
& \frac{\vert \wh{\bm\gamma}_{n + 1 \vert n} - \bm\gamma_{n + 1} \vert_2^2}
{\vert \bm\gamma_{n + 1} \vert_2}, \quad \text{and}
\\
\label{eq:forecast:xmax}
& \frac{\vert \wh{\bm\gamma}_{n + 1 \vert n} - \bm\gamma_{n + 1} \vert_\infty^2}
{\vert \bm\gamma_{n + 1} \vert_\infty},
\end{align}
see Tables~\ref{table:forecast:one:xavg}, \ref{table:forecast:one:xmax} (under~\ref{e:one}), \ref{table:forecast:two:xavg} and~\ref{table:forecast:two:xmax} (under~\ref{e:two}--\ref{e:three}).
Additionally, Table~\ref{table:forecast:benchmark} contains results in the above error measures obtained from the benchmark case when $\mbf X_t = \bm\xi_t$ under~\ref{m:oracle} when $\bm\xi_t$ is generated according to~\ref{e:one}.

For FNETS, the forecasting performance improves as $n$ increases regardless of the error measures.
The estimation error for $\bm\chi_{n + 1 \vert n}$ decreases with $p$
while it increases for $\bm\xi_{n + 1 \vert n}$,
which is due to that the factor-adjustment step enjoys the blessing of dimensionality
while VAR estimation tends to suffer from the increase of the dimensionality.
This observation is consistent with Propositions~\ref{thm:2005} and~\ref{prop:idio:pred}, and the former tends to offset latter in the estimation error of $\mbf X_{n + 1 \vert n}$.
The forecasting method based on the unrestricted GDFM (Appendix~\ref{sec:2017}) exhibits some numerical instabilities due to the instability of the singular VAR equation system adopted for this purpose (see Remark~\ref{rem:2017}~\ref{rem:2017:hormann})
which, in turn, may be attributed to the possible
over-specification of the VAR order \citep{hormann2021prediction}. 
As such, the performance of $\wh{\bm\chi}_{n + 1 \vert n}^{\static}$ is generally superior even when $\bm\chi_t$ does not admit a static representation (under~\ref{m:ar}),
and the gap between the two estimators gets wider when a static representation exists (under~\ref{m:ma}) as $n$ and $p$ increase.
In general, we do not observe any systematic effect of the innovation distribution on the forecasting performance.

FARM performs reasonably well when $n \ge 200$. In particular, when $\bm\chi_t$ is generated under~\ref{m:ma}, the static factor model offers a valid alternative to GDFM such that FARM marginally outperforms FNETS in in-sample estimation and estimating $\bm\chi_{n + 1 \vert n}$.
However, the performance in estimating the VAR parameters carries forward to producing the forecast of $\bm\xi_{n + 1}$, which sometimes result in slightly worse forecasts of $\mbf X_{n + 1}$.

Under~\ref{m:ma}, occasionally the best linear predictor $\bm\chi_{n + 1 \vert n}$
has all its elements close to zero and 
the small value of $\vert \bm\chi_{n + 1 \vert n} \vert_2$ inflates 
the relative estimation error measured as in~\eqref{eq:forecast:avg};
this phenomenon is not observed 
from the forecasting errors measured with~\eqref{eq:forecast:xavg}.

\begin{table}[htb!]
\caption{\small Errors in forecasting $\mbf X_{n + 1}$ by the FNETS measured by~\eqref{eq:forecast:avg}--\eqref{eq:forecast:xmax}
averaged over $100$ realisations (also reported are the standard errors)
under the model~\ref{e:one} for the generation of $\bm\xi_t$
and~\ref{m:oracle} for $\bm\chi_t$ with varying $n$ and $p$, which serve as a benchmark.}
\label{table:forecast:benchmark}
\centering
{\footnotesize 
\begin{tabular}{cccr cc cccr cc}
\toprule																
&	$n$ &	$p$ &	Method &	Mean &	SD &		&	$n$ &	$p$ &	Method &	Mean &	SD 	\\	\cmidrule(lr){1-4} \cmidrule(lr){5-6} \cmidrule(lr){7-10} \cmidrule(lr){11-12}		
\eqref{eq:forecast:avg} &	100 &	50 &	Lasso &	0.437 &	0.253 &		\eqref{eq:forecast:xavg} &	100 &	50 &	Lasso &	0.898 &	0.090	\\			
&	&	&	DS &	0.378 &	0.122 &		&	&	&	DS &	0.896 &	0.077	\\	\cmidrule(lr){2-4} \cmidrule(lr){5-6} \cmidrule(lr){8-10} \cmidrule(lr){11-12}		
&	100 &	100 &	Lasso &	0.748 &	0.208 &		&	100 &	100 &	Lasso &	0.958 &	0.044	\\			
&	&	&	DS &	0.463 &	0.133 &		&	&	&	DS &	0.913 &	0.045	\\	\cmidrule(lr){2-4} \cmidrule(lr){5-6} \cmidrule(lr){8-10} \cmidrule(lr){11-12}		
&	200 &	50 &	Lasso &	0.176 &	0.056 &		&	200 &	50 &	Lasso &	0.853 &	0.078	\\			
&	&	&	DS &	0.277 &	0.092 &		&	&	&	DS &	0.87 &	0.069	\\	\cmidrule(lr){2-4} \cmidrule(lr){5-6} \cmidrule(lr){8-10} \cmidrule(lr){11-12}		
&	200 &	100 &	Lasso &	0.207 &	0.036 &		&	200 &	100 &	Lasso &	0.873 &	0.054	\\
\bottomrule		
\end{tabular}}
\end{table}

\clearpage

\begin{table}[htb!]
\caption{\footnotesize Forecasting errors of FNETS and FARM measured by~\eqref{eq:forecast:avg}
averaged over $100$ realisations (also reported are the standard errors)
under the model~\ref{e:one} for the generation of $\bm\xi_t$
and~\ref{m:ar}--\ref{m:ma} for $\bm\chi_t$ with varying $n$ and $p$.
We also report the errors of restricted and unrestricted in-sample estimators of
$\bm\chi_t, \, 1 \le t \le n$.}
\label{table:forecast:one:avg}
\centering
{\tiny 
\begin{tabular}{ccc rr cc cc cc cc}
\toprule														
&	&	&	\multicolumn{2}{c}{Method} &		\multicolumn{2}{c}{In-sample} &		\multicolumn{2}{c}{$\bm\chi_{n + 1 \vert n}$} &		\multicolumn{2}{c}{$\bm\xi_{n + 1 \vert n}$} &		\multicolumn{2}{c}{$\mbf X_{n + 1 \vert n}$} 		\\	
&	$n$ &	$p$ &	$\bm\chi_t$ &	$\bm\xi_t$ &	Mean &	SD &	Mean &	SD &	Mean &	SD &	Mean &	SD	\\	\cmidrule(lr){1-5} \cmidrule(lr){6-7} \cmidrule(lr){8-9} \cmidrule(lr){10-11} \cmidrule(lr){12-13}
\ref{m:ar} &	100 &	50 &	Restricted &	Lasso &	0.355 &	0.095 &	0.517 &	0.358 &	0.732 &	0.181 &	0.491 &	0.206	\\	
&	&	&	Unrestricted &	&	0.37 &	0.104 &	0.543 &	0.309 &	0.814 &	0.326 &	0.491 &	0.198	\\	\cmidrule(lr){4-5} \cmidrule(lr){6-7} \cmidrule(lr){8-9} \cmidrule(lr){10-11} \cmidrule(lr){12-13}
&	&	&	Restricted &	DS &	- &	- &	- &	- &	0.746 &	0.177 &	0.493 &	0.209	\\	
&	&	&	Unrestricted &	&	- &	- &	- &	- &	0.845 &	0.33 &	0.493 &	0.189	\\	\cmidrule(lr){4-5} \cmidrule(lr){6-7} \cmidrule(lr){8-9} \cmidrule(lr){10-11} \cmidrule(lr){12-13}
&	&	&	&	FARM &	0.381 &	0.143 &	2.88 &	7.88 &	1.07 &	0.482 &	1.76 &	2.460	\\	\cmidrule(lr){2-5} \cmidrule(lr){6-7} \cmidrule(lr){8-9} \cmidrule(lr){10-11} \cmidrule(lr){12-13}
&	100 &	100 &	Restricted &	Lasso &	0.279 &	0.091 &	0.455 &	0.28 &	0.797 &	0.137 &	0.462 &	0.186	\\
\bottomrule
\end{tabular}}
\end{table}

\begin{table}[htb!]
\caption{\footnotesize Forecasting errors of FNETS and FARM measured by~\eqref{eq:forecast:max}
averaged over $100$ realisations (also reported are the standard errors)
under the model~\ref{e:one} for the generation of $\bm\xi_t$
and~\ref{m:ar}--\ref{m:ma} for $\bm\chi_t$ with varying $n$ and $p$.}
\label{table:forecast:one:max}
\centering
{\tiny 
\begin{tabular}{ccc rr cc cc cc cc}
\toprule												
&	&	&	\multicolumn{2}{c}{Method} &		\multicolumn{2}{c}{$\bm\chi_{n + 1 \vert n}$} &		\multicolumn{2}{c}{$\bm\xi_{n + 1 \vert n}$} &		\multicolumn{2}{c}{$\mbf X_{n + 1 \vert n}$} 		\\	
&	$n$ &	$p$ &	$\bm\chi_t$ &	$\bm\xi_t$ &	Mean &	SD &	Mean &	SD &	Mean &	SD	\\	\cmidrule(lr){1-5} \cmidrule(lr){6-7} \cmidrule(lr){8-9} \cmidrule(lr){10-11}
\ref{m:ar} &	100 &	50 &	Restricted &	Lasso &	0.682 &	0.277 &	0.87 &	0.158 &	0.654 &	0.208	\\	
&	&	&	Unrestricted &	&	0.731 &	0.28 &	0.929 &	0.314 &	0.646 &	0.208	\\	\cmidrule(lr){4-5} \cmidrule(lr){6-7} \cmidrule(lr){8-9} \cmidrule(lr){10-11} \cmidrule(lr){12-13}
&	&	&	Restricted &	DS &	- &	- &	0.883 &	0.201 &	0.653 &	0.210	\\	
&	&	&	Unrestricted &	&	- &	- &	0.971 &	0.354 &	0.652 &	0.199	\\	\cmidrule(lr){4-5} \cmidrule(lr){6-7} \cmidrule(lr){8-9} \cmidrule(lr){10-11} \cmidrule(lr){12-13}
&	&	&	&	FARM &	1.05 &	1.03 &	1.08 &	0.388 &	0.984 &	0.641	\\	\cmidrule(lr){2-5} \cmidrule(lr){6-7} \cmidrule(lr){8-9} \cmidrule(lr){10-11} \cmidrule(lr){12-13}
&	100 &	100 &	Restricted &	Lasso &	0.636 &	0.288 &	0.925 &	0.159 &	0.594 &	0.223	\\
\bottomrule
\end{tabular}}
\end{table}

\begin{table}[htb!]
\caption{\footnotesize Forecasting errors of FNETS and FARM measured by~\eqref{eq:forecast:xavg}
averaged over $100$ realisations (also reported are the standard errors)
under the model~\ref{e:one} for the generation of $\bm\xi_t$
and~\ref{m:ar}--\ref{m:ma} for $\bm\chi_t$ with varying $n$ and $p$.}
\label{table:forecast:one:xavg}
\centering
{\tiny 
\begin{tabular}{ccc rr cc cc cc cc}
\toprule												
&	&	&	\multicolumn{2}{c}{Method} &		\multicolumn{2}{c}{$\bm\chi_{n + 1}$} &		\multicolumn{2}{c}{$\bm\xi_{n + 1}$} &		\multicolumn{2}{c}{$\mbf X_{n + 1}$} 		\\	
&	$n$ &	$p$ &	$\bm\chi_t$ &	$\bm\xi_t$ &	Mean &	SD &	Mean &	SD &	Mean &	SD	\\	\cmidrule(lr){1-5} \cmidrule(lr){6-7} \cmidrule(lr){8-9} \cmidrule(lr){10-11}
\ref{m:ar} &	100 &	50 &	Restricted &	Lasso &	0.677 &	0.316 &	0.957 &	0.063 &	0.79 &	0.144	\\	
&	&	&	Unrestricted &	&	0.676 &	0.249 &	0.972 &	0.081 &	0.787 &	0.129	\\	\cmidrule(lr){4-5} \cmidrule(lr){6-7} \cmidrule(lr){8-9} \cmidrule(lr){10-11} \cmidrule(lr){12-13}
&	&	&	Restricted &	DS &	- &	- &	0.962 &	0.062 &	0.794 &	0.145	\\	
&	&	&	Unrestricted &	&	- &	- &	0.976 &	0.079 &	0.789 &	0.125	\\	\cmidrule(lr){4-5} \cmidrule(lr){6-7} \cmidrule(lr){8-9} \cmidrule(lr){10-11} \cmidrule(lr){12-13}
&	&	&	&	FARM &	1.65 &	1.96 &	1.02 &	0.131 &	1.23 &	0.798	\\	\cmidrule(lr){2-5} \cmidrule(lr){6-7} \cmidrule(lr){8-9} \cmidrule(lr){10-11} \cmidrule(lr){12-13}
&	100 &	100 &	Restricted &	Lasso &	0.631 &	0.296 &	0.97 &	0.042 &	0.781 &	0.146	\\
\bottomrule
\end{tabular}}
\end{table}

\begin{table}[htb!]
\caption{\footnotesize Forecasting errors of FNETS and FARM measured by~\eqref{eq:forecast:xmax}
averaged over $100$ realisations (also reported are the standard errors)
under the model~\ref{e:one} for the generation of $\bm\xi_t$
and~\ref{m:ar}--\ref{m:ma} for $\bm\chi_t$ with varying $n$ and $p$.}
\label{table:forecast:one:xmax}
\centering
{\tiny 
\begin{tabular}{ccc rr cc cc cc cc}
\toprule												
&	&	&	\multicolumn{2}{c}{Method} &		\multicolumn{2}{c}{$\bm\chi_{n + 1}$} &		\multicolumn{2}{c}{$\bm\xi_{n + 1}$} &		\multicolumn{2}{c}{$\mbf X_{n + 1}$} 		\\	
&	$n$ &	$p$ &	$\bm\chi_t$ &	$\bm\xi_t$ &	Mean &	SD &	Mean &	SD &	Mean &	SD	\\	\cmidrule(lr){1-5} \cmidrule(lr){6-7} \cmidrule(lr){8-9} \cmidrule(lr){10-11}
\ref{m:ar} &	100 &	50 &	Restricted &	Lasso &	0.738 &	0.257 &	0.978 &	0.08 &	0.838 &	0.181	\\	
&	&	&	Unrestricted &	&	0.77 &	0.219 &	1.01 &	0.195 &	0.843 &	0.168	\\	\cmidrule(lr){4-5} \cmidrule(lr){6-7} \cmidrule(lr){8-9} \cmidrule(lr){10-11} \cmidrule(lr){12-13}
&	&	&	Restricted &	DS &	- &	- &	0.982 &	0.092 &	0.835 &	0.179	\\	
&	&	&	Unrestricted &	&	- &	- &	1.01 &	0.204 &	0.84 &	0.170	\\	\cmidrule(lr){4-5} \cmidrule(lr){6-7} \cmidrule(lr){8-9} \cmidrule(lr){10-11} \cmidrule(lr){12-13}
&	&	&	&	FARM &	0.983 &	0.586 &	1.02 &	0.147 &	1.03 &	0.441	\\	\cmidrule(lr){2-5} \cmidrule(lr){6-7} \cmidrule(lr){8-9} \cmidrule(lr){10-11} \cmidrule(lr){12-13}
&	100 &	100 &	Restricted &	Lasso &	0.685 &	0.243 &	0.99 &	0.062 &	0.809 &	0.168	\\
\bottomrule
\end{tabular}}
\end{table}

\clearpage

\begin{table}[htb!]
\caption{\footnotesize Forecasting errors of FNETS measured by~\eqref{eq:forecast:avg}
averaged over $100$ realisations (also reported are the standard errors)
under the models~\ref{e:two}--\ref{e:three} for the generation of $\bm\xi_t$
and~\ref{m:ar} for $\bm\chi_t$ with varying $n$ and $p$.
We also report the errors of restricted and unrestricted in-sample estimators of
$\bm\chi_t, \, 1 \le t \le n$.}
\label{table:forecast:two:avg}
\centering
{\tiny 
\begin{tabular}{ccc rr cc cc cc cc}
\toprule														
&	&	&	\multicolumn{2}{c}{Method} &		\multicolumn{2}{c}{In-sample} &		\multicolumn{2}{c}{$\bm\chi_{n + 1 \vert n}$} &		\multicolumn{2}{c}{$\bm\xi_{n + 1 \vert n}$} &		\multicolumn{2}{c}{$\mbf X_{n + 1 \vert n}$} 		\\	
&	$n$ &	$p$ &	$\bm\chi_t$ &	$\bm\xi_t$ &	Mean &	SD &	Mean &	SD &	Mean &	SD &	Mean &	SD	\\	\cmidrule(lr){1-5} \cmidrule(lr){6-7} \cmidrule(lr){8-9} \cmidrule(lr){10-11} \cmidrule(lr){12-13}
\ref{e:two} &	100 &	50 &	Restricted &	Lasso &	0.344 &	0.103 &	0.594 &	0.52 &	0.775 &	0.3 &	0.509 &	0.26	\\	
&	&	&	Unrestricted &	&	0.342 &	0.099 &	0.582 &	0.507 &	0.912 &	0.75 &	0.493 &	0.231	\\	\cmidrule(lr){4-5} \cmidrule(lr){6-7} \cmidrule(lr){8-9} \cmidrule(lr){10-11} \cmidrule(lr){12-13}
&	&	&	Restricted &	DS &	- &	- &	- &	- &	0.764 &	0.231 &	0.518 &	0.266	\\	
&	&	&	Unrestricted &	&	- &	- &	- &	- &	0.904 &	0.682 &	0.502 &	0.236	\\	\cmidrule(lr){2-5} \cmidrule(lr){6-7} \cmidrule(lr){8-9} \cmidrule(lr){10-11} \cmidrule(lr){12-13}
&	100 &	100 &	Restricted &	Lasso &	0.252 &	0.1 &	0.426 &	0.256 &	0.821 &	0.175 &	0.442 &	0.189	\\	
&	&	&	Unrestricted &	&	0.302 &	0.118 &	0.507 &	0.245 &	0.994 &	0.434 &	0.452 &	0.173	\\
\bottomrule
\end{tabular}}
\end{table}

\begin{table}[htb!]
\caption{\footnotesize Forecasting errors of FNETS measured by~\eqref{eq:forecast:max}
averaged over $100$ realisations (also reported are the standard errors)
under the models~\ref{e:two}--\ref{e:three} for the generation of $\bm\xi_t$
and~\ref{m:ar} for $\bm\chi_t$ with varying $n$ and $p$.}
\label{table:forecast:two:max}
\centering
{\tiny 
\begin{tabular}{ccc rr cc cc cc cc}
\toprule												
&	&	&	\multicolumn{2}{c}{Method} &		\multicolumn{2}{c}{$\bm\chi_{n + 1 \vert n}$} &		\multicolumn{2}{c}{$\bm\xi_{n + 1 \vert n}$} &		\multicolumn{2}{c}{$\mbf X_{n + 1 \vert n}$} 		\\	
&	$n$ &	$p$ &	$\bm\chi_t$ &	$\bm\xi_t$ &	Mean &	SD &	Mean &	SD &	Mean &	SD	\\	\cmidrule(lr){1-5} \cmidrule(lr){6-7} \cmidrule(lr){8-9} \cmidrule(lr){10-11} \cmidrule(lr){12-13}
\ref{e:two} &	100 &	50 &	Restricted &	Lasso &	0.725 &	0.322 &	0.874 &	0.215 &	0.661 &	0.270	\\	
&	&	&	Unrestricted &	&	0.751 &	0.318 &	0.991 &	0.504 &	0.664 &	0.269	\\	\cmidrule(lr){4-5} \cmidrule(lr){6-7} \cmidrule(lr){8-9} \cmidrule(lr){10-11} \cmidrule(lr){12-13}
&	&	&	Restricted &	DS &	- &	- &	0.863 &	0.196 &	0.663 &	0.261	\\	
&	&	&	Unrestricted &	&	- &	- &	0.978 &	0.499 &	0.668 &	0.261	\\	\cmidrule(lr){2-5} \cmidrule(lr){6-7} \cmidrule(lr){8-9} \cmidrule(lr){10-11} \cmidrule(lr){12-13}
&	100 &	100 &	Restricted &	Lasso &	0.617 &	0.284 &	0.946 &	0.311 &	0.617 &	0.241	\\	
&	&	&	Unrestricted &	&	0.768 &	0.35 &	1.08 &	0.464 &	0.623 &	0.214	\\
\bottomrule
\end{tabular}}
\end{table}

\begin{table}[htb!]
\caption{\footnotesize Forecasting errors of FNETS measured by~\eqref{eq:forecast:xavg}
averaged over $100$ realisations (also reported are the standard errors)
under the models~\ref{e:two}--\ref{e:three} for the generation of $\bm\xi_t$
and~\ref{m:ar} for $\bm\chi_t$ with varying $n$ and $p$.}
\label{table:forecast:two:xavg}
\centering
{\tiny 
\begin{tabular}{ccc rr cc cc cc cc}
\toprule												
&	&	&	\multicolumn{2}{c}{Method} &		\multicolumn{2}{c}{$\bm\chi_{n + 1}$} &		\multicolumn{2}{c}{$\bm\xi_{n + 1}$} &		\multicolumn{2}{c}{$\mbf X_{n + 1}$} 		\\	
&	$n$ &	$p$ &	$\bm\chi_t$ &	$\bm\xi_t$ &	Mean &	SD &	Mean &	SD &	Mean &	SD	\\	\cmidrule(lr){1-5} \cmidrule(lr){6-7} \cmidrule(lr){8-9} \cmidrule(lr){10-11} \cmidrule(lr){12-13}
\ref{e:two} &	100 &	50 &	Restricted &	Lasso &	0.731 &	0.349 &	0.962 &	0.073 &	0.809 &	0.172	\\	
&	&	&	Unrestricted &	&	0.7 &	0.301 &	0.99 &	0.164 &	0.802 &	0.165	\\	\cmidrule(lr){4-5} \cmidrule(lr){6-7} \cmidrule(lr){8-9} \cmidrule(lr){10-11} \cmidrule(lr){12-13}
&	&	&	Restricted &	DS &	- &	- &	0.958 &	0.073 &	0.811 &	0.177	\\	
&	&	&	Unrestricted &	&	- &	- &	0.987 &	0.151 &	0.803 &	0.167	\\	\cmidrule(lr){2-5} \cmidrule(lr){6-7} \cmidrule(lr){8-9} \cmidrule(lr){10-11} \cmidrule(lr){12-13}
&	100 &	100 &	Restricted &	Lasso &	0.678 &	0.318 &	0.97 &	0.044 &	0.791 &	0.197	\\	
&	&	&	Unrestricted &	&	0.729 &	0.261 &	0.994 &	0.085 &	0.796 &	0.166	\\
\bottomrule
\end{tabular}}
\end{table}

\begin{table}[htb!]
\caption{\footnotesize Forecasting errors of FNETS measured by~\eqref{eq:forecast:xmax}
averaged over $100$ realisations (also reported are the standard errors)
under the models~\ref{e:two}--\ref{e:three} for the generation of $\bm\xi_t$
and~\ref{m:ar} for $\bm\chi_t$ with varying $n$ and $p$.}
\label{table:forecast:two:xmax}
\centering
{\tiny 
\begin{tabular}{ccc rr cc cc cc cc}
\toprule												
&	&	&	\multicolumn{2}{c}{Method} &		\multicolumn{2}{c}{$\bm\chi_{n + 1}$} &		\multicolumn{2}{c}{$\bm\xi_{n + 1}$} &		\multicolumn{2}{c}{$\mbf X_{n + 1}$} 		\\	
&	$n$ &	$p$ &	$\bm\chi_t$ &	$\bm\xi_t$ &	Mean &	SD &	Mean &	SD &	Mean &	SD	\\	\cmidrule(lr){1-5} \cmidrule(lr){6-7} \cmidrule(lr){8-9} \cmidrule(lr){10-11} \cmidrule(lr){12-13}
\ref{e:two} &	100 &	50 &	Restricted &	Lasso &	0.789 &	0.295 &	0.99 &	0.078 &	0.859 &	0.185	\\	
&	&	&	Unrestricted &	&	0.802 &	0.303 &	1.04 &	0.23 &	0.858 &	0.185	\\	\cmidrule(lr){4-5} \cmidrule(lr){6-7} \cmidrule(lr){8-9} \cmidrule(lr){10-11} \cmidrule(lr){12-13}
&	&	&	Restricted &	DS &	- &	- &	0.98 &	0.081 &	0.856 &	0.181	\\	
&	&	&	Unrestricted &	&	- &	- &	1.03 &	0.22 &	0.851 &	0.182	\\	\cmidrule(lr){2-5} \cmidrule(lr){6-7} \cmidrule(lr){8-9} \cmidrule(lr){10-11} \cmidrule(lr){12-13}
&	100 &	100 &	Restricted &	Lasso &	0.71 &	0.265 &	0.973 &	0.066 &	0.805 &	0.189	\\	
&	&	&	Unrestricted &	&	0.789 &	0.26 &	0.996 &	0.114 &	0.818 &	0.174	\\
\bottomrule
\end{tabular}}
\end{table}

\clearpage

\section{Proofs}
\label{sec:proof}

In what follows, for a random variable $X$ and $\nu \ge 1$, we denote $\Vert X \Vert_\nu = (\E|X|^\nu)^{1/\nu}$.

We work under the following assumption which extends Assumption~\ref{assum:innov:new} by considering the special case of Gaussianity separately.
\begin{assum} 
\label{assum:innov}
\it{
\begin{enumerate}[wide, label = (\roman*)]
\item \label{cond:iid} $\{\mbf u_t\}_{t \in \Z}$ is a sequence of zero-mean, $q$-dimensional martingale difference vectors with $\Cov(\mbf u_t) = \mbf I_q$, and $u_{it}$ and $u_{jt}$ are independent for all $1 \le i, j \le q$ with $i \ne j$ and all $t\in \Z$. 
Also, $\{\bm\vep_t\}_{t \in \Z}$ is a sequence of zero-mean, $p$-dimensional martingale difference vectors with $\Cov(\bm\vep_t) = \mbf I_p$, and $\vep_{it}$ and $\vep_{jt}$ are independent for all $1 \le i, j \le p$ with $i \ne j$ and all $t\in \Z$.

\item \label{cond:uncor} 
The common and idiosyncratic shocks are uncorrelated, i.e.\
$\E(u_{jt} \vep_{it'}) = 0$ for any $1 \le i \le p$, $1 \le j \le q$ and $t, t' \in \Z$.

\item \label{cond:dist} Either one of the following conditions holds.
\begin{enumerate}[label = (\alph*) ]
\item \label{cond:moment} There exists $\nu > 4$ such that
$\max\l\{ \max_{1 \le j \le q} \Vert u_{jt} \Vert_\nu^\nu, 
\max_{1 \le i \le p} \Vert \vep_{it} \Vert_\nu^\nu) \r\} \le \mu_\nu$
for some constant $\mu_\nu>0$.

\item \label{cond:gauss} $\mbf u_t \sim_{\iid} \mc N_q(\mbf 0, \mbf I)$ and 
$\bm\vep_t \sim_{\iid} \mc N_p(\mbf 0, \mbf I)$.
\end{enumerate}
\end{enumerate}
}
\end{assum}
Accordingly, we define $\psi_n$ and $\vartheta_{n, p}$ as in~\eqref{eq:rates:one}
under Assumption~\ref{assum:innov}~\ref{cond:dist}~\ref{cond:moment}, while 
\begin{align}
\label{eq:rates:two}
\psi_n = \sqrt{\frac{m\log(m)}{n}} \quad \text{and} \quad
\vartheta_{n, p} = \sqrt{\frac{m\log(mp)}{n}}
\end{align}
under Assumption~\ref{assum:innov}~\ref{cond:dist}~\ref{cond:gauss}.

\subsection{Proof of Proposition~\ref{prop:idio:eval}}

Let $\mc D(z) = \sum_{\ell = 0}^\infty \mbf D_\ell z^\ell$.
Under Assumption~\ref{assum:idio},  we can find a constant $B_\xi>0$ 
which depends only on $M_\vep$, $\Xi$ and $\varsigma$
such that, uniformly over $\omega \in [-\pi, \pi]$,
\begin{align*}
\mu_{\xi, 1}(\omega) =& \Vert \bm\Sigma_\xi(\omega) \Vert
= \frac{1}{2\pi} \Vert \mc D(e^{-\iota \omega}) \bm\Gamma \mc D^*(e^{-\iota \omega}) \Vert
\le \frac{M_\vep}{2\pi} 
\Vert \mc D(e^{-\iota \omega}) \Vert_1 \; \Vert \mc D(e^{-\iota \omega}) \Vert_\infty
\\
\le& \frac{M_\vep}{2\pi} 
\l(\max_{1 \le i \le p} \sum_{j = 1}^p \sum_{\ell = 0}^\infty \vert D_{\ell, ij} \vert \r)
\l(\max_{1 \le j \le p} \sum_{i = 1}^p \sum_{\ell = 0}^\infty \vert D_{\ell, ij} \vert \r)
\\
\le& \frac{M_\vep}{2\pi} 
\l(\max_i \sum_{j = 1}^p \sum_{\ell = 0}^\infty \frac{C_{ij}}{(1 + \ell)^{\varsigma}} \r)
\l(\max_j \sum_{i = 1}^p \sum_{\ell = 0}^\infty \frac{C_{ij}}{(1 + \ell)^{\varsigma}} \r)
\le \frac{\Xi^2 M_\vep}{2\pi} \l(\sum_{\ell = 0}^\infty \frac{1}{(1 + \ell)^\varsigma}\r)^2 \le B_\xi.
\end{align*}

\subsection{Results in Section~\ref{sec:fnets:network:theory}}
\label{sec:dpcs:proof}

\subsubsection{Preliminary lemmas}

In the following lemmas, we operate 
under Assumptions~\ref{assum:common}, \ref{assum:factor}, \ref{assum:idio} and~\ref{assum:innov}.

\cite{zhang2021} extend the functional dependence measure 
introduced in \cite{wu2005} for high-dimensional time series. 
Denote by $\mc F_t = \sigma\{(\mbf u_v, \bm\vep_v), \, v \le t\}$
and $\mc G(\cdot) = (g_1(\cdot), \ldots, g_p(\cdot))^\top$ 
a $\R^p$-valued measurable function such that 
$\mbf X_t = \mc G(\mc F_t)$ and $X_{it} = g_i(\mc F_t)$.
Also let $\mc F_{t, \{0\}} = \sigma\{\ldots, (\mbf u_{-1}, \bm\vep_{-1}), (\mbf u^\prime_0, \bm\vep^\prime_0), 
(\mbf u_1, \bm\vep_1)^\top, \ldots, (\mbf u_t, \bm\vep_t)\}$ 
denote a coupled version of $\mc F_t$ with an independent copy
$(\mbf u^\prime_0, \bm\vep^\prime_0)$ replacing $(\mbf u_0, \bm\vep_0)$.
Then, the element-wise functional dependence measure is defined as
\begin{align*}
\delta_{t, \nu, i} = \l\Vert g_i(\mc F_t) - g_i(\mc F_{t, \{0\}}) \r\Vert_\nu,
\end{align*}
the uniform functional dependence measure as
\begin{align*}
\delta_{t, \nu} = \l\Vert \vert \mc G(\mc F_t) - \mc G(\mc F_{t, \{0\}}) \vert_\infty \r\Vert_\nu,
\end{align*}
the dependence adjusted norms as
\begin{align*}
\Vert \mbf X_{i \cdot} \Vert_{\nu, \alpha} 
= \sup_{\ell \ge 0} \, (\ell + 1)^\alpha \sum_{t = \ell}^\infty \delta_{t, \nu, i}
\quad \text{and} \quad
\Vert \vert \mbf X_{\cdot} \vert_\infty \Vert_{\nu, \alpha} 
= \sup_{\ell \ge 0} \, (\ell + 1)^\alpha \sum_{t = \ell}^\infty \delta_{t, \nu},
\end{align*}
and the overall and the uniform dependence adjusted norms as
\begin{align*}
\Psi_{\nu, \alpha} = \l( \sum_{i = 1}^p \Vert \mbf X_{i \cdot} \Vert_{\nu, \alpha}^{\nu/2} \r)^{2/\nu}
\quad \text{and} \quad
\Phi_{\nu, \alpha} = \max_{1 \le i \le p} \Vert \mbf X_{i \cdot} \Vert_{\nu, \alpha}.
\end{align*}

\begin{lem}
\label{lem:func:dep}
Let $\alpha \le \varsigma - 1$. 
\begin{enumerate}[wide, label = (\roman*)]
\item 
Under Assumption~\ref{assum:innov}~\ref{cond:dist}~\ref{cond:moment}, we have
\begin{align*}
\Psi_{\nu, \alpha} \le C_{\nu, \Xi, \varsigma} M_\vep^{1/2} p^{2/\nu} \mu_\nu^{1/\nu}
\quad \text{and} \quad
\Vert \vert \mbf X_{\cdot} \vert_\infty \Vert_{\nu, \alpha} 
\le C_{\nu, \Xi, \varsigma} M_\vep^{1/2} \log^{1/2}(p) p^{1/\nu} \mu_\nu^{1/\nu}
\end{align*}
for some constant $C_{\nu, \Xi, \varsigma} > 0$ depending only on $\nu$ and $\varsigma$
(varying on each occasion).
\item Under Assumption~\ref{assum:innov}~\ref{cond:dist}~\ref{cond:moment}--\ref{cond:gauss}, we have 
$\Phi_{\nu, \alpha} \le C_{\nu, \Xi, \varsigma} M_\vep^{1/2} \mu_\nu^{1/\nu}$
for any $\nu$ for which $\Vert u_{jt} \Vert_\nu$ and $\Vert \vep_{it} \Vert_\nu$ exist.
\end{enumerate}
\end{lem}

\begin{proof}
By Minkowski inequality,
\begin{align*}
\delta_{t, \nu, i} &= \Vert \mbf B_{t, i\cdot} \mbf u_0 + \mbf D_{t, i\cdot} \bm\Gamma^{1/2} \bm\vep_0 \Vert_\nu
\le \Vert \mbf B_{t, i\cdot} \mbf u_0 \Vert_\nu + \Vert \mbf D_{t, i\cdot} \bm\Gamma^{1/2} \bm\vep_0 \Vert_\nu,
\quad \text{and}
\\
\delta_{t, \nu} &= \Vert \vert \mbf B_t \mbf u_0 + \mbf D_t \bm\Gamma^{1/2} \bm\vep_0 \vert_\infty \Vert_\nu
\le \Vert \vert \mbf B_t \mbf u_0 \vert_\infty \Vert_\nu + 
\Vert \vert \mbf D_t \bm\Gamma^{1/2} \bm\vep_0 \vert_\infty \Vert_\nu.
\end{align*}
Due to independence of $u_{jt}, \, 1 \le j \le q$, assumed in Assumption~\ref{assum:innov}~\ref{cond:iid}, Assumption~\ref{assum:common} and Lemma~D.3 of \cite{zhang2021}, there exists $C_\nu > 0$ depends only on $\nu$
(varying from one instance to another) such that
\begin{align*}
& \Vert \mbf B_{t, i\cdot} \mbf u_0 \Vert_\nu
= \l\Vert \sum_{j = 1}^q B_{t, ij} u_{j0} \r\Vert_\nu
\le C_\nu \l( \sum_{j = 1}^q \Vert B_{t, ij} u_{j0} \Vert_\nu^2 \r)^{1/2}
\le C_\nu \vert \mbf B_{t, i\cdot} \vert_2 \; \mu_\nu^{1/\nu} 
\le C_\nu \Xi (1 + t)^{-\varsigma} \; \mu_\nu^{1/\nu} 
\end{align*}
for all $1 \le i \le p$, and
\begin{align*}
& \Vert \vert \mbf B_t \mbf u_0 \vert_\infty \Vert_\nu
\le C_\nu \log^{1/2}(p) \l(\sum_{k = 1}^q \vert \mbf B_{t, \cdot k} \vert_\infty^2\r)^{1/2} 
q^{1/\nu} \mu_\nu^{1/\nu}
\le C_\nu \log^{1/2}(p) \Xi (1 + t)^{-\varsigma}
q^{1/\nu} \mu_\nu^{1/\nu}.
\end{align*}
Similarly, from Assumption~\ref{assum:idio} and independence of $\vep_{it}$ assumed in Assumption~\ref{assum:innov}~\ref{cond:iid}, we have
\begin{align*}
& \Vert \mbf D_{t, i\cdot} \bm\Gamma^{1/2} \bm\vep_0 \Vert_\nu 
\le C_\nu \vert \mbf D_{t, i\cdot} \bm\Gamma^{1/2} \vert_2 \; \mu_\nu^{1/\nu} 
\le C_\nu M_\vep^{1/2} \vert \mbf D_{t, i\cdot} \vert_2 \; \mu_\nu^{1/\nu} 
\\
&\le C_\nu M_\vep^{1/2} \l(\sum_{k = 1}^p C_{ik}^2\r)^{1/2} (1 + t)^{-\varsigma} \; \mu_\nu^{1/\nu}
\le C_\nu M_\vep^{1/2} \Xi (1 + t)^{-\varsigma} \; \mu_\nu^{1/\nu} 
\end{align*}
for all $1 \le i \le p$.
Then,
\begin{align}
& \Vert \vert \mbf D_t \bm\Gamma^{1/2} \bm\vep_0 \vert_\infty \Vert_\nu^\nu
\le \sum_{i = 1}^p \Vert \mbf D_{t, i\cdot} \bm\Gamma^{1/2} \bm\vep_0 \Vert_\nu^\nu
\le p (C_\nu M_\vep^{1/2} \Xi (1 + t)^{-\varsigma})^\nu \mu_\nu
\label{eq:d:gam:vep:inf}
\end{align}
such that
$\Vert \vert \mbf D_t \bm\Gamma^{1/2} \bm\vep_0 \vert_\infty \Vert_\nu
\le C_\nu M_\vep^{1/2} \Xi (1 + t)^{-\varsigma} p^{1/\nu} \mu_\nu^{1/\nu}$.
Then, for some constant $C_{\nu, \Xi} > 0$,
\begin{align*}
& \delta_{t, \nu, i} \le C_{\nu, \Xi} M_\vep^{1/2} (1 + t)^{-\varsigma} \mu_\nu^{1/\nu},
\quad
\delta_{t, \nu} \le C_{\nu, \Xi} M_\vep^{1/2} \log^{1/2}(p) (1 + t)^{-\varsigma} p^{1/\nu} \mu_\nu^{1/\nu}
\end{align*}
and setting $\alpha \le \varsigma - 1$ leads to
\begin{align*}
& \Phi_{\nu, \alpha}  \le C_{\nu, \Xi, \varsigma} M_\vep^{1/2} \mu_\nu^{1/\nu},
\quad 
\Psi_{\nu, \alpha} \le C_{\nu, \Xi, \varsigma} M_\vep^{1/2} p^{2/\nu} \mu_\nu^{1/\nu}
\quad \text{and}
\\
& \Vert \vert \mbf X_{\cdot} \vert_\infty \Vert_{\nu, \alpha} 
\le C_{\nu, \Xi, \varsigma} M_\vep^{1/2} \log^{1/2}(p) p^{1/\nu} \mu_\nu^{1/\nu}.
\end{align*}
\end{proof}

\begin{lem}
\label{lem:decay} 
Denote by $\bm\Gamma_x(\ell) = [\gamma_{x, ii'}(\ell), \, 1 \le i, i' \le p]$.
Then, there exists some constant $C_{\Xi, \varsigma, \vep}$
depending only on $\Xi, \varsigma$ defined in Assumptions~\ref{assum:common} and \ref{assum:idio}~\ref{cond:idio:coef} and $M_\vep$, such that
$\max_{1 \le i, i' \le p} \gamma_{x, ii'}(\ell) \le C_{\Xi, \varsigma, \vep} (1 + \vert \ell \vert)^{-\varsigma}$
and consequently,
$\max_{1 \le i, i' \le p} \sum_{\vert \ell \vert > m} \gamma_{x, ii'}(\ell) = O(m^{-\varsigma + 1})= o(m^{-1})$.
\end{lem}

\begin{proof}
By Assumption~\ref{assum:innov}~\ref{cond:uncor}, 
we have $\bm\Gamma_x(\ell) = \bm\Gamma_\chi(\ell) + \bm\Gamma_\xi(\ell)$ for all $\ell \in \Z$,
and define $\gamma_{\chi, ii'}(\ell)$ and $\gamma_{\xi, ii'}(\ell)$
analogously as $\gamma_{x, ii'}(\ell)$. 
From Assumption~\ref{assum:common}, for some $h \ge 0$,
\begin{align}
\vert \gamma_{\chi, ii'}(h) \vert = \l\vert
\E\l(\sum_{\ell, \ell' = 0}^\infty \sum_{j, j' = 1}^q B_{\ell, ij} B_{\ell', i'j'} u_{j, t - \ell - h}u_{j', t -\ell'} \r)
\r\vert
= \sum_{\ell = 0}^\infty \l\vert \sum_{j = 1}^q B_{\ell, ij} B_{\ell + h, i'j} \r\vert
\nn \\
= \sum_{\ell = 0}^\infty \vert \mbf B_{\ell, i \cdot} \vert_2 \; \vert \mbf B_{\ell + h, i' \cdot} \vert_2 
\le \sum_{\ell = 0}^\infty \frac{\Xi^2}{(1 + \ell)^{\varsigma} (1 + \ell + h)^{\varsigma}}
\le \sum_{\ell = 0}^\infty \frac{\Xi^2}{(1 + \ell)^{\varsigma} (1 + h)^{\varsigma}}
\le C_{\Xi, \varsigma} (1 + h)^{-\varsigma}
\label{eq:common:decay}
\end{align}
uniformly in $1 \le i, i' \le p$ for some $C_{\Xi, \varsigma} > 0$ 
depending only on $\Xi$ and $\varsigma$. Therefore,
$\sum_{\vert \ell \vert > m} \vert \gamma_{\chi, ii'}(\ell) \vert \le 
2C_{\Xi, \varsigma} \sum_{\ell > m} \ell^{-\varsigma}
= O(m^{-\varsigma + 1})$.
Similarly, from Assumption~\ref{assum:idio}~\ref{cond:idio:coef}, 
\begin{align}
\vert \gamma_{\xi, ii'}(h) \vert =& \l\vert
\E\l(\sum_{\ell, \ell' = 0}^\infty 
(\mbf D_{\ell, i \cdot} \bm\Gamma^{1/2} \bm\vep_{t - \ell - h})
(\mbf D_{\ell', i' \cdot} \bm\Gamma^{1/2} \bm\vep_{t - \ell'})\r)
\r\vert
= \sum_{\ell = 0}^\infty \l\vert 
\mbf D_{\ell, i \cdot} \bm\Gamma \mbf D^\top_{\ell + h, i' \cdot} 
\r\vert
\nn \\
=& \sum_{\ell = 0}^\infty \Vert \bm\Gamma \Vert
\vert \mbf D_{\ell, i \cdot} \vert_2 \; \vert \mbf D_{\ell + h, i' \cdot} \vert_2 
\le M_\vep
\sum_{\ell = 0}^\infty \frac{\Xi^2}{(1 + \ell)^{\varsigma} (1 + \ell + h)^{\varsigma}}
\nn \\
\le& M_\vep
\sum_{\ell = 0}^\infty \frac{\Xi^2}{(1 + \ell)^{\varsigma} (1 + h)^{\varsigma}}
\le C_{\Xi, \varsigma, \vep} (1 + h)^{-\varsigma}
\label{eq:idio:decay}
\end{align}
uniformly in $1 \le i, i' \le p$
for some $C_{\Xi, \varsigma, \vep} > 0$ 
depending only on $\Xi$, $\varsigma$ and $M_\vep$.
Therefore,
$\sum_{\vert \ell \vert > m} \vert \gamma_{\xi, ii'}(\ell) \vert = O(m^{-\varsigma+ 1})$,
which completes the proof.
\end{proof}

\begin{lem}
\label{lem:sigma:bound}
Denote by $\bm\Sigma_x(\omega) = [\sigma_{x, ii'}(\omega), 1 \le i, i' \le p]$.
Then, there exists a constant $B_\sigma>0$ such that
$\sup_{\pi \in [-\pi, \pi]} \max_{1 \le i, i' \le p} \sigma_{x, ii'}(\omega) \le B_\sigma$.
\end{lem}

\begin{proof}
By Lemma~\ref{lem:decay}, we can find $B_\sigma$ that depends only on $\Xi$, $\varsigma$ and $\bm\Gamma$
such that
\begin{align*}
\sup_{\omega \in [-\pi, \pi]} \max_{1 \le i, i' \le p} \l\vert \sigma_{x, ii'}(\omega) \r\vert \le 
\frac{1}{2\pi} \max_{1 \le i, i' \le p} \sum_{\ell = -\infty}^\infty  \l\vert \gamma_{x, ii'}(\ell) \r\vert
\le \frac{C_{\Xi, \varsigma, \vep}}{2\pi} \sum_{\ell = -\infty}^\infty \frac{1}{(1 + \vert \ell \vert)^\varsigma}
\le B_\sigma.
\end{align*}
\end{proof}

\begin{lem}
\label{lem:sigma:deriv}
For all $1 \le i, i' \le p$, the functions $\omega \mapsto \sigma_{\chi, ii'}(\omega)$
possess derivatives of any order and are of bounded variation,
i.e.\ there exists $B^\prime_\sigma > 0$ such that
$\sum_{k = 1}^N \vert \sigma_{\chi, ii'}(\omega_n) - \sigma_{\chi, ii'} (\omega_{k - 1}) \vert \le B^\prime_\sigma$
uniformly in $1 \le i, i' \le p$, $k \in \N$ and any partition of $[-\pi, \pi]$,
$-\pi = \omega_0 < \omega_1 < \ldots < \omega_N = \pi$.
\end{lem}

\begin{proof}
From Lemma~\ref{lem:decay}, 
$\max_{1 \le i, i' \le p} \vert \gamma_{\chi, ii'} (\ell) \vert \le C_{\Xi, \varsigma} (1 + \vert \ell \vert)^{-\varsigma}$
for all $\ell$, which implies that
\begin{align*}
\sigma_{\chi, ii'}(\omega) = \frac{1}{2\pi} \sum_{\ell = -\infty}^\infty \gamma_{\chi, ii'}(\ell) e^{-\iota \omega \ell}
\end{align*}
has derivatives of all orders. Moreover,
\begin{align*}
\l\vert \frac{d}{d\omega} \sigma_{\chi, ii'}(\omega) \r\vert = \frac{1}{2\pi} \sum_{\ell = -\infty}^\infty
\l\vert (-\iota \ell) \gamma_{\chi, ii'}(\ell) e^{-\iota \omega \ell} \r\vert
\le \frac{C_{\Xi, \varsigma}}{\pi} \sum_{\ell = 0}^\infty \frac{\ell}{(1 + \ell)^{\varsigma}} 
\le C^\prime_{\Xi, \varsigma}
\end{align*}
for some constant $C^\prime_{\Xi, \varsigma} > 0$ not depending on $1 \le i, i' \le p$ or $\omega \in [-\pi, \pi]$,
which entails the bounded variation of $\sigma_{\chi, ii'}(\omega)$.
\end{proof}

\begin{lem}
\label{lem:acv:x}
Let $\bm\Gamma_x(\ell) = [\gamma_{x, ii'}(\ell)]_{i, i' = 1}^p$
and $\wh{\bm\Gamma}_x(\ell) = [\wh\gamma_{x, ii'}(\ell)]_{i, i' = 1}^p$.
For some fixed $s \in \mathbb{N}$, the following holds.
\begin{enumerate}[wide, label = (\roman*)]
\item \label{lem:acv:x:one}
$p^{-1} \max_{- s \le \ell \le s}
\Vert \wh{\bm\Gamma}_{x}(\ell) - \bm\Gamma_x(\ell) \Vert_F = O_P(n^{-1/2})$.
\item \label{lem:acv:x:two}
$\max_{- s \le \ell \le s}
\vert \wh{\bm\Gamma}_{x}(\ell) - \bm\Gamma_x(\ell) \vert_\infty = O_P(\wt\vartheta_{n, p})$ where
\begin{align*}
\wt\vartheta_{n, p} = \l\{\begin{array}{ll}
\frac{p^{2/\nu} \log^3(p)}{n^{1 - 2/\nu}} \vee \sqrt{\frac{\log(p)}{n}} & \text{under Assumption~\ref{assum:innov}~\ref{cond:dist}~\ref{cond:moment},} 
\\
\sqrt{\frac{\log(p)}{n}} & \text{under Assumption~\ref{assum:innov}~\ref{cond:dist}~\ref{cond:gauss},} 
\end{array}\r.
\end{align*}
see also~\eqref{eq:tilde:vartheta}.
\end{enumerate}
\end{lem}

\begin{proof}
In Lemma~\ref{lem:func:dep},
for $\varsigma > 2$, we can always set $\alpha = \varsigma - 1 > 1/2 - 2/\nu$.
By Proposition~3.3 of \cite{zhang2021},
there exist universal constants $C_1, C_2 >0$
and constants $C_\alpha, C_{\nu, \alpha}>0$ that depend only on their subscripts,
such that for any $z > 0$, 
\begin{align*}
& \p\l( \max_{-s \le \ell \le s} \l\vert \wh\gamma_{x, ii'}(\ell) - \gamma_{x, ii'}(\ell) \r\vert > z \r) \le
\\
& \l\{\begin{array}{ll}
\frac{C_{\nu, \alpha} n s^{\nu/4} \Phi_{\nu, \alpha}^\nu}{(nz)^{\nu/2}}
+ C_1 s \exp\l(-\frac{nz^2}{C_\alpha \Phi_{4, \alpha}^4}\r) 
& \text{under Assumption~\ref{assum:innov}~\ref{cond:dist}~\ref{cond:moment}},
\\
2 s \exp\l[- C_2 \min\l( \frac{n z^2}{\Phi_{2, 0}^4}, \frac{nz}{\Phi_{2, 0}^2} \r) \r] 
& \text{under~Assumption~\ref{assum:innov}~\ref{cond:dist}~\ref{cond:gauss}.}
\end{array}\r.
\end{align*}
Then, due to Lemma~\ref{lem:func:dep}, the finiteness of $s$ and since $\nu > 4$ under Assumption~\ref{assum:innov}~\ref{cond:dist}~\ref{cond:moment},
there exists some constant $C >0$ not dependent on $i, i'$ such that
\begin{align}
\label{eq:lem:acv:x:one}
\E\l(\max_{-s \le \ell \le s} 
\l\vert \wh\gamma_{x, ii'}(\ell) - \gamma_{x, ii'}(\ell)) \r\vert^2\r) \le \frac{C}{n}
\end{align}
and~\ref{lem:acv:x:one} follows by Chebyshev's inequality (see also Remark~\ref{rem:unif:op}).
By the same proposition and Lemma~\ref{lem:func:dep}, 
\begin{align*}
& \p\l( \max_{-s \le \ell \le s} \l\vert \wh{\bm\Gamma}_x(\ell) - \bm\Gamma_x(\ell) \r\vert_\infty > z \r) \le
\\
& \l\{\begin{array}{ll}
\frac{C_{\nu, \alpha} n s^{\nu/4}
(\log(p) \Vert \vert \mbf X_{\cdot} \vert_\infty \Vert_{\nu, \alpha} 
\wedge \Psi_{\nu, \alpha})^\nu}{(nz)^{\nu/2}}
+ C_1 s p^2 \exp\l(-\frac{nz^2}{C_\alpha \Phi_{4, \alpha}^4}\r) 
& \text{under Assumption~\ref{assum:innov}~\ref{cond:dist}~\ref{cond:moment}},
\\
2 s p^2 \exp\l[- C_2 \min\l( \frac{n z^2}{\Phi_{2, 0}^4}, \frac{nz}{\Phi_{2, 0}^2} \r) \r] 
& \text{under~Assumption~\ref{assum:innov}~\ref{cond:dist}~\ref{cond:gauss},}
\end{array}\r.
\end{align*}
from which~\ref{lem:acv:x:two} follows.
\end{proof}

In what follows, we assume that $\wh{\mbf S}  = [\wh s_{ii'}, \, 1 \le i, i' \le p]$ is an estimator of
a Hermitian, positive semi-definite matrix $\mbf S \in \C^{p \times p}$ of a finite rank $q < p$.
In addition, we have $\vert \mbf S \vert_\infty = O(1)$.
Denoting the eigenvalues and eigenvectors of $\mbf S$ by
$(\mu_j, \mbf e_j), \, 1 \le j \le q$, we assume that there exist 
pairs of constants $(\alpha_j, \beta_j)$ and constants $\rho_j \in (0, 1]$
with $\rho_1 \ge \ldots \ge \rho_q$ satisfying
\begin{align}
\label{cond:s:one}
\beta_1 \ge \frac{\mu_1}{p^{\rho_1}} \ge \alpha_1 > \beta_2 \ge \frac{\mu_2}{p^{\rho_2}}
\ge \alpha_2 > \ldots > \beta_q \ge \frac{\mu_q}{p^{\rho_q}} \ge \alpha_q > 0.
\end{align}
Also, we assume that there exist a (deterministic) matrix $\wt{\mbf S} = [\wt s_{i, i'}, \, 1 \le i, i' \le p]$,
a constant $C>0$ not dependent on $i, i'$ and
$\zeta_{n, p}, \bar{\zeta}_{n, p} \to 0$ as $n, p \to 0$, such that
\begin{align}
\Vert \wt{\mbf S} - \mbf S \Vert & = O(1), \label{cond:s:two} \\
\E[(\wh s_{ii'} - \wt s_{ii'})^2] & \le C (\zeta_{n, p})^2, \label{cond:s:three} \\
\max_{1 \le i, i' \le p} \vert \wh s_{ii'} - \wt s_{ii'} \vert & = O_P(\bar{\zeta}_{n, p}). \label{cond:s:four}
\end{align}
Finally, we assume that $p^{1 - \rho_q} \zeta_{n, p} \to 0$ as $n, p \to \infty$.
Denoting the eigenvalues and eigenvectors of $\mbf S$ by
$(\wh\mu_j, \wh{\mbf e}_j), \, j \ge 1$,
let $\wh{\mbf E} = [\wh{\mbf e}_j, \, 1 \le j \le q]$
and $\wh{\bm{\mc M}} = \text{diag}(\wh\mu_j, \, 1 \le j \le q)$
and analogously define $\mbf E$ and $\bm{\mc M}$ with $(\mu_j, \mbf e_j)$.

\begin{lem}
\label{lem:one}
There exists a unitary, diagonal matrix 
$\bm{\mc O} \in \C^{q \times q}$ such that
\begin{align*}
\frac{\mu_q}{p}
\l\Vert \wh{\mbf E} - \mbf E \bm{\mc O} \r\Vert_F = 
O_P\l( \zeta_{n, p} \vee \frac{1}{p}\r). 
\end{align*}
\end{lem}
\begin{proof}
By~\eqref{cond:s:two}, \eqref{cond:s:three} and Chebyshev's inequality,
we have
\begin{align}
& \frac{1}{p} \Vert \wh{\mbf S} - \mbf S \Vert
\le 
\frac{1}{p}\ \Vert \wh{\mbf S} - \wt{\mbf S} \Vert
+ \frac{1}{p} \Vert \wt{\mbf S} - \mbf S \Vert
\le 
\frac{1}{p} \Vert \wh{\mbf S} - \wt{\mbf S} \Vert_F
+ \frac{1}{p} \Vert \wt{\mbf S} - \mbf S \Vert
= O_P\l( \zeta_{n, p} \vee \frac{1}{p} \r).
\label{lem:dk:eq:cov}
\end{align}
Then by Theorem~2 of \cite{yu2015}, 
there exist such $\bm{\mc O}$ satisfying
\begin{align*}
\Vert \wh{\mbf E} - \mbf E \bm{\mc O} \Vert_F \le 
\frac{2\sqrt{2q} \Vert \wh{\mbf S} - \mbf S \Vert}
{\min(\mu_0 - \mu_1, \mu_q)}
\end{align*}
with $\mu_0 = \infty$ which, combined with~\eqref{lem:dk:eq:cov}, concludes the proof.
\end{proof}

\begin{lem}
\label{lem:two}
\begin{align*}
\frac{\mu_q}{p}
\l\Vert \l(\frac{\wh{\bm{\mc M}}}{p}\r)^{-1} - 
\l(\frac{\bm{\mc M}}{p}\r)^{-1} \r\Vert_F
&= O_P\l(\sqrt{q}\l(\zeta_{n, p} \vee \frac{1}{p}\r) \r),
\\
\frac{\mu_q}{p}
\l\Vert \l(\frac{\wh{\bm{\mc M}}}{p}\r)^{-1} - 
\l(\frac{\bm{\mc M}}{p}\r)^{-1} \r\Vert
&= O_P\l(\zeta_{n, p} \vee \frac{1}{p}\r).
\end{align*}
\end{lem}

\begin{proof}
As a consequence of~\eqref{lem:dk:eq:cov} and Weyl's inequality, 
for all $1 \le j \le q$,
\begin{align}
\label{eq:first:q:eval}
\frac{1}{p} \vert \wh\mu_j- \mu_j \vert
\le \frac{1}{p} \Vert \wh{\mbf S} - \mbf S \Vert
= O_P\l(\zeta_{n, p} \vee \frac{1}{p} \r).
\end{align}
Also from~\eqref{cond:s:one},
there exists $\alpha_q$ such that 
$p^{-1} \mu_q \ge p^{\rho_q - 1}\alpha_q$ and 
thus $p^{-1} \wh{\mu}_q \ge 
p^{\rho_q - 1}\alpha_q + O_P(\zeta_{n, p} \vee p^{-1})
= p^{\rho_q - 1}\alpha_q(1 + o_P(1))$,
which implies that the matrix $p^{-1}\bm{\mc M}$ is invertible and 
the inverse of $p^{-1}\wh{\bm{\mc M}}$ exists for large enough $n$ and $p$. 
Therefore,
\begin{align*}
\l\Vert \l( \frac{\bm{\mc M}}{p} \r)^{-1} \r\Vert = \frac{p}{\mu_q}
\quad \text{and} \quad
\l\Vert \l(\frac{\wh{\bm{\mc M}}}{p}\r)^{-1} \r\Vert = \frac{p}{\wh{\mu}_q} = 
\frac{1}{p^{-1} \mu_q(1 + o_P(1))}.
\end{align*}
Then from~\eqref{eq:first:q:eval}, we have 
\begin{align*}
& \l\Vert \l(\frac{\wh{\bm{\mc M}}}{p}\r)^{-1} - 
\l(\frac{\bm{\mc M}}{p}\r)^{-1} \r\Vert_F
= \sqrt{p^2 \sum_{j = 1}^{q} 
\l(\frac{1}{\wh\mu_j} - \frac{1}{\mu_j} \r)^2}
= \sqrt{\sum_{j = 1}^{q}
\l(\frac{p^{-1}(\wh\mu_j - \mu_j)}{p^{-1} \wh\mu_j \cdot p^{-1} \mu_j} \r)^2}
\\
\le & \, \sqrt{\sum_{j = 1}^{q}
\l(\frac{p^{-1}(\wh\mu_j - \mu_j)}{(p^{-1} \mu_j)^2 (1 + o_P(1))} \r)^2}
= O_P\l( \frac{\sqrt{q} p}{\mu_q} \l(\zeta_{n, p} \vee \frac{1}{p}\r) \r).
\end{align*}
The second claim follows similarly.
\end{proof}

\begin{lem}
\label{lem:three}
Let $\bm\varphi_i$ denote the $p$-vector whose $i$-th element is one
and the rest are set to be zero. Then, with $\bm{\mc O}$ defined in Lemma~\ref{lem:one},
\begin{enumerate}[wide, label = (\roman*)]
\item \label{lem:three:one} uniformly for all $1 \le i \le p$ (in the sense described in Remark~\ref{rem:unif:op}),
\begin{align*}
\sqrt{p} \cdot
\frac{\mu_q}{p}
\l\vert \bm\varphi_i^\top \l(\wh{\mbf E} - \mbf E \bm{\mc O} \r) \r\vert_2
= O_P\l( \zeta_{n, p} \vee \frac{1}{\sqrt p}\r).
\end{align*}
\item \label{lem:three:two} In addition,
\begin{align*}
\sqrt{p} \cdot \frac{\mu_q}{2}
\max_{1 \le i \le p}
\l\vert \bm\varphi_i^\top \l(\wh{\mbf E} - \mbf E \bm{\mc O} \r) \r\vert_2
= O_P\l( \bar{\zeta}_{n, p} \vee \frac{1}{\sqrt p}\r).
\end{align*}
\end{enumerate}
\end{lem}

\begin{proof}
By~\eqref{cond:s:two}, \eqref{cond:s:three} and Chebyshev's inequality, we have
\begin{align}
& \frac{1}{\sqrt p} 
\l\vert \bm\varphi_i^\top (\wh{\mbf S} - \mbf S) \r\vert_2
\le 
\frac{1}{\sqrt p} 
\l\vert \bm\varphi_i^\top ( \wh{\mbf S} - \wt{\mbf S} ) \r\vert_2
+ \frac{1}{\sqrt p} \Vert \wt{\mbf S} - \mbf S \Vert
= O_P\l( \zeta_{n, p} \vee \frac{1}{\sqrt p}\r)
\label{lem:three:eq:one}
\end{align}
uniformly in $1 \le i \le p$.
Then, by~\eqref{lem:three:eq:one}, \eqref{cond:s:one},
Lemmas~\ref{lem:one}--\ref{lem:two} and the finiteness of the elements of $\mbf S$, we have
\begin{align*}
& \sqrt p \cdot \frac{\mu_q}{p}
\l\vert \bm\varphi_i^\top (\wh{\mbf E} - \mbf E \bm{\mc O}) \r\vert_2 
= \frac 1 {\sqrt p} \cdot \frac{\mu_q}{p} \l\vert
\bm\varphi_i^\top \l[ \wh{\mbf S} \wh{\mbf E}
\l( \frac{\wh{\bm{\mc M}}}{p} \r)^{-1}
- \mbf S \mbf E
\l( \frac{\bm{\mc M}}{p}\r)^{-1} \bm{\mc O}\r] \r\vert_2 
\\
\le& \frac{\mu_q}{p}
\l\{
\frac 1 {\sqrt p} \l\vert
\bm\varphi_i^\top (\wh{\mbf S} - \mbf S) \r\vert_2 \,
\l\Vert\l(\frac{\wh{\bm{\mc M}}}{p}\r)^{-1}\r\Vert \r.
+ 
\frac 1 {\sqrt p} \l\vert \bm\varphi_i^\top \mbf S \r\vert_2\,
\l\Vert\l(\frac{\wh{\bm{\mc M}}}{p}\r)^{-1} - 
\l(\frac{\bm{\mc M}}{p}\r)^{-1}\r\Vert 
\\ 
& + \l. \frac 1 {\sqrt p} \l\vert \bm\varphi_i^\top \mbf S\r\vert_2 \,
\l\Vert \l(\frac{\bm{\mc M}}{p} \r)^{-1} \r\Vert\,
\l\Vert \wh{\mbf E} - \mbf E \bm{\mc O} \r\Vert \r\}
= O_P\l( \zeta_{n, p} \vee \frac{1}{\sqrt p} \r) 
\end{align*} 
uniformly over $1 \le i \le p$, and thus~\ref{lem:evec:vec:one} follows.
The claim~\ref{lem:three:two} follows analogously, except that~\eqref{lem:three:eq:one} is replaced by
\begin{align*}
& \frac{1}{\sqrt p} 
\l\vert \bm\varphi_i^\top (\wh{\mbf S} - \mbf S ) \r\vert_2
\le 
\frac{1}{\sqrt p} 
\max_{1 \le i \le p} \l\vert \bm\varphi_i^\top (\wh{\mbf S} - \wt{\mbf S} ) \r\vert_2
+ \frac{1}{\sqrt p} \Vert \wt{\mbf S} - \mbf S \Vert 
= O_P\l( \bar{\zeta}_{n, p} \vee \frac{1}{\sqrt p}\r),
\end{align*}
where the inequality follows from~\eqref{cond:s:two} and~\eqref{cond:s:four}.
\end{proof}

\subsubsection{Proof of Theorem~\ref{thm:common:spec}}

In what follows, by $\mbf e_{x, j}(\omega)$, we denote the eigenvectors of $\bm\Sigma_x(\omega)$
corresponding to the eigenvalues $\mu_{x, j}(\omega)$,
and we similarly define $\mbf e_{\chi, j}(\omega)$ and $\mbf e_{\xi, j}(\omega)$.

Recall the decomposition of $\bm\Sigma_\chi(\omega)$ as
\begin{align*}
\bm\Sigma_\chi(\omega) &= \bm\E_\chi(\omega) \bm{\mc M}_\chi(\omega) \bm\E_\chi^*(\omega)
= \sum_{j = 1}^{q} \mu_{\chi, j}(\omega) \mbf e_{\chi, j}(\omega) \mbf e_{\chi, j}^*(\omega),
\end{align*}
where $\bm{\mc M}_\chi(\omega)$ denotes the diagonal matrix with $\mu_{\chi, 1}(\omega) \ge \ldots \ge \mu_{\chi, q}(\omega)$, the $q$ eigenvalues of $\bm\Sigma_\chi(\omega)$, on its diagonal, and  $\mbf E_\chi(\omega) = [\mbf e_{\chi, j}(\omega), \, 1 \le j \le q]$
the matrix of the normalised eigenvectors associated with $\mu_{\chi, j}(\omega), \, 1 \le j \le q$. 
In what follows, we assume the conditions made in Theorem~\ref{thm:common:spec} are met
and prove a series of results results.
Then, Theorem~\ref{thm:common:spec} is a direct consequence of Proposition~\ref{prop:acv:common}~\ref{prop:acv:common:two}.

\begin{prop}
\label{prop:spec:x} 
There exists a constant $C>0$ not dependent on $1 \le i, i' \le p$
such that
\begin{align*}
\E\l( \sup_{\omega \in [-\pi, \pi]} \l\vert \wh{\sigma}_{x, ii'}(\omega) - \sigma_{x, ii'}(\omega) \r\vert^2 \r)
\le C\l(\psi_n \vee \frac{1}{m}\r)^2.
\end{align*}
\end{prop}

\begin{proof}
Noting that
\begin{align}
& \E\l( \sup_{\omega} \l\vert \wh{\sigma}_{x, ii'}(\omega) - \sigma_{x, ii'}(\omega) \r\vert^2 \r)
\nn 
\\
& \le 2 \E\l( \sup_{\omega} \l\vert \wh{\sigma}_{x, ii'}(\omega) - \E(\wh\sigma_{x, ii'}(\omega)) \r\vert^2 \r)
+
2 \sup_{\omega} \l\vert \E(\wh{\sigma}_{x, ii'}(\omega)) - \sigma_{x, ii'}(\omega) \r\vert^2,
\label{prop:spec:x:eq:one}
\end{align}
we first address the first term in the RHS of \eqref{prop:spec:x:eq:one}.
In Lemma~\ref{lem:func:dep},
for $\varsigma > 2$, we can always set $\alpha = \varsigma - 1 > 1/2 - 2/\nu$.
Then, from the finiteness of $\Phi_{\nu, \alpha}$ shown therein and 
by Proposition~4.3 of \cite{zhang2021},
there exist universal constants $C_1, C_2, C_3 >0$
and constants $C_\alpha, C_{\nu, \alpha} >0$ that depend only on their subscripts,
such that for any $z > 0$,
\begin{align}
& \p\l( \sup_\omega \l\vert \wh\sigma_{x, ii'}(\omega) 
- \E(\wh{\sigma}_{x, ii'}(\omega)) \r\vert > z \r) \le
\label{eq:zhangwu}
\\
& \l\{\begin{array}{ll}
\frac{C_{\nu, \alpha} nm^{\nu/2} \Phi_{\nu, \alpha}^\nu}{(nz)^{\nu/2}}
+ C_1 m \exp\l(-\frac{nz^2}{C_\alpha m \Phi_{4, \alpha}^4}\r) 
& \text{under Assumption~\ref{assum:innov}~\ref{cond:dist}~\ref{cond:moment}},
\\
C_2 m \exp\l[- C_3 \min\l( \frac{n z^2}{m \Phi_{2, 0}^4}, \frac{nz}{m \Phi_{2, 0}^2} \r) \r] 
& \text{under~Assumption~\ref{assum:innov}~\ref{cond:dist}~\ref{cond:gauss}.}
\end{array}\r. \nn
\end{align}
Noting that for any positive random variable $Y$, we have $\E(Y) = \int_0^\infty \p(Y > y) dy$,
we have
$\E( \sup_\omega \vert \wh{\sigma}_{x, ii'}(\omega) - \E(\wh\sigma_{x, ii'}(\omega)) \vert^2)
\le C \psi_n^2$ for some constant $C >0$ independent of $i, i'$ (see Remark~\ref{rem:unif:op}).
Turning our attention to the second term in the RHS of~\eqref{prop:spec:x:eq:one},
\begin{align*}
& \max_{1 \le i, i' \le p} \sup_\omega 
2\pi \l\vert \E(\wh{\sigma}_{x, ii'}(\omega)) - \sigma_{x, ii'}(\omega) \r\vert
\\
& \le \max_{1 \le i, i' \le p} \sum_{\ell = -m}^m \frac{\vert \ell \vert}{m} \l\vert \gamma_{x, ii'}(\ell) \r\vert
+ \max_{1 \le i, i' \le p} \sum_{\vert \ell \vert > m} \l\vert \gamma_{x, ii'} (\ell) \r\vert
= I + II.
\end{align*}
From Lemma~\ref{lem:decay} and that $\varsigma > 2$, there exists a constant $C^\prime_{\Xi, \varsigma, \vep}>0$
such that
\begin{align*}
I \le 2C_{\Xi, \varsigma, \vep} \sum_{\ell = 1}^m \frac{\ell}{m (1 + \ell)^\varsigma}
\le \frac{2C_{\Xi, \varsigma, \vep}}{m} \sum_{\ell = 1}^m \frac{1}{(1 + \ell)^{\varsigma - 1}}
\le \frac{C^\prime_{\Xi, \varsigma, \vep}}{m}
\end{align*}
and $II = O(m^{-\varsigma + 1})$, thus $I + II = O(m^{-1})$ uniformly in $i, i'$ and $\omega$,
which completes the proof.
\end{proof}

\begin{rem}
\label{rem:unif:op}
All the probabilistic statements leading to the Frobenius norm bound in Theorem~\ref{thm:common:spec}
are due to Proposition~\ref{prop:spec:x}, which allows for deriving non-asymptotic bounds.
As in \cite{forni2017dynamic}, in what follows, we state that $\vert a_{i, i'} \vert = O_P(\psi_n)$
{\it uniformly} in $1 \le i, i' \le p$ when for any $\epsilon > 0$, there exists $\eta(\epsilon) > 0$
independent of $i, i'$ such that 
\begin{align*}
\p\l(\psi_n^{-1} \vert a_{i, i'} \vert > \eta(\epsilon)\r) < \epsilon \quad \text{for all $n$ and $p$}.
\end{align*}
Then, since $\psi_n^{-2} \E(\vert a_{i, i'} \vert^2) = \int_{0}^\infty \p(\psi_n^{-2} \vert a_{i, i'} \vert^2 > z) dz$,
we interchangeably write that there exists a constant $C>0$ independent of $i, i'$ such that
$\E(\vert a_{i, i'} \vert^2) \le C \psi_n^2$ for all $n$ and $p$.
To see this, let $t = C\psi_n$ with $\psi_n$ defined as in Theorem~\ref{thm:common:spec}.
Then, under Assumption~\ref{assum:innov}~\ref{cond:dist}~\ref{cond:moment},
there exists large enough constants $C, C' > 0$ 
that do not depend on $i, i'$ such that
\begin{align*}
& \E\l(\sup_\omega \l\vert \wh{\sigma}_{x, ii'}(\omega) - 
\E(\wh{\sigma}_{x, ii'}(\omega)) \r\vert^2 \r)   
= \int_0^{\infty} \p\l(\sup_\omega \l\vert \wh{\sigma}_{x, ii'}(\omega) - 
\E(\wh{\sigma}_{x, ii'}(\omega)) \r\vert^2 \ge z \r) dz
\\
&= \int_0^{t^2} \p\l(\sup_\omega \l\vert \wh{\sigma}_{x, ii'}(\omega) - 
\E(\wh{\sigma}_{x, ii'}(\omega)) \r\vert^2 \ge z \r) dz
+ \int_{t^2}^\infty \p\l(\sup_\omega \l\vert \wh{\sigma}_{x, ii'}(\omega) - 
\E(\wh{\sigma}_{x, ii'}(\omega))\r\vert^2 \ge z \r) dz
\\
&\le t^2 + \int_{t^2}^\infty \p\l(\sup_\omega \l\vert \wh{\sigma}_{x, ii'}(\omega) - \E(\wh{\sigma}_{x, ii'}(\omega)) \r\vert \ge \sqrt{z} \r) dz
\\
&\le t^2 + 
\int_{t^2}^{\infty} \frac{C_{\nu, \alpha} n m^{\nu/2} \Phi^\nu_{\nu, \alpha}}{(n\sqrt{z})^{\nu/2}} dz
+ \int_{t^2}^{\infty} C_1 m \exp\l( - \frac{nz}{C_\alpha m \Phi_{4, \alpha}^4} \r) dz
\\
&= t^2 +  \l[ - \frac{C_{\nu, \alpha} n^{1 - \nu/2} m^{\nu/2} \Phi^\nu_{\nu, \alpha}}{(\nu/4 - 1)} z^{-\nu/4 + 1} \r]_{t^2}^\infty
+ \l[- \frac{C_1C_\alpha\Phi_{4, \alpha}^4 m^2}{n} \exp\l( - \frac{nz}{C_{\alpha} m \Phi_{4, \alpha}^4} \r) \r]_{t^2}^{\infty}
\\
&\le \l(1 + \frac{C_{\nu, \alpha} \Phi^\nu_{\nu, \alpha}}{(\nu/4 - 1)C^{\nu/2}} 
+ \frac{C_1C_\alpha\Phi_{4, \alpha}^4 m^{1 - C^2/(C_\alpha \Phi_{4, \alpha}^4)} }{C^2\log(m)} \r) t^2 \le C' t^2,
\end{align*}
where we use~\eqref{eq:zhangwu} and that $t = C(m n^{-1 + 2/\nu} \vee \sqrt{m\log(m)/n})$ under the assumption.
The case under Assumption~\ref{assum:innov}~\ref{cond:dist}~\ref{cond:gauss}
is handled analogously.
In the remainder of the paper, we continue to write 
the boundedness in mean square interchangeably with 
the $O_P$ bound where similar arguments apply due to 
the upper bound on the tail probability.
\end{rem}

\begin{prop}
\label{prop:spec:x:max}
\begin{align*}
\sup_{\omega \in [-\pi, \pi]} \l\vert \wh{\bm\Sigma}_x(\omega) - \bm\Sigma_x(\omega) \r\vert_\infty
= O_P\l(\vartheta_{n, p} \vee \frac{1}{m}\r).
\end{align*}
\end{prop}

\begin{proof}
In Lemma~\ref{lem:func:dep},
for $\varsigma > 2$, we can always set $\alpha = \varsigma - 1 > 1/2 - 2/\nu$.
Then, by Proposition~4.3 of \cite{zhang2021},
there exist universal constants $C_1, C_2, C_3 >0$
and constants $C_\alpha, C_{\nu, \alpha}>0$ that depend only on their subscripts,
such that for any $z > 0$,
\begin{align*}
& \p\l( \sup_\omega \l\vert \wh{\bm\Sigma}_x(\omega) 
- \E(\wh{\bm\Sigma}_x(\omega)) \r\vert_\infty > z \r) \le
\\
& \l\{\begin{array}{ll}
\frac{C_{\nu, \alpha} n m^{\nu/2} (\log^{5/4}(p) \Vert \vert \mbf X_{\cdot} \vert_\infty \Vert_{\nu, \alpha} 
\wedge \Psi_{\nu, \alpha})^\nu}{(nz)^{\nu/2}}
+ C_1 mp^2 \exp\l(-\frac{nz^2}{C_\alpha m \Phi_{4, \alpha}^4}\r) 
& \text{under Assumption~\ref{assum:innov}~\ref{cond:dist}~\ref{cond:moment}},
\\
C_2 mp^2 \exp\l[- C_3 \min\l( \frac{n z^2}{m \Phi_{2, 0}^4}, \frac{nz}{m \Phi_{2, 0}^2} \r) \r] 
& \text{under~Assumption~\ref{assum:innov}~\ref{cond:dist}~\ref{cond:gauss}.}
\end{array}\r.
\end{align*}
Thanks to Lemma~\ref{lem:func:dep}, it follows that
$\sup_\omega \vert \wh{\bm\Sigma}_x(\omega) - \E(\wh{\bm\Sigma}_x(\omega)) \vert_\infty = O_P(\vartheta_{n, p})$.
As for the bias term,
we apply the arguments analogous to those adopted in the proof of Proposition~\ref{prop:spec:x}
and show $\sup_\omega \vert \E(\wh{\bm\Sigma}_x(\omega)) - \bm\Sigma_x(\omega) \vert_\infty = O(m^{-1})$,
which concludes the proof. 
\end{proof}

The following three lemmas follow from Lemmas~\ref{lem:one}--\ref{lem:three}
with $\wh{\mbf S} = \wh{\bm\Sigma}_x(\omega)$ and $\mbf S = \bm\Sigma_\chi(\omega)$,
as conditions~\eqref{cond:s:one}--\eqref{cond:s:four} are met by
Assumption~\ref{assum:factor} and 
Propositions~\ref{prop:idio:eval}, \ref{prop:spec:x} and~\ref{prop:spec:x:max},
with the supremum over $\omega \in [-\pi, \pi]$ taken where relevant.

\begin{lem}
\label{lem:dk}
Let $\wh{\mbf E}_{x, 1:q}(\omega) = [\wh{\mbf e}_{x, j}(\omega), \, j = 1, \ldots, q]$.
There exists a unitary, diagonal matrix 
$\bm{\mc O}(\omega) \in \C^{q \times q}$ for each $\omega \in [-\pi, \pi]$, such that
\begin{align*}
\sup_{\omega \in [-\pi, \pi]} \frac{\mu_{\chi, q}(\omega)}{p}
\l\Vert \wh{\mbf E}_{x, 1:q}(\omega) - \mbf E_\chi(\omega) \bm{\mc O}(\omega) \r\Vert_F = 
O_P\l( \psi_n \vee \frac{1}{m} \vee \frac{1}{p}\r). 
\end{align*}
\end{lem}

\begin{lem}
\label{lem:evals}
Let $\wh{\bm{\mc M}}_{x, 1:q}(\omega) = \text{diag}(\wh\mu_{x, j}(\omega), \, j = 1, \ldots, q)$.
Then,
\begin{align*}
\sup_{\omega \in [-\pi, \pi]}
\frac{\mu_{\chi, q}^2(\omega)}{p^2}
\l\Vert \l(\frac{\wh{\bm{\mc M}}_{x, 1:q}(\omega)}{p}\r)^{-1} - 
\l(\frac{\bm{\mc M}_\chi(\omega)}{p}\r)^{-1} \r\Vert_F
&= O_P\l(q\l(\psi_n \vee \frac{1}{m} \vee \frac{1}{p}\r) \r),
\\
\sup_{\omega \in [-\pi, \pi]}
\frac{\mu_{\chi, q}^2(\omega)}{p^2}
\l\Vert \l(\frac{\wh{\bm{\mc M}}_{x, 1:q}(\omega)}{p}\r)^{-1} - 
\l(\frac{\bm{\mc M}_\chi(\omega)}{p}\r)^{-1} \r\Vert
&= O_P\l(\psi_n \vee \frac{1}{m} \vee \frac{1}{p}\r).
\end{align*}
\end{lem}

\begin{lem}
\label{lem:evec:vec}
Let $\bm\varphi_i$ denote the $p$-vector whose $i$-th element is one
and the rest are set to be zero. Then with $\bm{\mc O}(\omega)$ defined in Lemma~\ref{lem:dk},
\begin{enumerate}[wide, label = (\roman*)]
\item \label{lem:evec:vec:one} uniformly for all $1 \le i \le p$,
\begin{align*}
\sqrt{p} \sup_{\omega \in [-\pi, \pi]}
\frac{\mu_{\chi, q}^2(\omega)}{p^2}
\l\vert \bm\varphi_i^\top \l(\wh{\mbf E}_{x, 1:q}(\omega) - 
\mbf E_{\chi}(\omega)\bm{\mc O}(\omega) \r) \r\vert_2
= O_P\l( \psi_n \vee \frac{1}{m} \vee \frac{1}{\sqrt p}\r).
\end{align*}
\item \label{lem:evec:vec:two} In addition,
\begin{align*}
\sqrt{p} \sup_{\omega \in [-\pi, \pi]}
\frac{\mu_{\chi, q}^2(\omega)}{p^2}
\max_{1 \le i \le p}
\l\vert \bm\varphi_i^\top \l(\wh{\mbf E}_{x, 1:q}(\omega) - 
\mbf E_{\chi}(\omega)\bm{\mc O}(\omega) \r) \r\vert_2
= O_P\l( \vartheta_{n, p} \vee \frac{1}{m} \vee \frac{1}{\sqrt p}\r).
\end{align*}
\end{enumerate}
\end{lem}

\begin{lem}
\label{lem:evec:size}
\begin{enumerate}[wide, label = (\roman*)]
\item \label{lem:evec:size:one} $\sup_{\omega \in [-\pi, \pi]}  \max_{1 \le j \le q} 
\mu_{\chi, j}^{1/2}(\omega) \max_{1 \le i \le p} \vert e_{\chi, ij}(\omega) \vert = O(1)$.

\item \label{lem:evec:size:two} If $\psi_n \to 0$ as $n, p \to \infty$, we have
$\sup_{\omega \in [-\pi, \pi]} \max_{1 \le j \le q} \wh\mu_{x, j}^{1/2}(\omega)
\vert \wh{e}_{x, ij}(\omega) \vert = O_P(1)$ uniformly in $1 \le i \le p$.

\item \label{lem:evec:size:three} If $\vartheta_{n, p} \to 0$ as $n, p \to \infty$, we have
$\sup_{\omega \in [-\pi, \pi]}  \max_{1 \le j \le q} 
\wh\mu_{\chi, j}^{1/2}(\omega) \max_{1 \le i \le p} \vert \wh{e}_{x, ij}(\omega) \vert = O_P(1)$.
\end{enumerate}
\end{lem}

\begin{proof}
Note that by Assumption~\ref{assum:innov}~\ref{cond:uncor} and 
Lemma~\ref{lem:sigma:bound},
$\sigma_{x, ii'}(\omega) = \sigma_{\chi, ii'}(\omega) + \sigma_{\xi, ii'}(\omega)$
and $\max_{i, i'} \sup_\omega \sigma_{x, ii'}(\omega) \le B_\sigma < \infty$. 
Then from that
$\sigma_{\chi, ii}(\omega)
= \sum_{j = 1}^q \mu_{\chi, j}(\omega) \vert e_{\chi, ij}(\omega) \vert^2 \le B_\sigma$,
the claim~\ref{lem:evec:size:one} follows. Next, by Proposition~\ref{prop:spec:x},
\begin{align*}
\sup_\omega \sum_{j = 1}^q \wh\mu_{x, j}(\omega) \vert \wh e_{x, ij}(\omega) \vert^2 \le
\sup_\omega \wh\sigma_{x, ii}(\omega)  
\le B_\sigma  + O_P\l(\psi_n \vee \frac{1}{m}\r)
\end{align*}
uniformly over $1 \le i \le p$
and~\eqref{eq:first:q:eval} indicates that $\wh\mu_{x, j}(\omega) \to \infty$ 
as $p \to \infty$ for $1 \le j \le q$,
from which~\ref{lem:evec:size:two} follows.
Also by Proposition~\ref{prop:spec:x:max} and Lemma~\ref{lem:sigma:bound},
\begin{align*}
\max_{1 \le i \le p} \sup_\omega \sum_{j = 1}^q \wh\mu_{x, j}(\omega) \vert \wh e_{x, ij}(\omega) \vert^2  \le
\max_{1 \le i \le p} \sup_\omega \wh\sigma_{x, ii}(\omega)  
\le B_\sigma + O_P\l(\vartheta_{n, p} \vee \frac{1}{m}\r)
\end{align*}
which, combined with~\eqref{eq:first:q:eval}, leads to~\ref{lem:evec:size:three}.
\end{proof} 

\begin{prop}
\label{prop:spec:common}
\begin{enumerate}[wide, label = (\roman*)]
\item \label{prop:spec:common:one} Uniformly over $1 \le i, i' \le p$,
\begin{align*}
\sup_{\omega \in [-\pi, \pi]}
\frac{\mu_{\chi, q}^2(\omega)}{p^2} 
\l\vert \wh{\sigma}_{\chi, ii'}(\omega) - \sigma_{\chi, ii'}(\omega) \r\vert
= O_P\l(q \l(\psi_n \vee \frac{1}{m} \vee \frac{1}{\sqrt{p}}\r) \r).
\end{align*}
\item \label{prop:spec:common:two} In addition,
\begin{align*}
\sup_{\omega \in [-\pi, \pi]}
\frac{\mu_{\chi, q}^2(\omega)}{p^2} 
\max_{1 \le i, i' \le p}
\l\vert \wh{\sigma}_{\chi, ii'}(\omega) - \sigma_{\chi, ii'}(\omega) \r\vert
= O_P\l(q \l(\vartheta_{n, p} \vee \frac{1}{m} \vee \frac{1}{\sqrt{p}}\r) \r).
\end{align*}
\end{enumerate}
\end{prop}

\begin{proof}
First, note that
\begin{align*}
& \sup_\omega \frac{\mu_{\chi, q}^2(\omega)}{p^2}
\l\vert \wh{\sigma}_{\chi, ii'}(\omega) - \sigma_{\chi, ii'}(\omega) \r\vert
=
\sup_\omega \frac{\mu_{\chi, q}^2(\omega)}{p^2}
\l\vert \bm\varphi_i^\top \l(\wh{ \bm\Sigma}_\chi(\omega) - \bm\Sigma_\chi(\omega) \r)
\bm\varphi_{i'} \r\vert
\\
\le & \sup_\omega \frac{\mu_{\chi, q}^2(\omega)}{p^2}
\l\{
\l\vert \bm\varphi_i^\top \l(\wh{\mbf E}_{x, 1:q}(\omega) - \mbf E_{\chi}(\omega) \bm{\mc O}(\omega) \r)
\wh{\bm{\mc M}}_{x, 1:q}(\omega) \wh{\mbf E}_{x, 1:q}^*(\omega) \bm\varphi_{i'} \r\vert
\r.
\\
& +
\l\vert \bm\varphi_i^\top \mbf E_{\chi}(\omega) \bm{\mc O}(\omega)
\l(\wh{\bm{\mc M}}_{x, 1:q}(\omega) - \bm{\mc M}_\chi(\omega)\r) \wh{\mbf E}_{x, 1:q}^*(\omega) \bm\varphi_{i'} \r\vert
\\
& +
\l.
\l\vert \bm\varphi_i^\top \mbf E_{\chi}(\omega) \bm{\mc O}(\omega)
\bm{\mc M}_{\chi}(\omega) \l(\wh{\mbf E}_{x, 1:q}(\omega) - 
\mbf E_\chi(\omega) \bm{\mc O}(\omega)\r)^* \bm\varphi_{i'} \r\vert \r\} = I + II + III.
\end{align*}
By Lemmas~\ref{lem:evals}, \ref{lem:evec:vec}~\ref{lem:evec:vec:one}, 
\ref{lem:evec:size}~\ref{lem:evec:size:one}--\ref{lem:evec:size:two} and
Cauchy-Schwarz inequality,
\begin{align*}
I =& \sup_\omega \frac{\mu_{\chi, q}^2(\omega)}{p^2}
\l\vert \sum_{j = 1}^q \wh{\mu}_{x, j}(\omega) \wh{e}^*_{x, i'j} \cdot
\bm\varphi_i^\top \l(\wh{\mbf E}_{x, 1:q}(\omega) - \mbf E_{\chi}(\omega) \bm{\mc O}(\omega) \r) \bm\varphi_{i'} \r\vert
\\
\le& \sup_\omega \frac{\mu_{\chi, q}^2(\omega)}{p^2} \cdot 
\sqrt{p} \l\vert \bm\varphi_i^\top \l(\wh{\mbf E}_{x, 1:q}(\omega) - \mbf E_{\chi}(\omega) \bm{\mc O}(\omega) \r) \r\vert_2 
\sqrt{\frac{1}{p} \sum_{j = 1}^q \wh{\mu}^2_{x, j}(\omega) \vert \wh{e}_{x, i'j}(\omega) \vert^2}
\\
=& O_P\l( \l(\psi_n \vee \frac{1}{m} \vee \frac{1}{\sqrt p} \r) \cdot 
\sqrt{\frac{1}{p} \sum_{j = 1}^q \wh{\mu}_{x, j}(\omega)} \r)
= O_P\l( \sqrt{q} \l(\psi_n \vee \frac{1}{m} \vee \frac{1}{\sqrt p} \r) \r)
\end{align*}
uniformly in $i$ and $i'$, 
and $III$ can be handled analogously.
By~\eqref{eq:first:q:eval} and Lemma~\ref{lem:evec:size}~\ref{lem:evec:size:two},
\begin{align*}
II \le & \sup_\omega \frac{\mu_{\chi, q}^2(\omega)}{p^2}
\sqrt{p} \l\vert \bm\varphi_i^\top \mbf E_{\chi}(\omega) \r\vert_2 \cdot
\frac{1}{p} \l\Vert \wh{\bm{\mc M}}_{x, 1:q}(\omega) - \bm{\mc M}_\chi(\omega) \r\Vert \cdot
\sqrt{p} \l\vert \bm\varphi_{i'}^\top \wh{\mbf E}_{x, 1:q}(\omega) \r\vert_2
\\
=& O_P\l( \l(\psi_n \vee \frac{1}{m} \vee \frac{1}{p} \r) 
\cdot \sup_\omega \frac{\mu_{\chi, q}^2(\omega)}{p^2}
\sqrt{p \sum_{j = 1}^q \frac{1}{\mu_{\chi, j}(\omega)} \cdot
p \sum_{j = 1}^q \frac{1}{\wh\mu_{x, j}(\omega)}}\r)
\\
=& O_P\l( q\l(\psi_n \vee \frac{1}{m} \vee \frac{1}{p} \r)\r)
\end{align*}
uniformly in $1 \le i \le p$, which proves~\ref{prop:spec:common:one}.
As for~\ref{prop:spec:common:two}, we take analogous steps
with Lemmas~\ref{lem:evec:vec}~\ref{lem:evec:vec:two} and \ref{lem:evec:size}~\ref{lem:evec:size:three} 
replacing Lemmas~\ref{lem:evec:vec}~\ref{lem:evec:vec:one} and \ref{lem:evec:size}~\ref{lem:evec:size:two}, respectively. 
\end{proof}

\begin{prop}
\label{prop:acv:common}
For $0 \le \ell \le s$ with some fixed $s \in \mathbb{N}$,
\begin{enumerate}[wide, label = (\roman*)]
\item \label{prop:acv:common:one} uniformly over $1 \le i, i' \le p$,
\begin{align*}
\l\vert \wh{\gamma}_{\chi, ii'}(\ell) - \gamma_{\chi, ii'}(\ell) \r\vert
= O_P\l(\frac{q p^2}{\inf_{\omega \in [-\pi, \pi]} \mu_{\chi, q}^2(\omega)} \l(\psi_n \vee \frac{1}{m} \vee \frac{1}{\sqrt p}\r) \r).
\end{align*}

\item \label{prop:acv:common:two} In addition,
\begin{align*}
\max_{1 \le i, i' \le p} \l\vert \wh{\gamma}_{\chi, ii'}(\ell) - \gamma_{\chi, ii'}(\ell) \r\vert
= O_P\l(\frac{q p^2}{\inf_{\omega \in [-\pi, \pi]} \mu_{\chi, q}^2(\omega)} \l(\vartheta_{n, p} \vee \frac{1}{m} \vee \frac{1}{\sqrt p}\r) \r).
\end{align*}
\end{enumerate}
\end{prop}

\begin{proof}
Note that
\begin{align*}
& \l\vert \wh{\gamma}_{\chi, ii'}(\ell) - \gamma_{\chi, ii'}(\ell) \r\vert
= \l\vert \frac{2\pi}{2m + 1} \sum_{k = -m}^m \wh\sigma_{\chi, ii'}(\omega_k) e^{\iota \omega_k \ell}
- \int_{-\pi}^\pi \sigma_{\chi, ii'}(\omega) e^{\iota \omega \ell} d\omega \r\vert
\\
& \le \frac{2\pi}{2m + 1} \sum_{k = -m}^m \l\vert \wh\sigma_{\chi, ii'}(\omega_k) - \sigma_{\chi, ii'}(\omega_k) \r\vert
\\
& + \l\vert \frac{2\pi}{2m + 1} \sum_{k = -m}^m \sigma_{\chi, ii'}(\omega_k) e^{\iota \omega_k \ell}
- \int_{-\pi}^\pi \sigma_{\chi, ii'}(\omega) e^{\iota \omega \ell} d\omega \r\vert = I + II.
\end{align*}
By Proposition~\ref{prop:spec:common}~\ref{prop:spec:common:one},
\begin{align*}
I \le 2\pi \sup_\omega \l\vert \wh\sigma_{\chi, ii'}(\omega) - \sigma_{\chi, ii'}(\omega) \r\vert
= O_P\l(\frac{qp^2}{\inf_\omega \mu_{\chi, q}^2(\omega)}
\l(\psi_n \vee \frac{1}{m} \vee \frac{1}{\sqrt p}\r)\r).
\end{align*}
Next, we can find $\{\omega_k^*\}_{k = -m}^{m - 1}$
and $\{\omega_k^\circ\}_{k = -m}^{m - 1}$
with $\omega_k^*, \omega_k^\circ \in [\omega_k, \omega_{k + 1}]$,
such that
\begin{align*}
II \le & 
\frac{2\pi}{2m + 1} \sum_{k = -m}^{m - 1} \max_{\omega_k \le \omega \le \omega_{k + 1}}
\l\vert \sigma_{\chi, ii'}(\omega_k) e^{\iota \omega_k \ell} - \sigma_{\chi, ii'}(\omega) e^{\iota \omega \ell} \r\vert
\\
\le &
\frac{2\pi}{2m + 1} \sum_{k = -m}^{m - 1} \max_{\omega_k \le \omega \le \omega_{k + 1}}
\l\vert \sigma_{\chi, ii'}(\omega_k) - \sigma_{\chi, ii'}(\omega) \r\vert
\\
& + 
\frac{2\pi \max_{1 \le i, i' \le p} \sup_\omega \vert \sigma_{\chi, ii'}(\omega) \vert}{2m + 1} 
\sum_{k = -m}^{m - 1} \max_{\omega_k \le \omega \le \omega_{k + 1}}
\l\vert e^{\iota \omega_k \ell} - e^{\iota \omega \ell} \r\vert 
\\
\le &
\frac{2\pi}{2m + 1} \sum_{k = -m}^{m - 1} 
\l( \l\vert \sigma_{\chi, ii'}(\omega_k) - \sigma_{\chi, ii'}(\omega_k^*) \r\vert +
\l\vert \sigma_{\chi, ii'}(\omega_{k + 1}) - \sigma_{\chi, ii'}(\omega_k^*) \r\vert \r) 
\\
& + 
\frac{2\pi B_\sigma}{2m + 1} 
\sum_{k = -m}^{m - 1} \l( \l\vert e^{\iota \omega_k \ell} - e^{\iota \omega_k^\circ \ell} \r\vert 
+ \l\vert e^{\iota \omega_{k + 1} \ell} - e^{\iota \omega_k^\circ \ell} \r\vert \r)
= III + IV,
\end{align*}
where the last inequality follows from Lemma~\ref{lem:sigma:bound}.
Then by Lemma~\ref{lem:sigma:deriv},
the functions $\omega \mapsto \sigma_{\chi, ii'}(\omega)$ possess derivatives of any order and are of bounded variation such that $III = O(m^{-1})$.
Also from the finiteness of $s$ (Assumption~\ref{assum:common:var}~\ref{cond:block:var})
and as the exponential function has bounded variation,
$IV = O(m^{-1})$ uniformly in $0 \le \ell \le s$, which completes the proof of~\ref{prop:acv:common:one}.
As for~\ref{prop:acv:common:two}, by Proposition~\ref{prop:spec:common}~\ref{prop:spec:common:two},
\begin{align*}
\frac{2\pi}{2m + 1} \sum_{k = -m}^m \max_{1 \le i, i' \le p} \l\vert \wh\sigma_{\chi, ii'}(\omega_k) - \sigma_{\chi, ii'}(\omega_k) \r\vert
= O_P\l(\frac{qp^2}{\inf_\omega \mu_{\chi, q}^2(\omega)}
\l(\vartheta_{n, p} \vee \frac{1}{m} \vee \frac{1}{\sqrt p}\r)\r),
\end{align*}
which replaces the above $I$. Since $II$ is deterministic and the bound derived above is uniform over $1 \le i, i' \le p$,
the proof is complete.
\end{proof}

\subsubsection{Proof of Proposition~\ref{prop:acv:idio}}

Proposition~\ref{prop:acv:idio} is a direct consequence of Corollary~\ref{prop:acv:idio:new}.

\begin{cor}
\label{prop:acv:idio:new}
For $0 \le \ell \le d$, the following statements hold.
\begin{enumerate}[wide, label = (\roman*)]
\item \label{prop:acv:idio:new:one} Uniformly over $1 \le i, i' \le p$, we have
\begin{align*}
\l\vert \wh{\gamma}_{\xi, ii'}(\ell) - \gamma_{\xi, ii'}(\ell) \r\vert
= O_P\l(\frac{q p^2}{\inf_{\omega \in [-\pi, \pi]} \mu_{\chi, q}^2(\omega)} \l(\psi_n \vee \frac{1}{m} \vee \frac{1}{\sqrt p}\r) \r).
\end{align*}

\item \label{prop:acv:idio:new:two} In addition,
\begin{align*}
\max_{1 \le i, i' \le p} \l\vert \wh{\gamma}_{\xi, ii'}(\ell) - \gamma_{\xi, ii'}(\ell) \r\vert
= O_P\l(\frac{q p^2}{\inf_{\omega \in [-\pi, \pi]} \mu_{\chi, q}^2(\omega)} \l(\vartheta_{n, p} \vee \frac{1}{m} \vee \frac{1}{\sqrt p}\r) \r).
\end{align*}
\end{enumerate}
\end{cor}

\begin{proof}
Noting that $\wh{\bm\Sigma}_\xi(\omega) = \wh{\bm\Sigma}_x(\omega) - \wh{\bm\Sigma}_\chi(\omega)$
by Assumption~\ref{assum:innov}~\ref{cond:uncor},
the result in~\ref{prop:acv:idio:new:one} is an immediate consequence of 
Lemma~\ref{lem:acv:x}~\ref{lem:acv:x:one} (see also~\eqref{eq:lem:acv:x:one}) and Proposition~\ref{prop:acv:common}~\ref{prop:acv:common:one}.
Similarly, \ref{prop:acv:idio:new:two} follows from 
Lemma~\ref{lem:acv:x}~\ref{lem:acv:x:two} and Proposition~\ref{prop:acv:common}~\ref{prop:acv:common:two}.
\end{proof}

\subsubsection{Proof of Proposition~\ref{prop:idio:est:lasso}}
\label{pf:prop:idio:est:lasso}

Throughout, we refer to the estimator obtained as in~\eqref{eq:lasso} as $\wh{\bm\beta}^{\las}$ and the tuning parameter by $\lambda^{\las}$.

The problem in~\eqref{eq:lasso} can be solved in parallel as
\begin{align*}
\wh{\bm\beta}^{\las}_{\cdot j} = {\arg\min}_{\mbf m \in \R^{pd}} \
\mbf m^\top \wh{\bbG} \mbf m - 2 \mbf m^\top\wh{\bbg}_{\cdot j} + \lambda^{\las} \vert \mbf m \vert_1,
\end{align*}
with $\wh{\bm\beta}^{\las} = [\wh{\bm\beta}^{\las}_{\cdot j}, \, 1 \le j \le p]$.

Throughout the proof, we condition our arguments on $\mc E_{n, p}$.
We first establish the followings:
\begin{enumerate}[wide, label = (C\arabic*)]
\item \label{cond:bbg} $\wh{\bbG}$ is positive semi-definite.

\item \label{cond:lre} {\bf Lower restricted eigenvalue:} 
$\wh{\bbG}$ satisfies $\mbf v^\top \wh{\bbG} \mbf v \ge \kappa \vert \mbf v \vert^2 - \tau \vert \mbf v \vert_1^2$
for all $\mbf v \in \R^{dp}$, with curvature $\kappa = 2\pi m_\xi$ and tolerance $\tau = C_\xi(\vartheta_{n, p} \vee m^{-1} \vee p^{-1/2})$.

\item \label{cond:dev} {\bf Deviation bound:} 
$\vert \wh{\bbg} - \wh{\bbG} \bm\beta \vert_\infty \le C_\xi (\Vert \bm\beta \Vert_1 + 1)(\vartheta_{n, p} \vee m^{-1} \vee p^{-1/2})$.
\end{enumerate}

To see~\ref{cond:bbg}, note that for any $\mbf a \in \R^p$,
\begin{align*}
\mbf a^\top (2\pi \wh{\bm\Sigma}_x(\omega)) \mbf a &= 
\sum_{\ell = -m}^m K\l(\frac{\ell}{m}\r) \wh{\gamma}_{\mbf a}(\ell) \exp(-\iota \ell \omega)
= \mbf e(\omega)^*
\bmx
\wh{\gamma}_{\mbf a}(0) & \ldots & \wh{\gamma}_{\mbf a}(m - 1) \\
& \ddots & \\
\wh{\gamma}_{\mbf a}(-m + 1) & \ldots & \wh{\gamma}_{\mbf a}(0) 
\emx \mbf e(\omega),
\end{align*}
where $\wh{\gamma}_{\mbf a}(\ell) = n^{-1} \sum_{t = \ell + 1}^n \mbf (a^\top\mbf X_{t - \ell})
(\mbf a^\top\mbf X_t)$
and $\mbf e(\omega) = (e^{-\iota \omega \ell}/\sqrt{m}, \, 1 \le \ell \le m)^\top$.
From the positive semi-definiteness of the sample autocovariance function $\wh{\gamma}_{\mbf a}(\cdot)$
(see e.g.\ \cite{mcleod1984nonnegative}), that of $\wh{\bm\Sigma}_x(\omega)$ follows.
Then, $\wh{\bm\Sigma}_\xi(\omega) = \sum_{j > q} \wh\mu_{x, j} \wh{\mbf e}_{x, j} \wh{\mbf e}_{x, j}^*$
is also positive semi-definite.
Noting that
\begin{align*}
\wh{\bbG} = 
\frac{2\pi}{2m + 1} \sum_{k = -m}^m \bmx
e^{\iota \omega_k \cdot 0} & e^{\iota \omega_k \cdot 1} & \ldots & e^{\iota \omega_k (d - 1)} \\
e^{\iota \omega_k(-1)} & e^{\iota \omega_k \cdot 0} & \ldots & e^{\iota \omega_k (d - 2)} \\
& & \ddots & \\
e^{\iota \omega_k(1 - d)} & e^{\iota \omega_k(2 - d)} & \ldots & e^{\iota \omega_k \cdot 0}
\emx \otimes \wh{\bm\Sigma}_\xi(\omega_k),
\end{align*}
where $\otimes$ denotes the Kronecker product, and that for all $\mbf a \in \C^d$,
\begin{align*}
\mbf a^* \bmx
e^{\iota \omega_k \cdot 0} & e^{\iota \omega_k \cdot 1} & \ldots & e^{\iota \omega_k (d - 1)} \\
e^{\iota \omega_k(-1)} & e^{\iota \omega_k \cdot 0} & \ldots & e^{\iota \omega_k (d - 2)} \\
& & \ddots & \\
e^{\iota \omega_k(1 - d)} & e^{\iota \omega_k(2 - d)} & \ldots & e^{\iota \omega_k \cdot 0}
\emx \mbf a = \l\vert \sum_{j = 1}^d a_j e^{- \iota \omega_k j} \r\vert^2 \ge 0,
\end{align*}
we have $\wh{\bbG}$ positive semi-definite.

To see~\ref{cond:lre}, conditional on $\mc E_{n, p}$, we have
\begin{align}
\label{eq:idio:surrogate}
\max\l( \l\vert \wh{\bbG} - \bbG \r\vert_\infty,
\l\vert \wh{\bbg} - \bbg \r\vert_\infty \r)
= \max_{0 \le \ell \le d} \l\vert \wh{\bm\Gamma}_\xi(\ell)- \bm\Gamma_\xi(\ell) \r\vert_\infty
\le C_\xi \l(\vartheta_{n, p} \vee \frac{1}{m} \vee \frac{1}{\sqrt p}\r).
\end{align}
Then, uniformly over $\mbf v \in \R^{dp}$, we have
\begin{align*}
\mbf v^\top \wh{\bbG} \mbf v 
\ge \mbf v^\top \bbG \mbf v -
\l\vert \mbf v^\top (\wh{\bbG} - \bbG) \mbf v \r\vert
\ge \Lambda_{\min}(\bbG )\vert \mbf v \vert_2^2
- \l\vert \wh{\bbG} - \bbG \r\vert_\infty \vert \mbf v \vert_1^2
\end{align*}
such that on $\mc E_{n, p}$, we can set
$\kappa = \Lambda_{\min}(\bbG) \ge 2\pi m_\xi$ and 
$\tau = C_\xi (\vartheta_{n, p} \vee m^{-1} \vee p^{-1/2})$.
Also from~\eqref{eq:idio:surrogate}, we have
\begin{align*}
\vert \wh{\bbg} - \wh{\bbG} \bm\beta \vert_\infty \le&
\vert \wh{\bbg} - \bbg \vert_\infty +
\vert (\wh{\bbG} - \bbG) \bm\beta \vert_\infty 
\le \l\vert \wh{\bbg} - \bbg \r\vert_\infty +
\Vert \bm\beta \Vert_1 \l\vert \wh{\bbG} - \bbG \r\vert_\infty 
\\
\le& C_\xi(\Vert \bm\beta \Vert_1 + 1)\l(\vartheta_{n, p} \vee \frac{1}{m} \vee \frac{1}{\sqrt p}\r)
\end{align*}
such that~\ref{cond:dev} follows.

Denote by $\mbf v = \wh{\bm\beta}^{\las}_{\cdot j} - \bm\beta_{\cdot j}$ and for any vector $\mbf a = (a_i, \, 1 \le i \le p)^\top$ and a set $\mc A \subset \{1, \ldots, p\}$, we write $\mbf a_{\mc A} = (a_i \mathbb{I}_{\{i \in \mc A\}}, 1 \le i  \le p)^\top$.
Also, let $\mc S_j$ denote the support of $\bm\beta_{\cdot j}$ such that $\vert \mc S_j \vert = s_{0, j}$.
Then by construction, 
\begin{align}
& (\wh{\bm\beta}^{\las}_{\cdot j})^\top \wh{\bbG} \wh{\bm\beta}^{\las}_{\cdot j} - 2 (\wh{\bm\beta}^{\las}_{\cdot j})^\top\wh{\bbg} + \lambda^{\las} \vert \wh{\bm\beta}^{\las}_{\cdot j} \vert_1 \le
\bm\beta_{\cdot j}^\top \wh{\bbG} \bm\beta_{\cdot j} - 2 \bm\beta_{\cdot j}^\top\wh{\bbg} + \lambda^{\las} \vert \bm\beta_{\cdot j}\vert_1, \text{ \ such that}
\nn \\
& \mbf v^\top \wh{\bbG} \mbf v \le \lambda^{\las} \l( \vert \bm\beta_{\cdot j}\vert_1 - \vert \wh{\bm\beta}^{\las}_{\cdot j} \vert_1 \r) - 2 \mbf v^\top \l(\wh{\bbG} \bm\beta_{\cdot j} - \wh{\bbg} \r)
\nn \\
\le & \, \lambda^{\las} \l( \vert \bm\beta_{\cdot j}\vert_1 - \vert \bm\beta_{\cdot j} + \mbf v \vert_1 \r) + \frac{\lambda^{\las}}{2} \vert \mbf v \vert_1
\le \lambda^{\las} \l( \vert \mbf v_{\mc S_j} \vert_1 - \vert \mbf v_{\mc S_j^c} \vert_1 \r) + \frac{\lambda^{\las}}{2} \vert \mbf v \vert_1,
\label{eq:lasso:one}
\end{align}
where the second inequality follows from~\ref{cond:dev} and the choice of $\lambda^{\las}$, and the last from that
\begin{align*}
\vert \bm\beta_{\cdot j} + \mbf v \vert_1 = \vert \bm\beta_{\cdot j} + \mbf v_{\mc S_j} \vert_1 + \vert \mbf v_{\mc S_j^c} \vert_1 \ge \vert \bm\beta_{\cdot j} \vert_1 - \vert \mbf v_{\mc S_j} \vert_1 + \vert \mbf v_{\mc S_j^c} \vert_1.
\end{align*}
From the above and~\ref{cond:bbg}, it follows that
\begin{align*}
& 0 \le \mbf v^\top \wh{\bbG} \mbf v \le \frac{\lambda^{\las}}{2} \l( 3\vert \mbf v_{\mc S_j} \vert_1 - \vert \mbf v_{\mc S_j^c} \vert_1 \r), \\
& \text{i.e. \ } \vert \mbf v_{\mc S_j^c} \vert_1 \le 3\vert \mbf v_{\mc S_j} \vert_1 \text{ \ and \ } \vert \mbf v \vert_1 \le 4 \vert \mbf v_{\mc S_j} \vert_1 \le 4 \sqrt{s_{0, j}} \vert \mbf v \vert_2.
\end{align*}
Combining these observations with~\ref{cond:lre} and~\eqref{eq:lasso:one}, we have
\begin{align*}
6 \lambda^{\las} \sqrt{s_{0, j}} \vert \mbf v \vert_2 &\ge 2\pi m_\xi \vert \mbf v \vert_2^2 - C_\xi \l( \vartheta_{n, p} \vee \frac{1}{m} \vee \frac{1}{\sqrt p} \r) \vert \mbf v \vert_1^2 
\\
&\ge 2 \pi m_\xi \vert \mbf v \vert_2^2 \l( 1 - \frac{8 C_\xi s_{0, j} (\vartheta_{n, p} \vee m^{-1} \vee p^{-1/2})}{\pi m_\xi} \r)
\ge \pi m_\xi \vert \mbf v \vert_2^2
\end{align*}
from the condition imposed on $s_{\text{\upshape in}} \ge s_{0, j}$. 
Hence, for all $1 \le j \le p$, we have
\begin{align*}
\l\vert \wh{\bm\beta}^{\las}_{\cdot j} - \bm\beta_{\cdot j} \r\vert_\infty \le
\l\vert \wh{\bm\beta}^{\las}_{\cdot j} - \bm\beta_{\cdot j} \r\vert_2 \le \frac{6 \sqrt{s_{0, j}} \lambda^{\las}}{\pi m_\xi}, \quad
\l\vert \wh{\bm\beta}^{\las}_{\cdot j} - \bm\beta_{\cdot j} \r\vert_1 \le \frac{24 s_{0, j} \lambda^{\las}}{\pi m_\xi}.
\end{align*}
Besides, by Karush-Kuhn-Tucker conditions, we have 
\begin{align*}
\mbf 0 \in \wh{\bbG} \wh{\bm\beta}^{\las}_{\cdot j} - \wh{\bbg} + \lambda^{\las} \mathrm{sgn} (\wh{\bm\beta}^{\las}_{\cdot j})
\end{align*}
where, defined as
\begin{align*}
\mathrm{sgn}(x) = \begin{cases}
1 & \text{ if } x > 0,\\
-1 & \text{ if } x < 0,\\
[-1, 1] & \text{ if } x = 0,
\end{cases}
\end{align*}
$\mathrm{sgn}(\cdot)$ acts element-wise. It implies that $\l\vert \wh{\bbG} \wh{\bm\beta}^{\las}_{\cdot j} - \wh{\bbg} \r\vert_\infty \le \lambda^{\las}$.
Also noting that from~\eqref{eq:lasso:one}, we have
\begin{align*}
0 \le \vert \bm\beta_{\cdot j} \vert_1 - \vert \wh{\bm\beta}_{\cdot j} \vert_1 + \frac{1}{2} \l( \vert \bm\beta_{\cdot j} \vert_1 + \vert \wh{\bm\beta}_{\cdot j} \vert_1 \r), \text{ \ i.e. \ }
\vert \wh{\bm\beta}_{\cdot j} \vert_1 \le 3 \vert \bm\beta_{\cdot j} \vert_1.
\end{align*}
Therefore, we obtain from~\ref{cond:dev},
\begin{align*}
& \l\vert \wh{\bm\beta}^{\las}_{\cdot j} - \bm\beta^{\las}_{\cdot j} \r\vert_\infty
= \l\vert \bbG^{-1} \l(\bbG \wh{\bm\beta}^{\las}_{\cdot j} - \wh{\bbG} \wh{\bm\beta}^{\las}_{\cdot j}
+ \wh{\bbG} \wh{\bm\beta}^{\las}_{\cdot j} -\wh{\bbg}_{\cdot j} + \wh{\bbg}_{\cdot j} - \bbg_{\cdot j} \r) \r\vert_\infty
\nn \\
& \le \l\Vert \bbG^{-1} \r\Vert_1 
\l(
\l\vert \bbG - \wh{\bbG} \r\vert_\infty \l\vert \wh{\bm\beta}^{\las}_{\cdot j} \r\vert_1
+ \l\vert \wh{\bbG} \wh{\bm\beta}^{\las}_{\cdot j} -\wh{\bbg}_{\cdot j} \r\vert_\infty
+ \l\vert \wh{\bbg}_{\cdot j} - \bbg_{\cdot j} \r\vert_\infty \r) 
\le 4 \l\Vert \bbG^{-1} \r\Vert_1 \lambda^{\las}
\end{align*}
which, together with the trivial bound $\vert \mbf v \vert_\infty \le \vert \mbf v \vert_2$, proves the final claim.

\subsubsection{Proof of Proposition~\ref{prop:idio:est:dantzig}}

As in Appendix~\ref{pf:prop:idio:est:lasso}, we condition our arguments on $\mc E_{n, p}$.
Note that solving~\eqref{eq:ds} is equivalent to solving
\begin{align*}
\wh{\bm\beta}^{\ds}_{\cdot j} = {\arg\min}_{\mbf m \in \R^{pd}} \vert \mbf m \vert_1
\quad \text{subject to} \quad
\l\vert \wh{\bbG} \mbf m - \wh{\bbg}_{\cdot j} \r\vert_\infty \le \lambda^{\ds}
\quad \text{for } 1 \le j \le p,
\end{align*}
(see e.g.\ Lemma~1 of \cite{cai2011}).
Then, $\bm\beta$ is a feasible solution to~\eqref{eq:ds} from~\ref{cond:dev},
and therefore $\vert \wh{\bm\beta}^{\ds}_{\cdot j} \vert_1 \le \vert \bm\beta_{\cdot j} \vert_1  \le \Vert \bm\beta \Vert_1$
for all $1 \le j \le p$.
Also, writing $\mbf v = \wh{\bm\beta}^{\ds}_{\cdot j} - \bm\beta_{\cdot j}$, we have
\begin{align*}
\mbf v^\top \wh{\bbG} \mbf v \le \vert \wh{\bbG} \mbf v \vert_\infty \vert \mbf v \vert_1 \le 2 \lambda^{\ds} \vert \mbf v \vert_1.
\end{align*}
Then using~\ref{cond:bbg}--\ref{cond:dev} and the arguments analogous to those adopted in Appendix~\ref{pf:prop:idio:est:lasso}, we obtain the results with slightly different multiplicative constants.
In particular,
\begin{align}
& \l\vert \wh{\bm\beta}^{\ds}_{\cdot j} - \bm\beta_{\cdot j} \r\vert_\infty
= \l\vert \bbG^{-1} \l(\bbG \wh{\bm\beta}^{\ds}_{\cdot j} - \wh{\bbG} \wh{\bm\beta}^{\ds}_{\cdot j}
+ \wh{\bbG} \wh{\bm\beta}^{\ds}_{\cdot j} -\wh{\bbg}_{\cdot j} + \wh{\bbg}_{\cdot j} - \bbg_{\cdot j} \r) \r\vert_\infty
\nn \\
& \le \l\Vert \bbG^{-1} \r\Vert_1 
\l(
\l\vert \bbG - \wh{\bbG} \r\vert_\infty \l\vert \wh{\bm\beta}^{\ds}_{\cdot j} \r\vert_1
+ \l\vert \wh{\bbG} \wh{\bm\beta}^{\ds}_{\cdot j} -\wh{\bbg}_{\cdot j} \r\vert_\infty
+ \l\vert \wh{\bbg}_{\cdot j} - \bbg_{\cdot j} \r\vert_\infty \r) 
\le 2 \l\Vert \bbG^{-1} \r\Vert_1 \lambda^{\ds}
\label{eq:prop:idio:est:dantzig:one}
\end{align}
for all $1 \le j \le p$.
In fact, without requiring exact sparsity, we can control the estimation error in $\ell_1$-norm.
Let $\wt\lambda \in (0, \infty)$ be a threshold level to be defined later, and denote by
\begin{align*}
\wt s_j = \sum_{i = 1}^{pd} \min\l( \frac{\vert \beta_{ij} \vert}{\wt\lambda}, 1 \r) \quad \text{and} \quad
\wt{\mc S}_j = \l\{ 1 \le i \le pd: \, \vert \beta_{ij} \vert > \wt\lambda \r\}.
\end{align*}
By definition, it is easily seen that
\begin{align}
\wt s_j = \vert \wt{\mc S}_j \vert + \sum_{i \in \wt{\mc S}^c_j} \frac{\vert \beta_{ij} \vert}{\wt\lambda} \ge \vert \wt{\mc S}_j \vert.
\label{eq:prop:idio:est:dantzig:two}
\end{align}
Then for all $1 \le j \le p$, 
\begin{align*}
& \l\vert \wh{\bm\beta}^{\ds}_{\cdot j} - \bm\beta_{\cdot j} \r\vert_1
\le \l\vert \wh{\bm\beta}^{\ds}_{\wt{\mc S}_j^c, j} \r\vert_1 + \l\vert \bm\beta_{\wt{\mc S}_j^c, j}  \r\vert_1
+ \l\vert \wh{\bm\beta}^{\ds}_{\wt{\mc S}_j, j} - \bm\beta_{\wt{\mc S}_j, j}  \r\vert_1
= \l\vert \wh{\bm\beta}^{\ds}_{\cdot j} \r\vert_1 - 
\l\vert \wh{\bm\beta}^{\ds}_{\wt{\mc S}_j, j} \r\vert_1 + \l\vert \bm\beta_{\wt{\mc S}_j^c, j}  \r\vert_1
+ \l\vert \wh{\bm\beta}^{\ds}_{\wt{\mc S}_j, j} - \bm\beta_{\wt{\mc S}_j, j}  \r\vert_1
\\
& \le \l\vert \bm\beta_{\cdot j} \r\vert_1 - 
\l\vert \wh{\bm\beta}^{\ds}_{\wt{\mc S}_j, j} \r\vert_1 + \l\vert \bm\beta_{\wt{\mc S}_j^c, j}  \r\vert_1
+ \l\vert \wh{\bm\beta}^{\ds}_{\wt{\mc S}_j, j} - \bm\beta_{\wt{\mc S}_j, j}  \r\vert_1
\le 2\l\vert \bm\beta_{\wt{\mc S}^c_j, j}  \r\vert_1 + 
2 \l\vert \wh{\bm\beta}^{\ds}_{\wt{\mc S}_j, j} - \bm\beta_{\wt{\mc S}_j, j}  \r\vert_1
\\
& \le 2 \wt\lambda \wt s_j + 4 \vert \wt{\mc S}_j \vert \Vert \bbG^{-1} \Vert_1 \lambda^{\ds}
\le 2 \wt s_j(\wt\lambda + 2 \Vert \bbG^{-1} \Vert_1 \lambda^{\ds}),
\end{align*}
where the second inequality follows from that $\vert \wh{\bm\beta}^{\ds}_{\cdot j} \vert_1 \le \vert \bm\beta_{\cdot j} \vert_1$,
the third from the triangular inequality,
the fourth from~\eqref{eq:prop:idio:est:dantzig:one},
and the last from~\eqref{eq:prop:idio:est:dantzig:two}.
Setting $\wt\lambda = \Vert \bbG^{-1} \Vert_1 \lambda^{\ds}$, by the definition of $s_{0, j}(\varrho)$, we have
\begin{align*}
\wt\lambda \wt s_j 
\le \wt\lambda \sum_{i = 1}^{pd} \min\l(\frac{\vert \beta_{ij} \vert^\varrho}{\wt\lambda^\varrho}, 1 \r) \le s_{0, j}(\varrho) \wt\lambda^{1 - \varrho}.
\end{align*}
Therefore,
\begin{align*}
\l\vert \wh{\bm\beta}^{\ds} - \bm\beta \r\vert_1
= \sum_{j = 1}^p \l\vert \wh{\bm\beta}^{\ds}_{\cdot j} - \bm\beta_{\cdot j} \r\vert_1
\le 6 \sum_{j = 1}^p s_{0, j}(\varrho) \wt\lambda^{1 - \rho} 
= 6 s_0(\varrho) \l(\Vert \bbG^{-1} \Vert_1 \lambda^{\ds}\r)^{1 - \varrho}.
\end{align*}

\subsubsection{Proof of Proposition~\ref{prop:idio:delta}~\ref{prop:idio:delta:one}}

We prove the following proposition from which Proposition~\ref{prop:idio:delta}~\ref{prop:idio:delta:one} follows.
Throughout, we refer to the estimator obtained as in~\eqref{eq:lasso} as $\wh{\bm\beta}^{\las}$.

\begin{prop}
\label{prop:idio:innov:cov:new}
{\it
\begin{enumerate}[wide, label = (\roman*)]
\item \label{prop:idio:innov:cov:one}
Suppose that Assumption~\ref{assum:idio} is met 
and the estimator $\wh{\bm\beta}$ fulfils
\begin{align}
\label{cond:prop:idio:innov:cov:one}
\l\Vert \wh{\bm\beta} - \bm\beta \r\Vert_1 \le C_\beta s_{\text{\upshape in}} \zeta_{n, p}
\end{align}
with some constant $C_\beta > 0$.
Then on $\mc E_{n, p}$ defined in~\eqref{eq:idio:set}, 
there exists a large enough constant $C > 0$ such that 
\begin{align*}
\l\vert \wh{\bm\Gamma} - \bm\Gamma \r\vert_\infty
\le C\phi_{n, p}
\quad \text{with}
\quad \phi_{n, p} = \l( s_{\text{\upshape in}} \zeta_{n, p} \vee 
\Vert \bm\beta \Vert_1 \l(\vartheta_{n, p} \vee \frac{1}{m} \vee \frac{1}{\sqrt p} \r) \r).
\end{align*}

\item \label{prop:idio:innov:cov:two}
Further, suppose that $\bm\Delta$ belongs to the parameter space
\begin{align*}
\mc D(\varrho, s_\delta(\varrho)) = \l\{ \mbf M = [m_{ii'}] \in \R^{p \times p}: \, 
\mbf M \text{ is positive definite}, \,
\max_{1 \le i \le p} \sum_{i' = 1}^p \vert m_{ii'} \vert^\varrho \le s_\delta(\varrho) \r\}
\end{align*}
for some $\varrho \in [0, 1)$,
and set $\eta = C \Vert \bm\Delta \Vert_1 \phi_{n, p}$.
Then on $\mc E_{n, p}$, we have
\begin{align*}
\l\vert \wh{\bm\Delta} - \bm\Delta \r\vert_\infty \le 4 \Vert \bm\Delta \Vert_1 \eta
\quad \text{and} \quad 
\l\Vert \wh{\bm\Delta} - \bm\Delta \r\Vert \le 12 s_\delta(\varrho)
\l(4 \Vert \bm\Delta \Vert_1 \eta\r)^{1 - \varrho}.
\end{align*}
\end{enumerate}}
\end{prop}

\begin{rem}
\label{rem:zeta}
From the proof of Proposition~\ref{prop:idio:est:lasso} and~\ref{prop:idio:est:dantzig}, both $\wh{\bm\beta}^{\las}$ and $\wh{\bm\beta}^{\ds}$ satisfy~\eqref{cond:prop:idio:innov:cov:one} with $\zeta_{n, p} = \lambda^{\las}$ and $\zeta_{n, p} = \Vert \bbG^{-1} \Vert_1 \lambda^{\ds}$, respectively.
Then, $\phi_{n, p} $ involved in the error bound in Proposition~\ref{prop:idio:innov:cov:new}~\ref{prop:idio:innov:cov:one} reduces to $\phi_{n, p} = s_{\text{\upshape in}} \zeta_{n, p}$.
\end{rem}

\begin{proof}[Proof of Proposition~\ref{prop:idio:innov:cov:new}~\ref{prop:idio:innov:cov:one}]
We write $\sum_{\ell = 1}^d \mbf A_{\ell} \bm\Gamma_\xi(\ell) = \bm\beta^\top \bbg$
and $\sum_{\ell = 1}^d \wh{\mbf A}_{\ell} \wh{\bm\Gamma}_\xi(\ell) = \wh{\bm\beta}^\top \wh{\bbg}$.
Then on $\mc E_{n, p}$,
we have $\vert \wh{\bbg} \vert_\infty \le \max_{0 \le \ell \le d} \vert \bm\Gamma_\xi(\ell) \vert_\infty + C_\xi(\vartheta_{n, p} \vee m^{-1} \vee p^{-1/2}) \le C_{\Xi, \varsigma, \vep}$ 
for a large enough $C_{\Xi, \varsigma, \vep} > 0$ and
by Lemma~\ref{lem:decay}.
Also,
\begin{align*}
\l\vert \wh{\bm\Gamma} - \bm\Gamma \r\vert_\infty
\le 
\l\vert \wh{\bm\Gamma}_\xi(0) - \bm\Gamma_\xi(0) \r\vert_\infty
+ \l\vert \wh{\bm\beta}^\top \wh{\bbg} - \bm\beta^\top \bbg \r\vert_\infty
=: I + II,
\end{align*}
where $I \le C_\xi(\vartheta_{n, p} \vee m^{-1} \vee p^{-1/2})$ and
\begin{align*}
II &\le
\l\vert \l(\wh{\bm\beta} - \bm\beta \r)^\top \wh{\bbg} \r\vert_\infty
+ \l\vert \bm\beta^\top \l(\wh{\bbg} - \bbg\r)\r\vert_\infty
\\
&\le \l\Vert \wh{\bm\beta} - \bm\beta \r\Vert_1 \; \max_{0 \le \ell \le d} \l\vert \wh{\bm\Gamma}_\xi(\ell) \r\vert_\infty
+ \Vert \bm\beta \Vert_1 \; \max_{0 \le \ell \le d} \l\vert \wh{\bm\Gamma}_\xi(\ell) - \bm\Gamma_\xi(\ell) \r\vert_\infty
\\
&\le C \l( s_{\text{\upshape in}} \zeta_{n, p} 
+ \Vert \bm\beta \Vert_1 \l(\vartheta_{n, p} \vee \frac{1}{m} \vee \frac{1}{\sqrt p} \r) \r)
\end{align*}
for large enough constant $C > 0$.
\end{proof}

\begin{proof}[Proof of Proposition~\ref{prop:idio:innov:cov:new}~\ref{prop:idio:innov:cov:two}]
The proof takes analogous steps as those in the proof of Theorem~6 of \cite{cai2011}.
On $\mc E_{n, p}$, by Proposition~\ref{prop:idio:innov:cov:new}~\ref{prop:idio:innov:cov:one}, we have
\begin{align*}
\l\vert \wh{\bm\Gamma} \bm\Delta - \mbf I \r\vert_\infty
\le \l\vert \l(\wh{\bm\Gamma} - \bm\Gamma\r) \bm\Delta \r\vert_\infty
\le \Vert \bm\Delta \Vert_1 \l\vert \wh{\bm\Gamma} - \bm\Gamma\r\vert_\infty \le 
\eta = C \Vert \bm\Delta \Vert_1 \phi_{n, p},
\end{align*}
and thus $\bm\Delta$ is a feasible solution to~\eqref{eq:inv:delta}.
Since solving the problem in~\eqref{eq:inv:delta} is equivalent to solving
\begin{align*}
\check{\bm\Delta}_{\cdot j} = {\arg\min}_{\mbf m \in \R^p} \vert \mbf m \vert_1 \quad 
\text{subject to} \quad 
\l\vert \wh{\bm\Gamma} \mbf m - \bm\varphi_j \r\vert_\infty \le \eta
\end{align*}
for all $1 \le j \le p$, it follows that 
$\Vert \check{\bm\Delta}\Vert_1 = \max_{1 \le j \le p} \vert \check{\bm\Delta}_{\cdot j} \vert_1 \le \Vert \bm\Delta \Vert_1$.
Then,
\begin{align*}
\l\vert \wh{\bm\Gamma}\l( \check{\bm\Delta} - \bm\Delta \r) \r\vert_\infty
&\le \l\vert \wh{\bm\Gamma}\check{\bm\Delta} - \mbf I \r\vert_\infty
+ \l\vert \l(\wh{\bm\Gamma} - \bm\Gamma\r) \bm\Delta \r\vert_\infty
\le \eta + \Vert \bm\Delta \Vert_1 \l\vert \wh{\bm\Gamma} - \bm\Gamma\r\vert_\infty \le 2\eta,
\\
\therefore \, \l\vert \bm\Gamma\l( \check{\bm\Delta} - \bm\Delta \r) \r\vert_\infty
&\le \l\vert \wh{\bm\Gamma}\l( \check{\bm\Delta} - \bm\Delta \r) \r\vert_\infty
+ \l\vert \l(\wh{\bm\Gamma} - \bm\Gamma\r) \l( \check{\bm\Delta}- \bm\Delta \r) \r\vert_\infty
\\
&\le 2 \eta +
\l\Vert \check{\bm\Delta} - \bm\Delta  \r\Vert_1
\l\vert \wh{\bm\Gamma} - \bm\Gamma\r\vert_\infty
\le 4 \eta.
\end{align*}
Therefore, it follows that
\begin{align*}
\l\vert \check{\bm\Delta} - \bm\Delta  \r\vert_\infty
\le \l\Vert \bm\Delta  \r\Vert_1 \;
\l\vert \bm\Gamma\l( \check{\bm\Delta} - \bm\Delta \r) \r\vert_\infty
\le 4 \l\Vert \bm\Delta  \r\Vert_1 \eta,
\end{align*}
and from the definition of $\wh{\bm\Delta}$, 
the bound on $\vert \wh{\bm\Delta} - \bm\Delta  \vert_\infty$ follows.

Next, let $\varpi = \vert \wh{\bm\Delta} - \bm\Delta \vert_\infty$
and define $\mbf h_j = \wh{\bm\Delta}_{\cdot j} - \bm\Delta_{\cdot j}$,
$\mbf h^{(1)}_j = (\wh\delta_{ij} \mathbb I_{\{\vert \wh\delta_{ij} \vert \ge 2\varpi\}}, \, 1 \le i \le p)^\top - \bm\Delta_{\cdot j}$
and $\mbf h^{(2)}_j = \mbf h_j - \mbf h^{(1)}_j$.
By definition, we have $\vert \wh{\bm\Delta}_{\cdot j} \vert_1 \le \vert \check{\bm\Delta}_{\cdot j} \vert_1 \le \vert \bm\Delta_{\cdot j} \vert_1$ for all $1 \le j \le p$.
Then,
\begin{align*}
\vert \bm\Delta_{\cdot j} \vert_1 - \vert \mbf h^{(1)}_j \vert_1 + \vert \mbf h^{(2)}_j \vert_1
\le \vert \bm\Delta_{\cdot j} + \mbf h^{(1)}_j \vert_1 + \vert \mbf h^{(2)}_j \vert_1
= \vert \wh{\bm\Delta}_{\cdot j} \vert_1 \le \vert \bm\Delta_{\cdot j} \vert_1,
\end{align*}
which implies that $\vert \mbf h^{(2)}_j \vert_1 \le \vert \mbf h^{(1)}_j \vert_1$
and thus $\vert \mbf h_j \vert_1 \le 2 \vert \mbf h^{(1)}_j \vert_1$.
The latter is bounded as
\begin{align*}
\vert \mbf h^{(1)}_j \vert_1 
=& \sum_{i = 1}^p \l\vert \wh\delta_{ij} \mathbb I_{\{\vert \wh\delta_{ij} \vert \ge 2\varpi\}} - \delta_{ij} \r\vert_1
\le 
\sum_{i = 1}^p \l\vert \delta_{ij} \mathbb I_{\{\vert \delta_{ij} \vert < 2\varpi\}} \r\vert_1
+
\\
&\sum_{i = 1}^p \l\vert \l(\wh\delta_{ij} - \delta_{ij}\r) 
\mathbb I_{\{\vert \wh\delta_{ij} \vert \ge 2\varpi\}}
+
\delta_{ij} \l(\mathbb I_{\{\vert \wh\delta_{ij} \vert \ge 2\varpi\}} -
\mathbb I_{\{\vert \delta_{ij} \vert \ge 2\varpi\}} \r)\r\vert_1
\\
\le& s_\delta(\varrho) (2\varpi)^{1 - \varrho} + 
\varpi \sum_{i = 1}^p \mathbb I_{\{\vert \wh\delta_{ij} \vert \ge 2\varpi\}}
+ \sum_{i = 1}^p \vert \delta_{ij} \vert \l\vert
\mathbb I_{\{\vert \wh\delta_{ij} \vert \ge 2\varpi\}} -
\mathbb I_{\{\vert \delta_{ij} \vert \ge 2\varpi\}} \r\vert_1
\\
\le& s_\delta(\varrho) (2\varpi)^{1 - \varrho} + 
\varpi \sum_{i = 1}^p \mathbb I_{\{\vert \delta_{ij} \vert \ge \varpi\}}
+ \sum_{i = 1}^p \vert \delta_{ij} \vert
\mathbb I_{\{\vert \vert \delta_{ij} \vert  - 2\varpi \vert \le \vert \wh\delta_{ij} - \delta_{ij} \vert\}}
\\
\le& s_\delta(\varrho) \varpi^{1 - \varrho} \l(1 + 2^{1 - \varrho} + 3^{1 - \varrho} \r),
\end{align*}
from which we drive that
\begin{align*}
\l\Vert \wh{\bm\Delta} - \bm\Delta \r\Vert
\le \l\Vert \wh{\bm\Delta} - \bm\Delta \r\Vert_1
\le 12s_\delta(\varrho) \varpi^{1 - \varrho}.
\end{align*}
\end{proof}

\subsubsection{Proof of Proposition~\ref{prop:idio:delta}~\ref{prop:idio:delta:two}}

In what follows, we omit the dependence of $\wh{\mbf A}_\ell$ on $\mathfrak{t}$.
Write
\begin{align*}
& \frac{1}{2\pi} \l\vert \wh{\bm\Omega} - \bm\Omega \r\vert_\infty 
\le 
\\
& \l\vert \l(\wh{\mc A}(1) - \mc A(1) \r)^\top \wh{\bm\Delta} \mc A(1) \r\vert_\infty
+
\l\vert \mc A^\top(1) \l(\wh{\bm\Delta} - \bm\Delta \r) \mc A(1) \r\vert_\infty
+
\l\vert \wh{\mc A}^\top(1) \wh{\bm\Delta} \l(\wh{\mc A}(1) - \mc A(1) \r) \r\vert_\infty
\\
&:= I + II + III.
\end{align*}
Note that from Corollaries~\ref{cor:idio:est:dir:lasso} and~\ref{cor:idio:est:dir:ds} and the conditions of the proposition, 
\begin{align*}
\l\Vert \wh{\mc A}(1) - \mc A(1) \r\Vert_1
& \le \max_{1 \le j \le p}
\sum_{\ell = 1}^d \sum_{i = 1}^p \l\vert \wh{A}_{\xi, \ell, ij} - A_{\ell, ij} \r\vert_1
\le s_{\text{\upshape out}} \mathfrak{t},
\\
\l\Vert \wh{\mc A}(1) \r\Vert_1 
& \le \l\Vert \mc A(1) \r\Vert_1 + \l\Vert \wh{\mc A}(1) - \mc A(1) \r\Vert_1 \le 2 \l\Vert \mc A(1) \r\Vert_1,
\\
\vert \wh{\bm\Delta} \vert_\infty 
& \le \vert \bm\Delta \vert_\infty + \vert \wh{\bm\Delta} - \bm\Delta \vert_\infty
\le \Vert \bm\Delta \Vert + 4 \Vert \bm\Delta \Vert_1 \eta.
\end{align*}
Also, for any compatible matrices $\mbf A$, $\mbf B$ and $\mbf C$,
we have $\vert \mbf A\mbf B\mbf C \vert_\infty \le \Vert \mbf A \Vert_\infty \vert \mbf B \vert_\infty \Vert \mbf C \Vert_1$.
Then,
\begin{align*}
I &\le \l\Vert \wh{\mc A}(1) - \mc A(1) \r\Vert_1 \l( \Vert \bm\Delta \Vert + 4 \Vert \bm\Delta \Vert_1 \eta \r) \Vert \mc A(1) \Vert_1 \le \Vert \mc A(1) \Vert_1 \l( \Vert \bm\Delta \Vert + 4 \Vert \bm\Delta \Vert_1 \eta \r) s_{\text{\upshape out}} \mathfrak{t},
\\
II &\le 4 \Vert \mc A(1) \Vert_1^2 \; \Vert \bm\Delta \Vert_1 \eta,
\\
III &\le 2\Vert \mc A(1) \Vert_1 \l( \Vert \bm\Delta \Vert + 4 \Vert \bm\Delta \Vert_1 \eta \r)
\l\Vert \wh{\mc A}(1) - \mc A(1) \r\Vert_1 
\le 2 \Vert \mc A(1) \Vert_1 \l( \Vert \bm\Delta \Vert + 4 \Vert \bm\Delta \Vert_1 \eta \r) s_{\text{\upshape out}} \mathfrak{t},
\end{align*}
and therefore 
\begin{align*}
I + II + III \le & \, \Vert \mc A(1) \Vert_1 \l[ 3\l( \Vert \bm\Delta \Vert + 4 \Vert \bm\Delta \Vert_1 \eta \r) s_{\text{\upshape out}} \mathfrak{t} + 4 \Vert \mc A(1) \Vert_1 \Vert \bm\Delta \Vert_1 \eta \r]
\\
\le& \, \Vert \mc A(1) \Vert_1 \l( 3 \Vert \bm\Delta \Vert  s_{\text{\upshape out}} \mathfrak{t} + 16 \Vert \mc A(1) \Vert_1 \Vert \bm\Delta \Vert_1 \eta \r),
\end{align*}
which concludes the proof.

\subsubsection{Proof of Remark~\ref{rem:static}}
\label{sec:rem:static}

In this section, we operate under the restricted GDFM in Section~\ref{sec:2005} 
and use the notations given therein.
The proofs arguments are analogous to those adopted in the proof of Proposition~\ref{prop:spec:common}~\ref{prop:spec:common:two}.

With $\wh{\mbf S} = \wh{\bm\Gamma}_x(0)$ and $\mbf S = \bm\Gamma_\chi(0)$,
the conditions~\eqref{cond:s:one}--\eqref{cond:s:four} are met by
Assumption~\ref{assum:static}~\ref{cond:static:linear}, Propositions~\ref{prop:idio:eval} and~\ref{lem:acv:x} 
with $\zeta_{n, p} = 1/\sqrt{n}$ and $\bar{\zeta}_{n, p} = \wt\vartheta_{n, p}$.
Then applying Lemmas~\ref{lem:one}--\ref{lem:three}, we obtain
\begin{align*}
\frac{1}{p} \max_{1 \le j \le r} \vert \wh\mu_{x, j} - \mu_{\chi, j} \vert &= O_P\l(\frac{1}{\sqrt n} \vee \frac{1}{p} \r),
\\
\sqrt{p} \cdot \frac{\mu_{\chi, r}^2}{p^2} \max_{1 \le i \le p} \l\vert \bm\varphi_i^\top \l(\wh{\mbf E}_x - \mbf E_\chi \bm{\mc O}\r) \r\vert_2 &= O_P\l(\wt{\vartheta}_{n, p} \vee \frac{1}{\sqrt p} \r)
\end{align*}
with some unitary, diagonal matrix $\bm{\mc O} \in \R^{r \times r}$.
Also using the arguments in the proof of Lemma~\ref{lem:evec:size}, we have\begin{align}
\label{eq:evec:size}
\max_{1 \le j \le r} \mu_{\chi, j}^{1/2} \max_{1 \le i \le p} \vert e_{\chi, ij} \vert = O(1) \quad \text{and} \quad
\max_{1 \le j \le r} \wh\mu_{\chi, j}^{1/2} \max_{1 \le i \le p} \vert \wh{e}_{x, ij} \vert = O_P(1).
\end{align}
Noting that $\wh{\bm\Gamma}_\chi(h) = \wh{\mbf E}_x \wh{\mbf E}_x^\top \wh{\bm\Gamma}_x(h) \wh{\mbf E}_x \wh{\mbf E}_x^\top$
and that $\bm\Gamma_\chi(h) = \mbf E_\chi\mbf E_\chi^\top \bm\Gamma_\chi(h) \mbf E_\chi\mbf E_\chi^\top$, we have
\begin{align*}
\l\vert \wh{\bm\Gamma}_\xi(h) - \bm\Gamma_\xi(h) \r\vert_\infty
\le & \l\vert \wh{\bm\Gamma}_x(h) - \bm\Gamma_x(h) \r\vert_\infty
+ \l\vert \wh{\mbf E}_x \wh{\mbf E}_x^\top \wh{\bm\Gamma}_x(h) \wh{\mbf E}_x \wh{\mbf E}_x^\top - \mbf E_\chi\mbf E_\chi^\top \bm\Gamma_x(h) \mbf E_\chi\mbf E_\chi^\top \r\vert_\infty
\\
& +
\l\vert \mbf E_\chi\mbf E_\chi^\top \bm\Gamma_\xi(h) \mbf E_\chi\mbf E_\chi^\top \r\vert_\infty = I + II + III,
\end{align*}
where $I = O_P(\wt\vartheta_{n, p} \vee p^{-1/2})$ from Lemma~\ref{lem:acv:x}.
As for $III$, we have 
\begin{align*}
III = \max_{1 \le i, i' \le p} \l\vert \bm\varphi_i^\top \mbf E_\chi\mbf E_\chi^\top \bm\Gamma_\xi(h) \mbf E_\chi\mbf E_\chi^\top \bm\varphi_{i'} \r\vert
= O_P\l( \frac{r}{\mu_{\chi, r}} \r)
\end{align*}
from~\eqref{eq:evec:size} and Proposition~\ref{prop:idio:eval}. Next,
\begin{align*}
II \le& \max_{i, i'} \l\vert \bm\varphi_i^\top\l( \wh{\mbf E}_x \wh{\mbf E}_x^\top - \mbf E_\chi\mbf E_\chi^\top\r) \wh{\bm\Gamma}_x(h) \wh{\mbf E}_x \wh{\mbf E}_x^\top \bm\varphi_{i'} \r\vert +
\max_{i, i'} \l\vert \bm\varphi_i^\top \mbf E_\chi\mbf E_\chi^\top
\l( \wh{\bm\Gamma}_x(h) - \bm\Gamma_x(h) \r) \wh{\mbf E}_x \wh{\mbf E}_x^\top \bm\varphi_{i'} \r\vert 
\\ &+
\max_{i, i'} \l\vert \bm\varphi_i^\top \mbf E_\chi\mbf E_\chi^\top
\bm\Gamma_x(h) \l( \wh{\mbf E}_x \wh{\mbf E}_x^\top - \mbf E_\chi\mbf E_\chi^\top\r) \bm\varphi_{i'} \r\vert
= IV + V + VI.
\end{align*}
Denoting the diagonal elements of $\bm{\mc O}$ by $O_{jj}, \, 1 \le j \le r$, we have
\begin{align*}
\max_i \l\vert \bm\varphi_i^\top\l( \wh{\mbf E}_x - \mbf E_\chi \bm{\mc O} \r) \wh{\mbf E}_x^\top \r\vert_2 \le 
\max_i \sqrt{ \sum_{i' = 1}^p \l(\sum_{j = 1}^r (\wh{e}_{x, ij} - O_{jj} e_{\chi, ij} ) \wh e_{x, i'j} \r)^2}
\\
\le \sqrt{  \max_i \sum_{j = 1}^r (\wh{e}_{x, ij} - O_{jj} e_{\chi, ij} )^2 \cdot 
\sum_{j = 1}^r \sum_{i' = 1}^p \wh e_{x, i'j}^2} 
= O_P\l( \frac{r p^{3/2}}{\mu_{\chi, r}^2} \l(\wt\vartheta_{n, p} \vee \frac{1}{\sqrt p}\r) \r)
\end{align*}
such that
\begin{align*}
\max_i \l\vert \bm\varphi_i^\top\l( \wh{\mbf E}_x \wh{\mbf E}_x^\top - \mbf E_\chi\mbf E_\chi^\top\r) \r\vert_2 &\le 
\max_i \l\vert \bm\varphi_i^\top\l( \wh{\mbf E}_x - \mbf E_\chi \bm{\mc O} \r) \wh{\mbf E}_x^\top \r\vert_2 +
\max_i \l\vert \bm\varphi_i^\top \mbf E_\chi \bm{\mc O}  \l( \wh{\mbf E}_x - \mbf E_\chi \bm{\mc O} \r)^\top \r\vert_2
\\
&= O_P\l( \frac{r p^{3/2}}{\mu_{\chi, r}^2} \l(\wt\vartheta_{n, p} \vee \frac{1}{\sqrt p}\r) \r).
\end{align*}
From this, \eqref{eq:evec:size} and Lemma~\ref{lem:acv:x}~\ref{lem:acv:x:one}, we have
\begin{align*}
IV = O_P\l( \frac{r p^{3/2}}{\mu_{\chi, r}^2} \l(\wt\vartheta_{n, p} \vee \frac{1}{\sqrt p}\r) \cdot \frac{p}{\sqrt{\mu_{\chi, r}}} \r) \quad \text{and} \quad
V = O_P\l( \frac{r p}{\mu_{x, r}} \cdot \l(\frac{1}{\sqrt n} \vee \frac{1}{\sqrt p}\r) \r)
\end{align*}
(assuming that $(p / \mu_{\chi, r})^{-5/2} (\wt\vartheta_{n, p} \vee p^{-1/2}) = o(1)$), and $VI$ is similarly handled as $IV$.
Finally, assuming that the factor strengths are $\varrho_1 = \ldots = \varrho_r = 1$
(in Assumption~\ref{assum:static}~\ref{cond:static:linear}), the conclusion follows.

\subsection{Results in Section~\ref{sec:forecast} and Appendix~\ref{sec:2017}}

\subsubsection{Results in Section~\ref{sec:2005}}

\begin{proof}[Proof of Proposition~\ref{thm:2005}]
The following three lemmas follow from Lemmas~\ref{lem:one}--\ref{lem:three}
with $\wh{\mbf S} = \wh{\bm\Gamma}_x(0)$ and $\mbf S = \bm\Gamma_\chi(0)$,
as conditions~\eqref{cond:s:one}--\eqref{cond:s:four} are met by
Assumption~\ref{assum:static}~\ref{cond:static:linear}, 
Proposition~\ref{prop:idio:eval} and Theorem~\ref{thm:common:spec}.

\begin{lem}
\label{lem:dk:two}
There exists an orthonormal, diagonal matrix 
$\bm{\mc O} \in \R^{r \times r}$ such that
\begin{align*}
\frac{\inf_\omega \mu_{\chi, q}^2(\omega) \cdot \mu_{\chi, r}}{p^3}
\l\Vert \wh{\mbf E}_\chi - \mbf E_\chi \bm{\mc O} \r\Vert = 
O_P\l( \psi_n \vee \frac{1}{m} \vee \frac{1}{\sqrt p}\r). 
\end{align*}
\end{lem}

\begin{lem}
\label{lem:evals:two}
\begin{align*}
\frac{\inf_\omega \mu_{\chi, q}^2(\omega) \cdot \mu^2_{\chi, r}}{p^4}
\l\Vert \l(\frac{\wh{\bm{\mc M}}_\chi}{p}\r)^{-1} - 
\l(\frac{\bm{\mc M}_\chi}{p}\r)^{-1} \r\Vert
&= O_P\l(\psi_n \vee \frac{1}{m} \vee \frac{1}{\sqrt p}\r).
\end{align*}
\end{lem}


We frequently use that by Lemma~\ref{lem:decay} and
\eqref{eq:common:decay}--\eqref{eq:idio:decay}, together with Chebyshev's inequality, we have
\begin{align}
\label{eq:bounded:norm}
\Vert \mbf X_t \Vert = O_P(\sqrt p), \quad
\Vert \bm\chi_t \Vert = O_P(\sqrt p) \quad \text{and} \quad
\Vert \bm\xi_t \Vert = O_P(\sqrt p)
\end{align}
for any given $t$.
Note that
\begin{align*}
\l\vert \wh{\bm\chi}^{\static}_{n + a \vert n}  - \bm\chi_{n + a \vert n} \r\vert_\infty
&\le \max_{1 \le i \le p} \l\vert \bm\varphi_i^\top \bm\Gamma_\chi(-a) \mbf E_\chi \bm{\mc M}_\chi^{-1} \mbf E_\chi^\top \bm\xi_n \r\vert
\\
& + \max_i \l\vert \bm\varphi_i^\top \l(\wh{\bm\Gamma}_\chi(-a) - \bm\Gamma_\chi(-a) \r) 
\wh{\mbf E}_\chi \wh{\bm{\mc M}}_\chi^{-1} \wh{\mbf E}_\chi^\top \mbf X_n \r\vert
\\
& + \max_i \l\vert \bm\varphi_i^\top \bm\Gamma_\chi(-a)  
\l(\wh{\mbf E}_\chi - \mbf E_\chi \bm{\mc O} \r) 
\wh{\bm{\mc M}}_\chi^{-1} \wh{\mbf E}_\chi^\top \mbf X_n \r\vert
\\
& + \max_i \l\vert \bm\varphi_i^\top \bm\Gamma_\chi(-a)  
\mbf E_\chi \bm{\mc O} 
\l(\wh{\bm{\mc M}}_\chi^{-1} - \bm{\mc M}^{-1}_\chi \r) \wh{\mbf E}_\chi^\top \mbf X_n \r\vert
\\
& + \max_i \l\vert \bm\varphi_i^\top \bm\Gamma_\chi(-a)  
\mbf E_\chi \bm{\mc O} \bm{\mc M}^{-1}_\chi 
\l(\wh{\mbf E}_\chi - \mbf E_\chi \bm{\mc O}\r)^\top \mbf X_n \r\vert
= I + II + III + IV + V.
\end{align*}
By~\eqref{eq:common:decay}, Assumption~\ref{assum:static}~\ref{cond:static:linear}
and combining the observation
$\Vert \Cov(\mbf E_\chi^\top \bm\xi_t) \Vert \le 2\pi B_\xi$ from Proposition~\ref{prop:idio:eval}
with Chebyshev's inequality, 
we have $I = O_P(\sqrt{p}/\mu_{\chi, r})$.
Using the arguments in the proof of Proposition~\ref{prop:acv:common},
the result therein is extended to showing the consistency of $\wh{\bm\Gamma}_\chi(a)$ at a given lag $a$.
Then combined with 
Lemmas~\ref{lem:evals:two} and~\eqref{eq:bounded:norm}, it leads to
\begin{align*}
II = O_P\l(\frac{p^3}{\inf_\omega \mu_{\chi, q}^2(\omega) \cdot \mu_{\chi, r}} \cdot 
\l(\vartheta_{n, p} \vee \frac{1}{m} \vee \frac{1}{\sqrt p}\r)\r).
\end{align*}
Next, from~\eqref{eq:common:decay}, \eqref{eq:bounded:norm}
and Lemmas~\ref{lem:dk:two}--\ref{lem:evals:two},
\begin{align*}
III = O_P\l(\frac{p^4}{\inf_\omega \mu_{\chi, q}^2(\omega) \cdot \mu^2_{\chi, r}} \cdot 
\l(\psi_n \vee \frac{1}{m} \vee \frac{1}{\sqrt p}\r)\r),
\end{align*}
and $V$ is handled analogously. Finally,
from~\eqref{eq:common:decay}, \eqref{eq:bounded:norm} and Lemma~\ref{lem:evals:two},
\begin{align*}
IV = O_P\l(\frac{p^4}{\inf_\omega \mu_{\chi, q}^2(\omega) \cdot \mu^2_{\chi, r}} \cdot 
\l(\psi_n \vee \frac{1}{m} \vee \frac{1}{\sqrt p}\r)\r).
\end{align*}
Putting together the bounds on $I$--$V$, the conclusion follows.
\end{proof}

We prove the following proposition which includes Proposition~\ref{prop:common:irreduce} as a special case.
\begin{prop}
\label{prop:common:irreduce:new}
{\it Suppose that Assumptions~\ref{assum:common} and~\ref{assum:innov} hold.
Then for any $a \ge 1$,
\begin{align*}
\vert \bm\chi_{n + a \vert n} - \bm\chi_{n + a} \vert_\infty = 
\l\{\begin{array}{ll}
O_P\l(q^{1/\nu} \mu_\nu^{1/\nu} \log^{1/2}(p)\r) &
\text{under Assumption~\ref{assum:innov}~\ref{cond:dist}~\ref{cond:moment},}
\\
O_P\l(\log^{1/2}(p)\r) &
\text{under Assumption~\ref{assum:innov}~\ref{cond:dist}~\ref{cond:gauss}.}
\end{array}\r.
\end{align*}}
\end{prop}

\begin{proof}
From~\eqref{eq:gdfm:best:lin} and Assumption~\ref{assum:innov}~\ref{cond:uncor},
$\bm\chi_{n + a \vert n} - \bm\chi_{n + a} 
= \sum_{\ell = 0}^{a - 1} \mbf B_\ell \mbf u_{n + a - \ell}$ for all $a \ge 1$.
By Lemma~D.3 of \cite{zhang2021},
Assumptions~\ref{assum:common} and~\ref{assum:innov}, 
there exist constants $C_\nu, C_{\nu, \Xi, \varsigma} > 0$
that depend only on their subscripts such that
\begin{align*}
& \l\Vert \l\vert \sum_{\ell = 0}^{a - 1} \mbf B_\ell \mbf u_{n + a - \ell} \r\vert_\infty \r\Vert_\nu^2
= \l\Vert \l\vert \sum_{\ell = 0}^{a - 1} \sum_{j = 1}^q 
\mbf B_{\ell, \cdot j} u_{j, n + a - \ell} \r\vert_\infty \r\Vert_\nu^2
\le C_\nu \log(p) \sum_{\ell = 0}^{a - 1} \sum_{j = 1}^q \l\Vert \l\vert \mbf B_{\ell, \cdot j} u_{j, n + a - \ell} \r\vert_\infty \r\Vert_\nu^2
\\
& \le C_\nu \log(p) \sum_{\ell = 0}^{a - 1} \sum_{j = 1}^q \l\vert \mbf B_{\ell, \cdot j} \r\vert_\infty^2 
\l\Vert \l\vert \mbf u_{t + a - \ell} \r\vert_\infty \r\Vert_\nu^2
\le C_\nu \log(p) \sum_{\ell = 0}^{a - 1} \Xi^2 (1 + \ell)^{-2\varsigma} q^{2/\nu} \mu_\nu^{2/\nu}
\\
& \le C_{\nu, \Xi, \varsigma} \log(p) q^{2/\nu} \mu_\nu^{2/\nu}
\end{align*}
under Assumption~\ref{assum:innov}~\ref{cond:dist}~\ref{cond:moment},
such that by Chebyshev's inequality, we have 
$\l\vert \sum_{\ell = 0}^{a - 1} \mbf B_\ell \mbf u_{n + a - \ell} \r\vert_\infty =
O_P(\log^{1/2}(p) q^{1/\nu} \mu_\nu^{1/\nu})$.
When Assumption~\ref{assum:innov}~\ref{cond:dist}~\ref{cond:gauss} holds,
notice that 
\begin{align*}
\max_{1 \le i \le p} \Var\l(\sum_{\ell = 0}^{a - 1} \mbf B_{\ell, i \cdot} \mbf u_{t + a - \ell}\r)
= \max_i \sum_{\ell = 0}^{a - 1} \l\vert \mbf B_{\ell, i \cdot} \r\vert_2^2
\le \sum_{\ell = 0}^{a - 1} \Xi^2 (1 + \ell)^{-\varsigma} < \infty
\end{align*}
such that by Gaussian maximal inequality, we have
$\l\vert \sum_{\ell = 0}^{a - 1} \mbf B_\ell \mbf u_{n + a - \ell} \r\vert_\infty =
O_P(\log^{1/2}(p))$.
\end{proof}

\subsubsection{Proof of Proposition~\ref{prop:idio:pred}}

Let $\wt\zeta^{(1)}_p = \log^{1/2}(p) p^{1/\nu} \mu_\nu^{1/\nu}$ and $\wt\zeta^{(2)}_p = p^{1/\nu} \mu_\nu^{1/\nu}$ under Assumption~\ref{assum:innov}~\ref{cond:dist}~\ref{cond:moment}, $\wt\zeta^{(1)}_p = \wt\zeta^{(2)}_p = \log^{1/2}(p)$ under Assumption~\ref{assum:innov}~\ref{cond:dist}~\ref{cond:gauss}.
We first show that for any given $t$,
\begin{align*}
\vert \bm\xi_t \vert_\infty &= O_P\l( \wt\zeta^{(1)}_p \r) \quad \text{and} \quad
\vert \bm\Gamma^{1/2} \bm\vep_t \vert_\infty = O_P(\wt\zeta^{(2)}_p),
\end{align*}
with which we establish
\begin{align}
\label{eq:pf:prop:idio:pred}
\l\vert \wh{\bm\xi}_{n + 1 \vert n} - \bm\xi_{n + 1} \r\vert_\infty
&= O_P\l( s_{\text{\upshape in}} \zeta_{n, p} \wt\zeta^{(1)}_p + \Vert \bm\beta \Vert_1 \bar{\zeta}_{n, p}
+ p^{1/\nu} \wt\zeta^{(2)}_p\r).
\end{align}

Under Assumption~\ref{assum:innov}~\ref{cond:dist}~\ref{cond:moment},
by~\eqref{eq:d:gam:vep:inf} and Assumption~\ref{assum:idio}, 
we have
\begin{align*}
\l\Vert \l\vert \mbf D_\ell \bm\Gamma^{1/2} \bm\vep_{t - \ell} \r\vert_\infty \r\Vert_\nu
& \le C_\nu M_\vep^{1/2} \Xi (1 + \ell)^{-\varsigma} p^{1/\nu} \mu_\nu^{1/\nu}
\end{align*}
such that applying Lemma~D.3 of \cite{zhang2021}, we have
\begin{align*}
\Vert \vert \bm\xi_t \vert_\infty \Vert_\nu^2
&= 
\l\Vert \l\vert \sum_{\ell = 0}^\infty 
\mbf D_\ell \bm\Gamma^{1/2} \bm\vep_{t - \ell} \r\vert_\infty \r\Vert_\nu^2
\le
C_\nu^\prime \log(p) \sum_{\ell = 0}^\infty 
\l\Vert \l\vert \mbf D_\ell \bm\Gamma^{1/2} \bm\vep_{t - \ell} \r\vert_\infty \r\Vert_\nu^2
\\
&\le
C_\nu^2 C_\nu^\prime M_\vep \log(p) p^{2/\nu} \mu_\nu^{2/\nu} 
\sum_{\ell = 0}^\infty \Xi^2 (1 + \ell)^{-2\varsigma} 
\le C_{\nu, \Xi, \varsigma} M_\vep \log(p) p^{2/\nu} \mu_\nu^{2/\nu} 
\end{align*}
with $C_\nu, C^\prime_\nu, C_{\nu, \Xi, \varsigma}$ denoting some positive constants
that depend only on their subscripts.
Therefore, by Chebyshev's inequality,
$\wt\zeta^{(1)}_p = \log^{1/2}(p) p^{1/\nu} \mu_\nu^{1/\nu}$.
Similarly, by Assumption~\ref{assum:idio} and Lemma~D.3 of \cite{zhang2021},
\begin{align*}
\l\Vert (\bm\Gamma^{1/2})_{i \cdot } \bm\vep_t  \r\Vert_\nu^2
= \l\Vert \sum_{k = 1}^p (\bm\Gamma^{1/2})_{ik} \vep_{kt} \r\Vert_\nu^2
\le C_\nu \vert (\bm\Gamma^{1/2})_{i \cdot} \vert_2^2 \mu_\nu^{2/\nu}
\le C_\nu M_\vep \mu_\nu^{2/\nu}
\end{align*}
for all $1 \le i \le p$, such that
$\Vert \vert \bm\Gamma^{1/2} \bm\vep_t \vert_\infty \Vert_\nu \le (C_\nu M_\vep)^{1/2} p^{1/\nu} \mu_\nu^{1\nu}$
and with Chebyshev's inequality, 
we have $\wt\zeta^{(2)}_p = p^{1/\nu} \mu_\nu^{1/\nu}$.

Under Assumption~\ref{assum:innov}~\ref{cond:dist}~\ref{cond:gauss}, 
by Gaussian maximal inequality and that $\mu_{\xi, 1} \le 2\pi B_\xi$
(which indicates that $\max_{1 \le i \le p} \Var(\xi_{it}) \le 2\pi B_\xi$),
we have $\wt\zeta^{(1)}_p = \log^{1/2}(p)$ and similarly,
$\wt\zeta^{(2)}_p = \log^{1/2}(p)$.

Next, by construction, we have
\begin{align*}
& \l\vert \wh{\bm\xi}_{n + 1 \vert n} - \bm\xi_{n + 1 \vert n} \r\vert_\infty
\le \l\vert \sum_{\ell = 1}^d \l(\wh{\mbf A}_{\ell} - \mbf A_{\ell}\r) \wh{\bm\xi}_{n + 1 - \ell}
\r\vert_\infty +
\l\vert \sum_{\ell = 1}^d \mbf A_{\ell}
\l(\wh{\bm\xi}_{n + 1 - \ell} - \bm\xi_{n + 1 - \ell}\r) \r\vert_\infty
\\
&\le \sum_{\ell = 1}^d \l\Vert \l(\wh{\mbf A}_{\ell} - \mbf A_{\ell}\r)^\top \r\Vert_1
\l( \l\vert \wh{\bm\xi}_{n + 1 - \ell} - \bm\xi_{n + 1 - \ell} \r\vert_\infty + \l\vert \bm\xi_{n + 1 - \ell} \r\vert_\infty \r)
+ \sum_{\ell = 1}^d  \l\Vert \mbf A_{\ell}^\top \r\Vert_1
\l\vert \wh{\bm\xi}_{n + 1 - \ell} - \bm\xi_{n + 1 - \ell} \r\vert_\infty
\\
&\le \l\Vert \wh{\bm\beta} - \bm\beta \r\Vert_1 \;
\max_{1 \le \ell \le d} \l( \l\vert \wh{\bm\xi}_{n + 1 - \ell} - \bm\xi_{n + 1 - \ell} \r\vert_\infty + \l\vert \bm\xi_{n + 1 - \ell} \r\vert_\infty \r)
+ \l\Vert \bm\beta \r\Vert_1 \; 
\max_{1 \le \ell \le d} \l\vert \wh{\bm\xi}_{n + 1 - \ell} - \bm\xi_{n + 1 - \ell} \r\vert_\infty
\\
&= O_P\l( s_{\text{\upshape in}} \zeta_{n, p} \l(\bar{\zeta}_{n, p} + \wt\zeta^{(1)}_p\r) + \Vert \bm\beta \Vert_1 \bar{\zeta}_{n, p} \r),
\end{align*}
where it is used that
$\max_{1 \le \ell \le d} \vert \bm\xi_{t + 1 - \ell} \vert_\infty = O_P(d\wt\zeta^{(1)}_p)$
for given $t$.
This, combined with the observation that $\bm\Gamma^{1/2}\bm\vep_{n + 1} = \bm\xi_{n + 1 \vert n} - \bm\xi_{n + 1}$, proves~\eqref{eq:pf:prop:idio:pred}.

\subsubsection{Results in Section~\ref{sec:2017}}

\begin{proof}[Proof of Proposition~\ref{prop:idio:eval:two}]
Let $\bm\Sigma_w(\omega)$ denote the spectral density matrix of $\mbf W_t$. Then,
$\bm\Sigma_w(\omega) = \mc A_\chi(e^{-\iota \omega}) \bm\Sigma_\xi(\omega) \mc A_\chi^\top(e^{\iota \omega})$
and for any $\mbf a \in \C^p$ with $\vert \mbf a \vert_2^2 = 1$, we have
\begin{align*}
\mbf a^* \bm\Sigma_w(\omega) \mbf a \le \mu_{\xi, 1}(\omega) \; \mbf a^* \mc A_\chi(e^{-\iota \omega}) \mc A_\chi^\top(e^{\iota \omega}) \mbf a
\le B_\xi \Vert \mc A_\chi(e^{-\iota \omega}) \mc A_\chi^\top(e^{\iota \omega}) \Vert,
\end{align*}
thanks to Proposition~\ref{prop:idio:eval}.
Further, by Assumption~\ref{assum:common:var}~\ref{cond:det}
and Lemma~\ref{lem:decay}, there exists a constant $B_\chi >0$ such that
\begin{align*}
\sup_{\omega \in [-\pi, \pi]} \l\Vert \mc A_\chi(e^{-\iota \omega}) \mc A_\chi^\top(e^{\iota \omega}) \r\Vert
\le \max_{1 \le h \le N} \sup_{\omega \in [-\pi, \pi]} 
\l\Vert \mc A_\chi^{(h)}(e^{-\iota \omega}) (\mc A_\chi^{(h)}(e^{\iota \omega}))^\top \r\Vert \le B_\chi.
\end{align*}
Consequently,
\begin{align*}
\mu_{w, 1} = \sup_{\mbf a \in \R^p} \mbf a^\top \bm\Gamma_w \mbf a
= \sup_{\mbf a \in \R^p} \int_{-\pi}^\pi \mbf a^\top\bm\Sigma_w(\omega) \mbf a \, d\omega \le 2\pi B_\xi B_\chi.
\end{align*}
\end{proof}

\begin{proof}[Proof of Proposition~\ref{thm:common:var}~\ref{thm:common:var:one}]
In what follows, $C_k, \, k = 1, 2, \ldots$ denote
positive constants not dependent on the indices $1 \le i, i' \le p$.
By Assumption~\ref{assum:common:var}~\ref{cond:block:var},
\begin{align}
& \frac{1}{p} \l\Vert \wh{\mbf A}_\chi - \mbf A_\chi \r\Vert_F^2
= \frac{1}{p} \sum_{h = 1}^N \l\Vert \wh{\mbf A}^{(h)}_\chi - \mbf A^{(h)}_\chi \r\Vert_F^2, \quad \text{where}
\nn \\
& \l\Vert \wh{\mbf A}^{(h)}_\chi - \mbf A^{(h)}_\chi \r\Vert_F
=
\l\Vert \wh{\mbf B}^{(h)}_\chi (\wh{\mbf C}^{(h)}_\chi)^{-1}
- \mbf B^{(h)}_\chi (\mbf C^{(h)}_\chi)^{-1} \r\Vert_F
\nn \\
\le &
\l\Vert \l(\wh{\mbf B}^{(h)}_\chi - \mbf B^{(h)}_\chi \r) (\wh{\mbf C}^{(h)}_\chi)^{-1} \r\Vert_F
+ \l\Vert \mbf B^{(h)}_\chi 
\l((\wh{\mbf C}^{(h)}_\chi)^{-1} - (\mbf C^{(h)}_\chi)^{-1} \r) \r\Vert_F
\nn \\
\le & 
\l\Vert \wh{\mbf B}^{(h)}_\chi - \mbf B^{(h)}_\chi \r\Vert_F \;
\l\Vert (\wh{\mbf C}^{(h)}_\chi)^{-1} \r\Vert_F
+ \l\Vert \mbf B^{(h)}_\chi \r\Vert_F \;
\l\Vert (\wh{\mbf C}^{(h)}_\chi)^{-1} - (\mbf C^{(h)}_\chi)^{-1} \r\Vert_F.
\label{thm:common:var:one:eq:one}
\end{align}
By~\eqref{eq:common:decay},
Assumption~\ref{assum:common:var}~\ref{cond:det}
and that the entries of $(\mbf C^{(h)}_\chi)^{-1}$
are rational functions of those of $\mbf C^{(h)}_\chi$,
we have
\begin{align}
\label{eq:bound:BC}
\max_{1 \le h \le N} 
\max\l( \Vert \mbf B^{(h)}_\chi \Vert_F,
\Vert \mbf C^{(h)}_\chi \Vert_F \r) \le C_1 \quad \text{and} \quad
\max_{1 \le h \le N} \Vert (\mbf C^{(h)}_\chi)^{-1} \Vert_F \le C_2
\end{align}
Also, by Proposition~\ref{prop:acv:common}~\ref{prop:acv:common:two},
\begin{align}
\max_{1 \le h \le N} \l\Vert \wh{\mbf B}^{(h)}_\chi - \mbf B^{(h)}_\chi \r\Vert_F^2
\le (q + 1)^2 \sum_{\ell = 1}^s \max_{1 \le i, i' \le p} 
\vert \wh{\gamma}_{\chi, ii'}(\ell) - \gamma_{\chi, ii'}(\ell) \vert^2
\nn \\
= O_P\l(\frac{s q^3 p^4}{\inf_\omega \mu_{\chi, q}^4(\omega)} 
\l(\vartheta_{n, p}^2 \vee \frac{1}{m^2} \vee \frac{1}{p}\r)\r)
\label{eq:bound:Bdiff:max}
\end{align}
and similarly,
\begin{align}
\max_{1 \le h \le N} \l\Vert \wh{\mbf C}^{(h)}_\chi - \mbf C^{(h)}_\chi \r\Vert_F^2
= O_P\l(\frac{s^2 q^3 p^4}{\inf_\omega \mu_{\chi, q}^4(\omega)} 
\l(\vartheta_{n, p}^2 \vee \frac{1}{m^2} \vee \frac{1}{p}\r)\r).
\label{eq:bound:Cdiff:max}
\end{align}
Also, \eqref{eq:bound:BC} and \eqref{eq:bound:Cdiff:max} imply that
\begin{align*}
\max_{1 \le h \le N} \l\Vert \l( \wt{\mbf C}^{(h)}_\chi \r)^{-1} \r\Vert &\le \max_{1 \le h \le N} \l\Vert \l( \mbf C^{(h)}_\chi \r)^{-1} \r\Vert_F + \max_{1 \le h \le N} \l\Vert \wh{\mbf C}^{(h)}_\chi - \mbf C^{(h)}_\chi \r\Vert_F = O_P(1)
\end{align*}
by Weyl's inequality, and
\begin{align*}
\max_{1 \le h \le N} \l\Vert (\wh{\mbf C}^{(h)}_\chi)^{-1} - (\mbf C^{(h)}_\chi)^{-1} \r\Vert_F &\le \max_{1 \le h \le N} \l\Vert \l( \wt{\mbf C}^{(h)}_\chi \r)^{-1} \r\Vert_F \l\Vert \l( \mbf C^{(h)}_\chi \r)^{-1} \r\Vert_F \l\Vert \wh{\mbf C}^{(h)}_\chi - \mbf C^{(h)}_\chi \r\Vert_F
\\
&= O_P\l(\frac{p^2}{\inf_\omega \mu_{\chi, q}^2(\omega)} 
\l(\vartheta_{n, p} \vee \frac{1}{m} \vee \frac{1}{\sqrt p}\r)\r).
\end{align*}
This, together with~\eqref{thm:common:var:one:eq:one}, \eqref{eq:bound:Bdiff:max}, \eqref{eq:bound:Cdiff:max}, and that $p = (q + 1)N$, proves the first claim.
It also follows that
\begin{align}
& \max_{1 \le h \le N} \l\Vert \wh{\mbf A}^{(h)}_\chi - \mbf A^{(h)}_\chi \r\Vert_F
\nn \\
\le & \,
\max_{1 \le h \le N} \l\Vert \wh{\mbf B}^{(h)}_\chi - \mbf B^{(h)}_\chi \r\Vert_F \;
\l\Vert (\wh{\mbf C}^{(h)}_\chi)^{-1} \r\Vert_F
+ \max_{1 \le h \le N} \l\Vert \mbf B^{(h)}_\chi \r\Vert_F \;
\l\Vert (\wh{\mbf C}^{(h)}_\chi)^{-1} - (\mbf C^{(h)}_\chi)^{-1} \r\Vert_F
\nn \\
=& \, O_P\l(\vartheta_{n, p} \vee \frac{1}{m} \vee \frac{1}{\sqrt p} \r).
\label{eq:var:block:norm}
\end{align}
Then, the second claim follows since
\begin{align*}
\max_{1 \le i \le p} \l\vert \bm\varphi_i^\top 
\l( \wh{\mbf A}_{\chi} - \mbf A_{\chi} \r) \r\vert_2
\le \max_{1 \le h \le N} \l\Vert \wh{\mbf A}^{(h)}_\chi - \mbf A^{(h)}_\chi \r\Vert_F.
\end{align*}
\end{proof}

Let $\wh{\bm\Gamma}_{\mathbb X} = (n - s)^{-1} \mathbb X \mathbb X^\top$
and $\bm\Gamma_{\mathbb X} = \E(\wh{\bm\Gamma}_{\mathbb X})$, where
\begin{align*}
\mathbb X = \bmx 
\mbf X_{s + 1} & \mbf X_{s + 2} & \ldots & \mbf X_n \\
& & \ddots & \\
\mbf X_{1} & \mbf X_{2} & \ldots & \mbf X_{n - s} 
\emx \in \R^{p(s + 1) \times (n - s)}.
\end{align*}
Then, $\bm\Gamma_{\mathbb X} = [\bm\Gamma_x(\ell - \ell'), \, 1 \le \ell, \ell' \le s + 1]$.
Further, we write that $\wh{\bm\Gamma}_{\mathbb X} = [\wh{\bm\Gamma}_x(\ell, \ell'), \, 1 \le \ell, \ell' \le s + 1]$
where $\wh{\bm\Gamma}_x(\ell, \ell') = [\wh\gamma_{x, ii'}(\ell, \ell'), \, 1 \le i, i' \le p]
= (n - s)^{-1} \sum_{t = s + 2 - \ell}^{n - \ell + 1} \mbf X_t \mbf X_{t - \ell' + \ell}^\top$.

\begin{prop}
\label{prop:acv:z}
Under the assumptions made in  Proposition~\ref{thm:common:var}~\ref{thm:common:var:two},
\begin{align*}
\frac{1}{p} \l\Vert \wh{\bm\Gamma}_z - \bm\Gamma_z \r\Vert
= O_P\l(\vartheta_{n, p} \vee \frac{1}{m} \vee \frac{1}{\sqrt p}\r), \quad
\l\vert \wh{\bm\Gamma}_z - \bm\Gamma_z \r\vert_\infty
= O_P\l(\vartheta_{n, p} \vee \frac{1}{m} \vee \frac{1}{\sqrt p}\r).
\end{align*}
\end{prop}

\begin{proof}
Recall that by Assumption~\ref{assum:common:var}~\ref{cond:block:var}, $s$ is finite.
By definition, $\bm\Gamma_z = \mathbb A_\chi \bm\Gamma_{\mathbb X} \mathbb A_\chi^\top$
where $\mathbb A_\chi = [\mbf I, - \mathbf A_\chi] \in \R^{p \times p(s + 1)}$
and $\wh{\bm\Gamma}_z = \wh{\mathbb A}_\chi \wh{\bm\Gamma}_{\mathbb X} \wh{\mathbb A}_\chi^\top$
with $\wh{\mathbb A}_\chi = [\mbf I, - \wh{\mbf A}_\chi]$.
Then,
\begin{align*}
\frac{1}{p} \l\Vert \wh{\bm\Gamma}_z - \bm\Gamma_z \r\Vert
&\le \frac{1}{p} \l\Vert (\wh{\mathbb A}_\chi - \mathbb A_\chi)\wh{\bm\Gamma}_{\mathbb X} \wh{\mathbb A}_\chi^\top \r\Vert
+ \frac{1}{p} \l\Vert \mathbb A_\chi (\wh{\bm\Gamma}_{\mathbb X} - \bm\Gamma_{\mathbb X}) \wh{\mathbb A}_\chi^\top \r\Vert
+ \frac{1}{p} \l\Vert \mathbb A_\chi \bm\Gamma_{\mathbb X} (\wh{\mathbb A}_\chi - \mathbb A_\chi)^\top \r\Vert
\\
&= I + II + III. 
\end{align*}
From the block structure of $\mbf A_\chi$ and
the proof of Proposition~\ref{thm:common:var}~\ref{thm:common:var:one},
\begin{align}
\l\Vert \wh{\mathbb A}_\chi - \mathbb A_\chi \r\Vert 
= \l\Vert \wh{\mbf A}_\chi - \mbf A_\chi \r\Vert
\le s \cdot \max_{1 \le h \le N} \l\Vert \wh{\mbf A}^{(h)}_\chi - \mbf A^{(h)}_\chi \r\Vert_F
= O_P\l(\vartheta_{n, p} \vee \frac{1}{m} \vee \frac{1}{\sqrt p}\r).
\label{prop:acv:z:eq:one}
\end{align}
With $C_2$ and $C_3$ given in~\eqref{eq:bound:BC}, we have
\begin{align}
\l\Vert \mathbb A_\chi \r\Vert \le 1 + s \cdot \max_{1 \le h \le N} \l\Vert \mbf B^{(h)}_\chi \r\Vert_F \;
\l\Vert (\mbf C^{(h)}_\chi)^{-1} \r\Vert_F \le 1 + sC_2C_3 < \infty.
\label{prop:acv:z:eq:two}
\end{align}
By Proposition~2.3 of \cite{basu2015}, Assumption~\ref{assum:factor} and Proposition~\ref{prop:idio:eval},
\begin{align}
\frac{1}{p} \l\Vert \bm\Gamma_{\mathbb X} \r\Vert \le 2\pi \sup_\omega \frac{\mu_{x, 1}(\omega)}{p}
\le 2\pi \sup_\omega \l(\beta_{\chi, 1}(\omega) + \frac{B_\xi}{p}\r) < \infty.
\label{prop:acv:z:eq:three}
\end{align}
Also by the arguments analogous to those adopted in the proof of Lemma~\ref{lem:acv:x}~\ref{lem:acv:x:one},
there exists a fixed constant $C > 0$ such that
\begin{align}
\label{eq:lem:acv:x}
\frac{1}{p^2} \E\l(\l\Vert \wh{\bm\Gamma}_{\mathbb X} - \bm\Gamma_{\mathbb X} \r\Vert_F^2\r)
= \frac{1}{p^2} \sum_{\ell, \ell' = 1}^{s + 1} \sum_{i, i' = 1}^p
\E\l(\l\vert \wh\gamma_{x, ii'}(\ell, \ell') - \gamma_{x, ii'}(\ell - \ell')) \r\vert^2\r) 
\le C \frac{(s + 1)^2}{n}.
\end{align}
Then by~\eqref{prop:acv:z:eq:one}--\eqref{eq:lem:acv:x},
$I = O_P(\vartheta_{n, p} \vee m^{-1} \vee p^{-1/2})$,
by~\eqref{prop:acv:z:eq:one}--\eqref{prop:acv:z:eq:two} and~\eqref{eq:lem:acv:x},
$II = O_P(n^{-1/2})$ and $III$ is bounded analogously as $I$, such that
$p^{-1} \Vert \wh{\bm\Gamma}_z - \bm\Gamma_z \Vert = O_P(\vartheta_{n, p} \vee m^{-1} \vee p^{-1/2})$.

For the second claim, we proceed similarly by noting that
\begin{align*}
\l\vert \wh{\bm\Gamma}_z - \bm\Gamma_z \r\vert_\infty
&\le \max_{i, i'} \l\Vert 
\bm\varphi_i^\top \l(\wh{\mathbb A}_{\chi} - \mathbb A_{\chi} \r)
\wh{\bm\Gamma}_{\mathbb X} \wh{\mathbb A}_\chi^\top \bm\varphi_{i'} \r\Vert
+ \max_{i, i'} \l\Vert 
\bm\varphi_i^\top \mathbb A_\chi 
\l( \wh{\bm\Gamma}_{\mathbb X} - \bm\Gamma_{\mathbb X} \r) 
\wh{\mathbb A}_\chi^\top\bm\varphi_{i'} \r\Vert\\
\\
& + \max_{i, i'} \l\Vert 
\bm\varphi_i^\top \mathbb A_\chi \wh{\bm\Gamma}_{\mathbb X} 
\l(\wh{\mathbb A}_{\chi} - \mathbb A_{\chi} \r)^\top \bm\varphi_{i'} \r\Vert
= IV + V + VI.
\end{align*}
From the block structure of $\mbf A_\chi$, 
the second claim of Proposition~\ref{thm:common:var}~\ref{thm:common:var:one},
Lemmas~\ref{lem:decay} and~\ref{lem:acv:x}~\ref{lem:acv:x:two}
and~\eqref{prop:acv:z:eq:one}--\eqref{prop:acv:z:eq:two},
$IV = O_P(\vartheta_{n, p} \vee m^{-1} \vee p^{-1/2})$,
and $V$ and $VI$ are bounded analogously.
\end{proof}

\begin{proof}[Proof of Proposition~\ref{thm:common:var}~\ref{thm:common:var:two}]
Denote the eigendecomposition of $\bm\Gamma_v$ by 
$\bm\Gamma_v = \mbf E_v\bm{\mc M}_v \mbf E_v^\top$.
With $\wh{\bm\Gamma}_z$ as $\wh{\mbf S}$ and $\bm\Gamma_v$ as $\mbf S$,
\eqref{cond:s:one}--\eqref{cond:s:four} are met by
Assumption~\ref{assum:static:two}, Proposition~\ref{prop:idio:eval:two}
and Proposition~\ref{prop:acv:z}
with $\varrho_1 = \ldots = \varrho_q = 1$ and
$\omega_{n, p}^{(1)} = \omega_{n, p}^{(2)} = (\vartheta_{n, p} \vee m^{-1} \vee p^{-1/2})$.
Therefore, applying Lemmas~\ref{lem:one} and~\ref{lem:three}, 
there exists a diagonal, orthogonal matrix $\bm{\mc O} \in \R^{p \times p}$ such that
\begin{align*}
\l\Vert \wh{\mbf E}_{z, 1:q} - \mbf E_v \bm{\mc O} \r\Vert
&= O_P\l( \vartheta_{n, p} \vee \frac{1}{m} \vee \frac{1}{\sqrt p}\r),
\\
\sqrt{p} 
\max_{1 \le i \le p}
\l\vert \bm\varphi_i^\top \l(\wh{\mbf E}_{z, 1:q} - \mbf E_v \bm{\mc O} \r) \r\vert_2
&= O_P\l( \vartheta_{n, p} \vee \frac{1}{m} \vee \frac{1}{\sqrt p}\r).
\end{align*}
Also, by Assumptions~\ref{cond:ration} and~\ref{assum:common:var}~\ref{cond:block:var}, we have
\begin{align}
\label{eq:v:norm:bound}
\max_i \, \Var(V_{it}) = \max_i \sum_{j = 1}^q \mu_{v, j} \vert e_{v, ij} \vert^2
= \max_i \sum_{j = 1}^q R_{ij}^2 \Var(u_{ij}) \le q (B_\chi^{(1)})^2 < \infty
\end{align}
where $\mbf e_{v, j} = (e_{v, ij}, \, 1 \le i \le p)^\top$ and $\mbf R = [R_{ij}, \, 1 \le i \le p, \, 1 \le j \le q]$,
such that $\max_{1 \le i \le p} \vert e_{v, ij} \vert = O(p^{-1/2})$
by Assumption~\ref{assum:static:two}.
Besides, by~\eqref{eq:var:block:norm},
Lemma~\ref{lem:decay} and Chebyshev's inequality,
\begin{align*}
\frac{1}{\sqrt p} \l\Vert \wh{\mbf Z}_t - \mbf Z_t \r\Vert
&= \frac{1}{\sqrt p} \l\Vert \sum_{\ell = 1}^s (\wh{\mbf A}_{\chi, \ell} - \mbf A_{\chi, \ell}) \mbf X_{t - \ell} \r\Vert
\\
&\le \frac{1}{\sqrt p} \sum_{\ell = 1}^s \l\Vert [ (\wh{\mbf A}^{(h)}_{\chi, \ell} - \mbf A^{(h)}_{\chi, \ell}) \mbf X^{(h)}_{t - \ell}, \, 1 \le h \le N ] \r\Vert
= O_P\l(\vartheta_{n, p} \vee \frac{1}{m} \vee \frac{1}{\sqrt p}\r)
\end{align*}
for given $t$, where $\mbf X^{(h)}_t$ is defined analogously as $\bm\chi^{(h)}_t$.
Also, from the finiteness of $\Vert \mathbb{A}_\chi \Vert$ in~\eqref{prop:acv:z:eq:two}
and~\eqref{eq:bounded:norm}, we have $\Vert \mbf Z_t \Vert = O_P(\sqrt p)$
for a given $t$. 
Finally, $\Vert \Cov(\mbf E_v^\top \mbf W_t) \Vert \le B_w$ by Proposition~\ref{prop:idio:eval:two}
such that $\Vert \mbf E_v^\top \mbf W_t \Vert = O_P(1)$ by Chebyshev's inequality.
Putting the above observations together,
\begin{align*}
\l\vert \wh{\mbf R}\wh{\mbf u}_t - \mbf R\mbf u_t \r\vert_\infty
\le& \max_i \l\vert \bm\varphi_i^\top
\l(\wh{\mbf E}_{z, 1:q} \wh{\mbf E}_{z, 1:q}^\top \wh{\mbf Z}_t 
- \mbf E_v \mbf E_v^\top \mbf Z_t\r) \r\vert
+ \max_i \l\vert \bm\varphi_i^\top \mbf E_v \mbf E_v^\top \mbf W_t \r\vert
\\
\le & 
\max_i \l\vert \bm\varphi_i^\top
\l(\wh{\mbf E}_{z, 1:q} - \mbf E_v \bm{\mc O} \r) \wh{\mbf E}_{z, 1:q}^\top \wh{\mbf Z}_t \r\vert
+
\max_i \l\vert \bm\varphi_i^\top \mbf E_v \bm{\mc O}
\l( \wh{\mbf E}_{z, 1:q} - \mbf E_v \bm{\mc O} \r)^\top \wh{\mbf Z}_t \r\vert
\\
& +
\max_i \l\vert \bm\varphi_i^\top \mbf E_v \mbf E_v^\top
(\wh{\mbf Z}_t - \mbf Z_t) \r\vert + O_P\l(\frac{1}{\sqrt p}\r)
= O_P\l(\vartheta_{n, p} \vee \frac{1}{m} \vee \frac{1}{\sqrt p}\r).
\end{align*}
\end{proof}

\begin{proof}[Proof of Proposition~\ref{thm:common:var}~\ref{thm:common:var:three}]
We can write each block in~\eqref{eq:gdfm:var} as
\begin{align*}
\bmx \bm\chi^{(h)}_t \\ \bm\chi^{(h)}_{t - 1} \\ \vdots \\ \bm\chi^{(h)}_{t - s + 1} \emx
&= \underbrace{\bmx 
\mbf A^{(h)}_{\chi, 1} & \mbf A^{(h)}_{\chi, 2} & \ldots & \mbf A^{(h)}_{\chi, s - 1} & \mbf A^{(h)}_{\chi, s} \\
\mbf I & \mbf O & & \mbf O & \mbf O \\
 && \ddots && \\
\mbf O & \mbf O & \ldots & \mbf I & \mbf O
\emx}_{\mathfrak{A}^{(h)}_\chi} \;
\bmx \bm\chi^{(h)}_{t - 1} \\ \bm\chi^{(h)}_{t - 2} \\ \vdots \\ \bm\chi^{(h)}_{t - s} \emx
+ \underbrace{\bmx \mbf R^{(h)} \mbf u_t \\ \mbf 0 \\ \vdots \\ \mbf 0 \emx}_{\mathfrak{R}^{(h)} \mbf u_t}, \text{ such that}
\\
\bm\chi^{(h)}_t 
&= \mbf J (\mbf I - \mathfrak{A}^{(h)}_\chi L)^{-1} 
\mathfrak{R}^{(h)} \mbf u_t
= \mbf J \sum_{\ell = 0}^\infty (\mathfrak{A}^{(h)}_\chi)^\ell 
\mathfrak{R}^{(h)} \mbf u_{t - \ell}
\end{align*}
with $\mbf J = [\mbf I_{q + 1}, \mbf O, \ldots, \mbf O] \in \R^{(q + 1) \times (q + 1)s}$.
Then, defining $\mbf B^{(h)}_{\ell}$ that consists of
the $(q + 1)$ rows of $\mbf B_{\ell}$ to satisfy
$\bm\chi^{(h)}_t = \sum_{\ell = 0}^\infty \mbf B^{(h)}_{\ell}\mbf u_{t - \ell}$
under~\eqref{eq:gdfm},
we have $\mbf B^{(h)}_{\ell} = \mbf J (\mathfrak{A}^{(h)}_\chi)^\ell \mathfrak{R}^{(h)}$.
Letting $\wh{\mathfrak A}^{(h)}_\chi$ and $\wh{\mathfrak{R}}^{(h)}_\chi$
denote the estimated counterparts of
$\mathfrak A^{(h)}_\chi$ and $\mathfrak{R}^{(h)}_\chi$, respectively,
we have
\begin{align*}
\l\vert \wh{\bm\chi}^{\va}_{t + a \vert t} - \bm\chi_{t + a \vert t} \r\vert_\infty
\le &
\max_{1 \le h \le N} \l\vert \mbf J \sum_{\ell = 0}^K 
\l[(\wh{\mathfrak A}^{(h)}_\chi)^\ell \wh{\mathfrak{R}}^{(h)}_\chi \wh{\mbf u}_{t - \ell}
- (\mathfrak A^{(h)}_\chi)^\ell \mathfrak{R}^{(h)}_\chi \mbf u_{t - \ell} \r] \r\vert_\infty
\\
& + \l\vert \sum_{\ell = K + 1}^\infty \mbf B_{\ell + a} \mbf u_{t - \ell} \r\vert_\infty
= I + II.
\end{align*}
From~\eqref{eq:var:block:norm} and~\eqref{prop:acv:z:eq:two},
we obtain
\begin{align*}
& \max_h \l\Vert (\wh{\mathfrak A}^{(h)}_\chi)^\ell -  (\mathfrak A^{(h)}_\chi)^\ell \r\Vert
\le 
\max_h \l\Vert \wh{\mathfrak A}^{(h)}_\chi \r\Vert \;
\l\Vert (\wh{\mathfrak A}^{(h)}_\chi)^{\ell - 1} -  (\mathfrak A^{(h)}_\chi)^{\ell - 1} \r\Vert
\\
& \qquad +
\max_h \l\Vert ({\mathfrak A}^{(h)}_\chi)^{\ell - 1} \r\Vert \;
\l\Vert \wh{\mathfrak A}^{(h)}_\chi - \mathfrak A^{(h)}_\chi \r\Vert
= O_P\l( \max_h \Vert \mathfrak{A}^{(h)}_\chi \Vert^{\ell - 1}
\l(\vartheta_{n, p} \vee \frac{1}{m} \vee \frac{1}{\sqrt p}\r) \r)
\end{align*}
by induction for any finite $\ell$, while
\begin{align*}
\max_h \l\vert 
\wh{\mathfrak{R}}^{(h)}_\chi \wh{\mbf u}_{t - \ell} - \mathfrak{R}^{(h)}_\chi \mbf u_{t - \ell}
\r\vert_\infty = O_P\l(\vartheta_{n, p} \vee \frac{1}{m} \vee \frac{1}{\sqrt p}\r)
\end{align*}
by Proposition~\ref{thm:common:var}~\ref{thm:common:var:two}. 
From these observations and~\eqref{eq:v:norm:bound} (in combination with Chebyshev's inequality), 
we derive
\begin{align*}
I \le & \sum_{\ell = 0}^K 
\l\{
\max_h \l\vert \mbf J \l[(\wh{\mathfrak A}^{(h)}_\chi)^\ell - (\mathfrak A^{(h)}_\chi)^\ell \r] \;
\wh{\mathfrak{R}}^{(h)}_\chi \wh{\mbf u}_{t - \ell} \r\vert_\infty
+
\max_h \l\vert \mbf J (\mathfrak A^{(h)}_\chi)^\ell
\l(\wh{\mathfrak{R}}^{(h)}_\chi \wh{\mbf u}_{t - \ell} -
\mathfrak{R}^{(h)}_\chi \mbf u_{t - \ell} \r) \r\vert_\infty
\r\}
\\
& = O_P\l( K \max_h \Vert \mathfrak{A}^{(h)}_\chi \Vert^K
\l(\vartheta_{n, p} \vee \frac{1}{m} \vee \frac{1}{\sqrt p}\r) \r).
\end{align*}
Also, from Assumption~\ref{assum:common}, there exist constants $C_\nu, C_{\nu, \Xi, \varsigma} > 0$
that depend only on their subscripts such that
\begin{align*}
& \l\Vert \l\vert \sum_{\ell = K + 1}^\infty \mbf B_{\ell + a} \mbf u_{t - \ell} \r\vert_\infty \r\Vert_\nu^2
\le C_\nu \log(p) \sum_{\ell  = K + 1}^\infty \sum_{j = 1}^q 
\l\Vert \l\vert \mbf B_{\ell + a, \cdot j} u_{j, t - \ell} \r\vert_\infty \r\Vert_\nu^2
\\
& \le C_\nu \log(p) \sum_{\ell  = K + 1}^\infty \sum_{j = 1}^q 
\l\vert \mbf B_{\ell + a, \cdot j} \r\vert_\infty^2 \l\Vert \l\vert \mbf u_{t - \ell} \r\vert_\infty \r\Vert_\nu^2
\le C_\nu \log(p) q^{2/\nu} \mu_\nu^{2/\nu} \sum_{\ell = K + 1}^\infty \Xi^2 (1 + \ell + a)^{-2\varsigma} 
\\
& \le C_{\nu, \Xi, \varsigma} \log(p) q^{2/\nu} \mu_\nu^{2/\nu} (K + a)^{-2(\varsigma - 1)},
\end{align*}
by Lemma~D.3 of \cite{zhang2021}.
Therefore by Chebyshev's inequality, we have
\begin{align*}
II = O_P\l(\log^{1/2}(p) q^{1/\nu} \mu_\nu^{1/\nu} (K + a)^{-\varsigma + 1} \r)
\end{align*}
under Assumption~\ref{assum:innov}~\ref{cond:dist}~\ref{cond:moment}.
When Assumption~\ref{assum:innov}~\ref{cond:dist}~\ref{cond:gauss} is met,
\begin{align*}
\max_{1 \le i \le p} \Var\l(\sum_{\ell = K + 1}^\infty \mbf B_{\ell + a, i \cdot} \mbf u_{t - \ell} \r)
= \max_i \sum_{\ell = K + 1}^\infty \l\vert \mbf B_{\ell + a, i \cdot} \r\vert_2^2
\\
\le \sum_{\ell = K + 1}^\infty \Xi^2 (1 + \ell + a)^{-2\varsigma} 
\le C_{\Xi, \varsigma} (K + a)^{-2(\varsigma - 1)}
\end{align*}
such that $II = O_P(\log^{1/2}(p) K^{-\varsigma + 1})$.
Combining the bounds on $I$ and $II$, the conclusion follows. 
\end{proof}

\clearpage

\section{Information on the real dataset}
\label{app:data}

Table~\ref{table:data:info} provides the list of the $46$ companies included in 
the application presented in Section~\ref{sec:real}
along with their tickers and
industry and sub-industry classifications according to Global Industry Classification Standard.

\begin{table}[htb]
\caption{\small Tickers, industry and sub-industry classifications of the $46$ companies.}
\label{table:data:info}
\centering
\resizebox{\columnwidth}{!}{\begin{tabular}{c c c c}
\toprule 
Name &	Ticker &	Industry &	Sub-industry	\\	
\cmidrule(lr){1-1} \cmidrule(lr){2-2} \cmidrule(lr){3-3} \cmidrule(lr){4-4}
JPMORGAN CHASE \& CO &	JPM &	Banks &	Diversified banks	\\	
COMERICA INC &	CMA &	Banks &	Regional banks	\\	
CITIGROUP INC &	C &	Banks &	Diversified banks	\\	
FIFTH THIRD BANCORP &	FITB &	Banks &	Regional banks	\\	
REGIONS FINANCIAL CORP &	RF &	Banks &	Regional banks	\\	
M \& T BANK CORP &	MTB &	Banks &	Regional banks	\\	
U S BANCORP &	USB &	Banks &	Diversified banks	\\	
HUNTINGTON BANCSHARES &	HBAN &	Banks &	Regional banks	\\	
BANK OF AMERICA CORP &	BAC &	Banks &	Diversified banks	\\	
WELLS FARGO \& CO &	WFC &	Banks &	Diversified banks	\\	
PNC FINANCIAL SVCS GROUP INC &	PNC &	Banks &	Regional banks	\\	
KEYCORP &	KEY &	Banks &	Regional banks	\\	
ZIONS BANCORPORATION NA &	ZION &	Banks &	Regional banks	\\	
TRUIST FINANCIAL CORP &	TFC &	Banks &	Regional banks	\\	
PEOPLE'S UNITED FINL INC &	PBCT &	Banks &	Regional banks	\\	
SVB FINANCIAL GROUP &	SIVB &	Banks &	Regional banks	\\	
\cmidrule(lr){1-1} \cmidrule(lr){2-2} \cmidrule(lr){3-3} \cmidrule(lr){4-4}
AMERICAN EXPRESS CO &	AXP &	 Diversified Financials &	Consumer finance	\\	
BANK OF NEW YORK MELLON CORP &	BK &	 Diversified Financials &	Asset Management \& Custody Banks	\\	
FRANKLIN RESOURCES INC &	BEN &	 Diversified Financials &	Asset Management \& Custody Banks	\\	
S\&P GLOBAL INC &	SPGI &	 Diversified Financials &	Financial Exchanges \& Data	\\	
NORTHERN TRUST CORP &	NTRS &	 Diversified Financials &	Asset Management \& Custody Banks	\\	
RAYMOND JAMES FINANCIAL CORP &	RJF &	 Diversified Financials &	Investment Banking \& Brokerage	\\	
STATE STREET CORP &	STT &	 Diversified Financials &	Asset Management \& Custody Banks	\\	
MORGAN STANLEY &	MS &	 Diversified Financials &	Investment Banking \& Brokerage	\\	
PRICE (T. ROWE) GROUP &	TROW &	 Diversified Financials &	Asset Management \& Custody Banks	\\	
SCHWAB (CHARLES) CORP &	SCHW &	 Diversified Financials &	Investment Banking \& Brokerage	\\	
INVESCO LTD &	IVZ &	 Diversified Financials &	Asset Management \& Custody Banks	\\	
CAPITAL ONE FINANCIAL CORP &	COF &	 Diversified Financials &	Consumer finance	\\	
GOLDMAN SACHS GROUP INC &	GS &	 Diversified Financials &	Investment Banking \& Brokerage	\\	
BLACKROCK INC &	BLK &	 Diversified Financials &	Asset Management \& Custody Banks	\\
\cmidrule(lr){1-1} \cmidrule(lr){2-2} \cmidrule(lr){3-3} \cmidrule(lr){4-4}	
AFLAC INC &	AFL &	Insurance &	Life \& Health Insurance	\\	
AMERICAN INTERNATIONAL GROUP &	AIG &	Insurance &	Multi-line Insurance	\\	
AON PLC &	AON &	Insurance &	Insurance Brokers	\\	
ARTHUR J GALLAGHER \& CO &	AJG &	Insurance &	Insurance Brokers	\\	
LINCOLN NATIONAL CORP &	LNC &	Insurance &	Life \& Health Insurance	\\	
LOEWS CORP &	L &	Insurance &	Property \& Casualty Insurance	\\	
MARSH \& MCLENNAN COS &	MMC &	Insurance &	Insurance Brokers	\\	
GLOBE LIFE INC &	GL &	Insurance &	Life \& Health Insurance	\\	
UNUM GROUP &	UNM &	Insurance &	Life \& Health Insurance	\\	
PROGRESSIVE CORP-OHIO &	PGR &	Insurance &	Property \& Casualty Insurance	\\	
BERKLEY (W R) CORP &	WRB &	Insurance &	Property \& Casualty Insurance	\\	
CINCINNATI FINANCIAL CORP &	CINF &	Insurance &	Property \& Casualty Insurance	\\	
CHUBB LTD &	CB &	Insurance &	Property \& Casualty Insurance	\\	
ALLSTATE CORP &	ALL &	Insurance &	Property \& Casualty Insurance	\\	
EVEREST RE GROUP LTD &	RE &	Insurance &	Reinsurance	\\	
HARTFORD FINANCIAL SERVICES &	HIG &	Insurance &	Multi-line Insurance	\\	\bottomrule
\end{tabular}}
\end{table}

\end{document}